\documentclass[10pt]{article}
\usepackage[a4paper,nohead]{geometry}
\usepackage[utf8x]{inputenc}
\usepackage[T1]{fontenc}
\usepackage{graphicx}
\usepackage{amsmath}
\usepackage[sort&compress,numbers]{natbib}
\usepackage{amsfonts}
\usepackage{amssymb}
\usepackage{amsmath}
\usepackage{latexsym}
\usepackage{color}
\usepackage{ntheorem}
\usepackage{braket}
\usepackage[colorlinks=true]{hyperref}
\usepackage{authblk}
\usepackage{tensor}

\newcounter{mnotecount}[section]

\renewcommand{\themnotecount}{\thesection.\arabic{mnotecount}}
\newcommand{\mnote}[1]%{}
{\protect{\stepcounter{mnotecount}}$^{\mbox{\footnotesize
$%\!\!\!\!\!\!\,
\bullet$\themnotecount}}$ \marginpar{%\color{red}%
\raggedright\tiny\em
$\!\!\!\!\!\!\,\bullet$\themnotecount: #1} }

\newtheorem{theorem}{\sc  Theorem\rm}[section]

\newtheorem{corollary}[theorem]{\sc  Corollary\rm}

\newtheorem{definition}[theorem]{\sc  Definition\rm}

\newtheorem{lemma}[theorem]{\sc Lemma\rm}

\newtheorem{proposition}[theorem]{\sc Proposition\rm}

\newtheorem{remark}[theorem]{\sc Remark\rm}

\newcommand{\ol}[1]{\overline{#1}{}}
\newcommand{\ul}[1]{\underline{#1}{}}

\newcommand{\jlcax}[1]{}
\newcommand{\eean}{\nonumber\end{eqnarray}}

\newcommand{\kk}[1]{}%{\mnote{{\bf If we consider the KK case:} #1}}

\newcommand{\beq}{\begin{equation}}

\newcommand{\FS}       %{F_1} %
                  {F}
\newcommand{\HS} %{F_2}
       {H_{\mbox{\scriptsize volume}}}

{\ptc{this should be removed in the oberwolfach version}}%

\newcommand{\eeal}[1]{\label{#1}\end{eqnarray}}
\newcommand{\bed}{\begin{deqarr}}
\newcommand{\eed}{\end{deqarr}}
\newcommand{\bedl}[1]{\begin{deqarr}\label{#1}}
\newcommand{\eedl}[2]{\arrlabel{#1}\label{#2}\end{deqarr}}

\newcommand{\bel}[1]{\begin{equation}\label{#1}}
\newcommand{\bea}{\begin{eqnarray}}
\newcommand{\bean}{\begin{eqnarray}\nonumber}
\newcommand{\beal}[1]{\begin{eqnarray}\label{#1}}
\newcommand{\eea}{\end{eqnarray}}

\newcommand{\qed}{\hfill $\Box$ \medskip}

\newcommand{\be}{\begin{equation}}
\newcommand{\eeq}{\end{equation}}
\newcommand{\ee}{\end{equation}}
\newcommand{\beqa}{\begin{eqnarray}}
\newcommand{\eeqa}{\end{eqnarray}}
\newcommand{\beqan}{\begin{eqnarray*}}
\newcommand{\eeqan}{\end{eqnarray*}}
\newcommand{\ba}{\begin{array}}
\newcommand{\ea}{\end{array}}

\newcommand{\mcM}{{\mycal M}}
\newcommand{\mcD}{{\mycal D}}

\newcommand{\scri}{{\mycal I}}%

\newcommand{\warn}[1]%{}%{}
{\protect{\stepcounter{mnotecount}}$^{\mbox{\footnotesize
$%\!\!\!\!\!\!\,
\bullet$\themnotecount}}$ \marginpar{%\color{red}%
\raggedright\tiny\em
$\!\!\!\!\!\!\,\bullet$\themnotecount: {\bf Warning:} #1} }

\newcommand{\eq}[1]{(\ref{#1})}

\newcommand{\ptc}[1]{\mnote{{\bf ptc:}#1}}

\newcommand{\beqar}{\begin{deqarr}}
\newcommand{\eeqar}{\end{deqarr}}

\newcommand{\beaa}{\begin{eqnarray*}}
\newcommand{\eeaa}{\end{eqnarray*}}

\DeclareFontFamily{OT1}{rsfs}{}
\DeclareFontShape{OT1}{rsfs}{m}{n}{ <-7> rsfs5 <7-10> rsfs7 <10-> rsfs10}{}
\DeclareMathAlphabet{\mycal}{OT1}{rsfs}{m}{n}

{\catcode `\@=11 \global\let\AddToReset=\@addtoreset}
\AddToReset{equation}{section}

{\catcode `\@=11 \global\let\AddToReset=\@addtoreset}
\AddToReset{figure}{section}

{\catcode `\@=11 \global\let\AddToReset=\@addtoreset}
\AddToReset{table}{section}

\newcommand{\rnabla}{\not\hspace{-.2em}\nabla}
\newcommand{\rGamma}{\not\hspace{-.1em}\Gamma}
\newcommand{\re}{\not\hspace{-.1em}e}
\newcommand{\rsigma}{\not\hspace{-.1em}\sigma}

\begin{document}

\title{On the smoothness of the critical sets of the cylinder at spatial infinity in vacuum spacetimes\protect\thanks{Preprint UWThPh-2018-14}}
\author{
Tim-Torben Paetz
\thanks{E-mail:  Tim-Torben.Paetz@univie.ac.at}  \vspace{0.5em}\\  \textit{Gravitational Physics, University of Vienna}  \\ \textit{Boltzmanngasse 5, 1090 Vienna, Austria }
}
\maketitle

\maketitle

%\vspace{-0.2em}

\begin{abstract}
We analyze the appearance of logarithmic terms at the critical sets of Friedrich's cylinder representation of spatial infinity.
It is shown that if the radiation field vanishes at all orders at the critical sets no logarithmic terms are produced in the formal expansions.
Conversely, it is proved  that, under the additional hypothesis that the spacetime has constant (ADM) mass aspect and vanishing dual (ADM)  mass aspect, this condition
is also necessary for a spacetime to admit a smooth representation at the critical sets.
\end{abstract}

\tableofcontents

\section{Introduction}

\subsection{Asymptotically flat spacetimes}

The notion of \emph{asymptotic flatness} is a delicate issue in general relativity due to the absence of non-dynamical background fields.
 Penrose \cite{p1, p2} provided an elegant geometric approach to resolve this issue, cf.\ \cite{ah, geroch}. A spacetime is regarded as asymptotically flat if it admits
a \emph{smooth conformal compactification at infinity}. By this it is meant that, after some appropriate conformal rescaling, one can attach a conformal boundary $\scri$ 
to the spacetime through which the rescaled metric admits a smooth extension.
The picture behind this notion  is that such an extension  is possible whenever the gravitational field admits an asymptotically Minkowski-like
fall-off behavior.

In this paper we are interested in the vacuum case $\widetilde R_{\mu\nu}=0$. Then $\scri$ is a null hypersurface. 
Imposing the ``natural'' topological restriction $\scri \cong \mathbb{R}\times S^2$ one can show that the Weyl tensor needs to vanish on a smooth $\scri$ (cf.\ \cite{geroch, p2}).
Even more, the smoothness of $\scri$ is related to specific peeling properties of the Weyl tensor, cf.\ \cite{F_17} for an overview.
This raised the question whether the Einstein equation are compatible with Penrose's notion of asymptotic flatness in the sense that they admit
a sufficiently large class of solutions for which the Weyl tensor shows this peeling behavior.
Meanwhile, Klainerman  and Nicol\`o have shown that appropriate, sufficiently small asymptotically Euclidean Cauchy data generate
vacuum spacetimes where the Weyl tensor does  have  the peeling properties \cite{kn}, leaving the question open, whether a smooth $\scri$ is generated  as well.

Besides, Penrose's approach  provided a tool to construct asymptotically flat vacuum spacetimes.
Instead of studying the long-term behavior of the gravitational field by limiting processes one can start from the outset in a conformally rescaled
spacetime and work on bounded domains, which is also very convenient from a PDE point of view.
The relevant substitute to Einstein's field equations are Friedrich's conformal field equations (CFE) \cite{F1,F2}  which are equivalent to Einstein's if the conformal factor
is non-zero, but which remain regular where it vanishes. 
In particular, this conformal approach permits the formulation of an \emph{asymptotic initial value problem} where data are prescribed either on a portion of $\scri^-$ and an incoming 
null hypersurface \cite{kannar}, or on $\scri^-$ as the future light-cone of a regular point $i^-$ representing past timelike infinity \cite{F_T, ChPaetzInfCone}.
While this shows that there is a large class of asymptotically flat vacuum spacetimes one does not gain any insights how `generic' these solutions are, as the smoothness of $\scri$ is build in from the outset.

One therefore needs to study initial value problems where data are prescribed on ordinary (i.e.\ non-asymptotic) hypersurfaces.
To avoid difficulties at spatial infinity one is led to study, as a first step, a hyperboloidal Cauchy problem, where data are prescribed on a spacelike hypersurface which intersects
$\scri$ in a spherical cross-section.
Supposing that the relevant initial data for the CFE admit  smooth extensions through $\scri$, Friedrich proved local well-posedness \cite{F_hyp}.
For small data one can use Cauchy stability of the underlying symmetric hyperbolic system contained in the CFE to show that the solution is complete
and even admits a smooth future timelike infinity $i^+$ \cite{F_hyp2}.
However, it turns out that generic solutions to the vacuum constraint equations do not admit a smooth but a polyhomogeneous expansion at $\scri^-$,
and that certain mild regularity conditions need to be imposed on the asymptotic behavior of the freely prescribable ``seed'' data to end up with a solution
to the constraints which is smooth  at $\scri$ \cite{andersson, acf}.

The same phenomenon can be observed for a characteristic Cauchy problem with data on either a future light-cone or two transversally intersecting null hypersurfaces.
Assuming that the data for the CFE are smoothly extendable  through the cross-section, where the initial surface intersects $\scri$, local well-posedness holds
an the spacetime admits a piece of a smooth $\scri$ \cite{CCTW}.
Again, one finds that generically  solutions of the characteristic constraint equations constructed from smooth ``seed'' data develop logarithmic terms at $\scri$, while  certain
mild regularity  condition ensure that this does not happen \cite{ChPaetz2, ttp3}.

While these types of initial value problems permit the construction of  large classes of non-generic asymptotically flat vacuum spacetimes,
they also show that this is only possible if the leading order terms of the seed data are subject to certain regularity conditions.

To fully analyze this issue, one needs to construct asymptotically flat spacetimes from an ordinary Cauchy problem.
In fact, this has been done via  gluing techniques. Friedrich's results on the hyperboloidal Cauchy problem can be used to 
construct asymptotically flat, and in fact asymptotically simple (where, in addition, a certain completeness condition is imposed), vacuum spacetimes from  Cauchy data which are glued to stationary data near spatial infinity~\cite{cd}.

All these constructions have in common that they circumvent issues arising at spatial infinity by either ignoring this part of the spacetime completely, or
choosing the spacetime to be stationary near $i^0$.
However, in order to gain a full understanding of the obstructions coming along with Penrose's notion of asymptotic flatness and whether this definition is broad enough to include all cases of physical interest,  an understanding of the behavior
of the gravitational field at spatial infinity is essential.

The results described above might be viewed  as an indication that a polyhomogeneous $\scri$ where
the asymptotic expansion of the gravitational field is allowed to have logarithmic terms is somewhat more natural and generic.
However, to give a fully satisfactory answer, spatial infinity, where Cauchy surfaces with asymptotically Euclidean data ``touch'' infinity, needs to be taken into account. 
In a recent breakthrough result Hintz and Vasy \cite{hiva} were able to construct spacetimes with a polyhomogeneous $\scri$
from asymptotically Euclidean Cauchy data sets. 
They use a gauge where log terms are inevitably produced unless the ADM mass vanishes, so a priori their result does not
give hints concerning the smoothness of  $\scri$ (in an appropriate gauge).

Ideally one would like to mimic the results on the hyperboloidal Cauchy problem and construct asymptotically Euclidean Cauchy data 
which admit  smooth extensions through spatial infinity and use standard result on symmetric hyperbolic systems.
However, it is well-known
that in the ``classical'' representation of spatial infinity as a point $i^0$, this point cannot be regular unless the ADM mass vanishes.
This is somehow intuitive as %the whole 
spatial infinity is compressed to a single point.
It  led Friedrich to introduce a blow-up of this point to a cylinder $I$ \cite{F_i0}.
While the metric is singular there, the frame field is not, and it becomes possible to construct non-trivial asymptotically Euclidean Cauchy data sets such that all the fields which appear in the CFE admit smooth expansions through the 2-sphere $I^0$ where the Cauchy surface intersects $I$ \cite{F_i0, kroon0}.

Nevertheless, some difficulties remain (cf.\ \cite{F3, kroon_book} for an overview). The hyperbolicity of the CFE breaks down at the critical sets $I^{\pm}$ where the cylinder
``touches'' $\scri^{\pm}$. Related to this is the following property: The CFE provide inner equations on the cylinder so that
in principle all fields (including all radial derivatives) can be determined on the cylinder from the (asymptotic part of the) Cauchy data. It turns out, though, that, in general, logarithms arise at the critical sets.
While this seems similar to e.g.\ the hyperboloidal Cauchy problem, the problem is in fact much more severe: In the latter case, once one  makes
sure that  no logarithmic terms arise when solving the constraints equations, i.e.\   that the restriction of
the relevant fields to the initial surface admits smooth extension through $\scri$, no logarithmic terms arise
in higher order transverse derivatives  (in an appropriate gauge).
At the critical sets this is no longer true. In principle, logarithmic terms can arise at any order, while all lower orders admit smooth expansions at $I^-$.

A priori, non-smoothness of the critical sets is irrelevant as the critical sets do not belong to $\scri$. However, one expects that these log terms spread over to $\scri$ and therefore produce a $\scri$ which is not smooth but  merely  polyhomogeneous.
In the spin-2 case this can explicitly be shown \cite{F_spin}, and there is no reason to expect that the situation is better in the non-linear case.
For this reason one would like to understand the meaning  of logarithmic  terms at the critical sets, and characterize initial data 
which admit smooth extensions through $I^{\pm}$.

Friedrich \cite{F_i0} (in the time-symmetric case)  and Valiente Kroon \cite{kroon0} (in the general case) analyzed this issue starting from an ordinary Cauchy problem.
They derived a couple of necessary conditions on the Cauchy data to get rid of the log terms.
One can also establish  sufficient conditions for the non-appearance of logarithmic terms at the critical sets, and at $\scri^-$:
Dain \cite{dain} has shown, cf.\ \cite{d_s}, that an asymptotically Euclidean  spacetime which is stationary near spatial infinity admits a smooth $\scri$ (at least near spacelike infinity).
Friedrich \cite{F_i0_2} showed that Cauchy data which are static also admit smooth critical sets, a result which has been generalized by
Ace\~na and  Valiente Kroon \cite{acena} to the stationary case.
One may regard this as another indication that there is a relation between the appearance of log terms at the critical sets and at $\scri$.
Their results  imply  that  the inner equations on the cylinder do not produce logarithmic terms if the data are
merely asymptotically stationary.

This raises the question  whether asymptotic stationarity (staticity in the time-symmetric~case) is also \emph{necessary} for the non-appearance of log terms.
That the notion of asymptotic staticity  plays a distinguished in view of the appearance of log terms at the critical sets was first observed in \cite{F_i0}.
Some evidence that this might be true is provided by the result~\cite{kroon1}  that  staticity is necessary for time-symmetric, conformally flat data.

The main purpose of this paper is to analyze these kind of  issues from $\scri^-$. More precisely, we consider an asymptotic characteristic
initial value problem, where data are prescribed on $\scri^-$ (the data on the incoming null hypersurface will be largely irrelevant),
and  analyze the appearance of log terms approaching the critical set $I^-$ from $\scri^-$.
 Assuming that the data  generate a spacetime with a smooth cylinder we will then analyze the  appearance of log terms  when approaching  $I^-$ from $I$, as well, where we assume that
the data for the transport equations on $I$ are induced  by the limit of the corresponding fields on $\scri^-$ to $I^-$.

A main  advantage of this approach as compared to the ordinary Cauchy problem is that the critical set arises as a future boundary of 
the initial surface, whence it is easier to control the fields there. Another advantage is that the
no-logs conditions one obtains from $\scri^-$ turn out to be somewhat easier to handle. Finally, the no-logs conditions
 depend crucially on the radiation field, which essentially provides  the freely prescribable data near $I^-$ (together with certain ``integration functions''
prescribed at $I^-$, cf.\ Appendix~\ref{app_ADM}).

Omitting some technical details our main result can be stated as follows:

\begin{theorem}
\begin{enumerate}
\item[(i)] Assume that a smooth vacuum spacetime with smooth $\scri^-$, $I$ and $I^-$  has been given
and assume further that it has constant (ADM) mass aspect and vanishing dual  (ADM)  mass aspect,%
\footnote{The dual mass aspect introduced later might be regarded as a generalized NUT-like parameter.}
 then the radiation field vanishes at $I^-$
at any order.
%\tim{one should mention how this is related to asymptotic stationarity}
\item[(ii)] Conversely, the restriction of all the fields appearing in the CFE to both $\scri^-$ and $I$, and all derivatives thereof, admit smooth extensions through  $I^-$ if the radiation field vanishes there at any order.
\end{enumerate}
\end{theorem}

\begin{remark}
{\rm
This raises the question, to be analyzed elsewhere, whether an asymptotic characteristic Cauchy problem with data on $\scri^-$ and some incoming null hypersurface
generates a vacuum spacetime which admits a smooth cylinder with smooth critical sets if the prescribed radiation field vanishes at $I^-$ at any oder.
}
\end{remark}

\subsection{Overview}

In Section~\ref{section1} we recall the conformal field equations as well as the conformal Gauss gauge, 
which  provides a natural, geometric gauge in the  conformal setting we will use. In particular  we explain how a conformal Gauss
gauge is constructed from $\scri^-$.
We  further describe the additional gauge data which one may specify on $\scri^-$.
Finally we  provide  the evolution equations in this gauge, and, as we will analyze the constraint equation using adapted null coordinates, we give the relation between frame components and coordinate components on $\scri^-$.

In Section~\ref{section2} we derive conditions on the gauge data at $\scri^-$ in order to end up with a spacetime which admits
a finite representation of spatial infinity. Starting from the Minkowski spacetime as an explicit example we then introduce the
cylinder representation of spatial infinity. We also extract another gauge freedom. The section is closed by defining a what we call
``(weakly) asymptotically Minkowski-like conformal Gauss gauge'', were a certain asymptotic behavior of the gauge data at $\scri^-$ is
imposed when approaching $I^-$ which turns out to be very convenient for the subsequent analysis.

Section~\ref{section3} is devoted to an analysis of the behavior of all the unknowns which appear in the CFE and all
transverse derivatives thereof when approaching the critical set $I^-$ from $\scri^-$. In the zeroth order this can and will be done explicitly. For higher order transverse derivatives we will merely derive the structure of the equations and of the no-logs condition, as the gauge will be  still quite general at this stage.

In Section~\ref{section4} we provide a corresponding analysis  of the behavior of the fields and their radial derivatives 
 when approaching the critical set $I^-$ from $I$.
In the zeroth order one recovers the Minkowskian values. We will also study the first-order radial derivatives explicitly
and  derive the structure of the equations and of the no-logs condition for radial derivatives of higher orders.
The no-logs conditions adopt here a somewhat more difficult form, as they are PDEs rather than ODEs.
However, expanding the fields in spherical harmonics they  can be transformed into  hypergeometric 
ODEs which are analyzed.

In Section~\ref{sec_gauge_ind} we show that transforming a spacetime which admits a smooth representation of $\scri^-\cup I^-\cup I$ from one weakly asymptotically Minkowski-like conformal Gauss gauge
into another one is accompanied by a  smooth (coordinate and conformal)  transformation unless the gauge data are badly chosen in such a way that the
congruence of conformal geodesics produces conjugate point directly on $I^-$.
From this we conclude  that if a spacetime is smooth near $I^-$  in one weakly asymptotically Minkowski-like conformal Gauss gauge
then the same is true in any other one.

This allows us to restrict attention in  Section~\ref{section6} to a more restricted gauge which we call 
 asymptotically Minkowski-like conformal Gauss gauge at each  order (which is also introduced in Section~\ref{section2}).
In this particular gauge we then show that for a radiation field which vanishes asymptotically at $I^-$ at any order
no logarithmic terms are produced when approaching $I^-$ from both $\scri^-$ and $I$.
This can be done because in this gauge one can show that simply no terms arise at the critical orders which produce
logarithmic terms. On $\scri^-$ it can be shown that the fields are basically polynomials of a sufficiently low degree, supplemented by terms which decay arbitrarily fast, while on $I$ all fields decay sufficiently fast at $I^-$.
It is  illuminating to bring the Kerr metric into such a gauge (at least up to some order).

 In  Section~\ref{section7} we study the massless spin-2 equation. On the one hand this provides a toy model for the full non-linear
case where the whole analysis concerning the appearance of log terms can be done very explicitly. In particular one finds if-and-only-if conditions for the  appearance of logarithmic terms at the critical sets. On the other hand this allows us to compare
the no-logs conditions of the spin-2 case with the general case, which gives some insights what the sources of additional or more restrictive no-logs conditions are in the general case (or even on a curved  background).

 In  Section~\ref{section8} we restrict attention to data with constant (ADM) mass aspect and vanishing (ADM) dual mass aspect (by which we mean
the limit of the Bondi mass and dual mass aspect to $I^-$).
We will explain why this simplifies the analysis considerably. We then explicitly show that the no-logs conditions are indeed
more restrictive as compared to the spin-2 case (in fact this is already known from the ordinary Cauchy problem \cite{kroon2}).
In a next step  we  show that logarithmic terms are inevitably produced, unless the radiation field vanishes at
any order at $I^-$, which  does not need to be the case in the spin-2 case.
This requires some rather lengthy computations as also next-to-leading order terms need to be taken into account.

In  Appendix~\ref{app_app} we review the constraint equations in adapted null coordinates for the
asymptotic characteristic initial value problem. We also provide a slight variation of the standard approach which allows
to shift the freedom to prescribe certain data on the intersection sphere of the two null hypersurfaces to $I^-$, so that
as many data as possible are directly prescribed on $\scri^-$ and its future boundary $I^-$.

\subsection{Notation}

Here let us given an overview over some frequently used notation:
\begin{enumerate}
\item $\eta=\mathrm{diag}(-1,1,1,1)$
\item $(.)_{\mathrm{tf}}$ denotes the trace-free part
\item $g=\Theta^2 \widetilde g$, where $\widetilde g$ is the physical metric, the inverse metric is denoted by $g^{\sharp}$
\item $\scri^{\pm}$ denotes future and past null infinity, $i^0$ spatial infinity, in particular if represented as a point, $I$ denotes the cylinder representation of spatial infinity, $I^{\pm}$  the critical sets where $\scri^{\pm}$ and $I$ ``touch''
\item
Coordinate spacetime indices are denoted by $\mu,\nu,\sigma,\dots$, spatial coordinate indices are denoted by $\alpha,\beta,\gamma,\dots$, and  angular coordinate
indices by $\mathring A, \mathring B, \mathring C\dots$.
\item
Frame  spacetime indices are denoted by $i,j,k,\dots $, spatial frame  indices are denoted by $a,b,c,\dots$, and angular coordinate
indices by $ A,  B, C,\dots$
\item Objects associated to the Weyl connection are decorated with $\widehat {.}$
\item Objects associated to $\not \hspace{-.2em}g=g_{\mathring A\mathring B}\mathrm{d}x^{\mathring A}\mathrm{d}x^{\mathring B}|_{\scri^-}$
are denoted by $\rnabla$, $\rGamma$ etc.
\item The Levi-Civita covariant derivative associated to the standard metric $s_{\mathring A\mathring B}\mathrm{d}x^{\mathring A}\mathrm{d}x^{\mathring B}$ on $S^2$ is denoted by $\mcD$, the Christoffel symbols by $\mathring\Gamma$ and the volume form by $\epsilon_{\mathring A\mathring B}$.
\item $\nu_{\tau}=g_{\tau r}|_{\scri^-}$, $\nu_{\mathring A}=g_{\tau\mathring A}|_{\scri^-}$
\item The Hodge decomposition of a 1-form $f_A$ on $\mathbb{S}^2:=(S^2,s_{\mathring A\mathring B})$ is written as $f_A=\mcD_A\ul f + \epsilon_A{}^B\mcD_B\ol f$, that of a symmetric trace-free tensor  $\mathfrak{x}_{AB}$ as
$
\mathfrak{x}_{AB} = (\mcD_{A}\mcD_{B}\ul {\mathfrak{x}})_{\mathrm{tf}}
+ \epsilon_{(A}{}^C \mcD_{B)}\mcD_{C}\ol{ \mathfrak{x}}$
\item $f^{(m)}=\frac{1}{m!}\partial_r^mf|_{I^-}$  denotes expansion coefficients in the radial coordinate $r$ ($\Theta^{(n)}$ and $b^{(n)}_i$ are exceptions, they denote the expansion coefficients in $1+\tau$)
\item $f^{(m,n)}=\frac{1}{m!n!}\partial^m_r\partial_{\tau}^n f|_{I^-}$
\item the  independent components of the rescaled Weyl tensor $W_{ijkl}$ we will use are: $U_{AB}=W_{0101}\eta_{AB} + W_{01AB}$, $V^{\pm}_{AB}=(W_{1A1B})_{\mathrm{tf}}\pm W_{0(AB)1}$ and $W^{\pm}_A=W_{010A}\pm W_{011A}$.
\item $M$ denotes the (ADM) mass aspect, and $N$ the dual (ADM) mass aspect, we also use the notation $\mathcal{M}_A=\mcD_A M + \epsilon_A{}^B\mcD_B N$
and  $\ol{\mathcal{M}}_A=\mcD_A M - \epsilon_A{}^B\mcD_B N$ (more  precisely, the limit of the Bondi mass and dual mass aspect on $\scri^-$ to $I^-$).
\end{enumerate}

\section{General conformal field equations and conformal Gauss gauge}
\label{section1}

\subsection{General conformal field equations}
A breakthrough on the way to gain a better understanding of the compatibility of Penrose's notion of asymptotic simplicity and Einstein's field equations was obtained by Friedrich \cite{F1,F2}. He derived a set of equations, the so-called \emph{conformal field equations (CFE)}, which substitute 
Einstein's vacuum field equations in Penrose's conformally rescaled  spacetimes. They are equivalent to the vacuum equations
in regions where the conformal factor does not vanish, and  remain regular  at points where it vanishes.
This result offered the possibility to study the evolution of initial data sets in the conformally rescaled spacetime from the outset.
In particular, it permitted the  formulation of an \emph{asymptotic initial value problem} where an appropriate set of data is prescribed  ``at infinity'', i.e.\  at $\scri^-$.

Beside the usual gauge freedom arising from the freedom to choose coordinates,  frame field etc., the CFE  contain an additional gauge freedom which arises from the artificially introduced conformal factor $\Theta$ which relates the physical spacetime with its conformally rescaled counterpart. For an analysis of the gravitational field near spacelike infinity, though, it turned out
that the introduction of additional gauge degrees of freedom, which even more exploits the conformal structure, can simplify the analysis considerably. They are obtained when replacing the Levi-Civita connection
by some appropriately chosen Weyl connection. This way one is led to the so-called \emph{general conformal field equations (GCFE)},
introduced by Friedrich in \cite{F_AdS}, cf.\ \cite{F_i0, F3, F_i0_2}.
In the following we will recall these equations and  sum up some of the results.

Let $(\widetilde\mcM, \widetilde g)$ be a smooth
Lorentzian manifold, and denote by $g=\Theta^2 \widetilde g$ a conformally rescaled metric.
We denote by $\widetilde \nabla$ and $\nabla$ the Levi-Civita connection of $\widetilde g$ and $g$, respectively.

Let $\widetilde f$ be a smooth 1-form on $\widetilde\mcM$.
There exists a unique 
torsion-free connection $\widehat\nabla$,  the  so-called \emph{Weyl connection}, which satisfies
\begin{equation}
\widehat\nabla_{\sigma} \widetilde g_{\mu\nu} = -2 \widetilde f_{\sigma} \widetilde g_{\mu\nu} \,.
\label{weyl_relation}
\end{equation}
Then
\begin{equation}
\widehat\nabla = \widetilde\nabla + S(\widetilde f)\,, \quad \text{where}\quad S( \widetilde  f)_{\mu}{}^{\sigma}{}_{\nu}:=2\delta_{(\mu}{}^{\sigma} \widetilde f_{\nu)} - \widetilde g_{\mu\nu}\widetilde g^{\sigma\rho}\widetilde f_{\rho}
\,,
\end{equation}
or, equivalently,
\begin{equation}
\widehat\nabla =\nabla + S(f)\,, \quad \text{where}\quad f=\widetilde f-\Theta^{-1}\mathrm{d}\Theta
\,.
\end{equation}
We observe that $S(\widetilde f)$ depends merely on the conformal class of $\widetilde g$.

Let $e_k$ be a frame field satisfying  $g(e_i,e_j)=\eta_{ij}\equiv \mathrm{diag}(-1,1,1,1)$.
We define the connection coefficients of $\widehat \nabla$ in this frame field by
\begin{equation}
\widehat\nabla_i e_j = \widehat\Gamma_i{}^k{}_j e_k
\,.
\end{equation}
Note that
\begin{equation}
 \widehat\Gamma_i{}^k{}_j =  \Gamma_i{}^k{}_j  +S(f)_i{}^k{}_j \,, \quad \text{and} \quad 
f_i =\frac{1}{4}\widehat\Gamma_i{}^k{}_k
\,.
\end{equation}
A Weyl connection respects the conformal class in the sense that for any $C^1$-curve $\gamma: (-\varepsilon, \varepsilon)
\rightarrow \widetilde\mcM$ and any frame field $e_k$ which is parallely transported along $\gamma$ w.r.t.\ $\widehat\nabla$,
there exists a function $\Omega_{\tau}>0$ along $\gamma(\tau)$  such that $\widetilde g(e_i,e_j)|_{\gamma(\tau)}
=(\Omega_{\tau})^2\widetilde g(e_i,e_j)|_{\gamma(0)}$.

Finally, set
\begin{equation}
b:=\Theta\widetilde f =\Theta f + \mathrm{d}\Theta
\,,
\label{dfn_d}
\end{equation}
and denote by
\begin{align}
\widehat W^{\mu}{}_{\nu\sigma\rho} &=  \Theta^{-1}\widehat C^{\mu}{}_{\nu\sigma\rho}
\,,
\\
\widehat L_{\mu\nu} &= \frac{1}{2}\widehat R_{(\mu\nu)} - \frac{1}{4} \widehat R_{[\mu\nu]} - \frac{1}{12}\widehat R g_{\mu\nu}
\,,
\end{align}
\emph{ rescaled Weyl tensor} and \emph{Schouten tensor} of $\widehat\nabla$, respectively.
We note that
\begin{equation}
\widehat L_{\mu\nu} =L_{\mu\nu} - \nabla_{\mu}f_{\nu}  + \frac{1}{2}S(f)_{\mu}{}^{\sigma}{}_{\nu} f_{\sigma}
\,,
\label{schouten_weylconnect}
\end{equation}
while the rescaled Weyl tensor does not depend on the Weyl connection,
\begin{equation}
\widehat W^{\mu}{}_{\nu\sigma\rho}=W^{\mu}{}_{\nu\sigma\rho}
\,.
\end{equation}

Let now $(\widetilde\mcM, \widetilde g)$ be a
solution to  Einstein's vacuum field equations 
\begin{equation}
\widetilde R_{\mu\nu}=\lambda \widetilde g_{\mu\nu}
\,.
\end{equation}
Then the tuple
\begin{equation}
\mathfrak{f}:= (e^{\mu}{}_k, \widehat\Gamma_i{}^k{}_j, \widehat L_{ij}, W^i{}_{jkl})
\,,
\end{equation}
where $e^{\mu}{}_k:=\langle\mathrm{d} x^{\mu},e_k\rangle$,
satisfies the \emph{general conformal field equations (GCFE)} \cite{F_AdS}
\begin{align}
[e_p,e_q] &=  2\widehat\Gamma_{[p}{}^l{}_{q]} e_l
\,,
\label{GCFE_1}
\\
e_{[p}(\widehat\Gamma_{q]}{}^i{}_j) &=
 \widehat\Gamma_k{}^i{}_j\widehat\Gamma_{[p}{}^k{}_{q]}
- \widehat\Gamma_{[p}{}^i{}_{|k|}\widehat\Gamma_{q]}{}^k{}_j
+ \delta_{[p}{}^i\widehat L_{q]j} - \delta_j{}^i \widehat L_{[pq]} - \eta_{j[p}\widehat L_{q]}{}^i +\frac{\Theta}{2}   W^i{}_{jpq}\,,
\\
2\widehat\nabla_{[k}\widehat L_{l]j} &= b_i W^{i}{}_{jkl}\,,
\label{GCFE_3}
\\
 \widehat\nabla_i W^i{}_{jkl}  &= \frac{1}{4}\widehat\Gamma_i{}^p{}_pW^i{}_{jkl} 
\label{GCFE_4}
\,.
\end{align}
The fields $\Theta$ and $b$ reflect the \emph{conformal gauge freedom}.

\subsection{Conformal geodesics and   conformal Gauss gauge}
\subsubsection{Definition and properties of conformal geodesics}

A \emph{conformal geodesic for $(\widetilde\mcM , \widetilde g)$}
(cf.\ e.g.\ \cite {F_cg, F_Schmidt})
 is a curve $x(\tau)$ in $\widetilde\mcM$
for which   a 1-form $\widetilde f= \widetilde f(\tau)$ exists along $x(\tau)$ such that the pair $(x, \widetilde f)$ solves
the \emph{conformal geodesics equations}
\begin{align}
(\widetilde\nabla_{\dot x} \dot x)^{\mu} + S(\widetilde f)_{\lambda}{}^{\mu}{}_{\rho} \dot x^{\lambda} \dot x^{\rho} &= 0
\,,
\label{conf_geo1}
\\
(\widetilde \nabla_{\dot x} \widetilde f)_{\nu} - \frac{1}{2} \widetilde f_{\mu} S(\widetilde f)_{\lambda}{}^{\mu}{}_{\nu}\dot x^{\lambda} &= \widetilde L_{\lambda\nu}\dot x^{\lambda}
\,.
\label{conf_geo2}
\end{align}
Given data $x_*\in \widetilde \mcM$, $\dot x_*\in T_{x_*}\widetilde\mcM$ and $\widetilde f_*\in T^*_{x_*}\widetilde\mcM$
there exists a unique solution $x(\tau)$, $\widetilde f(\tau)$  to  \eq{conf_geo1}-\eq{conf_geo2} near  $x_*$
satisfying, for given $\tau_*\in\mathbb{R}$,
\begin{equation}
x(\tau_*) = x_*\,, \quad \dot x(\tau_*) =  \dot x_* \,, \quad \widetilde f (\tau_*) = \widetilde f_*
\,.
\end{equation}
Conformal geodesics are  curves which are associated with the conformal structure in a similar
way as geodesics are associated with the metric.
They are conformally invariant in the following sense:
Let $b$ be a smooth 1-form on $\mcM$. Then $(x(\tau), \widetilde f(\tau))$ solves \eq{conf_geo1}-\eq{conf_geo2}
if and only if $(x(\tau), \widetilde f(\tau)-b|_{x(\tau)})$ solves \eq{conf_geo1}-\eq{conf_geo2} with $\widetilde\nabla$ and $\widetilde L$
replaced by $\widehat \nabla =\widetilde\nabla + S(b)$ and $\widehat L$, respectively.
The conformal geodesic $x(\tau)$ and its parameter $\tau$ are independent of the Weyl connection in the conformal
class w.r.t.\ which \eq{conf_geo1}-\eq{conf_geo2} are written
(in particular, they do not depend on the metric in the conformal class chosen to write the conformal geodesics equations).

It follows from \eq{conf_geo1} that the sign of $\widetilde g(\dot x,\dot x)$ is preserved  along any conformal geodesic,
\begin{equation}
\widetilde \nabla_{\dot x} \widetilde g(\dot x, \dot x) = -2 \langle \widetilde f, \dot x\rangle \widetilde g(\dot x, \dot x)
\,,
\label{norm_conf_geod}
\end{equation}
i.e.\   conformal geodesics preserve their causal character.

\begin{lemma}[\cite {F_cg}]
\label{lemma_repara_conf_geods}
Consider a conformal geodesic $x(\tau)$. 
Changes of its  initial data $\dot x_*$ and $f_*$ for the conformal geodesics equations 
which  locally preserve the point set spread out by the curve $x(\tau)$  (i.e.\ which change only its parameterization) are given by
\begin{align}
\dot x_* &\mapsto \varphi_* \dot x_* \,,  \quad \varphi_*\in\mathbb{R}\setminus\{0\} \,, 
\label{repara1}
\\
 f_* &\mapsto f_* + \psi_*   g( \dot x_*,\cdot)
\,,  \quad \psi_*\in\mathbb{R}
\,.
\label{repara2}
\end{align}
\end{lemma}

Consider now a congruence $\{ x(\tau,\rho), f(\tau,\rho)\}$ of conformal geodesics, set $x':=\partial x/\partial \rho$, and denote by $F:= \nabla_{x'}f$ the \emph{deviation 1-form}.
The \emph{conformal Jacobi equation} reads \cite{F_cg}
\begin{equation}
\nabla_{\dot x}\nabla_{\dot x} x' = \mathrm{Ric}( \dot x, x') \dot x
- S(F)(\dot x, \cdot,\dot x) - 2S(f)(\dot x, \cdot, \nabla_{\dot x} F)
\,.
\label{conf_dev}
\end{equation}
It is convenient to introduce \emph{conformal Gauss coordinates}, a geometrically defined coordinate system, 
where the time-like coordinate lines are generated by timelike conformal geodesics.
As metric geodesics, conformal geodesics may develop caustics. Even worse, they may, in addition,  become tangent to each other. 
However, because of   \eq{conf_dev} one may hope  that curvature induced tendencies to develop caustics may be counteracted by the 1-form $f$.
This is a reason why one expects a  gauge   which is based on conformal geodesics to provide  a more convenient  setting to cover large domains of  spacetime than the classical Gauss gauge does.
In fact, Friedrich showed that the Schwarzschild-Kruskal spacetime permits a global coordinate system based on conformal geodesics \cite{F_cg}.
Also in the case of  a hyperboloidal initial value problem with data sufficiently close to Minkowskian hyperboloidal data
 it can be shown that conformal Gauss  coordinates exist globally \cite{luebbe_kroon}.

\subsubsection{Conformal Gauss gauge}

We  construct  a conformal Gauss gauge adapted to a congruence of conformal geodesics, which employs the fact that a congruence of conformal
geodesics distinguishes the Weyl connection associated to the 1-form $f$.
This will be done from an initial surface $\mathcal{S}$  (spacelike or null)  which intersects the congruence transversally and meets
each of the curves exactly once. 
This way we mimic the construction in \cite{F_AdS,  F_i0}. Here, though, it provides the starting point to construct such a gauge for the special case
where $\mathcal{S}$ is  identified with past null infinity (for non-negative cosmological constant).

Consider a smooth congruence of conformal geodesics which covers an open set $U$ of $\widetilde \mcM$ (we do not need to impose the  vacuum equations at this stage),
 and assume that
the associated 1-form $\widetilde f$ defines a smooth tensor field on $U$.
As above, we denote by $\widehat \nabla$  the Weyl connection which satisfies $\widehat\nabla = \widetilde\nabla + S(\widetilde f)$.
Then \eq{conf_geo1}-\eq{conf_geo2} adopt the simple form
\begin{equation}
\widehat\nabla_{\dot x} \dot x  = 0
\,,
\quad 
 \widehat L(\dot x, \cdot) =0
\,.
\label{conf_geob}
\end{equation}
We would like to construct a gauge where  $\widehat \nabla$ preserves the conformal structure, i.e.\ where 
\begin{equation}
\widehat\nabla_{\dot x} g =0
\,.
\label{conf_gauge1}
\end{equation}
It follows from \eq{weyl_relation} that
\begin{equation}
\eq{conf_gauge1}
\quad \Longleftrightarrow \quad
\widetilde \nabla_{\dot x}\Theta = \Theta\langle \dot x, \widetilde f\rangle
\,,
\label{equiv_relation_Theta}
\end{equation}
i.e.\ such a conformal gauge can always be realized by an appropriate choice of the conformal  factor $\Theta$.

Let us start with the physical spacetime.
Let  $\widetilde {\mathcal{S}}\subset \widetilde \mcM$ be a hypersurface
(for definiteness take a spacelike or characteristic one).
We require the fields $x$, $\widetilde f$ and $\Theta$ to satisfy \eq{conf_geo1}, \eq{conf_geo2} and \eq{conf_gauge1}.
This conformal gauge  needs to be  complemented  by the ``gauge data''
\begin{equation}
\dot x|_{\widetilde{\mathcal{S}}}\,, \quad \widetilde f|_{\widetilde{\mathcal{S}}} \,, \quad \Theta|_{\widetilde{\mathcal{S}}}>0
\,,
\label{gauge_data_o_cauchy}
\end{equation}
 with  $\dot x|_{\widetilde{\mathcal{S}}}$ being  transversal to $\widetilde {\mathcal{S}}\subset \widetilde \mcM$. In this article we will always assume
 $\dot x|_{\widetilde{\mathcal{S}}}$ to be  timelike, $\widetilde g (\dot x, \dot x)|_{\widetilde{\mathcal{S}}} <0$.
It follows from \eq{equiv_relation_Theta} that instead of  \eq{gauge_data_o_cauchy} one
may prescribe
\begin{equation}
\dot x|_{\widetilde{\mathcal{S}}}\,, \quad \widetilde f_{\widetilde {\mathcal{S}}} \,, \quad \Theta|_{\widetilde{\mathcal{S}}}>0\;, \quad \nabla_{\dot x}\Theta |_{\widetilde{\mathcal{S}}}
\,,  
\end{equation}
where $\widetilde g (\dot x, \dot x)|_{\widetilde{\mathcal{S}}} <0$.
Here $\widetilde  f_{\widetilde {\mathcal{S}}}$ denotes the pull back of  $\widetilde f$ on $\widetilde{\mathcal{S}}$.

Furthermore, we choose the frame field on $\widetilde{\mathcal{S}}$  in such a way that 
\begin{equation}
\Theta^{2} \widetilde g(e_{i},e_{j})|_{\widetilde{\mathcal{S}}}= \eta_{ij}
\quad \Longleftrightarrow \quad    g(e_{i},e_{j})|_{\widetilde{\mathcal{S}}}=\eta_{ij}
\,.
\end{equation}
The frame field is then parallely propagated w.r.t\ $\widehat \nabla$ along the conformal geodesics, whence, by \eq{conf_gauge1}
\begin{equation}
 g(e_i,e_j) =\eta_{ij} \quad \text{on $U$}
\,.
\end{equation}

Let us now pass to a conformally rescaled space-time and replace $\widetilde f$ by $f\equiv \widetilde f- \mathrm{d}\log \Theta$.
Expressed in terms of the Levi-Civita connection of $g$
and the 1-form $f$ the conformal geodesics
equations read
\begin{align}
\nabla_{\dot x} \dot x  &= -S( f)(\dot x,\cdot, \dot x)
\,,
\label{f_eqn1}
\\
\nabla_{\dot x}  f  &= \frac{1}{2}  S( f)(\dot x,f,\cdot)+  L(\dot x,\cdot)
\label{f_eqn2}
\,,
\end{align}
while  \eq{equiv_relation_Theta} becomes
\begin{equation}
 \langle \dot x , f \rangle =0
\,.
\label{equiv_relation_Theta2}
\end{equation}
These equations contain the conformal factor only implicitly.
However, from the above considerations  it is clear how such a gauge can be constructed starting from the physical, non-rescaled spacetime $(\widetilde \mcM,\widetilde g)$ and the 1-form $\widetilde f$.
The free   gauge data on the (spacelike or null)  hypersurface $\mathcal{S}\subset\mcM$ for a conformal gauge adapted to a congruence of timelike conformal geodesics
are
\begin{equation}
\dot x|_{{\mathcal{S}}}\,, \quad  f_{ {\mathcal{S}}} \,, \quad \Theta|_{{\mathcal{S}}}>0\;, \quad \nabla_{\dot x}\Theta|_{{\mathcal{S}}}
\,, \quad \text{with}  \quad   g (\dot x, \dot x)|_{{\mathcal{S}}} <0
\,.
\end{equation}

\subsubsection{Construction of the conformal Gauss gauge from null infinity}
\label{sec_constr_gauss_gauge}

Let us analyze the construction of such a conformal gauge  more detailed in the case where the initial surface belongs to (past) null infinity,
  ${\mathcal{S}}\subset\scri^-$.
More specifically, we consider a smooth $\lambda\geq 0$-vacuum spacetime $(\widetilde \mcM, \widetilde g)$ which admits a conformal representation
$(\mcM, g)$  and a smooth $\scri^-$ \`a la Penrose \cite{p1, p2} (so that $\scri^-$ is either a spacelike or  a null hypersurface).
In the conformal picture this corresponds to a   metric $g$ and a conformal factor $\Theta$,
 which satisfies   $\Theta|_{\scri^-}=0$ and $\mathrm{d}\Theta|_{\scri^-}\ne 0$,    such that $(g, \Theta)$ solves the conformal field equations
with non-negative cosmological constant $\lambda$.

Given initial data $\dot x|_{\scri^-}$ and $f|_{\scri^-}$ on $\scri^-$, we  solve the conformal geodesic equations \eq{f_eqn1}-\eq{f_eqn2} by standard results on ODEs. In particular, this singles out  the Weyl connection $\widehat\nabla=\nabla + S(f)$. 
In addition, though, we want to choose $\Theta$ in such a way that $\langle \dot x, f\rangle =0$ holds in the region $U$ covered by the congruence of conformal geodesics, which will, in general,   not be the case.
The ``wrong'' conformal factor needs to be rescaled by some positive function $\phi$.
 The gauge condition \eq{equiv_relation_Theta2}  adopts the form 
\begin{equation}
\nabla_{\dot x}\phi = \phi \langle \dot x, f\rangle
\,,
\label{ODE_conf_factor}
\end{equation}
and, once this equation has been solved, $g$, $\Theta$ and $f$ need to be replaced by
\begin{equation}
g^{\mathrm{new}}  =\phi^2 g\,, \quad \Theta^{\mathrm{new}} = \phi\Theta\,, \quad f^{\mathrm{new}} = f- \mathrm{d}\log\phi
\,.
\end{equation}
When solving the ODE \eq{ODE_conf_factor} for $\phi$ there remains the freedom to prescribe $\phi|_{\scri^-}$.
We further observe that the freedom to prescribe $\langle \dot x, f\rangle|_{\scri^-}$
 can be identified with the freedom to prescribe
$\nabla_{\dot x} \phi|_{\scri^-}$, while there remains the freedom to prescribe the pull back $f^{\mathrm{new}}_{\scri^-}$
 of $f^{\mathrm{new}}$ on $\scri^-$.

We want to identify   these gauge degrees of freedom with certain gauge data for the initial  value problem. 
Since it is the case we are mainly interested in we will focus on the $\lambda=0$-case.
In that case $\scri^-$ represents a null hypersurface in $(\mcM, g)$. 
In contrast to the case of a positive cosmological constant the Weyl tensor does not need to vanish on $\scri$.
Here, we assume that $\scri^-$ has the ``natural'' topology   \cite{geroch, hawking}
\begin{equation}
\scri^- \cong \mathbb{R}\times S^2  
\,.
\end{equation}
 This topology will be crucial for some of the computations below,%
\footnote{e.g.\ we will employ that, given a smooth 1-form $v_A$ such that its Hodge decomposition scalars do not contain $\ell=1$-spherical harmonics, the PDE $\mcD_B w_A{}^B=v_A$ admits a unique symmetric and trace-free solution $w_{AB}$ on the round  sphere, and that there are no non-trivial harmonic 1-forms on the round sphere.}
and in that case it is  well-known \cite{geroch} that the Weyl tensor vanishes on $\scri^-$.

The initial data for the GCFE need to satisfy certain constraint equations on $\scri^-$.
Now, the frame we are dealing with will be adapted to an ordinary spacelike Cauchy problem and to the cylinder at spacelike infinity rather than to
$\scri^-$.
For this reason it  is  convenient to choose coordinates which are adapted to $\scri^-$.
Solutions to the constraint equation are then constructed in these coordinates. The frame coefficients needed for the GCFE will be computed  afterwards.

Assuming  $\lambda=0$,  let us introduce \emph{adapted null coordinates} $(\tau, r, x^{\mathring A})$
on $\scri^-\cong \mathbb{R}\times S^2$.
They are defined in such a way that $\scri^-=\{\tau=-1\}$, $r$ parameterizes the null geodesic generators of $\scri^-$,  and the $x^{\mathring A}$'s are local coordinates
on the $\Sigma_r:=\{\tau=-1, r=\mathrm{const.}\}\cong S^2$-level sets (cf.\ \cite{CCM2} for more details).
Because $\scri^-$ is a null hypersurface $g(\partial_r, \partial_r)|_{\scri^-}=0$.
Since $n=\mathrm{grad}(\tau)$ is another null vector  normal and tangent to $\scri^-$, $n$ and $\partial_r$ have to be proportional, which implies that
$g(\partial_r,\partial_{\mathring A})|_{\scri^-}=0$,
so that the metric adopts the form
\begin{equation}
g|_{\scri^-} = g_{\tau\tau}\mathrm{d} \tau^2 + 2\nu_{\tau}\mathrm{d}\tau\mathrm{d}r + 2\nu_{\mathring A}\mathrm{d}\tau\mathrm{d}x^{\mathring A}  + g_{\mathring A\mathring B} \mathrm{d}x^{\mathring A}\mathrm{d}x^{\mathring B}
\,.
\label{adapted_null_gen}
\end{equation}
Here an henceforth we use $\mathring{}$ to denote angular coordinate indices.
%We extend the adapted null coordinates off $\scri$ in such a way that $\dot x=\partial_{\tau}$.
The remaining metric coefficients are determined by the constraint equations and how the coordinates are extended off $\scri^-$.

 At each $p\in \Sigma_r$ we denote by $\ell^{\pm}$ the future-directed null vectors orthogonal to $\Sigma_r$
and normalized in such a way that $ g(\ell^+,\ell^-) = -2$. In adapted null coordinates they read
\begin{equation}
 \ell^+=\partial_r\;, \quad \ell^- = -2\nu^{\tau}\partial_{\tau} -  g^{rr}\partial_r -2  g^{r\mathring A}\partial_{\mathring A}
\;,
\label{ell_eqn}
\end{equation}
where $\nu^{\tau}:=\nu_{\tau}^{-1}$.
We  denote by $\theta^{\pm}$ the \emph{divergences} of the null hypersurfaces emanating from $\Sigma_r$ tangentially to $\ell^{\pm}$.
For the computation of $\theta^{\pm}$ it does not matter how $\ell^{\pm}$ are extended off $\scri^-$, so we may use \eq{ell_eqn} for all values of $\tau$.
A somewhat  lengthy calculation making extensively use of the formulae in \cite[Appendix~A]{CCM2} reveals that
\begin{align}
 \theta^+(r,x^{\mathring A}) &\equiv [ g^{\mu\nu} + (\ell^+)^{(\mu} (\ell^-)^{\nu)}]\nabla_{\mu}\ell^+_{\nu}|_{\Sigma_r}
\nonumber
\\
 &= \frac{1}{2}g^{\mathring A\mathring B} \partial_r g_{\mathring A\mathring B}
\;,
\label{dfn_theta+}
\\
\theta^-(r,x^{\mathring A}) &\equiv  [ g^{\mu\nu} + (\ell^+)^{(\mu} (\ell^-)^{\nu)}]\nabla_{\mu}\ell^-_{\nu}|_{\Sigma_r}
\nonumber
%\\
%&=[2\nu^0 ( \tilde\nabla_{A}-\xi_A)\nu^A
%-( \tau -6\ol \Gamma^0_{01}) \ol g^{11}
% -2\nu^0( \ol\Gamma^0_{00} +  3\ol\Gamma^1_{01} +\ol\Gamma^A_{0A})
%\nonumber
%\\
%&  +2(\nu^0)^2(\partial_{1}\ol g_{00} +\ol{\partial_{0}g_{10}})
%+2\nu^0\ol g^{1A}(2\partial_{1}\nu_A + \ol{\partial_{0}g_{1A}}  -2\chi_A{}^B\nu_B)]|_{\Sigma_r}
%\nonumber
\\
 &=2\nu^{\tau}  \rnabla^{\mathring A}\nu_{\mathring A} -\theta^+  g^{rr}  -\nu^{\tau}  g^{\mathring A\mathring B}\partial_{\tau}g_{\mathring A\mathring B}
%\nonumber
%\\
%&= 2 g^{\mathring A\mathring B}\Gamma^r_{\mathring A\mathring B} + \theta^+ g^{rr}
\;,
\label{dfn_theta-}
\end{align}
where $\not\hspace{-.25em}\nabla$ denotes the Levi-Civita connection associated to the one-parameter family $r\mapsto \not\hspace{-.2em} g =g_{\mathring A\mathring B}|_{\scri^-}\mathrm{d}x^{\mathring A}\mathrm{d}x^{\mathring B}$ on $S^2$.

Let us consider the behavior of  $\theta^{\pm}$ under conformal rescaling $\Theta\mapsto \phi\Theta$,
\begin{align}
%\nabla_{\dot x}\Theta_{\mathrm{new}}|_{\scri^-} &= \phi  \nabla_{\dot x}\Theta
%\\
\theta^+_{\mathrm{new}}
%&\equiv \frac{1}{2}\ol\phi^{-2} \ol g^{AB} \partial_1 (\ol\phi^2 \ol g_{AB})
&= \theta^+ + 2  \partial_r\log  \phi
\,,
\label{behave_+expansion}
\\
\theta^-_{\mathrm{new}} 
%&\equiv 2 \phi^{-2} g^{\mathring A\mathring B}( \Gamma_{\mathrm{new}})^r_{\mathring A\mathring B}  + \frac{1}{2} \phi^{-4} g^{rr} g^{\mathring A\mathring B} \partial_r(\phi^2 g_{\mathring A\mathring B})
%\\
&=
  \phi^{-2}\Big(\theta^-   -4 \nu^{\tau}\partial_{\tau}\log\phi -4   g^{r\mathring A}\partial_{\mathring A}\log\phi - 2  g^{rr} \partial_r\log\phi\Big)
\,.
\label{behave_-expansion}
\end{align}
We conclude that for a timelike congruence of conformal geodesics the freedom to prescribe $\phi|_{\scri^-}$
and  $\nabla_{\dot x} \phi|_{\scri^-}$ can be employed to prescribe the divergences  $\theta^{\pm}$ on $\scri^-$.
However, we observe that \eq{behave_+expansion} does not fully determine $\phi$ since it leaves the  gauge freedom   $\phi\mapsto \alpha( x^{\mathring A})\phi$.
For this reason it is more convenient to employ this gauge freedom to prescribe $\nabla_{\dot x}\Theta|_{\scri^-}$, which transforms as%
\footnote{
%Instead of $\theta^+$   one may also prescribe $\nabla_{\dot x}\Theta|_{\scri^-}$.
One may think that, instead of $\theta^-$, it should be possible to exploit the gauge  freedom   $\nabla_{\dot x} \phi|_{\scri^-}$  to prescribe the function
$\nabla_{\dot x}\nabla_{\dot x}\Theta|_{\scri^-}$. However, this does  not work  as will become clear later.
The reason for this basically  is that  $\nabla_{\dot x}\nabla_{\dot x}\Theta|_{\scri^-}$ involves a transverse derivative of $\dot x$
which is only determined by the conformal geodesics equations. Instead, $\nabla_{\dot x}\nabla_{\dot x}\Theta|_{\scri^-}$ can be identified with a certain gauge freedom
to choose coordinates on the initial surface
(cf.\ \eq{expr_Theta2}).
}
\begin{equation}
\nabla_{\dot x}\Theta_{\mathrm{new}}|_{\scri^-} = \phi  \nabla_{\dot x}\Theta
\,.
\end{equation}
Let us merely  remark that for $\lambda>0$ a corresponding analysis of the behavior of the Ricci scalar $ R^{(3)}$ of the induced metric 
%$h$  on $\scri^-=\{\tau=-1, x^{\alpha}\}$, $\alpha=1,2,3$,
and  the mean curvature $K$ on $\scri^-$
%$K=\frac{1}{2}g^{\alpha\beta}\mcL_{\dot x}g_{\alpha\beta}|_{\scri^-}$ 
under conformal rescalings shows that these two functions can be identified as gauge degrees of freedom.

By way of summary,  a   gauge  based on a congruence of timelike conformal geodesics, which requires the equations \eq{f_eqn1}-\eq{equiv_relation_Theta2} to be fulfilled,  comes along with the additional gauge  freedom to prescribe 
\begin{align}
\dot x|_{\scri^-}\,, \quad  f_{ \scri^- } \,, \quad \widehat \nabla_{\dot x}\Theta|_{\scri^-}>0
 \;, \quad \theta^-
\quad \text{for}  \quad   \lambda=0
\,,
\label{rel_gauge_data}
\\
\dot x|_{\scri^-}\,, \quad  f_{ \scri^- } \,, \quad    R^{(3)} 
 , \quad K
 \quad \text{for}  \quad   \lambda>0
\,,
\label{rel_gauge_data1b}
\end{align}
in either cases we choose  $g (\dot x, \dot x)|_{\scri^-} <0$.

%%
%\tim{reword footnote!!!}
%\footnote{
%When constructing this gauge from some given spacetime $(g,\Theta)$, one prescibes the data  $ f_* +  \mathrm{d}\log(\nabla_{\dot x}\Theta)_* $
%to solve the conformal geodesics equations, because then, after conformal rescaling, one ends up with the desired data  $f_*$.
%}

We introduce  the initial frame field as follows: Let $e_{0*}$ be a future-directed timelike vector field on $\scri^-$, and denote by $e_{a*}$, $a=1,2,3$, spacelike
vectors which complement $e_{0*}$ to an orthonormal frame $g_*(e_{i*},e_{j*})=\eta_{ij}$.
For $\lambda=0$,
to obtain $(e_{a*})$ we
consider any 2-sphere which is transversally intersected
by the null geodesic generators of $\scri^-$. We choose two spacelike orthonormal frame vectors $e_A$, $A=2,3$ tangent to that sphere which we complement  by another spacelike vector $e_1$ and a future-directed timelike vector $e_0$ to an orthonormal frame on this 2-sphere.
Parallel transport along the null geodesic generators of $\scri^-$ yields an orthonormal frame on $\scri^-$.
%%
%\begin{equation}
%%\phi_*^2
%g_*(e_{i*},e_{j*})=\eta_{ij}
%\,.
%\end{equation}
%

The  frame field  is then   parallely  propagated w.r.t.\ $\widehat\nabla$  along the congruence of conformal geodesics, i.e.\ it is required to satisfy   the transport equation
\begin{equation}
\widehat\nabla_{\dot x} e_k= 0
\,,
\label{eqn_frame}
\end{equation}
starting from the initial frame field $e_{k*}$ on $\scri^-$ (i.e.\  it is Fermi-propagated w.r.t.\ $\nabla$). 
 Then, by \eq{conf_gauge1}, 
\begin{equation}
%\phi^2 
g(e_{i},e_{j})=\eta_{ij}
\,.
\end{equation}
Finally, a coordinate system is  obtained as follows: We choose $x^0 =\tau $. Moreover, let  $(x^{\alpha})$, $\alpha=1,2,3$, be  local coordinates on $\scri^-$
(for $\lambda=0$ we will take adapted null coordinates $(r, x^{\mathring A})$).
The coordinates $(x^{\alpha})$ are then dragged along the conformal geodesics.

For given  ``conformal gauge data''
\eq{rel_gauge_data} and \eq{rel_gauge_data1b}, respectively, 
at least locally
a conformal gauge which satisfies \eq{conf_geob}, \eq{conf_gauge1} and \eq{eqn_frame}
can be constructed:
Through each point $x_*\in\scri^-$  there exists a unique solution  
$$\tau\mapsto (x(\tau),f(\tau),\Theta(\tau), e_k(\tau))$$
which yields a smooth orthonormal frame field $e_k$, conformal factor $\Theta$, and a coordinate system.
The parameter $\tau$  along the conformal geodesics is chosen such that  $\scri^- =\{ \tau =-1\}$.
The so-obtained conformal geodesics define in some neighborhood of $\scri^-$  a smooth caustic-free congruence.

Coordinates, frame field, and conformal factor constructed this way are said to form a \emph{conformal Gauss gauge} (cf.\ \cite{F_AdS, F_i0}).
There still remains some gauge freedom which arises from the freedom to choose coordinates $x^{\alpha}$  and frame field $e_{a}$
 on $\scri^-$. This freedom will be addressed below.
In a conformal Gauss gauge the following relations are fulfilled,
\begin{equation}
\dot x= e_0=\partial_{\tau} \,, \quad g(e_i,e_j)=\eta_{ij}\,, \quad \hat L_{0k}=0\,, \quad \hat\Gamma_0{}^k{}_j =0
\,.
\label{0_gauge_conditions}
\end{equation}

\subsubsection{Gauge freedom associated to  $\nabla_{\dot x}\Theta|_{\scri^-}$ and $\theta^-$}
Geometrically the freedom to prescribe $\dot x|_{\scri^-}:= \dot x_*$ and $f_{\scri^-}:= f_*$ clearly corresponds to the choice of a congruence of conformal geodesics.
As the remaining conformal gauge data we have identified  $\nabla_{\dot x}\Theta|_{\scri^-}=:\Theta^{(1)}$ and $\theta^-$ (which are supplemented by the gauge freedom to choose  frame and  coordinates).

Recall Lemma~\ref{lemma_repara_conf_geods}.
The transformation $\dot x_*\mapsto \alpha_*\dot x_*$ and $f_*\mapsto f_* + \psi_*   g( \dot x,\cdot)$
corresponds only to a change of the parameterization of the conformal geodesics,
i.e.\ these transformations  locally preserve the point set spread out by each conformal geodesic in the congruence.
The first transformation will in general violate $g(\dot x, \dot x)=-1$ and thus require a conformal rescaling of $g$ as well, $g\mapsto \alpha^{-2}g$. Applying both transformations yields $\Theta^{(1)}\mapsto \alpha \Theta^{(1)}$ and we observe that $\Theta^{(1)}$  can be any positive prescribed function.

The second transformation will in general lead to a violation of the gauge condition $\langle \dot x, f\rangle=0$ and therefore also requires a conformal transformation
$g\mapsto \phi^2 g$ with 
\begin{equation}
\nabla_{\dot x}\phi = \phi \langle \dot x, f\rangle + \psi_* \phi g(\dot x, \dot x)
\,.
\label{metric_trans}
\end{equation}
Since $ \phi_* g(\dot x_*, \dot x_*)$ is non-zero in our setting, the initial datum $\phi_*$ can be adjusted in such a way that $\nabla_{\dot x}\phi  |_{\scri^-}$ and thus $\theta^-$ becomes any prescribed function.

The freedom to prescribe $\Theta^{(1)}$ and $\theta^-$ is in this sense related to  the freedom to choose a parameterization of the conformal geodesics.

\subsubsection{Some crucial relations}

The GCFE have been formulated in terms of the gauge fields $\Theta$ and $b\equiv \Theta f + \mathrm{d}\Theta$.
These are determined by the gauge conditions \eq{f_eqn1}-\eq{equiv_relation_Theta2} which, expressed in terms of
$\Theta$ and  $b$,  read
\begin{equation}
\widehat\nabla_{\dot x} \dot x  =0
\,,
\quad
 \widehat L(\dot x, \cdot) =0
\,,
\quad
\widehat\nabla_{\dot x}\Theta =\langle \dot x, b\rangle \,.
\label{eqn_Theta}
\end{equation}
In principle these equations need to be employed to supplement the GCFE  to a closed system.
However, a very remarkable result by Friedrich \cite{F_AdS} shows that this is not necessary.
The fields $\Theta$ and  $b$ can be \emph{explicitly} determined in the conformal Gauss gauge, the latter one  in terms of its frame components $b_k=e^{\mu}{}_k{}b_{\mu}$ of a parallely propagated w.r.t.\ $\widehat\nabla$  frame $e_k$.
The corresponding expressions can then be simply inserted into the GCFE.
%For $\lambda=0$ one is thus led to the evolution equations \eq{evolution1}-\eq{evolution9} below.

\begin{lemma}[\cite{F_AdS}]
\label{lemma_b_theta}
In the conformal Gauss gauge the following relations hold:
\begin{enumerate}
\item[(i)] $\nabla_{\dot x}\nabla_{\dot x}\nabla_{\dot x}\Theta = 0$, and
\item[(ii)] $ \nabla_{\dot x}\nabla_{\dot x}b_k  = 0$, where $b_k\equiv \langle b, e_k\rangle$.
\end{enumerate}
\end{lemma}

\begin{remark}
{\rm
For $\lambda=0$  the initial data for these ODEs in terms of the ``gauge data'' \eq{rel_gauge_data} at $\scri^-$ are computed    in Section~\ref{sec_b_theta} below.
}
\end{remark}

\begin{remark}
{\rm
The conformal Gauss gauge is distinguished by the fact that it is geometric and deeply intertwined with the
conformal structure.
It has the   decisive  property to  supply explicit knowledge  about the fields $\Theta$ and $b_k$.
These fields  can be computed explicitly along any conformal geodesic of the congruence at hand.
In particular, one gains an a priori knowledge  of the location of $\scri^+$ (supposing that the solution extends that far).
}
\end{remark}

\begin{proof}
In terms of the Levi-Civita connection $\nabla$ \eq{eqn_frame} and \eq{eqn_Theta} read
\begin{align}
\Theta\nabla_{\dot x} \dot x  &= %g(\dot x,  \dot x)  
-g^{\sharp}(b - \mathrm{d}\Theta ,\cdot)
\,,
\label{constraint1b}
\\
\Theta \nabla_{\dot x}(b-  \mathrm{d}\Theta ) &=
  \langle \dot x, d\rangle (b -  \mathrm{d}\Theta )  -  \frac{1}{2} g^{\sharp}(b- \mathrm{d}\Theta ,b- \mathrm{d}\Theta) g(\dot x,\cdot)
+ \Theta^2 L(\dot x, \cdot)
\,,
\label{constraint2b}
\\
\nabla_{\dot x}\Theta &= \langle \dot x, b\rangle
\,,
\label{constraint3b}
\\
\Theta\nabla_{\dot x} e_k 
&= -\langle  b -  \mathrm{d}\Theta, e_k\rangle \dot x +   g(\dot x, e_k)  g^{\sharp}(b -  \mathrm{d}\Theta,\cdot) 
\,,
\label{constraint4b}
\end{align}
where $g^{\sharp}$ denotes the inverse metric.
We contract \eq{constraint2b} with $\dot x$ and $b$, respectively. Using also \eq{constraint1b} we deduce
\begin{align}
 \Theta(\langle\dot x, \nabla_{\dot x}b \rangle
-  \nabla_{\dot x} \nabla_{\dot x}\Theta )
&=
%g(\dot x,\dot x)
-\frac{  1}{2}\Big(g^{\sharp}( \mathrm{d}\Theta , \mathrm{d}\Theta)
-   g^{\sharp}(b ,b)  \Big)
+ \Theta^2 L(\dot x, \dot x)
\,,
\label{constr3_contrx}
\\
\Theta g^{\sharp}(b,\nabla_{\dot x}(b -  \mathrm{d}\Theta))
&=
  \frac{ \langle \dot x,b\rangle }{2}\Big(g^{\sharp}(b ,b)
- g^{\sharp}( \mathrm{d}\Theta , \mathrm{d}\Theta)  \Big)
+ \Theta^2g^{\sharp}(b,\cdot)  L(\dot x, \cdot)
\,.
\label{constr3_contrd}
\end{align}
On the other hand, it follows from \eq{constraint1b} and \eq{constraint3b} that
\begin{align}
\Theta\nabla_{\dot x} \nabla_{\dot x}\Theta& = \Theta\nabla_{\dot x} \langle \dot x, b\rangle=
 \langle\Theta\nabla_{\dot x} \dot x, b\rangle + \Theta\langle \dot x, \nabla_{\dot x} b\rangle
\\
&=
  %g(\dot x,  \dot x)
-  g^{\sharp}(b ,b)
+ %  g(\dot x,  \dot x) 
  g^{\sharp}( \mathrm{d}\Theta ,b)
+ \Theta \langle \dot x, \nabla_{\dot x} b\rangle
\,.
\end{align}
Combined, that yields
\begin{equation}
\Theta^2 L(\dot x, \dot x)
=
 %g(\dot x,  \dot x) %%
\frac{1}{2}  g^{\sharp}(b ,b)
-     g^{\sharp}( \mathrm{d}\Theta ,b)
+ \frac{1}{2}g^{\sharp}( \mathrm{d}\Theta , \mathrm{d}\Theta)
\,.
\label{schouten1}
\end{equation}
The CFE (\eq{conf3} and \eq{conf5} in Appendix~\ref{app_ADM}, cf.\ \cite{F3}) imply
\begin{equation}
\nabla_{\mu}\nabla_{\nu}\Theta = -\Theta L_{\mu\nu} + \frac{1}{2}\Theta^{-1}\Big(\nabla_{\alpha}\Theta\nabla^{\alpha}\Theta + \frac {\lambda}{3}\Big) g_{\mu\nu}
\,.
\label{cfe_relation}
\end{equation}
Contracting this twice with $\dot x$ and using \eq{constraint1b} and \eq{schouten1} yields
\begin{equation}
\nabla_{\dot x}\nabla_{\dot x}\Theta
  =g(\nabla_{\dot x}\dot x,\mathrm{d}\Theta) -\Theta L(\dot x,\dot x) -\frac{1}{2}\Theta^{-1}\Big( g^{\sharp}(\mathrm{d}\Theta,\mathrm{d}\Theta)
+ \frac{\lambda}{3}\Big)% g(\dot x, \dot x)
=-\frac{1}{2}\Theta^{-1}%  g(\dot x,  \dot x) 
\Big(  g^{\sharp}(b ,b)+ \frac{\lambda}{3} \Big) 
\,,
\label{2nd_deriv_Theta}
\end{equation}
while contraction with $\dot x$ and $b$ leads to
\begin{equation}
g^{\sharp}(b,\nabla_{\dot x} \mathrm{d}\Theta)
= -\Theta g^{\sharp}(b,\cdot)  L(\dot x, \cdot) +\frac{1}{2}\Theta^{-1} \Big( g^{\sharp}(\mathrm{d}\Theta,\mathrm{d}\Theta) + \frac{\lambda}{3} \Big) \langle \dot x, b\rangle
\,.
\end{equation}
We insert the latter equation into \eq{constr3_contrd},
\begin{equation}
 \nabla_{\dot x} g^{\sharp}(b,b )
= 
 \Theta^{-1}\langle \dot x,b\rangle\Big( g^{\sharp}(b ,b)
+   \frac{\lambda}{3} 
\Big)
\,.
\label{deriv_d^2}
\end{equation}
Taking now  the derivative of \eq{2nd_deriv_Theta} along the conformal geodesics and using \eq{constraint3b}
 and \eq{deriv_d^2}, we obtain (i).

Next, we insert \eq{cfe_relation}, contracted with $\dot x$, into \eq{constraint2b},
\begin{equation}
\Theta\nabla_{\dot x}b =
 \Big(g^{\sharp}(b  ,\mathrm{d}\Theta) -  \frac{1}{2} g^{\sharp}(b ,b) + \frac{\lambda}{6} \Big)g(\dot x, \cdot) 
+  \langle \dot x, b\rangle (b - \mathrm{d}\Theta ) 
\,.
\label{deriv_d}
\end{equation}
%
%whence, with \eq{constraint1b} and \eq{constraint3b},
%%
%\begin{equation}
%\nabla_{\dot x}\langle\dot x,b \rangle
%= 
%-\frac{1}{2}\Theta^{-1} \Big(   g^{\sharp}(b ,b) + \frac{\lambda}{3} \Big)
%%g(\dot x, \dot x) 
%\,.
%\end{equation}
%%
%
From \eq{deriv_d} and \eq{constraint4b} we deduce
\begin{equation}
\nabla_{\dot x}\langle b, e_k\rangle 
 =
\frac{1}{2} \Theta^{-1}\Big(  g^{\sharp}(b ,b) 
 + \frac{\lambda}{3} \Big) g(\dot x, e_k)  
\,.
\label{deriv_dx}
\end{equation}
Differentiating this one more time along the conformal geodesics and using   \eq{constraint3b} and  \eq{deriv_d^2},  we find (ii).
\qed
\end{proof}

\subsection{Evolution equations for Schouten tensor, connection and frame coefficients}
\label{sec_evolution1}

In the conformal Gauss gauge the GCFE split into evolution and constraint equations.
The constraint equations will be analyzed in Section~\ref{sec_solution_constr_gen}
for $\lambda=0$ and in adapted null coordinates (and then rewritten in terms of frames).
Here we want to derive a somewhat more explicit form of the evolution equations.
Let us start with Schouten tensor, connection coefficients and frame field.
Recall  \eq{0_gauge_conditions}, so we do not need equations for $\widehat L_{0i}$, $\widehat\Gamma_0{}^i{}_j$ and $e^{\mu}{}_0$.
The GCFE \eq{GCFE_1}-\eq{GCFE_3} imply the following system of evolution equations for the remaining components (cf.\ e.g.\ \cite{F3}),
\begin{align}
\partial_{\tau}\widehat L_{aj} 
&= b_i W^{i}{}_{j0a}-  \widehat\Gamma_a{}^b{}_0\widehat L_{bj} 
\,,
\\
\partial_{\tau}\widehat\Gamma_{a}{}^i{}_j 
 &=
-  \widehat\Gamma_b{}^i{}_j\widehat\Gamma_{a}{}^b{}_{0} + 
2\delta_{(0}{}^i\widehat L_{|a|j)}   - \eta_{j0}\widehat L_{a}{}^i +\Theta W^i{}_{j0a}
\,,
\\
\partial_{\tau}e^{\mu}{}_a&= -\widehat\Gamma_{a}{}^k{}_{0} e^{\mu}{}_k
\,.
\label{evolution_frame}
\end{align}
The Levi-Civita connection satisfies $\Gamma_{i(jk)}=0$, equivalently, $\widehat\Gamma _{i(jk)} = \eta_{jk}f_i$. If follows that the Weyl connection
has the following (anti-)symmetric properties, we will make extensively use of,
\begin{eqnarray}
\widehat\Gamma_a{}^1{}_0=\widehat\Gamma_a{}^0{}_1\,, \quad \widehat\Gamma_a{}^0{}_0=\widehat\Gamma_a{}^1{}_1=\frac{1}{2}\widehat\Gamma_a{}^A{}_A\,, \quad\eta_{AB}\widehat\Gamma_a{}^B{}_1=-\widehat\Gamma_a{}^1{}_A\,, 
\quad \eta_{AB}\widehat\Gamma_a{}^B{}_0=\widehat\Gamma_a{}^0{}_A\,.
\label{relations_connection1}
\end{eqnarray}
As the relevant independent components
 of the Weyl connection one may regard
\begin{equation}
\widehat\Gamma_a{}^0{}_b
\,,\quad
\widehat\Gamma_a{}^1{}_b
\,,\quad
\widehat\Gamma_{a[BC]}
\,.
\end{equation}
%
%(in fact one merely needs $\widehat\Gamma_{a[BC]}$).
Using the algebraic symmetries of the Weyl tensor (cf.\ the next section) one  ends up with a system of evolution equations for Schouten tensor, connection  and frame coefficients,
\begin{align}
\partial_{\tau}\widehat L_{a0} 
&= b_i W^{i}{}_{00a}-  \widehat\Gamma_a{}^b{}_0\widehat L_{b0} 
\,,
\label{evolution1}
\\
\partial_{\tau}\widehat L_{ab} 
&= b_i W^{i}{}_{b0a}-  \widehat\Gamma_a{}^c{}_0\widehat L_{cb} 
\,,
\\
\partial_{\tau}\widehat\Gamma_{a}{}^0{}_b 
 &=
-  \widehat\Gamma_c{}^0{}_b\widehat\Gamma_{a}{}^c{}_{0} + \widehat L_{ab}   -\Theta W_{0a0b}
\,,
\\
\partial_{\tau}\widehat\Gamma_{a}{}^1{}_b 
 &=
-  \widehat\Gamma_c{}^1{}_b\widehat\Gamma_{a}{}^c{}_{0} +\delta^1{}_{b}\widehat L_{a0}   -\Theta W_{0ab1}
\,,
\\
\partial_{\tau}\widehat\Gamma_{1}{}^A{}_B 
 &=
-  \widehat\Gamma_c{}^A{}_B\widehat\Gamma_{1}{}^c{}_{0} + \delta^A{}_{B}\widehat L_{10}   +\Theta W^A{}_{B01}
\,,
\\
\partial_{\tau}\widehat\Gamma_{A}{}^B{}_C 
 &=
-  \widehat\Gamma_d{}^B{}_C\widehat\Gamma_{A}{}^d{}_{0} + \delta^B{}_{C}\widehat L_{A0}   -2\Theta \eta_{A[B}W_{C]110}
\,,
\\
\partial_{\tau}e^{\mu}{}_a&= -\widehat\Gamma_{a}{}^0{}_{0} \delta^{\mu}{}_0
 -\widehat\Gamma_{a}{}^b{}_{0} e^{\mu}{}_b
\,.
\label{evolution7}
\end{align}

\subsection{Bianchi equation}
\label{sec_evolution2}

%\subsubsection{Independent components of the Weyl tensor}

As 10 independent components of the rescaled Weyl tensor in an orthonormal frame one can identify
(``$\mathrm{tf}$'' denotes the trace-free part w.r.t.\ to the $(AB)$-``angular''-components),
\begin{equation}
W_{0101}\,, \quad
W_{011A}\,, \quad
W_{010A}\,, \quad
W_{01AB} \,, \quad 
( W_{1A1B})_{\mathrm{tf}}\,,\quad
( W_{0(AB)1})_{\mathrm{tf}} \,.
\label{independent_Weyl}
\end{equation}
The remaining components  are related to these ones in the following way:
\begin{align}
W_{0[AB]1} = -\frac{1}{2}W_{01AB}
\,,
\label{Weyl_components1}
\quad
\eta^{AB}W_{0AB1} = 0
\,,
\\
W_{0ABC} =
-2W_{011[C}\eta_{B]A}
\,,
\quad
W_{1ABC} =
 -2W_{010[C}\eta_{B]A}
\,,
\\
\eta^{AB}W_{0A0B} = -W_{0101}
\,,
\quad
( W_{0A0B} )_{\mathrm{tf}}
=
 ( W_{1A1B} )_{\mathrm{tf}}
\,,
\\
W_{ABCD} 
=2\eta_{C[B}\eta_{A]D}W_{0101}
\,,
\quad
\eta^{AB}W_{1A1B} = W_{0101}
\,,
\label{Weyl_components8}
\end{align}
where we have employed all the algebraic symmetries of the  Weyl tensor.
It is convenient to make the following definitions:
\begin{align}
V^{\pm}_{AB}& :=( W_{1A1B})_{\mathrm{tf}}\pm  W_{0(AB)1}
\,,
\label{rel_comp1}
\\
W^{\pm}_A &:=W_{010A}\pm W_{011A}
\,,
\\
U_{AB} &:= W_{01AB}  + \eta_{AB} W_{0101}
\,,
\label{rel_comp3}
\end{align}
which capture all independent components. We will use $U_{AB}$ (instead of $W_{0101}$ and $W_{01AB}$)  only occasionally.

%\subsubsection{Evolution and constraint  equations}
Let us consider the Bianchi equation \eq{GCFE_4}.
The independent components are provided by the $j=a$ components.
Expressed in terms of the connection coefficients these components read
\begin{align}
\partial_{\tau}W_{0a0b}     =&  e^{\mu}{}_c\partial_{\mu} W^c{}_{a0b}  - \widehat\Gamma_c{}^c{}_0 W_{0a0b} 
+ \widehat\Gamma_c{}^c{}_d W^d{}_{a0b} 
- \widehat\Gamma_c{}^0{}_a W^c{}_{00b} 
- \widehat\Gamma_c{}^d{}_a W^c{}_{d0b} 
\nonumber
\\
&
-2 \widehat\Gamma_c{}^0{}_{0} W^c{}_{a0b} 
- \widehat\Gamma_c{}^d{}_{0} W^c{}_{adb} 
+ \widehat\Gamma_c{}^d{}_{b} W^c{}_{ad0} 
\,,
\label{Bianchi equation1}
\\
\partial_{\tau}W_{0abc}     =&  e^{\mu}{}_d\partial_{\mu} W^d{}_{abc}  - \widehat\Gamma_d{}^d{}_0 W_{0abc} 
+ \widehat\Gamma_d{}^d{}_e W^e{}_{abc} 
- \widehat\Gamma_d{}^0{}_a W^d{}_{0bc}
- \widehat\Gamma_d{}^e{}_a W^d{}_{ebc}  
\nonumber
\\
&
- \widehat\Gamma_d{}^0{}_{b} W^d{}_{a0c}
+ \widehat\Gamma_d{}^0{}_{c} W^d{}_{a0b}
- \widehat\Gamma_d{}^e{}_{b} W^d{}_{aec}
+ \widehat\Gamma_d{}^e{}_{c} W^d{}_{aeb}
-  \widehat\Gamma_d{}^0{}_0W^d{}_{abc} 
\,.
\label{Bianchi equation2}
\end{align}
%
%Not all of them are evolution equations, nevertheless 
We will need all of them for our  analysis of the critical sets where spatial infinity touches null infinity.
For this we rewrite them in terms of \eq{rel_comp1}-\eq{rel_comp3}. 
It is convenient to define the operator  $\check\nabla$ as follows,
\begin{equation}
\check\nabla_Av_B := e^{\mu}{}_A\partial_{\mu}v_B -\widehat\Gamma_A{}^C{}_Bv_C
\,.
\end{equation}
and similarly for  tensors of higher valence.
From a lengthy calculation we obtain,
\begin{align}
\partial_{\tau}W_{0101} 
 =& 
\Big(-\frac{1}{2}\check\nabla^A
+2 \widehat\Gamma^{A0}{}_{0}
+\frac{1}{2} \widehat\Gamma^{A0}{}_1\Big) W^{+}_A
+\Big(\frac{1}{2}\check\nabla^A
-2 \widehat\Gamma^{A0}{}_{0}
+\frac{1}{2} \widehat\Gamma^{A0}{}_1\Big)W^{-}_A
\nonumber
\\
&
-\frac{3}{2} \widehat\Gamma_A{}^A{}_0W_{0101}
-\frac{3}{2} \widehat\Gamma^{AB}{}_1W_{01AB}
- \frac{1}{2}(\widehat\Gamma^{AB}{}_{1}+ \widehat\Gamma^{AB}{}_{0})V^+_{AB}
\nonumber
\\
&
+ \frac{1}{2}(\widehat\Gamma^{AB}{}_{1}-\widehat\Gamma^{AB}{}_{0})V^-_{AB}
\label{evolutionW1b}
\,,
\\
\partial_{\tau}W_{01AB}     =&
(\check\nabla_{[A}  
-  2\widehat\Gamma_{[A}{}^0{}_{|0|} 
- \widehat\Gamma_{[A}{}^0{}_{|1|} )W^{+}_{B]}
+(\check\nabla_{[A} 
-  2\widehat\Gamma_{[A}{}^0{}_{|0|} 
+ \widehat\Gamma_{[A}{}^0{}_{|1|}) W^{-}_{B]}
\nonumber
\\
&
-3\widehat\Gamma_{[A}{}^{1}{}_{B]} W_{0101}
  - \frac{3}{2}\widehat\Gamma_C{}^C{}_0 W_{01AB} 
+( \widehat\Gamma^{C1}{}_{[A}- \widehat\Gamma^{C0}{}_{[A})V^+_{B]C}
\nonumber
\\
&
+( \widehat\Gamma^{C1}{}_{[A}+ \widehat\Gamma^{C0}{}_{[A}) V^-_{B]C}
\label{evolutionW2b}
\,,
\\
\partial_{\tau}W^-_A
 =& ( \check\nabla^B
-3 \widehat\Gamma^{B0}{}_{0}
+2\widehat\Gamma^{B0}{}_1)V^+_{AB}
 +\frac{1}{2}\check\nabla^B U_{BA}
-\frac{3}{2} \widehat\Gamma^{B0}{}_{0} U_{BA}
\nonumber
\\
&
+ \Big( \widehat\Gamma_{[A}{}^0{}_{B]} 
+3\widehat\Gamma_{[A}{}^{1}{}_{B]} 
 \Big) W^{-B}
+\frac{1}{2}\Big( \widehat\Gamma_B{}^B{}_{1}  
 -3 \widehat\Gamma_B{}^B{}_0 \Big)W^-_{A} 
\nonumber
\\
&
+\Big( \widehat\Gamma_{(A}{}^0{}_{B)}
+ \widehat\Gamma_{(A}{}^1{}_{B)} 
\Big) W^{+B}
+\frac{1}{2}\Big( \widehat\Gamma_B{}^B{}_{1}  - \widehat\Gamma_B{}^B{}_0\Big) W^+_A
\,,
\label{evolutionW3b}
\end{align}
\begin{align}
\partial_{\tau}W^+_A
 =&( -  \check\nabla^B
+3 \widehat\Gamma^{B0}{}_{0}
+2\widehat\Gamma^{B0}{}_1)V^-_{AB}
 -\frac{1}{2}\check\nabla^B U_{AB}
+\frac{3}{2} \widehat\Gamma^{B0}{}_{0} U_{AB} 
\nonumber
\\
&
+ \Big( \widehat\Gamma_{[A}{}^0{}_{B]}
-3\widehat\Gamma_{[A}{}^{1}{}_{B]} \Big) W^{+B}
 - \frac{1}{2}\Big(3\widehat\Gamma_B{}^B{}_0 
+ \widehat\Gamma_B{}^B{}_{1}\Big) W^+_A
\nonumber
\\
&
+\Big(   \widehat\Gamma_{(A}{}^0{}_{B)}
-\widehat\Gamma_{(A}{}^1{}_{B)} 
 \Big)W^{-B}
-\frac{1}{2}\Big( \widehat\Gamma_B{}^B{}_{1} +\widehat\Gamma_B{}^B{}_0  \Big)W^-_A
\,,
\label{evolutionW4b}
\\
(e^{\tau}{}_1 + 1)\partial_{\tau}  V^+_{AB}
 =&
- e^{\alpha}{}_1\partial_{\alpha} V^+_{AB}
+( \widehat\Gamma_1{}^0{}_{0}
- 2\widehat\Gamma_1{}^1{}_{0} 
- \widehat\Gamma_C{}^C{}_0
- \widehat\Gamma_C{}^C{}_1 )V^+_{AB}
\nonumber
\\
&
+ \Big( ( \widehat\Gamma_C{}^0{}_{(A} - \widehat\Gamma_C{}^1{}_{(A}+2 \widehat\Gamma_{1C(A} )V^+_{B)}{}^{C}
- \frac{3}{2}(\widehat\Gamma_C{}^0{}_{(A} + \widehat\Gamma_C{}^1{}_{(A}) U^C{}_{B)} 
\nonumber
\\
&
+ (\check\nabla_{(A} 
+ 2\widehat\Gamma_1{}^0{}_{(A}
-2 \widehat\Gamma_{(A}{}^0{}_{|0|}
+ \widehat\Gamma_{(A}{}^0{}_{|1|}
+ 2\widehat\Gamma_1{}^1{}_{(A} )W^-_{B)}
\Big)_{\mathrm{tf}}
\,,
\label{evolutionW5b}
\\
( e^{\tau}{}_1-1)\partial_{\tau}
 V^-_{AB} 
=&
-e^{\alpha}{}_1\partial_{\alpha} V^-_{AB} 
+ (\widehat\Gamma_1{}^0{}_0+ 
2 \widehat\Gamma_1{}^1{}_{0} 
  + \widehat\Gamma_C{}^C{}_0 
- \widehat\Gamma_C{}^C{}_1 ) V^-_{AB} 
\nonumber
\\
&
+ \Big( -( \widehat\Gamma_C{}^0{}_{(A} + \widehat\Gamma_C{}^1{}_{(A}- 2\widehat\Gamma_{1C(A})V^-_{B)}{}^C
+ \frac{3}{2}( \widehat\Gamma_C{}^0{}_{(A} 
-  \widehat\Gamma_C{}^1{}_{(A} )U_{B)}{}^C
\nonumber
\\
&
+(\check\nabla_{(A}
- 2\widehat\Gamma_1{}^0{}_{(A} 
- 2 \widehat\Gamma_{(A}{}^0{}_0
- \widehat\Gamma_{(A}{}^0{}_{|1|}
+2 \widehat\Gamma_1{}^1{}_{(A} )W^+_{B)}
\Big)_{\mathrm{tf}}
\,,
\label{evolutionW6b}
\end{align}
and,
\begin{align}
(\partial_{\tau} -e^{\mu}{}_1\partial_{\mu}) W_{0101}
 =& 
\Big(\check\nabla^A
-4 \widehat\Gamma^{A0}{}_{0}
+\widehat\Gamma^{A0}{}_1- \widehat\Gamma_1{}^A{}_0 
-\widehat\Gamma_1{}^A{}_{1} \Big)W^{-}_A
+\Big( \widehat\Gamma_1{}^A{}_0 
-\widehat\Gamma_1{}^A{}_{1}\Big)W^{+}_{A}
\nonumber
\\
&
-(  \widehat\Gamma_A{}^B{}_1+ \widehat\Gamma_A{}^B{}_0)V^+_{A}{}^B
+\frac{3}{2} (\widehat\Gamma_A{}^B{}_{0} - \widehat\Gamma_A{}^B{}_{1} )W^A{}_{B01}
\nonumber
\\
&
-3\Big(  \widehat\Gamma_1{}^0{}_0+\frac{1}{2} \widehat\Gamma_A{}^A{}_0 -\frac{1}{2}\widehat\Gamma_A{}^A{}_{1} \Big)W_{0101}
\label{evolution_WA-_2A}
\,,
\\
(\partial_{\tau} -  e^{\mu}{}_1\partial_{\mu}) W_{01AB}   =&
2\Big(\check\nabla_{[A} 
-  2\widehat\Gamma_{[A}{}^0{}_{|0|} 
+ \widehat\Gamma_{[A}{}^0{}_{|1|}- \widehat\Gamma_1{}^0{}_{[A} 
+ \widehat\Gamma_1{}^1{}_{[A}\Big) W^{-}_{B]}
-2\Big(  \widehat\Gamma_1{}^0{}_{[A} 
+ \widehat\Gamma_1{}^1{}_{[A} \Big)W^{+}_{B]}
\nonumber
\\
&
+2( \widehat\Gamma_C{}^{1}{}_{[A}- \widehat\Gamma_C{}^{0}{}_{[A})V^+_{B]}{}^{C}
-3(\widehat\Gamma_{[A}{}^0{}_{B]} +\widehat\Gamma_{[A}{}^{1}{}_{B]})W_{0101} 
\nonumber
\\
&
- 3\Big(\widehat\Gamma_1{}^0{}_{0} + \frac{1}{2}\widehat\Gamma_C{}^C{}_0 -\frac{1}{2}\widehat\Gamma_C{}^C{}_1 \Big)W_{01AB} 
\,,
\\
(\partial_{\tau}-e^{\mu}{}_1\partial_{\mu} )W^-_{A} 
 =& \Big(2 \check\nabla^B
-6 \widehat\Gamma^{B0}{}_{0}
+4\widehat\Gamma^{B0}{}_1
-\widehat\Gamma_1{}^B{}_{0} - \widehat\Gamma_1{}^B{}_{1} \Big)V^+_{AB}
+\frac{3}{2}\Big( (\widehat\Gamma_1{}^B{}_{1}- \widehat\Gamma_1{}^B{}_0  
\Big)U_{BA}
\nonumber
\\
&
-\Big(2 \widehat\Gamma_1{}^0{}_{0}  -\widehat\Gamma_1{}^1{}_0+2 \widehat\Gamma_B{}^B{}_0-2\widehat\Gamma_B{}^B{}_1 \Big)W^-_{A}
\nonumber
\\
&
+\Big( 4\widehat\Gamma_{[A}{}^0{}_{B]} 
+4\widehat\Gamma_{[A}{}^{1}{}_{B]} 
- \widehat\Gamma_{1BA}\Big) W^{-B}
\,,
\label{evolution_WA-_2AB}
\\
(\partial_{\tau}  -e^{\mu}{}_1\partial_{\mu})W^+_{A}   =&
 -\check\nabla^B U_{AB}
+\frac{3}{2}\Big(  \widehat\Gamma_1{}^0{}_{B}-\widehat\Gamma_1{}^1{}_{B}
 +2\widehat\Gamma_B{}^0{}_{0} \Big)U_{A}{}^{B}
+\Big( \widehat\Gamma_1{}^B{}_{0} -\widehat\Gamma_1{}^B{}_{1}\Big) V^-_{AB}
\nonumber
\\
&
+\Big( \widehat\Gamma_B{}^B{}_1  - \widehat\Gamma_B{}^B{}_0  - \widehat\Gamma_1{}^1{}_0 -2\widehat\Gamma_1{}^0{}_{0}\Big) W^+_{A}
-\Big( 2\widehat\Gamma_{[A}{}^0{}_{B]}+ 2\widehat\Gamma_{[A}{}^1{}_{B]}  + \widehat\Gamma_{1BA} \Big) W^{+B}
\nonumber
\\
&
 -\Big( \widehat\Gamma_B{}^B{}_0 + \widehat\Gamma_B{}^B{}_1\Big)   W^-_A
+2\Big( \widehat\Gamma_{(A}{}^0{}_{B)}-\widehat\Gamma_{(A}{}^1{}_{B)}  \Big)W^{-B}
\,.
\label{evolution_WA-_2}
\end{align}
At times we will call \eq{evolutionW1b}-\eq{evolutionW6b} ``evolution equations'' and \eq{evolution_WA-_2A}-\eq{evolution_WA-_2},
whose evaluation on $\scri^-$ does not contain transverse derivatives,  ``constraint equations''.
However, in this work we do not attempt to solve the evolution problem or show preservation of the constraints
under evolution (for an ordinary Cauchy problem this has been done  in \cite{F_AdS}).
We therefore   do not care here whether this is the ``appropriate'' split of the Bianchi equation.

\subsection{Frame field and coordinates at $\scri^-$ for $\lambda=0$}

\subsubsection{Adapted null coordinates}
We assume henceforth that the cosmological constant vanishes, 
\begin{equation}
\lambda=0
\,.
\end{equation}
We introduce \emph{adapted null coordinates} $(\tau,r,x^{\mathring A})$ on $\scri^-\cong \mathbb{R}\times S^2$ (cf.\ their definition prior to \eq{adapted_null_gen}).
In conformal Gaussian coordinates we require, in addition, $g_{\tau\tau}=g(\dot x,\dot x)=-1$.
It is well-known that the shear tensor vanishes on $\scri$.
This implies that $g_{\mathring A\mathring B}=\Omega^2(r,x^{\mathring C})  h_{\mathring A\mathring B}$ for some $r$-independent Riemannian metric on $S^2$.
Since any smooth Riemannian metric on $S^2$ is conformal to the  standard metric  $s_{\mathring A\mathring B} \mathrm{d}x^{\mathring A}\mathrm{d}x^{\mathring B}$ we may simply assume by redefining $\Omega$ that $h_{\mathring A\mathring B}=s_{\mathring A\mathring B}$.
By way of summary, the line element  takes the following form on $\scri^-$,
\begin{equation}
 g|_{\scri^-}= - \mathrm{d}\tau ^2 + 2\nu_{\tau}(r,x^{\mathring C}) \mathrm{d}\tau\mathrm{d}r 
+ 2\nu_{\mathring A }(r,x^{\mathring C}) \mathrm{d}\tau\mathrm{d}x^{\mathring A}+\Omega^2(r,x^{\mathring C})  s_{\mathring A\mathring B}(x^{\mathring  C}) \mathrm{d}x^{\mathring A}\mathrm{d}x^{\mathring B}
\,,
\label{adapted_null}
\end{equation}
which is a regular Lorentzian metric supposing that $\nu_{\tau},\Omega \ne 0$.
Adapted null coordinates are used as ``initial coordinates'' which are then dragged along the congruence of conformal geodesics.
The following relation holds between $\Omega$ and the divergence $\theta^+$ of the null geodesic generators of $\scri^-$ \cite{CCM2},
\begin{equation}
\theta^+=2\partial_r\log\Omega
\,.
\label{div_conffactor}
\end{equation}
We have already mentioned  that there is still a gauge freedom left, namely to reparameterize  the null geodesic generators of $\scri^-$.
This gauge freedom, $r\mapsto \tilde r=\tilde r(r, x^{\mathring A})$, can be employed to prescribe the function $\kappa$ \cite{CCM2},
given by
\begin{equation}
\rnabla_{\ell}\ell=\kappa\ell\,.
\end{equation}
It  measures the deviation of the coordinate $r$ to be an affine parameter.
This does not completely fix the $r$-coordinate. The remaining gauge freedom will be considered below.
There also remains the gauge freedom to choose coordinates $(x^{\mathring A})$ on $\{\tau=-1,r=\mathrm{const.}\}\cong S^2$, whose specific choice will be  irrelevant for us.

A  list of all the relevant  Christoffel symbols, or rather their restriction to $\scri^-$, is provided in Appendix~\ref{app_charact_constraints},
\eq{coordChristoffel1}-\eq{coordChristoffel13} (recall that in conformal Gauss coordinates we, in addition, require  $g_{\tau\tau}=-1$).
We remark that \eq{dfn_kappa} and \eq{dfn_xi} may be regarded as definitions of $\kappa$ and $\xi_{\mathring A}$, while the trace-free part of 
\eq{coordChristoffel11} may be regarded as the definition of $\Xi_{\mathring A\mathring B}$. Equivalently, they can be defined as
\begin{eqnarray}
\kappa &=& \nu^{\tau}\partial_r\nu_{\tau}-\frac{1}{2}\nu^{\tau}\partial_{\tau} g_{rr}|_{\scri^-}
\,,
\\
\xi_{\mathring A} &=& -\nu^{\tau}(\partial_{\mathring A} \nu_{\tau} -\partial_{\tau} g_{r\mathring A}|_{\scri^-}+ \partial_r\nu_{\mathring A} -\theta^+\nu_{\mathring A})
\,,
\\
\Xi_{\mathring A\mathring B} &=&  \nu^{\tau} (\partial_{\tau} g_{\mathring A\mathring B}|_{\scri^-}
-2 \rnabla_{(\mathring A}\nu_{\mathring B)})_{\mathrm{tf}}
\,,
\end{eqnarray}
 where $\nu^{\tau}:= (\nu_{\tau})^{-1}$.
The indices of $\xi_{\mathring A}$ and $\Xi_{\mathring A\mathring B}$ will be  raised and lowered with $ g_{\mathring A\mathring B}|_{\scri^-}=\Omega^2 s_{\mathring A\mathring B}$.
%(Here we have set $\xi^{\mathring C}= g^{\mathring C\mathring D} \xi_{\mathring D}$.)

\subsubsection{Frame field}

We need to choose an initial  frame field $\{e_{i*}\}$ on  $\scri^-$  which satisfies
\begin{equation}
 g (e_{i*},e_{j*}) = \eta_{ij} \quad \text{and} \quad  e_{0*}=\partial_{\tau}
\,.
\label{frame_field_conds}
\end{equation}
A frame field which fulfills these requirements is provided by
\begin{align}
e_{0*} =& \partial_{\tau}
\,,
\label{frame_field1}
\\
e_{1*} =& \partial_{\tau} + \nu^{\tau}\partial_r
\,,
\label{frame_field2}
\\
e_{A*} =&   \Omega^{-1}\mathring e^{\mathring A}{}_A( \partial_{\mathring A} - \nu^{\tau}\nu_{\mathring A}\partial_r )
\,,
\label{frame_field3}
\end{align}
where $(\mathring e_A)$, $A=2,3$, denotes an orthonormal frame field on the  round sphere $\mathbb{S}^2:=(S^2, s_{\mathring A\mathring B})$.
All other frame fields which satisfy \eq{frame_field_conds} arise from this one as
\begin{equation}
\check e_{0*}= e_{0*}\,, \quad \check e_{a*}=  M(r, x^{\mathring A}) \cdot e_{a*}
\,, \quad M(r, x^{\mathring A})\in O(3)
\,.
\end{equation}
We consider a conformal Gauss gauge based on adapted null coordinates and a frame field given by \eq{frame_field1}-\eq{frame_field3}.
%\tim{results in global orthogonal transformation of the frame and has therefore no impact on the production of log terms as long as the transformation is smooth...?}

\subsubsection{Initial data for $\Theta$ and $b$}
\label{sec_b_theta}

According to Lemma~\ref{lemma_b_theta}  the conformal factor $\Theta$ and the 1-form $b$
are globally of the form
\begin{align}
\Theta =& \Theta^{(1)}(1+\tau) + \Theta^{(2)}(1+\tau)^2
\,,
\\
b_i =& b^{(0)}_i + b^{(1)}_i (1+\tau)
\,.
\end{align}
We want relate the values of the integration functions $\Theta^{(n)}=\Theta^{(n)}(r,x^{\mathring A})$ and 
$b_i^{(n)}=b_i^{(n)}(r,x^{\mathring A})$, $n=0,1$, in terms of the gauge data \eq{rel_gauge_data}.
First of all, we observe that $\Theta^{(1)}= \widehat \nabla_{\dot x}\Theta|_{\scri^-}>0$ can be directly identified as a conformal gauge freedom.

Let us express  $\Theta^{(2)}=\frac{1}{2}(\widehat\nabla_{\dot x}\widehat\nabla_{\dot x} \Theta  )|_{\scri^-}$ in terms of data on $\scri^-$. For this,  
we contract  equation \eq{cfe_relation} twice with $\dot x$  as well as  with $\dot x$ and $\ell$. Eliminating
the second term on the right-hand side  yields
\begin{equation}
\Theta^{(2)}  =- \frac{1}{2}
\nu^{\tau} \Big(
\rnabla_{\ell}-\Gamma^{\tau}_{\tau r}-\nu_{\tau} \Gamma^{\tau}_{\tau\tau}
\Big)\Theta^{(1)}
\,.
\end{equation}
Below (cf.\ \eq{frame_deriv1}) we will show that 
$\Gamma^{\tau}_{\tau\tau}|_{\scri^-}=-f^{\tau}=-\nu^{\tau}f_r$.
using also \eq{dfn_kappa2}, it follows that 
\begin{equation}
\Theta^{(2)}  =- \frac{1}{2}
 \Big(
\rnabla_{\ell}+\kappa+ \langle \ell, f\rangle
\Big)\Big(\frac{\Theta^{(1)}}{  g(\dot x_*, \ell)}\Big)
\,.
\label{expr_Theta2}
\end{equation}
The freedom to prescribe $\kappa$ can be replaced by the freedom to prescribe $\Theta^{(2)}$, whence
one  may regard $\Theta^{(2)}$ as a (coordinate) gauge freedom.

Let us consider the 1-form $b\equiv \Theta f + \mathrm{d}\Theta$.
It  follows straightforwardly from \eq{frame_field1}-\eq{frame_field3}  that 
\begin{equation}
b^{(0)}_0=\Theta^{(1)}
\,,
\quad
b^{(0)}_1= \Theta^{(1)}
\,,
\quad
b^{(0)}_A= 0
\,.
\end{equation}
To obtain the first-oder expansion coefficients we employ  \eq{evolution7},
\begin{align}
b^{(1)}_0=&2\Theta^{(2)}
\,,
\\
b^{(1)}_1=&-\widehat\Gamma_{1}{}^1{}_{0} \Theta^{(1)}+ 2\Theta^{(2)}+\nu^{\tau}\partial_{r}\Theta^{(1)}
\,,
\\
b^{(1)}_A=& -\widehat\Gamma_{A}{}^1{}_{0} \Theta^{(1)}+ e^{r}{}_A\partial_{r}\Theta^{(1)} + e^{\mathring A}{}_A\partial_{\mathring A}\Theta_1
\,.
\end{align}
Using   the formulas \eq{GammaA10_gen} and \eq{Gamma110_gen}
derived below,
\begin{align}
\widehat\Gamma_{1}{}^1{}_{0} |_{\scri^-} =& -(\partial_r + \kappa+ f_r)\nu^{\tau} 
\,,
\\
\widehat\Gamma_{A}{}^1{}_{0}  |_{\scri^-} =& \frac{1}{2}\xi_A + \nu^{\tau}\rnabla_A\nu_{\tau}
+\nu_{ A}\Big(\partial_{r}- \frac{1}{2}\theta^+ + \kappa \Big)  \nu^{\tau}
\,,
\end{align}
the constraint equations (cf.\ \eq{constraint3} and \eq{constraint_xiA}),
\begin{align}
(\partial_r -\frac{1}{2}\theta^+ + \kappa)(\nu^{\tau}\Theta_1)=&0
\,,
\\
\xi_A-2\rnabla_A\log(\nu^{\tau}\Theta_1) =&0
\,,
\end{align}
as well as \eq{expr_Theta2}, we end up with the following expressions for $b$,
\begin{equation}
b^{(1)}_0=2\Theta^{(2)}
\,,
\quad
b^{(1)}_1= 0
\,,
\quad
b^{(1)}_A= 0
\,.
\end{equation}
This shows that the frame components of $b$ are fully determined by  $\Theta^{(1)}$ and $\Theta^{(2)}$.  They do not depend on $f$. In particular,
the gauge data $f_{\scri^-}$ cannot be identified with certain components of $b$.

\subsubsection{Gauge data at $\scri^-$}
\label{sec_gauge_freedom}

Let us analyze the freedom to choose the initial direction $\dot x_*$ of the conformal geodesics somewhat more detailed.
For this let $v$ be an arbitrary timelike vector field on $\scri^-$, i.e.\ with  $v^{\mu}v_{\mu}<0$.
We introduce arbitrary  adapted null coordinates $(\tau,r,x^{\mathring A})$, 
$v|_{\scri^-}=(\ol v^{\tau}, \ol v^{r}, \ol v^{ \mathring A})$. We want to transform into new adapted null coordinates where $\hat v|_{\scri^-}=(1,0,0,0)$.
For this, we introduce  new coordinates $(\hat \tau,\hat r,\hat x^{\mathring A})$ by
\begin{equation}
1+\hat \tau =(1+\tau)/\ol v^{\tau}\,, \quad  \hat r =r - (1+\tau)\ol v^r/\ol v^{\tau}\,, \quad \hat x^{\mathring A}= x^{\mathring A} - (1+\tau) \ol v^{\mathring A}/\ol v^{\tau}
\,.
\end{equation}
Note that the $r$- and $x^{\mathring A}$-coordinates remain unchanged on $\scri^-=\{\tau=-1\}$ under such a transformation.
We find that
\begin{align}
\hat v^{\tau}|_{\scri^-}  =&  \frac{\partial \hat \tau}{\partial x^{\mu}}\ol v^{\mu}= 1
\,,
\\
\hat v^{r} |_{\scri^-} =& \frac{\partial \hat r}{\partial x^{\mu}}\ol v^{\mu}= 0
\,,
\\
\hat v^{\mathring A}|_{\scri^-}  =& \frac{\partial \hat x^{\mathring A}}{\partial x^{\mu}}\ol v^{\mu} = 0
\,.
\end{align}
Under this coordinate transformation we  have
\begin{align}
 g_{\tau r}|_{\scri^-}  =&  \frac{\partial \hat x^{\mu}}{\partial \tau} \frac{\partial \hat x^{\nu}}{\partial r}\hat g_{\mu\nu}
\,=\,\frac{\hat g_{\tau r}}{ \ol v^{\tau}}
\,,
\\
 g_{\tau \mathring A}|_{\scri^-}  =&  \frac{\partial \hat x^{\mu}}{\partial \tau} \frac{\partial \hat x^{\nu}}{\partial x^{\mathring A}}\hat g_{\mu\nu}
\,=\,  \frac{\hat g_{\tau \mathring A} }{ \ol v^{\tau}}- \frac{g_{ \mathring A\mathring B} \ol v^{\mathring B}}{\ol  v^{\tau}}  
\,,
\\
 g_{\tau\tau}|_{\scri^-}  =&  \frac{\partial \hat x^{\mu}}{\partial \tau} \frac{\partial \hat x^{\nu}}{\partial \tau}\hat g_{\mu\nu}
\,=\,\frac{\hat g_{\tau\tau}}{(\ol  v^{\tau})^2}
- 2\frac{\hat g_{\tau r}\ol  v^r}{(\ol  v^{\tau})^2}
- 2\frac{\hat g_{\tau \mathring A} \ol v^{\mathring A}}{( \ol v^{\tau})^2} 
+ \frac{g_{\mathring A\mathring B}\ol  v^{\mathring A} \ol v^{\mathring B}}{( \ol v^{\tau})^2} 
\,.
\end{align}
We conclude that the gauge freedom to prescribe $\dot x_*$ can be identified with the freedom to prescribe, in a fixed adapted null coordinate system,  
the metric coefficients $g_{\tau\mu}$ on $\scri^-$,
%(which, in a wave-map gauge \cite{CCM2},  is captured by  the freedom to prescribe the gauge source functions $W^{\mu}$  on $\scri^-$ \cite{ChPaetz}),
%
\begin{equation}
g_{\tau\tau}|_{\scri^-} \,, \quad \nu_{\tau}\,, \quad \nu_{\mathring A}
\,.
\end{equation}
In the conformal Gauss gauge the vector $\dot x$ is normalized to 1, whence $g_{\tau\tau}|_{\scri^-}=-1$.

In other words, the gauge freedom to choose the initial direction of the conformal geodesics is chosen in such a way that
 the metric components $\nu_{\tau}$ and $\nu_{\mathring A}$ take certain prescribed values in the associated
conformal Gauss coordinates. Instead of $\dot x_*$ they may therefore be regarded as gauge degrees of freedom.

Let us also take a look at the 1-form $f$. From \eq{frame_deriv1} below  we deduce  (recall that $f_0=\langle f,\dot x\rangle=0$)
\begin{equation}
\Gamma^{\mu}_{\tau\tau} |_{\scri^-} = - f^{\mu}  
\quad
\Longleftrightarrow \quad\begin{cases} f_{r}|_{\scri^-}  = -\nu_{\tau}  \Gamma^{\tau}_{\tau\tau}  
= \frac{1}{2}\partial_rg_{\tau\tau} - \partial_{\tau}g_{\tau r}
 \,, 
\\
 f_{\mathring A}  |_{\scri^-} = -  g_{\mathring A\mathring B} \Gamma^{\mathring B}_{\tau\tau} -\nu_{\mathring A}  \Gamma^{\tau}_{\tau\tau}    
 =  \frac{1}{2}\partial_{\mathring A} g_{\tau\tau}  -  \partial_{\tau} g_{\tau\mathring A}
\,,
\\
\partial_{\tau}g_{\tau\tau} |_{\scri^-} = \nu_{\tau} g^{r\mathring A} (\partial_{\mathring A}g_{\tau\tau} - 2\partial_{\tau} g_{\tau\mathring A})
\,.
\end{cases}
\label{0g00-eqn}
\end{equation}
The freedom to prescribe $f_{\alpha}$, $\alpha=1,2,3$, on $\scri^-$ therefore corresponds to 
 the freedom to prescribe $\partial_{\tau} g_{\tau\alpha}|\scri^-$.
%(or, equivalently, to prescribe the gauge source functions $\partial_{\tau} W^{\tau}|_{\scri^-}$ and $\partial_{\tau} W^{\mathring A}|_{\scri^-}$ in a wave-map gauge \cite{ttp1}).
Contrary to that, $\partial_{\tau}g_{\tau\tau} |_{\scri^-}$  cannot be considered as a gauge function (or rather it could if the normalization
condition on $\dot x$ is dropped).
However, in this work we prefer to regard $f_{\scri^-}$ as gauge functions.

%Note further that, by \eq{expr_Theta2}, the coordinate gauge  freedom to prescribe $\kappa$ can be transferred to the freedom to prescribe $\Theta^{(2)}$.
Accordingly, as gauge data to realize a conformal Gauss gauge from null infinity one can identify
\begin{equation}
\nu_{\tau}\,, \quad \nu_{\mathring A}\,, \quad f_{\scri^-}\,, \quad \kappa\,, \quad  \theta^-\,, \quad \Theta^{(1)}
\,.
\label{gauge_data_scri}
\end{equation}
As indicated above,   the gauge freedom  $r \mapsto r' =r'(r, x^{\mathring A}) $ is not completely exhausted by these gauge data.
The remaining freedom can be used to prescribe certain functions at spatial infinity, by which we mean the future boundary of $\scri^-$.
%Which functions one may prescribe  depends on the representation of spatial infinity, and 
We will analyze this for the cylinder representation in Section~\ref{sec_yet_another}.

\subsection{Realization of conformal Gauss coordinates}
\label{sec_realization}

Consider a solution $(\mcM, \widetilde g, \widetilde \Theta)$ of the CFE with $\lambda=0$ which admits a smooth $\scri^-$.
%At this stage we do not need to assume that it is smooth at $I$ (nor that there is a representation as a cylinder).
Moreover, choose any functions 
\begin{equation}
\Theta^{(1)}(r,x^{\mathring A})>0\,,
\quad \kappa(r,x^{\mathring A})\,, \quad
\theta^-(r, x^{\mathring A})\,, \quad
\nu_{\tau}(r,x^{\mathring A})>0
\,,\quad
\nu_{\mathring A}
\,, \quad f_{r*}
\,, \quad f_{\mathring A*}
\quad \text{on $\scri^-$.}
\end{equation}
We will describe how
conformal Gauss coordinates with this choice of gauge data on $\scri^-$ can be realized.

For this choose  an adapted null coordinate system $(\widetilde\tau,\widetilde r,\widetilde x^{\mathring A})$ (in particular \ $\scri^-=\{\widetilde\tau=-1\}$) and extend it off $\scri^-$ in any way.
One would like to start with a conformal transformation which realizes $\Theta^{(1)}$ (and $\theta^-$)  followed by a coordinate transformation which realizes
the remaining gauge data. Then a solution to the conformal geodesics equations would determine the coordinate transformation off $\scri^-$.
However, there is a problem: $\Theta^{(1)}$ is given w.r.t.\ the new $r$-coordinate. The relation between the new and the old $r$-coordinate is determined
by $\kappa$ and $\widetilde \kappa$, which are \emph{not} invariant under conformal transformations.
The transformations to $\Theta^{(1)}$ and $ \kappa$ therefore need to be accomplished simultaneously.
We further not that $\Theta^{(1)}$ is not invariant under rescaling of $\tau$, so also the transformation to $\nu_{\tau}$ needs to be taken into account.

We therefore consider a coordinate transformation of the form
\begin{equation}
\widetilde r \mapsto r= r(\widetilde r,x^{\mathring A}) 
\,,
\qquad
\widetilde \Theta \mapsto \Theta = \psi(\widetilde r, x^{\mathring A})\widetilde \Theta
\,, 
\quad 
1+\widetilde\tau \mapsto  1+\tau =h(\widetilde r, x^{\mathring A})(1+\widetilde \tau)
\,.
\label{kappa_trafo1}
\end{equation}
Taking the behavior of connection coefficients under conformal and coordinate transformations into account, we find
that the function $r$ is
given by (we suppress dependence on the angular coordinates),
\begin{align}
\kappa(r(\widetilde r))=&\frac{\partial \widetilde x^{\alpha}}{\partial r}\frac{\partial \widetilde x^{\beta}}{\partial r}\frac{\partial  r}{\partial  \widetilde x^{\gamma}}
\Big(\widetilde \Gamma^{\gamma}_{\alpha\beta}- \psi^{-1}(2\delta_{(\alpha}{}^{\gamma}\partial_{\beta)}\psi -g_{\alpha\beta}g^{\gamma\lambda}\partial_{\lambda}\psi\Big)
+ \frac{\partial  r}{\partial  \widetilde x^{\alpha}}\frac{\partial^2 \widetilde x^{\alpha}}{\partial r^2}
\nonumber
\\
=&\frac{\partial \widetilde r}{\partial r}\big[\widetilde \kappa(\widetilde r) -2\partial_{\widetilde r}\log\psi(\widetilde r)\big]+ \frac{\partial r}{\partial \widetilde r}\frac{\partial^2 \widetilde r}{\partial r^2}
\,,
\label{kappa_trafo2}
\end{align}
or, equivalently, 
\begin{equation}
\frac{\partial^2 r}{\partial \widetilde r^2}
=\frac{\partial r}{\partial \widetilde r}\big[\widetilde\kappa(\widetilde r)-2\partial_{\widetilde r}\log\psi(\widetilde r)\big]
  - \Big(\frac{\partial r}{\partial \widetilde r}\Big)^2\kappa( r(\widetilde r))
\,.
\label{kappa_trafo3}
\end{equation}
The function $\psi$ is given by
\begin{equation}
\psi(\widetilde r) =h(\widetilde r)\frac{\Theta^{(1)}(r(\widetilde r))}{\widetilde \Theta^{(1)}(\widetilde r)}
\,, \quad \text{where}\quad 
h(\widetilde r)=\frac{\partial\widetilde r}{\partial r} \frac{\nu^{\tau}(r(\widetilde r))}{\widetilde \nu^{\tau}(\widetilde r)}
\,.
\label{kappa_trafo4}
\end{equation}
We construct the gauge from some cut $\{\widetilde r=\mathrm{const.}\}$ of $\scri^-$.
The ODE \eq{kappa_trafo3} is of the following form
\begin{equation}
\frac{\partial^2 r}{\partial \widetilde r^2}
= F(r,\partial_{\widetilde r} r, \widetilde r)
\,,
\label{kappa_trafo5}
\end{equation}
where $F$ is some smooth function, and this equation can at least locally be solved.
The freedom to choose the initial data will be specified later (and not  at a cut of $\scri^-$ but at the critical set $I^-$ of spatial infinity).
Choosing $\psi$ and $h$ as above the desired values for $\Theta^{(1)}$,  $\kappa$ and $\nu_{\tau}$ are realized.

Next, a coordinate transformation of the form
\begin{equation}
x^{\alpha}\mapsto x^{\alpha} + f^{\alpha}(x^{\beta})(1+\tau)
\label{coord_trafo2}
\end{equation}
 realizes
$g_{\tau\tau}|_{\scri^-}=-1$ and the prescribed value for $\nu_{\mathring A}$.
A conformal transformation
\begin{equation}
\Theta\mapsto [1+ \phi(r,x^{ \mathring A})(1+\tau)]\Theta
\,,
\label{coord_trafo3}
\end{equation}
with an appropriately chosen $\phi$ transforms to the right value for  $\theta^-$, cf.\ \eq{behave_-expansion}.
Note that $\Theta^{(1)}$, $\kappa$ and $\nu_{\tau}$ remain invariant under \eq{coord_trafo2}-\eq{coord_trafo3}.
Then  we solve the conformal geodesics equations with initial data 
\begin{equation}
\dot x|_{\scri^-}=\partial_{\tau}
\,,
\quad
 {f}_{\tau}|_{\scri^-}  = 0
\,,
\quad
{f}_{r}|_{\scri^-}  =  f_{r*} 
\,,
\quad
{f}_{\mathring A}|_{\scri^-}  =  f_{\mathring A*} 
\,.
\end{equation}
That  yields a vector field  $\dot x$ and a 1-form $ f$ on  $\mcM$ (at least in some neighborhood of $\scri^-$).
The  gauge condition $\langle \dot x, f\rangle=0$ is realized by another conformal transformation
$\Theta\mapsto \Psi\Theta$. There is no freedom to choose the initial datum $\Psi|_{\scri^-}$
which  needs to be $1$ in order to preserve the gauge functions we have already realized.
%Under this transformation $f\mapsto f-\mathrm{d}\log\Psi$, in particular $f_{\tau}|_{\scri^-}\mapsto f_{\tau}|_{\scri^-}-\psi=0$ as desired.
Finally, a coordinate transformation is necessary  to transform $\dot x$ to $\partial_{\tau}$.
Since $\dot x|_{\scri^-}=\partial_{\tau}$ 
it is of the form $x^{\mu}\mapsto x^{\mu} + O(1+\tau)^2$ and therefore does not affect the gauge functions we have realized in the previous steps.

\subsection{Connection coefficients}
\label{sec_connection_gen}

We want to compute the connection coefficients of the Weyl connection w.r.t.\ the frame field $(e_i)$ on $\scri^-$
 in terms of the connection coefficients associated with the adapted null coordinates \eq{adapted_null}.
We have
\begin{equation}
\widehat\Gamma_i{}^k{}_j e^{\mu}{}_k
=e^{\nu}{}_i \widehat\nabla_{\nu} e^{\mu}{}_j 
=e^{\nu}{}_i (\partial_{\nu} e^{\mu}{}_j  + \Gamma^{\mu}_{\nu\sigma} e^{\sigma}{}_j )
 + 2e^{\mu}{}_{(i}f_{j)} - \eta_{ij} f^{\mu}
\,.
\label{dfn_coennect}
\end{equation}
Recall that the frame field $(e_i)$ has been constructed such that  $ \widehat\Gamma_0{}^k{}_j  =0$,
so that the  $i=0$-components of  \eq{dfn_coennect} yield with   \eq{frame_field1}-\eq{frame_field3}
\begin{equation}
\Gamma^{\mu}_{\tau\tau} =  - f^{\mu}
\,,
\label{frame_deriv1}
\end{equation}
(this relation holds globally),
and
\begin{align}
\partial_{\tau} e^{\mu}{}_1   |_{\scri^-}  =&  - f_{1}  e^{\mu}{}_0
- \Gamma^{\mu}_{\tau\tau} 
- \nu^{\tau}\Gamma^{\mu}_{\tau r}
\,,
\label{frame_deriv2}
\\
\partial_{\tau} e^{\mu}{}_A   |_{\scri^-}  =&-  f_{A}  e^{\mu}{}_0
- e^{\mathring A}{}_A( \Gamma^{\mu}_{\tau\mathring A} -\Gamma^{\mu}_{\tau r} \nu^{\tau}\nu_{\mathring A})
\,.
\label{frame_deriv3}
\end{align}
We set  $\re^{\mathring A}{}_A:=e^{\mathring A}{}_A|_{\scri^-}=\Omega^{-1}\mathring e^{\mathring A}{}_A$. Then $(\re_A)$ is an orthonormal frame for $g_{\mathring A\mathring B}|_{\scri^-}$.

For $i=A$ we obtain from \eq{dfn_coennect} (set $\nu_A:=  \re^{\mathring A}{}_A\nu_{\mathring A}$ and
use \eq{coordChristoffel1}, \eq{coordChristoffel2} and \eq{coordChristoffel4})
 %note that $\Gamma^{\tau}_{r\alpha }$ and $\Gamma^{\mathring A}_{rr}$ vanish in adapted null coordinates, while $\Gamma^{\mathring B}_{r\mathring A}|_{\scri^-}=\partial_r\log\Omega\,\delta^{\mathring B}{}_{\mathring A}$)
%
\begin{align}
\widehat\Gamma_A{}^k{}_0 e^{\mu}{}_k |_{\scri^-} =& \re^{\mathring A}{}_A (\Gamma^{\mu}_{\tau\mathring A }
- \nu^{\tau}\nu_{\mathring A}\Gamma^{\mu}_{ \tau r} )
 + \delta^{\mu}{}_{\tau}f_{A} 
\,,
\\
\widehat\Gamma_A{}^k{}_1 e^{\mu}{}_k  |_{\scri^-}=&\re^{\mathring A}{}_A (\partial_{\mathring A} e^{\mu}{}_1
  + \Gamma^{\mu}_{\tau\mathring A }   +\nu^{\tau} \Gamma^{\mu}_{r \mathring A } 
 )
 -\nu^{\tau}\nu_{ A} (\partial_{r} e^{\mu}{}_1+ \Gamma^{\mu}_{\tau r} +\nu^{\tau} \Gamma^{\mu}_{rr}  )
 +2 e^{\mu}{}_{(A}f_{1)} 
\,,
\\
\widehat\Gamma_A{}^k{}_B e^{\mu}{}_k |_{\scri^-} =&\re^{\mathring A}{}_A \partial_{\mathring A} e^{\mu}{}_B 
+\re^{\mathring A}{}_A\re^{\mathring B}{}_B  \Gamma^{\mu}_{\mathring A \mathring B} 
-2\nu^{\tau}\nu_{( A} \re^{\mathring A}{}_{B)} \Gamma^{\mu}_{r\mathring A} 
\nonumber
\\
&
- \nu^{\tau}\nu_{ A} \partial_{r} e^{\mu}{}_B  
+  (\nu^{\tau})^2\nu_{ A} \nu_{B}\Gamma^{\mu}_{rr}
 + 2e^{\mu}{}_{(A}f_{B)} - \eta_{AB} f^{\mu}
\,.
\end{align}
We deduce the following relations,where we denote by $ (  \rsigma^{ A})$ the co-frame of $ ( \re_A)$, and by
$\rGamma_A{}^C{}_B $  the connection coefficients of $ ( \re_A)$,
%the orthonormal frame $( \re_A)=  \re^{\mathring A}{}_A\partial_{\mathring A}$,
%
\begin{align}
 \widehat\Gamma_A{}^1{}_0   |_{\scri^-}=& \re^{\mathring A}{}_A (\Gamma^{\tau}_{\tau\mathring A }
- \nu^{\tau}\nu_{\mathring A}\Gamma^{\tau}_{ \tau r} ) 
\,,
\\
 \widehat\Gamma_A{}^B{}_0|_{\scri^-}  =& \re^{\mathring A}{}_A  \rsigma^B{}_{\mathring B}\Gamma^{\mathring B}_{\tau\mathring A }
- \nu^{\tau}\nu_{ A}  \rsigma^B{}_{\mathring B}\Gamma^{\mathring B}_{ \tau r} 
\,,
\\
\widehat\Gamma_A{}^B{}_1   |_{\scri^-}=&  \re^{\mathring A}{}_A \rsigma^B{}_{\mathring B} ( \Gamma^{\mathring B}_{\tau\mathring A }   +\nu^{\tau} \Gamma^{\mathring B}_{r \mathring A } )
 -  \nu^{\tau}\nu_{ A} \rsigma^B{}_{\mathring B} \Gamma^{\mathring B}_{\tau r} 
 +\delta^B{}_A  f_{1} 
\,,
\\
\widehat\Gamma_A{}^C{}_B  |_{\scri^-} =&    \rGamma_A{}^C{}_B 
+\frac{1}{2}\theta^+ \nu^{\tau}(\nu^C \eta_{AB}-\nu_{B}\delta^C{}_A)
 + 2\delta^C{}_{(A}f_{B)} - \eta_{AB} f^{C}
\,.
\end{align}

For $i=1$ we obtain from \eq{dfn_coennect}, using \eq{frame_deriv2}-\eq{frame_deriv3}
\begin{align}
\widehat\Gamma_1{}^k{}_0 e^{\mu}{}_k|_{\scri^-}  =& \Gamma^{\mu}_{\tau \tau} +\nu^{\tau}  \Gamma^{\mu}_{\tau r} 
 + \delta^{\mu}{}_{0}f_{1} 
\,,
\\
\widehat\Gamma_1{}^k{}_1 e^{\mu}{}_k|_{\scri^-}  =&
 \nu^{\tau} \partial_{r} e^{\mu}{}_1 
+\nu^{\tau}  \Gamma^{\mu}_{\tau r}  
+ (\nu^{\tau})^2\Gamma^{\mu}_{r r}   
  - f_{1}  \delta^{\mu}{}_0
 + 2e^{\mu}{}_{1}f_{1} -  f^{\mu}
\,,
\\
\widehat\Gamma_1{}^k{}_A e^{\mu}{}_k|_{\scri^-}  =&
 \nu^{\tau}  \partial_{r} e^{\mu}{}_A
- (\nu^{\tau})^2\nu_{ A}\Gamma^{\mu}_{r r} 
+ \nu^{\tau}\Gamma^{\mu}_{r \mathring A}   \re^{\mathring A}{}_A 
 -  f_{A}  \delta^{\mu}{}_0+ e^{\mu}{}_{1}f_{A}   + e^{\mu}{}_{A}f_{1} 
\,,
\end{align}
whence
\begin{align}
\widehat\Gamma_1{}^1{}_0 |_{\scri^-}  =& \Gamma^{\tau}_{\tau \tau} +\nu^{\tau}  \Gamma^{\tau}_{\tau r} 
\,,
\\
\widehat\Gamma_1{}^A{}_0 |_{\scri^-}  =&  \rsigma^A{}_{\mathring A}(\Gamma^{\mathring A}_{\tau \tau} +\nu^{\tau}  \Gamma^{\mathring A}_{\tau r} )\,,
\\
\widehat\Gamma_1{}^A{}_1|_{\scri^-}  =&
 \rsigma^A{}_{\mathring A}\nu^{\tau}  \Gamma^{\mathring A}_{\tau r}  -  f^{ A}
\,,
\\
\widehat\Gamma_1{}^A{}_B |_{\scri^-}  =&
 \nu^{\tau}\Gamma^{\mathring A}_{r \mathring B}  \re^{\mathring B}{}_B   \rsigma^{ A}{}_{\mathring A} 
+(f_1 - \frac{1}{2}\theta^+ \nu^{\tau}   )\delta^A{}_B
\,.
\end{align}
Finally, we insert the expressions \eq{coordChristoffel1}-\eq{coordChristoffel13} to end up with the following list for the relevant components of the
connection coefficients
\begin{align}
 \widehat\Gamma_A{}^1{}_1   |_{\scri^-}=& f_A
\,,
\label{Weyl_connectA_1}
\\
 \widehat\Gamma_A{}^1{}_0   |_{\scri^-}=&  \frac{1}{2}\xi_{ A} + \nu^{\tau}\rnabla_{ A}\nu_{\tau}
+ \nu_{ A}\Big(\partial_{r}- \frac{1}{2}\theta^+ + \kappa \Big) \nu^{\tau}
\,,
\label{GammaA10_gen}
\\
 \widehat\Gamma_A{}^B{}_0|_{\scri^-}  =& 
 \frac{1}{2} \nu_{\tau}\Xi_{ A}{}^{ B} - \frac{1}{4}\nu_{\tau}\Big( \theta^- + \theta^+ (\nu^{\tau})^2(1+\nu_{ C}\nu^{ C})\Big)\delta_{ A}{}^{ B}
- \frac{1}{2} \xi^{ B} \nu_{ A}
\nonumber
\\
&
+  \Big( \nu_{\tau} \rnabla_{ A}- \frac{1}{2}\xi_{ A}
+ \frac{1}{2}\theta^+\nu^{\tau}\nu_{ A} -  \kappa\nu_{ A}  -\nu_{ A}\partial_r   \Big)(\nu^{\tau}\nu^{ B})
\,,
\label{GammaAB0_gen}
\\
\widehat\Gamma_A{}^B{}_1   |_{\scri^-}=& 
 \widehat\Gamma_A{}^B{}_0
% \frac{1}{2} \nu_{\tau}\Xi_{ A}{}^{ B} - \frac{1}{4}\nu_{\tau}\Big( \theta^- + \theta^+ (\nu^{\tau})^2(1+\nu_{ A}\nu^{ A})\Big)\delta_{ A}{}^{ B}
%+\Big( \rnabla^{ B}- \frac{1}{2} \xi^{ B}\Big) \nu_{ A}
%\nonumber
%\\
%&
%-\nu^{\tau}\nu^{ B}\Big(\rnabla_{ A}\nu_{\tau}  +\frac{1}{2}\nu_{\tau}\xi_{ A}
% + (  \kappa- \frac{1}{2}\theta^+-\nu^{\tau}\partial_r\nu_{\tau})\nu_{ A}\Big)
% -  \nu^{\tau}\nu_{ A} \partial_r \nu^{ B}
%\\
%&
+\Big(\frac{1}{2}\theta^+ \nu^{\tau} +f_1\Big)\delta^{B}{}_A
\,,
\\
\widehat\Gamma_A{}^C{}_B  |_{\scri^-} =&    \rGamma_A{}^C{}_B 
+\frac{1}{2}\theta^+ \nu^{\tau}(\nu^C \eta_{AB}-\nu_{B}\delta^C{}_A)
 + 2\delta^C{}_{(A}f_{B)} - \eta_{AB} f^{C}
\,.
\label{Weyl_connectA_5}
\end{align}
Here we have set $\xi_A= \re^{\mathring A}{}_A\xi_{\mathring A}$ and $\Xi_{AB}= \re^{\mathring A}{}_A\re^{\mathring B}{}_B\Xi_{\mathring A\mathring B}$,
$\rnabla_A=\re^{\mathring A}{}_A\rnabla_{\mathring A}$   refers to the Levi-Civita covariant derivative associated to the family  $r\mapsto \not\hspace{-.1em} g =g_{\mathring A\mathring B}\mathrm{d}x^{\mathring A}\mathrm{d}x^{\mathring B}$ of Riemannian metrics.

For the remaining connection coefficients we find with \eq{frame_deriv1}
\begin{align}
\widehat\Gamma_1{}^1{}_1 |_{\scri^-}  =& f_1
\,,
\label{Weyl_connect1_1}
\\
\widehat\Gamma_1{}^1{}_0 |_{\scri^-}  =&   -(\partial_{r}+\kappa )\nu^{\tau} -f_1
\,,
\label{Gamma110_gen}
\\
\widehat\Gamma_1{}^A{}_0 |_{\scri^-}  =&
\nu^{\tau}(\partial_r + \kappa-\nu^{\tau}\partial_r\nu_{\tau})\nu^{ A}  + \frac{1}{2}\xi^{ A}  -  f^{ A}
\,,
\\
\widehat\Gamma_1{}^A{}_1|_{\scri^-}  =&
\widehat\Gamma_1{}^A{}_0
\,,
\\
\widehat\Gamma_1{}^A{}_B |_{\scri^-}  =&
 f_{1}  \delta^A{}_B
\label{Weyl_connect1_5}
\,.
\end{align}

\subsection{Schouten tensor}
\label{sec_schouten_gen}

We  compute the  Schouten tensor associated to  the Weyl connection.
First of all we express its  frame components in terms of coordinate components of the adapted null coordinate system \eq{adapted_null},
%
%\begin{align}
%\widehat L_{ 1 1}|_{\scri^-} 
%%=&  \widehat L_r{}^r -(\nu^{\tau})^2\nu^{ A}\nu_{ A} \widehat L_{rr} + \nu^{\tau}\nu^{\mathring A}\widehat L_{r\mathring A}
%%\\
%=&  \widehat L_r{}^r  + \nu^{\tau}\nu_{\mathring A}\widehat L_{r}{}^{\mathring A}
%\,,
%\\
%%\widehat L_{ 1 A}|_{\scri^-} =& 
%%\nu^{\tau}\re^{\mathring A}{}_A\widehat L_{r\mathring B} -(\nu^{\tau})^2 \nu_A\widehat L_{rr}
%%\\
%%=&
%%\nu^{\tau}\re^{\mathring A}{}_A g_{\mathring A\mathring B}\widehat L_r{}^{\mathring B}
%\widehat L_{ 1}{}^{ A}|_{\scri^-} =& 
%\nu^{\tau}\rsigma^{A}{}_{\mathring A}\widehat L_r{}^{\mathring A}
%\,,
%\\
%\widehat L^A{}_{ 1}|_{\scri^-} =& \rsigma^{ A}{}_{\mathring A}(\nu_{\tau}\widehat L^{\mathring Ar} 
%+\nu_{\mathring B}\widehat L^{\mathring A\mathring B})
%\,,
%\\
%\widehat L^{ AB}|_{\scri^-} 
%%=&
%%\re^{\mathring A}{}_A\re^{\mathring B}{}_B\widehat L_{\mathring A\mathring B}
%%- \re^{\mathring A}{}_A\nu^{\tau}\nu_B \widehat L_{\mathring A r}
%%- \re^{\mathring B}{}_B\nu^{\tau}\nu_A \widehat L_{\mathring B r}
%%+(\nu^{\tau})^2\nu_A\nu_B\widehat L_{rr}
%%\\
%=&
%\rsigma^{ A}{}_{\mathring A}\rsigma^{ B}{}_{\mathring B}\widehat L^{\mathring A\mathring B}
%\,,
%\\
%\widehat L_{ 10}|_{\scri^-} =& \widehat L_{ 1 1} -(\nu^{\tau})^2 L_{rr} 
%\,,
%\\
%\widehat L^A{}_{ 0}|_{\scri^-} 
%%=&\widehat L_{ A1} -\nu^{\tau}\re^{\mathring A}{}_A\widehat L_{\mathring A r} + (\nu^{\tau})^2\nu_A \widehat L_{rr}
%%\\
%=&\widehat L_{ A1} -\nu^{\tau}\rsigma^{ A}{}_{\mathring A}\widehat L^{\mathring A}{}_{ r} 
%\,.
%\end{align}
\begin{align}
\widehat L_{ 1 1}|_{\scri^-} 
=&  \widehat L_r{}^r -(\nu^{\tau})^2\nu^{  A}\nu_{ A} \widehat L_{rr} + \nu^{\tau}\nu^{\mathring A}\widehat L_{r\mathring A}
\,,
\label{11_com_Schouten}
\\
\widehat L_{ 1 A}|_{\scri^-} =& 
\nu^{\tau}\re^{\mathring A}{}_A\widehat L_{r\mathring A} -(\nu^{\tau})^2 \nu_A\widehat L_{rr}
\,,
\\
\widehat L_{ A1}|_{\scri^-} =& \re^{\mathring A}{}_A(\nu_{\tau}\widehat L_{\mathring A}{}^{r} -\nu^{\tau}\nu_{ B}\nu^{ B}\widehat L_{\mathring Ar} + \nu^{\mathring B}\widehat L_{\mathring A\mathring B})
-\nu_A(  L_{r}{}^{r} -(\nu^{\tau})^2  \nu_{ B} \nu^B\widehat L_{rr}+\nu^{\tau} \nu^{\mathring B}\widehat L_{r\mathring B} )
\,,
\\
\widehat L_{ AB}|_{\scri^-} 
=&
\re^{\mathring A}{}_A\re^{\mathring B}{}_B\widehat L_{\mathring A\mathring B}
- 2\re^{\mathring A}{}_{(A}\nu^{\tau}\nu_{B)} \widehat L_{\mathring A r}
+(\nu^{\tau})^2\nu_A\nu_B\widehat L_{rr}
\,,
\\
\widehat L_{ 10}|_{\scri^-} =& \widehat L_{ 1 1} -(\nu^{\tau})^2 \widehat L_{rr} 
\,,
\\
\widehat L_{ A 0}|_{\scri^-} 
=&\widehat L_{ A1} -\nu^{\tau}\re^{\mathring A}{}_A\widehat L_{\mathring A r} + (\nu^{\tau})^2\nu_A \widehat L_{rr}
\,.
\end{align}
Moreover, by \eq{schouten_weylconnect} we have
%
%\begin{equation}
%\widehat L_{\mu\nu}|_{\scri^-}  
%%=  L_{\mu\nu} -\nabla_{\mu} f_{\nu}  + \frac{1}{2}S(f)_{\mu}{}^{\sigma}{}_{\nu}f_{\sigma}
%=  L_{\mu\nu} -\nabla_{\mu} f_{\nu}  +f_{\mu}f_{\nu} - \frac{1}{2}g_{\mu\nu} f^{\sigma}f_{\sigma}
%\,,
%\end{equation}
%%
%i.e.
%
\begin{align*}
\widehat L_{rr}|_{\scri^-} =&  L_{rr} -(\partial_{r} -\kappa-f_r)f_r
\,,
\\
\widehat L_{r\mathring A} |_{\scri^-}=&  L_{r\mathring A} -(\partial_{r}-\frac{1}{2}\theta^+ ) f_{\mathring A}
 +f_r (f_{\mathring A} -\frac{1}{2}\xi_{\mathring A} )
\,,
\\
\widehat L_{\mathring A r}|_{\scri^-} =&  L_{r\mathring A} -(\rnabla_{\mathring A}+\frac{1}{2}\xi_{\mathring A}) f_{r}  
+( \frac{1}{2}\theta^+ +f_r)f_{\mathring A} 
\,,
\\
\widehat L_{r}{}^r|_{\scri^-} =&  L_{r}{}^r -(\partial_{r}+\kappa-\frac{1}{2}f_r) f^r  
- \frac{1}{2}f^{\mathring A}( f_{\mathring A}-\xi_{\mathring A} )
+\frac{1}{2}\nu^{\tau}\nu_{\mathring A}\xi^{\mathring A} f_r
+\frac{1}{2}f_r(\partial_r+2\kappa) g^{rr}
\,,
\\
\widehat L_{\mathring A\mathring B}|_{\scri^-} =&  L_{\mathring A\mathring B} -\frac{1}{2} f_{r} \Xi_{\mathring A\mathring B} -(\rnabla_{\mathring A}-f_{\mathring A} ) f_{\mathring B} 
+ \frac{1}{4}\Big(2\theta^+\nu^{\tau}\nu^{\mathring C}f_{\mathring C}- 2f^{\mu}f_{\mu}+(\theta^-  - \theta^+  g^{rr}) f_{r}\Big)g_{\mathring A\mathring B}
\,,
\\
\widehat L_{\mathring A}{}^r|_{\scri^-} =&  L_{\mathring A}{}^r -(\rnabla_{\mathring A}-\frac{1}{2}\xi_{\mathring A} -f_{\mathring A}) f^r 
+  \frac{1}{2} f_r(\rnabla_{\mathring A}-\xi_{\mathring A}) g^{rr}
 +\frac{1}{2} f_r\nu^{\tau}\nu^{\mathring B}\Xi_{\mathring A\mathring B}
\\
& - \frac{1}{4} f_r(\theta^-  - \theta^+  g^{rr}) \nu^{\tau}\nu_{\mathring A}
+\frac{1}{2}f^{\mathring B}\Xi_{\mathring A\mathring B} - \frac{1}{4}(\theta^-  - \theta^+  g^{rr}) f^{\mathring B}g_{\mathring A\mathring B}
\,.
\end{align*}
The relevant components of  the Schouten tensor  associated to the Levi-Civita connection
are given in Appendix~\ref{app_charact_constraints} (the $L^{rr}$-component is not needed),
\eq{constraint_Lrr}, \eq{constraint_LrA}, \eq{constraint_LABtr}, \eq{constraint12}, \eq{LAB_constraint}, \eq{LAr_constraint},
\begin{align}
 L_{rr} |_{\scri^-}
=& 
-\frac{1}{2}\Big(\partial_r + \frac{1}{2}\theta^+ -\kappa\Big)\theta^+
\,,
\\
 L_{r\mathring A} |_{\scri^-}
 =& - \frac{1}{2}\Big(\rnabla_{\mathring A}  +\frac{1}{2}\xi_{\mathring A}\Big)\theta^+ 
\,,
\\
 L_r{}^r |_{\scri^-}
 =&  \frac{1}{4}\Big(\partial_r + \kappa\Big)\theta^-  +   \frac{1}{4}\Big(\rnabla_{\mathring A}- \frac{1}{2}\xi_{\mathring A}\Big)\xi^{\mathring A}  -\frac{1}{4} g^{rr}\Big(\partial_r+\frac{1}{2}\theta^+ -\kappa \Big)\theta^+
- \frac{1}{4}\not \hspace{-0.2em}R
\,,
\\
 L_{\mathring A\mathring B}|_{\scri^-}
=&
-\frac{1}{2}\Big(\partial_{r}-\frac{1}{2} \theta^+ +\kappa\Big)\Xi_{\mathring A\mathring B}
+\frac{1}{2}(\rnabla_{(\mathring A}\xi_{\mathring B)})_\mathrm{tf}
-\frac{1}{4}  (\xi_{\mathring A}\xi_{\mathring B})_{\mathrm{tf}}
+\frac{1}{4}\Big(\not \hspace{-0.2em}R+ \frac{1}{2}\theta^+\theta^-  \Big)g_{\mathring A\mathring B}
\,,
\\
 L_{\mathring A}{}^r|_{\scri^-} 
=&
\frac{1}{2} \Big(\rnabla^{\mathring B} - \frac{1}{2}\xi^{\mathring B} \Big)\Big(\Xi_{\mathring A\mathring B} + \frac{1}{2}\theta^- g_{\mathring A\mathring B} \Big)
- \frac{1}{4}  \mathring  g^{rr} \Big( \rnabla_{\mathring A} +\frac{1}{2}\xi_{\mathring A}\Big)   \theta^+
\,.
\end{align}
This allows us to compute  $\widehat L_{ij}$ in terms of the coordinate components of the Schouten tensor in adapted null coordinates as computed from the constraint equations given in Appendix~\ref{app_charact_constraints}.
%The restriction of the  rescaled Weyl tensor to $\scri^-$ is computed in  Section~\ref{sec_solution_constr_gen} below.
This will be done explicitly in Section~\ref{sec_confG_firstoder} for a specific choice of the gauge data \eq{gauge_data_scri}.

\section{Cylinder representation of spatial infinity}
\label{section2}

So far, we have described the construction of a gauge based on conformal geodesics starting from $\scri^-$ which does not care about any representation of 
spatial infinity. 
In fact, depending on the choice of the conformal gauge data at $\scri^-$, $(\nu_{\tau},\nu_{\mathring A}, f_{\scri^-}, \kappa, \theta^-, \Theta^{(1)})$, the conformal Gauss gauge leads to different  representations of spatial infinity such as the ``classical''  point representation or  Friedrich's  cylinder representation.
The behavior of the fields near  the critical sets of a cylinder representing  spatial infinity, tough, is what we are interested in.

\subsection{Spatial infinity}
\label{sec_spatial_inf}

We consider  a conformally rescaled vacuum  spacetime $(\mcM,g, \Theta)$ which admits a smooth $\scri^-$, and we introduce adapted null coordinates 
at $\scri^-$.
For $(\mcM,g, \Theta)$ to admit a (finite) representation of spatial infinity, $\mathrm{d}\Theta$ needs to vanish
along  each null geodesic 
generator of  $\scri^-$ for some (finite) value of $r$, i.e.\  for each $x^{\mathring A}$ the function 
$\partial_{\tau}\Theta|_{\scri^-}$
needs to have a zero for some (finite) value $r=r_1(x^{\mathring A})$.
We are interested in the possible   behavior of the functions $\partial_{\tau}\Theta|_{\scri^-}$, $\nu_{\tau}$, $\theta^+$,
and $\kappa$ near $i^0$.
The constraint equations on $\scri^-$  (cf.\  \eq{constraint3} in Appendix~\ref{app_charact_constraints})
imply that the function $\partial_{\tau}\Theta|_{\scri^-}$ satisfies the  ODE
\begin{equation}
\Big(\partial_r-\frac{1}{2} \theta^++ \kappa-\nu^{\tau}\partial_r\nu_{\tau}\Big)\partial_{\tau}\Theta|_{\scri^-} =0
\,.
\label{important_constr}
\end{equation}
This equation can be integrated,
\begin{equation}
\partial_{\tau}\Theta|_{\scri^-} (r,x^{\mathring A}) = e^{\int_{r_0}^r (\frac{1}{2}\theta^+ -\kappa+\nu^{\tau}\partial_r\nu_{\tau})\mathrm{d}\hat r}\partial_{\tau}\Theta|_{\scri^-} (r_0, x^{\mathring A})
\label{important_contraint}
\end{equation}
for some initial value $\partial_{\tau}\Theta|_{\scri^-}(r_0,x^{\mathring A})$.
The solution  will vanish at $r_1$ if and only if
\begin{equation}
\int^{r_1} \Big(\frac{1}{2}\theta^+-\kappa+\nu^{\tau}\partial_r\nu_{\tau}\Big)\mathrm{d}\hat r=-\infty
\,.
\end{equation}
We deduce that whenever
%\begin{lemma}
%\label{lemma_div_i0}
%Assume that
  a vacuum spacetime admits a piece of a smooth $\scri^-$ as well as some  representation of spatial infinity $i^0$,
then, along any null geodesic generator of $\scri^-$ at least one of the following scenarios happens in adapted null coordinates on $\scri^-$:
\begin{enumerate}
\item[(i)]$\lim_{r\rightarrow i^0}\nu_{\tau} =0$,
\item[(ii)] $\int^{i^0} \kappa=\infty$,
\item[(iii)] 
$\int^{i^0} \theta^+= -\infty$.
\end{enumerate}
%\end{lemma}

\begin{remark}
{\rm
The divergence of $\theta^+$ along each null geodesic generator of $\scri^-$
indicates
the presence of a conjugate point. One therefore should expect  that a gauge where (iii) is realized at $r_1<\infty$  leads to the usual representation of $i^0$
as a point. 
This point is known to be singular for non-vanishing ADM mass.
In any gauge where (i) or (ii) are realized  the inverse metric  or the   derivative $ \partial_{\tau}g_{rr}|_{\scri^-}=2(\partial_r-\kappa)\nu_{\tau}$
%(which is positive near  $i^0$)
 diverge.
A certain singular behavior at spatial infinity  therefore seems to be unavoidable regardless of the gauge condition.
}
\end{remark}

Let us  compute how the conditions (i)-(iii) behave under reparameterizations $r\mapsto r'= r'(r,x^A)$,
\begin{align}
\lim_{r'\rightarrow i^0}\nu'_{\tau} =& \lim_{r\rightarrow i^0}\Big(\frac{\partial r}{\partial r'}\nu_{\tau} \Big)
\\
\int^{i^0} \kappa'\mathrm{d}r' =&\int^{i^0}\Big(\frac{\partial r}{\partial r'} \kappa(r(r'))  -\partial_{r'} \log\Big| \frac{\partial r'}{\partial r}\Big| \Big)\mathrm{d}r' 
=\int^{i^0} \kappa\,\mathrm{d}r  - \lim_{r\rightarrow i^0} \log\Big| \frac{\partial r'}{\partial r}\Big| + \mathrm{const.}
\,,
\\
\int^{i^0} \theta^{+'}\mathrm{d}r'  =&   \int^{i^0}\frac{\partial r}{\partial r'} \theta^+(r(r'))\mathrm{d}r'  
\,=\,  \int^{i^0} \theta^+\,\mathrm{d}r
\,.
\end{align}
While (iii) is invariant, (ii) is invariant at  least as long as $\lim_{r\rightarrow i^0}|  \frac{\partial r'}{\partial r}|\ne \infty$.
However, these considerations  suggest to combine (i) and (ii) into one condition
\begin{equation}
\int^{i^0}\nu^{\tau}\partial_{\tau}g_{rr}|_{\scri^-}\equiv 2 \int^{i^0}(\nu^{\tau}\partial_r\nu_{\tau}-\kappa)=-\infty
\quad \Longleftrightarrow  \quad \lim_{r\rightarrow i^0}\log|\nu^{\tau}| + \int^{i^0}\kappa   =\infty
\label{inv_expr}
\end{equation}
Indeed, under the transformation  $r\mapsto r'= r'(r,x^A)$ this behaves as
\begin{equation}
\lim_{r'\rightarrow i^0}\log|{\nu^{\tau}}' | + \int^{i^0}\kappa'\mathrm{d}r' 
= \lim_{r\rightarrow i^0}\log|{\nu^{\tau}} | +\int^{i^0} \kappa\,\mathrm{d}r + \mathrm{const.}
\,,
\end{equation}
so that \eq{inv_expr} is invariant under reparameterizations of $r$.

We have proved:
\begin{lemma}
\label{lemma_div_i0}
Assume that  a vacuum spacetime admits a piece of a  smooth $\scri^-$ as well as some representation of spatial infinity $i^0$.
Consider  any adapted null coordinate system  at $\scri^-$ which admits a finite coordinate representation of $i^0$.
Then along each  null geodesic generator of $\scri^-$
 at least one of the following scenarios happens:
\begin{enumerate}
\item[(i)] $\int^{i^0}\nu^{\tau}\partial_{\tau}g_{rr}|_{\scri^-}=-\infty$ (equivalently \eq{inv_expr}),
% (i.e.\ any affine parameter diverges at $i^0$), 
or
\item[(ii)] 
$\int^{i^0} \theta^+= -\infty$ (which indicates that  $i^0$ is a conjugate point).
\end{enumerate}
These conditions are invariant under reparameterizations of $r$.
%Conversely, it follows from xx 
%\tim{add}
%that if (i) or (ii) holds (if (ii) holds we need to assume, in addition, that $\lim_{r\rightarrow i^0}\nu^0\ne 0$), then $\mathrm{d}\Theta|_{i^0} =0$.
%\tim{dfn of $i^0$...}
%(if (i) holds and $\lim_{r\rightarrow i^0}\nu^0= 0$, we transform to e.g.\ $\nu^0=-1$ via \eq{trafo_nu0=1} and use that (ii) is invariant under such
%a transformation).
\end{lemma}

\begin{remark}
{\rm
A similar analysis can be applied to timelike infinity.
}
\end{remark}

Next, we present and discuss  two  explicit  gauge choices for the Minkowski spacetime  where the    different scenarios  (i)-(ii) are realized
and yield qualitatively different representations on spatial infinity.
We will see that (ii) corresponds to the classical point representation of spatial infinity while (i) yields a representation as a cylinder.

%Under such a coordinate transformation an $i^0$ which corresponds to a finite value of $r$ would be represented by a divergent $r'$.
%
%any affine parameter diverges at $i^0$

\subsection{Example: Minkowski spacetime}

\subsubsection{Point representation of spatial infinity}

Via a  conformal rescaling and suitable coordinate transformations
(compare \cite{p2})
the Minkowski metric  $\widetilde \eta = -(\mathrm{d}T)^2 + (\mathrm{d}R)^2+R^2s_{\mathring A\mathring B}\mathrm{d}x^{\mathring A}\mathrm{d}x^{\mathring B}$ can be brought into the form
\begin{equation}
  \eta =\Theta^2 \widetilde \eta =  -\mathrm{d}\tau^2 -2\mathrm{d}\tau\mathrm{d}r  + \sin^2(r)s_{\mathring A\mathring B}\mathrm{d}x^{\mathring A}\mathrm{d}x^{\mathring B}
\,,
\label{rep1_Mink}
\end{equation}
with
\begin{equation}
\Theta 
= 4\sin\frac{1+\tau}{2}\sin\Big(r+\frac{1+\tau}{2}\Big)
\,,
\end{equation}
and with $s_{\mathring A\mathring B}\mathrm{d}x^{\mathring A}\mathrm{d}x^{\mathring B}$   being the standard metric on $S^2$.
This is realized as follows:
First of all one introduces the retarded time $U$,
\begin{equation}
 U := T-R
\,,
\end{equation}
so that the Minkowski metric becomes
\begin{equation}
\widetilde \eta = -\mathrm{d}U^2 -2 \mathrm{d}U\mathrm{d}R+R^2s_{\mathring A\mathring B}\mathrm{d}x^{\mathring A}\mathrm{d}x^{\mathring B}
\,.
\end{equation}
We then apply the coordinate transformation
\begin{align}
R \enspace\mapsto \enspace& r:= \mathrm{arccot}(2U) -  \mathrm{arccot}(2(U+2R))\,,
\\
U \enspace\mapsto \enspace& \tau:= 2 \mathrm{arccot}(2(U+2R))-1
\,.
\end{align}
The inverse transformation reads
\begin{equation}
\tau\,\mapsto U
=\frac{1}{2}\cot\Big(r+ \frac{1+\tau}{2}\Big) \,, \quad r\,\mapsto  R\,=\,\frac{\sin r}{\Theta}
\,,
\end{equation}
and we have
\begin{align}
 \mathrm{d}U
 =&
 -\frac{4\sin^2 \frac{1+\tau}{2}}{\Theta^2}\mathrm{d}\tau -  \frac{8\sin^2 \frac{1+\tau}{2}}{\Theta^2 }\mathrm{d}r
\,,
\\
\mathrm{d}  R
=&-2\frac{ \sin^2\big( r+\frac{1+\tau}{2}\big) - \sin^2 \frac{1+\tau}{2} }{\Theta^2}\mathrm{d} \tau
+ \frac{4\sin^2 \frac{1+\tau}{2}}{\Theta^2}\mathrm{d} r
\,.
\end{align}
In the conformally rescaled spacetime,
past timelike infinity $i^-$  can be identified with the point $(\tau=-1,r=0)$,
 past null infinity $\scri^-$ corresponds to the set $\{\tau=-1, \enspace r\in (0,\pi)\}$
and spacelike infinity $i^0$ is given by the \emph{point} $(\tau=-1, r=\pi)$.

We find that
\begin{align}
% R[\eta] &=& 6\,,
%\\
%s &=& \frac{1}{4}\Box_\eta \Theta  + \frac{1}{24}\Theta R[\eta] \,=\, -2 \cos\Big( r + \frac{u}{2}\Big) \cos\frac{u}{2}
%\,,
%\\
%s|_{\scri^-} &=&  -2 \cos r 
%\,,
%\\
\theta^+ = 2\cot r
\,,
\quad
\theta^-=-2\cot r
\,,
\quad
\kappa =0
\,,
\\
%\Sigma = -2\sin r 
\partial_{\tau}\Theta|_{\scri^-}= 2\sin r
\,,
\quad
\nu_{\tau}= -1
\,,
\quad
\nu_{\mathring A} =0
\,, \quad
\partial_{\tau}g_{rr}|_{\scri^-}=0
\,.
\end{align}
which clearly belongs to case (ii) of Lemma~\ref{lemma_div_i0} (the integrand in (i) vanishes).
The null geodesics emanating from past timelike infinity $i^-$ meet again at $i^0$, as indicated by the divergence of the expansion $\theta^+$.

\subsubsection{Cylinder representation of spatial infinity and conformal Gauss coordinates}

In fact, we are more interested in case (i) of Lemma~\ref{lemma_div_i0}.
Again, as an example let us study the Minkowski spacetime for which we want to find a conformal representation which admits
a cylinder representation of spatial infinity and which we aim to express  in conformal Gauss coordinates (cf.\ \cite{kroon}).

Consider the Minkowski spacetime in standard Cartesian coordinates $(y^{\mu})$,
\begin{equation}
\widetilde \eta = - (\mathrm{d}y^0)^2 + (\mathrm{d}y^1)^2+ (\mathrm{d}y^2)^2+ (\mathrm{d}y^3)^2
\,.
\label{Mink_cart}
\end{equation}
In the domain $\{y_{\mu}y^{\mu}>0\}$ we introduce new coordinates $(x^{\mu})$ via
\begin{equation}
x^{\mu} := -\frac{y^{\mu}}{y^{\nu}y_{\nu}} \quad\Longrightarrow \quad  y^{\mu} = -\frac{x^{\mu}}{x^{\nu}x_{\nu}}
\label{Mink_coord_trafo}
\end{equation}
This coordinate patch excludes causal future and past of the origin, whence there will be no representation of timelike infinity.

The Minkowski line element becomes
\begin{equation}
\widetilde \eta = \frac{1}{(x^{\mu}x_{\mu})^2}\Big( - (\mathrm{d}x^0)^2 + (\mathrm{d}x^1)^2+ (\mathrm{d}x^2)^2+ (\mathrm{d}x^3)^2\Big)
\,.
\end{equation}
Let now $r$ denote the standard radial coordinate associated with the spatial coordinates $x^{\alpha}$, $\alpha=1,2,3$, and set $\tau:= x^0/r$.
Replacing $(x^{\alpha})$ by polar coordinates $(r, x^{\mathring A})$, $\widetilde \eta$ takes the form
\begin{equation}
\widetilde \eta = \frac{1}{r^2(1-\tau^2)^2}\Big(
-\mathrm{d}\tau^2 - 2\frac{\tau}{  r}\mathrm{d}\tau\mathrm{d}r - \frac{\tau^2-1}{r^2} \mathrm{d}r^2
 +s_{\mathring A\mathring B}\mathrm{d}x^{\mathring A}\mathrm{d}x^{\mathring B} \Big)
\,.
\end{equation}
%
%where $s_{\mathring A\mathring B}\mathrm{d}x^{\mathring A}\mathrm{d}x^{\mathring B}$ denotes the standard line element of the round 2-sphere.
We choose the conformal factor,
\begin{equation}
\Theta := r (1-\tau^2)
\,,
\label{Mink_conf_fac}
\end{equation}
which yields the following conformal representation of Minkowski spacetime,
\begin{equation}
 \eta =\Theta^2\widetilde \eta= 
-\mathrm{d}\tau^2 - 2\frac{\tau}{  r}\mathrm{d}\tau\mathrm{d}r + \frac{1-\tau^2}{r^2} \mathrm{d}r^2
 +s_{\mathring A\mathring B}\mathrm{d}x^{\mathring A}\mathrm{d}x^{\mathring B} 
\,,
\quad |\tau|<1\,, \quad r>0
\,.
\label{Mink_cyl}
\end{equation}
Future and past null infinity can be identified with $\scri^{\pm}=\{\tau = \pm 1, r> 0\}$.
The set $\{r=0\}$ represents spacelike infinity.
By introducing $\widehat r:= -\log r$ as a new coordinate one  shows  that the set $\{r=0\}$, where the metric coefficients in \eq{Mink_cyl} become singular,
has cylinder topology $[-1,1]\times \mathbb{S}^2$.
We denote the 2-spheres where the cylinder touches $\scri$ by $I^{\pm}:=\{\tau = \pm 1, r= 0\}$,
while the ``proper part'' of spacelike infinity is denoted $I:= \{|\tau|<1, r=0\}$.
$I^{\pm}$ are called \emph{critical sets}.
We have
\begin{eqnarray}
L_{\tau\tau} = \frac{1}{2}
\,,
\quad
L_{\tau r} =\frac{\tau}{2r}
\,,
\quad
L_{\tau \mathring A} = 0
\,,
\end{eqnarray}
and one checks that
\begin{equation}
\dot x = \partial_{\tau}
\,, \quad
f_{\tau}= 0
\,, \quad
f_r = r^{-1}
\,, \quad
f_{\mathring A} = 0
\end{equation}
solves the conformal geodesics equations \eq{conf_geo1}-\eq{conf_geo2}, so that  \eq{Mink_cyl} provides a conformal representation of (a subset of) Minkowski
spacetime in conformal Gauss coordinates.

The coordinate transformation which relates \eq{Mink_cart} and \eq{Mink_cyl} is given by
\begin{align}
y^0= \frac{-\tau}{r(1-\tau^2)} 
\,,
\qquad 
y^1= \frac{-\sin\theta\cos\phi}{r(1-\tau^2)}
\,,
\label{Mink_trafo1}
\qquad
y^2= \frac{-\sin\theta\sin\phi}{r(1-\tau^2)}
\,, \qquad 
y^3= \frac{-\cos\theta}{r(1-\tau^2)}
\,.
%\label{Mink_trafo2}
\end{align}
The inverse transformation takes the form
\begin{align}
\tau =& \frac{y^0}{\sqrt{(y^1)^2+(y^2)^2+(y^3)^2}}
\,,
\label{inv_Mink_trafo1}
\qquad
r = \frac{-\sqrt{(y^1)^2+(y^2)^2+(y^3)^2}}{-(y^0)^2+ (y^1)^2+(y^2)^2+(y^3)^2}
\,,
\\
\theta =& \arccos\Big(\frac{y^3}{\sqrt{(y^1)^2+(y^2)^2+(y^3)^2}}\Big)
\,,
\qquad
\phi = \arcsin\Big(\frac{y^2}{\sqrt{(y^1)^2+(y^2)^2}}\Big)
\,.
\label{inv_Mink_trafo4}
\end{align}
These conformal Gauss coordinates correspond to the following gauge data,
\begin{equation}
\nu_{\tau} = \frac{1}{r}
\,,\quad  \nu_{\mathring A} = 0
\,,
\quad
f_1|_{\scri^-} =1
\,,\quad 
f_A|_{\scri^-} = 0
\,,\quad 
\kappa=-\frac{2}{r} 
\,, \quad 
\theta^-=  0
\,, 
\quad
\Theta^{(1)}=2r
\,.
\label{Mink_gauge}
\end{equation}
Moreover, we anticipate (this will be relevant in view of Section~\ref{sec_yet_another} below, where $v^{(2)}_{\mathring A}$ is defined),
\begin{equation}
g_{\mathring A\mathring B}|_{I^-}=s_{\mathring A\mathring B}
\,, \quad \mcD^{\mathring A}v^{(2)}_{\mathring A}=0
\,,
\label{Mink_gauge2}
\end{equation}
where $\mcD$ denotes the Levi-Civita connection of the standard metric on $S^2$.
It follows from \eq{constraint3} that in  this gauge $\scri^-$ has vanishing divergence, $\theta^+=0$. That means that this gauge cannot admit a (finite) representation of a (regular) past timelike infinity $i^-$ as a tip of a cone.
Instead, $i^-$  is shifted to infinity in these coordinates
(we  observe this directly when applying the coordinate transformation \eq{Mink_coord_trafo}).
We also note that case (ii) of Lemma~\ref{lemma_div_i0} is violated while (i) is fulfilled (here we have $\partial_{\tau}g_{rr}|_{\scri^-}=2/r^2$ whence  $\nu^{\tau}\partial_{\tau}g_{rr}=2/r$).

\subsection{Cylinder representation and a priori restrictions on the gauge functions}
\label{section_apriori}

We want to derive restrictions on the asymptotic behavior of the gauge functions appearing in the conformal Gauss gauge scheme at $I^-$, necessary to obtain a spacetime
which admits a smooth cylinder representation of spatial infinity.
As already mentioned before, to  obtain some representation of spatial infinity the differential of the conformal factor $\Theta$ needs to become
zero somewhere along the null geodesic generators of $\scri^-$. We will choose the $r$-coordinates in such a way that spatial infinity is
located at $r=0$. It follows that the gauge  function $\Theta^{(1)}$ needs to satisfy $\Theta^{(1)}=o(1)$.
We are  interested in the construction of smooth spacetimes, which admit a smooth extension through null infinity,  spatial infinity, which we want to represent
as a cylinder,  and therefore also through its critical sets.
This leads to more restrictions than those obtained  Section~\ref{sec_spatial_inf}.
First of all we need to require  
\begin{equation}
\Theta^{(1)}=   \mathfrak{O}(r)
\,.
\end{equation}
Here the symbol $\mathfrak{O}(r)$ is defined as follows: We say that a function $f= \mathfrak{O}(r^n)$, $n\geq 0$, if it is a smooth function of $r$
and $x^{\mathring A}$, and if it Taylor expansion at $r=0$ starts with a term of $n$th-order.
We say that $f= \mathfrak{O}(r^{-n})$ if $r^n f= \mathfrak{O}(1)$.

Let us now  focus on  the specifics of the cylinder representation.
It is obtained by imposing a specific behavior on the gauge functions \eq{gauge_data_scri} near spatial infinity.
A characteristic feature of the cylinder representation is that the  Riemannian metric $g_{\mathring A\mathring B}|_{\scri^-}$ does not degenerate
at spatial infinity (as compared to e.g.\ the point representation of spatial infinity).
It follows from \eq{adapted_null} and \eq{div_conffactor} that  $g_{\mathring A\mathring B}|_{\scri^-}$ satisfies
\begin{equation}
g_{\mathring A\mathring B}|_{\scri^-} =  e^{\int ^r\theta^+\mathrm{d}\hat r }s_{\mathring A\mathring B}
\,.
\end{equation}
We thus need to require
\begin{equation}
-\infty< {\int ^{I^-}\theta^+\mathrm{d} r }< \infty
\,,
\end{equation}
i.e.
\begin{equation}
\theta^+ = \mathfrak{O}(1)
\,.
\end{equation}
In a conformal Gauss gauge the conformal factor satisfies globally $\Theta=\Theta^{(1)} (1+\tau) + \Theta^{(2)}(1+\tau)^2$.
To end up with a spacetime where $\scri^+=\{\tau=+1\}$ the relation
\begin{equation}
\Theta^{(2)}=-\frac{1}{2}\Theta^{(1)}
\label{conf_fac_relation}
\end{equation}
needs to be satisfied.

``Natural'' requirements on $\Theta$ at $I$ are, as on $\scri$, $\Theta|_I=0$ and $\mathrm{d}\Theta|_I\ne 0$.
The gauge function $\Theta^{(1)}$ therefore needs to satisfy the following condition,
\begin{equation}
\Theta^{(1)}= \Theta^{(1,1)}(x^{\mathring A}) r +  \mathfrak{O}(r^2)
\,, \quad  \Theta^{(1,1)}>0
\,.
\label{expansion_conf_factor}
\end{equation}
It is clear that in our smooth setting we need to impose
\begin{equation}
f_a |_{\scri^-}=  \mathfrak{O}(1)
\,.
\end{equation}
Morever, we require the frame coefficients (which appear as unknowns in the GCFE) to be regular at spatial infinity.
This will be the case if (cf.\ \eq{frame_field1}-\eq{frame_field3})
\begin{equation}
\nu^{\tau} =   \mathfrak{O}(1)\,,\quad   \nu_{\mathring A} =  \mathfrak{O}(1)
\,.
\label{apriori_nutau}
\end{equation}
We apply Lemma~\ref{lemma_div_i0}
to deduce that  necessarily
\begin{equation}
{\int ^{I^-}\kappa\mathrm{d}\hat r }=\infty
\,.
\end{equation}
It follows from \eq{Gamma110_gen} that
\begin{equation}
(\partial_r+\kappa)\nu^{\tau}=  \mathfrak{O}(1)
\quad \overset{\eq{apriori_nutau}}{\Longleftrightarrow} \quad
\kappa\nu^{\tau}=  \mathfrak{O}(1)
\,,
\end{equation}
which is only possible if
\begin{equation}
\nu^{\tau} =   \mathfrak{O}(r)
\,.
\end{equation}
Because of this behavior, the frame vectors $e_0$ and $e_1$ \eq{frame_field1}-\eq{frame_field2} become linearly depend at spatial infinity
which implies that $I$ is a total characteristic.
This also implies that $\nu_{\tau}$ diverges at spatial infinity. Let us impose the condition that this divergence is as weak as possible,
\begin{equation}
\nu^{\tau} = \nu^{\tau(1)}r+  \mathfrak{O}(r^2)
\,, \quad \nu^{\tau(1)} \ne 0 
\,,
\end{equation}
equivalently,
\begin{equation}
\mathrm{d}\sqrt{-\det g^{\sharp}}|_{I^-}
%=\sqrt{\det \not g^{\sharp}}\partial_r |\nu^{\tau}|\mathrm{d}r
 \ne 0
\,.
\label{additional_assumption}
\end{equation}
Taking  \eq{important_constr} into account it follows   that $\kappa$ cannot  diverge faster than  $r^{-1}$ and that
\begin{equation}
\kappa= -\frac{2}{r}+  \mathfrak{O}(1)
\,.
\end{equation}
In particular,  any affine parameter along the null geodesics generating $\scri^-$ diverges when approaching $I^-$,

\begin{equation}
\kappa_{\mathrm{aff}}=0 \quad \Longleftrightarrow \quad
r_{\mathrm{aff}}(r,x^{\mathring A}) = \Big(\frac{\partial r_{\mathrm{aff}}}{\partial r}\Big)\Big|_{r=r_0} \int_{r_0}^re^{\int_{r_0}^{\hat r}\kappa\mathrm{d}\hat{\hat r}} \mathrm{d}\hat r + r_{\mathrm{aff}}|_{r=r_0}
\,.
\label{affine_para_trafo}
\end{equation}
From the trace of \eq{GammaAB0_gen}  we deduce that 
\begin{equation}
\theta^-=  \mathfrak{O}(r)
\,.
\end{equation}
Moreover,  \eq{expr_Theta2} together with \eq{conf_fac_relation} gives  ($\langle \ell,f\rangle|_{\scri^-} =\nu_{\tau}f_1$)
\begin{equation}
 f_1|_{\scri^-}  =
1-(\Theta^{(1)} )^{-1}  (\partial_r +\kappa )(\nu^{\tau} \Theta^{(1)}) 
\,,
\label{restriction_f1}
\end{equation}
i.e.\ to make sure that $\scri^+=\{\tau=+1\}$ the freedom to choose $f_1|_{\scri^-}$ is lost.

%
%In fact, it is a  characteristic feature of the cylinder  that $\Theta_2$ goes to zero as well.
%%
%\begin{equation}
% \Theta_1|_I=\Theta_2|_I=0
%\,.
%\end{equation}
%
Because the conformal factor $\Theta$ vanishes on $I$,  the equations \eq{evolution1}-\eq{evolution7} for connection coefficients, Schouten tensor and frame field decouple, on the cylinder,  from
those for the rescaled Weyl tensor (cf.\ \cite{F3}), 
%which is a  characteristic feature of the cylinder,
%
\begin{align}
\partial_{\tau}\widehat L_{a0} |_I
&= -  \widehat\Gamma_a{}^b{}_0\widehat L_{b0} 
\,,
\label{ev_eqn_gauge_1}
\\
\partial_{\tau}\widehat L_{ab} |_I
&=-  \widehat\Gamma_a{}^c{}_0\widehat L_{cb} 
\,,
\label{ev_eqn_gauge_2}
\\
\partial_{\tau}\widehat\Gamma_{a}{}^0{}_b |_I
 &=
-  \widehat\Gamma_c{}^0{}_b\widehat\Gamma_{a}{}^c{}_{0} + \widehat L_{ab}   
\,,
\label{ev_eqn_gauge_3}
\\
\partial_{\tau}\widehat\Gamma_{a}{}^1{}_b |_I
 &=
-  \widehat\Gamma_c{}^1{}_b\widehat\Gamma_{a}{}^c{}_{0} +\delta^1{}_{b}\widehat L_{a0} 
\,,
\label{ev_eqn_gauge_4}
\\
\partial_{\tau}\widehat\Gamma_{1}{}^A{}_B |_I
 &=
-  \widehat\Gamma_c{}^A{}_B\widehat\Gamma_{1}{}^c{}_{0} + \delta^A{}_{B}\widehat L_{10}   
\,,
\label{ev_eqn_gauge_5}
\\
\partial_{\tau}\widehat\Gamma_{A}{}^B{}_C |_I
 &=
-  \widehat\Gamma_D{}^B{}_C\widehat\Gamma_{A}{}^D{}_{0} + \delta^B{}_{C}\widehat L_{A0}  
\,,
\label{ev_eqn_gauge_6}
\\
\partial_{\tau}e^{\mu}{}_a|_I
&= -\widehat\Gamma_{a}{}^0{}_{0} \delta^{\mu}{}_0
 -\widehat\Gamma_{a}{}^b{}_{0} e^{\mu}{}_b
\,.
\label{ev_eqn_gauge_7}
\end{align}
The divergence of  $\nu_{\tau}\equiv (\nu^{\tau})^{-1}$ does not matter  as  the frame field remains regular and
the metric itself does no appear as an unknown in the GCFE.

Finally, it follows from \eq{GammaAB0_gen} that necessarily
\begin{equation}
\Xi_{AB} = \Xi^{(1)}_{AB} r+ \mathfrak{O}(r^{2}) 
\,.
\label{apriori_Xi}
\end{equation}

\begin{lemma}
\label{lemma_apriori}
The gauge data need to satisfiy the following a priori restrictions to obtain a smooth representation of spatial infinity as a cylinder $I=\{r=0, |\tau|<1\}$  and a smooth representation of null infinity  $\scri^{\pm}=\{\tau=\pm 1, r>0\}$, which, in addition, satisfies \eq{additional_assumption},
\begin{align}
\nu_{\tau} = \nu_{\tau}^{(1)}r^{-1} +  \mathfrak{O}(1)
\,,\quad  \nu_{\mathring A} =  \mathfrak{O}(1)
\\
f_a =  \mathfrak{O}(1)
\,,\quad 
\kappa=-\frac{2}{r} +   \mathfrak{O}(1)
\,, \quad 
\theta^-=  \mathfrak{O}(r)
\,, 
\\
\Theta^{(1)}= \Theta^{(1,1)}(x^{\mathring A}) r +  \mathfrak{O}(r^2)
\,,
\end{align}
where $ \nu_{\tau}^{(1)}\ne0$ and $\Theta^{(1,1)}>0$.
Moreover, the data $\Xi_{AB}$ need to be of the form \eq{apriori_Xi}, and the gauge function $f_1|_{\scri^-}$ needs to fulfill \eq{restriction_f1}.
%\tim{non-deg of $\nabla_i\nabla_j\Theta$}
\end{lemma}

\begin{remark}
{\rm
In a next step these expansions are inserted into the constraint equations  computed in Appendix~\ref{app_charact_constraints}.
In turns out that further restrictions need to be imposed to make sure that the restriction of the rescaled Weyl tensor is bounded
at $I^-$ and  does not produce logarithmic terms there.

However, in view of an analysis of the constraint equations on the cylinder it is very convenient if the gauge functions approach the
``Minkowskian values'' \eq{Mink_gauge}. In particular, this makes sure that the system \eq{ev_eqn_gauge_1}-\eq{ev_eqn_gauge_7} can  be solved explicitly.
% and  these  additional restrictions are not restrictive enough.
The analysis of the no-logs condition will therefore be carried out only for gauge functions of a form as in Definition~\ref{dfn_weak_asympt_gauge} \& \ref{dfn_asympt_Minik_gauge}  below.
}
\end{remark}

\subsection{Yet another gauge freedom}
\label{sec_yet_another}

Before we proceed it is important to note that there is still some gauge freedom left.
We have already mentioned that the gauge function 
$\kappa$ does not fully determine  the $r$-coordinate.
When transforming to a prescribed $\kappa$ via the transformation $r\mapsto \tilde r =\tilde r (r, x^{\mathring A})$ one solves a second-order ODE,  so
 there remains the  freedom to choose the  integration functions.
The precise role of these integration functions depends on the asymptotic behavior  of $\kappa$ near spatial infinity.
In a  conformal Gauss gauge which satisfies \eq{main_gauge0B} below 
we can work out   what these gauge freedom corresponds to.
For this let us assume that all the other gauge data have already been transformed to their desired values.

We consider a transformation as in \eq{kappa_trafo1} which leads us to the ODE  \eq{kappa_trafo3}
\begin{equation}
\frac{\partial^2 r}{\partial \widetilde r^2}
=\frac{\partial r}{\partial \widetilde r}\big[\kappa(\widetilde r)-2\partial_{\widetilde r}\log\psi(\widetilde r)\big]
  - \Big(\frac{\partial r}{\partial \widetilde r}\Big)^2\kappa( r(\widetilde r))
\,,
\end{equation}
with 
\begin{equation}
\psi(\widetilde r) =\frac{\partial\widetilde r}{\partial r} \frac{\nu^{\tau}(r(\widetilde r))}{ \nu^{\tau}(\widetilde r)}\frac{\Theta^{(1)}(r(\widetilde r))}{ \Theta^{(1)}(\widetilde r)}
\,.
\end{equation}
Here, we want to solve this equation from $I^-$.
In a conformal Gauss gauge which satisfies \eq{main_gauge0B} below  it has the form
\begin{equation}
\frac{\partial^2 r}{\partial \widetilde r^2}
=
\Big(\frac{\partial r}{\partial \widetilde r} \Big)^2\Big( \frac{2}{r} +\mathfrak{O}(r)\Big)
-\frac{\partial r}{\partial \widetilde r} \Big( \frac{2}{\widetilde r}+\mathfrak{O}(\widetilde r)\Big) 
\,.
\label{some_trafo1}
\end{equation}
Set $u:=\partial_{\widetilde r}\log(\tilde r/ r)$ and  $v:=r/\tilde r$. Then this singular ODE becomes a regular first-order system,
\begin{align}
 \partial_{\widetilde r} u
=&-u^2- v^2(1-u\widetilde r)^2\mathfrak{O}((v\widetilde r)^0) +(1-u\widetilde r) \mathfrak{O}(\widetilde r^0) 
\,,
\label{some_trafo2}
\\
\partial_{\widetilde r}v =&-uv
\,.
\label{some_trafo3}
\end{align}
The solution is of the form
\begin{equation}
  r =f_{p,q}(\widetilde r,\mathring x^A)=\mathfrak{O}(\widetilde r)
\,, \quad \text{where} \quad  \partial_{\widetilde r}f_{p,q}|_{I^-}=p(x^A)>0
\quad \text{and}\quad   \partial^2_{\widetilde r}f_{p,q}|_{I^-}=q(x^A)
\end{equation}
are the initial data.
Note that this transformation does not change the location of $I^-=\{r=0\}$.
(In the special case where  $\kappa=-2/r$ the solution can be determined explicitly, $f_{p,q}=2p^2 \widetilde r/(2p- q\widetilde r)$.)

One then proceeds us described in Section~\ref{sec_realization}, where coordinate and conformal transformations are chosen in such a way that the
other gauge data remain invariant.
Under these transformations
\begin{equation}
\widetilde g_{\mathring A\mathring B} (\widetilde r, x^{\mathring C})|_{\scri^-}\mapsto  g_{\mathring A\mathring B} ( r, x^{\mathring C})|_{\scri^-}=(\psi(\widetilde r( r),x^{\mathring C}))^2g_{\mathring A\mathring B} (\widetilde  r( r), x^{\mathring C})
\,,
\end{equation}
We have
\begin{equation}
\psi|_{I^-} =p(x^{\mathring A})
\,, 
\end{equation}
whence
\begin{equation}
 g_{\mathring A\mathring B} |_{I^-}=(p(x^{\mathring A}))^{2}\widetilde g_{\mathring A\mathring B} 
\,.
\label{conf_class_gauge}
\end{equation}
The gauge freedom coming along with  $p(x^{\mathring A})$  can therefore be employed to conformally rescale $g_{\mathring A\mathring B} $
in any convenient manner.
Let us  consider the behavior of $\Xi_{\mathring  A\mathring  B}$  under the conformal and coordinate transformations  of Section~\ref{sec_realization}. 
 Since we want to leave $g_{\mathring A\mathring B} $ invariant we set $p(x^{\mathring A})=1$. Then
\begin{align}
 \Xi_{\mathring  A\mathring  B}
 =&
-2\Big(\frac{\partial \widetilde x^{\alpha}}{\partial x^A}\frac{\partial \widetilde x^{\beta}}{\partial x^B}\frac{\partial  r}{\partial  \widetilde x^{\gamma}}
\widetilde\Gamma^{\gamma}_{\alpha\beta}+ \frac{\partial  r}{\partial  \widetilde x^{\alpha}}\frac{\partial^2 \widetilde x^{\alpha}}{\partial x^A\partial x^B}
\Big)_{\mathrm{tf}}
\nonumber
\\ 
=&
-2\Big(\frac{\partial \widetilde r}{\partial x^A}\frac{\partial \widetilde r}{\partial x^B}\frac{\partial  r}{\partial \widetilde r}
\widetilde\Gamma^{r}_{rr}
+ 2\frac{\partial \widetilde r}{\partial x^A}\frac{\partial  r}{\partial  \widetilde r}
\widetilde\Gamma^{r}_{r\mathring B}
+ 2\frac{\partial \widetilde r}{\partial x^A}\frac{\partial  r}{\partial  \widetilde x^{\mathring C}}
\widetilde\Gamma^{\mathring C}_{r\mathring B}
+\frac{\partial  r}{\partial \widetilde r}\widetilde\Gamma^{r}_{\mathring A\mathring B}
+\frac{\partial  r}{\partial  \widetilde x^{\mathring C}}\widetilde\Gamma^{\mathring C}_{\mathring A\mathring B}
+ \frac{\partial  r}{\partial \widetilde  r}\frac{\partial^2\widetilde r}{\partial x^A\partial x^B}
\Big)_{\mathrm{tf}}
\nonumber
\\ 
=&
 \frac{\partial  r}{\partial \widetilde r}\Big(-2(\widetilde\kappa -\widetilde \theta^+)\rnabla_{\mathring A}\widetilde r\rnabla_{\mathring B}\widetilde r
+2 \widetilde\xi_{(\mathring A}\rnabla_{\mathring B)}\widetilde r
+\widetilde\Xi_{\mathring A\mathring B}
-2\rnabla_{\mathring A}\rnabla_{\mathring B}\widetilde r
\Big)_{\mathrm{tf}}
\nonumber
\\
=&
(1+q \widetilde r)\Big(\widetilde \Xi_{\mathring A\mathring B}-\widetilde r^2 \widetilde\xi_{(\mathring A}\rnabla_{\mathring B)}q
+\widetilde r^2\rnabla_{\mathring A}\rnabla_{\mathring B}q
\Big)_{\mathrm{tf}} + \mathfrak{O}( \widetilde r^3)
\,.
\label{trafo_Xi}
\end{align}
In a conformal Gauss gauge which satisfies \eq{main_gauge0B}  below  we have
$\xi_{\mathring A}=\mathfrak{O}(\widetilde r)$ (cf.\  \eq{constraint_xiA}).

In Section~\ref{sec_solution_constr_gen}
 we will  see that boundedness of the rescaled Weyl tensor at $I^-$ requires the data $\Xi_{\mathring A\mathring B}$ to be of the form
$\Xi_{\mathring A\mathring B}=\Xi^{(2)}_{\mathring A\mathring B} r^2 + \mathfrak{O}(r^3)$. It follows from \eq{trafo_Xi} that the leading order term transforms as
\begin{equation}
\Xi_{\mathring A\mathring B}^{(2)}\mapsto \Xi^{(2)}_{\mathring A\mathring B}
+(\rnabla_{\mathring A}\rnabla_{\mathring B}q)_{\mathrm{tf}}
\,
\end{equation}
It is convenient to set
\begin{equation}
v_{\mathring A} :=\rnabla_{\mathring B}\Xi_{\mathring A}{}^{\mathring B} 
\,.
\end{equation}
It follows from the Hodge-decomposition theorem (cf.\ e.g.\ \cite{hodge}) that on a closed Riemannian manifold $(\Sigma,h)$ a (smooth) 1-form $\omega$
 admits the decomposition
\begin{equation}
\omega_A = \rnabla_A\ul \omega+  \not\hspace{-.1em}\epsilon_A{}^B\rnabla_B\ol \omega+ \lambda_A
\quad \text{with}
\quad \Delta_h \lambda_A=0
\,,
\end{equation}
where $ \not\hspace{-.1em}\epsilon_{AB}$ denotes the volume form associated  with $h$.
If $(\Sigma,h)$ is compact and has non-negative Ricci curvature which is positive at one point it follows from Bochner's theorem (cf.\ e.g.\ \cite{bochner}) that
all harmonic 1-forms are identically zero.
In that case any 
(smooth)
vector field  admits a decomposition of the form
\begin{equation}
\omega_A = \rnabla_A\ul \omega +   \not\hspace{-.1em}\epsilon_A{}^B\rnabla_B\ol \omega
\,.
\label{vf_decomposition}
\end{equation}
In particular  on a Riemannian 2-sphere
 all  vector fields can be decomposed this way, whence
the expansion coefficients of $v_{\mathring A}$ can be written as
\begin{equation}
v^{(n)}_{\mathring A}= \rnabla_{\mathring A}\ul{v}^{(n)} + \not\hspace{-.1em} \epsilon_{\mathring A}{}^{\mathring B}\rnabla_{\mathring B}\ol{v}^{(n)}
\,.
\end{equation}
Consider a gauge where $ g_{\mathring A\mathring B} |_{I^-}$ is the standard metric $s_{\mathring A\mathring B}$ on $S^2$.
We observe that in that case $\ul v^{(2)}$ transforms as
\begin{equation}
\ul v^{(2)} \mapsto  \ul v^{(2)} + \frac{1}{2}(\Delta_s +2)q 
\,.
\label{add_gauge_trafo}
\end{equation}
Recall  that $v^{(2)}_{\mathring A}$ arises  as the divergence of a symmetric trace-free tensor, $v^{(2)}_A=\mcD^B\Xi^{(2)}_{AB}$.
As a consequence of  York splitting (cf.\ e.g.\ \cite{c_lecture}), the fact that there are no non-trivial TT-tensors on $\mathbb{S}^2$, and
Hodge decomposition \cite{beig}, on $\mathbb{S}^2$, 
  any   symmetric trace-free  tensor $\mathfrak{x}$ admits a decomposition of the form
\begin{equation}
\mathfrak{x}_{AB} = (\mcD_{(A}\mathfrak{x}_{B)})_{\mathrm{tf}}= (\mcD_{A}\mcD_{B}\ul {\mathfrak{x}})_{\mathrm{tf}}
+ \epsilon_{(A}{}^C \mcD_{B)}\mcD_{C}\ol{ \mathfrak{x}}
\,,
\label{tensor_decomposition}
\end{equation}
for  appropriately chosen 1-form $\mathfrak{x}_A$ and functions $\ul {\mathfrak{x}}$ and $\ol {\mathfrak{x}}$.
Its divergence reads
\begin{equation}
\mcD^B\mathfrak{x}_{AB} =\frac{1}{2}\mcD_A(\Delta_s+2)\ul{\mathfrak{x}}
+\frac{1}{2} \epsilon_{A}{}^B\mcD_{B}(\Delta_s+2)\ol{\mathfrak{x}}
\,.
\label{tensor_decomposition2}
\end{equation}
It follows that $\ul v^{(2)}$ and $\ol v^{(2)}$ cannot have $\ell=0,1$-spherical harmonics in their harmonic decomposition.
%\tim{!!!}
We will make frequently use of the Hodge decompositions described here.

By way of summary, assuming a conformal Gauss gauge which satisfies  \eq{main_gauge0B} below  (and $g_{\mathring A\mathring B} |_{I^-}=s_{\mathring A\mathring B}$)
the remaining gauge freedom  can be employed to prescribe
 the divergence of $v^{(2)}_{\mathring  A}$ by solving a Laplace-like equation.
Since  $\ul v^{(2)}$ does not contain  $\ell=1$-spherical harmonics,  the kernel of the operator in \eq{add_gauge_trafo} does not
provide any obstructions.
One then may proceed as described in Section~\ref{sec_realization}  to transform the remaining gauge functions %$f_{\scri^-}$
to their desired form.

\subsubsection{Dual mass aspect}
\label{sec_dual_mass}

It is convenient so set
\begin{equation}
N:=\frac{1}{8}\Delta_s\ol v^{(2)} \quad \Longleftrightarrow \quad 
N=-\frac{1}{8}\epsilon^{AB}\mcD_Av^{(2)}_B
\,.
\label{definig_eqn_N}
\end{equation}
Later on we shall see (cf.\ \eq{data_gen1}-\eq{data_gen2})  that the function $N$ can be identified with the leading order term of a certain rescaled Weyl tensor component at $I^-$,
and this  component is dual to the one which involves the (ADM) mass aspect $M$, by which we mean the limit of the Bondi mass aspect at $I^-$.
In the case of e.g.\ the Taub-NUT spacetime,  cf.\ \cite{gp}, this component is  constant and can be identified with the NUT-parameter (note that this spacetime is not asymptotically flat, whence $N$ can be constant and non-zero, which it cannot be in our setting).
In this sense $N$ may  be regarded as a generalized NUT-like or twist parameter. 

In  \cite{dual}, cf.\ \cite{ashtekar}, a so-called `dual Bondi 4-momentum' has been introduced, leading in particular to the notion of a
 `dual Bondi mass', or `magnetic Bondi mass'. As the Bondi mass it is defined  as the integral of a `dual Bondi mass aspect' over cuts of $\scri$.
Since  the function $N$ arises as a limit thereof at $I^-$, we will call it \emph{dual (ADM) mass aspect}.

It follows immediately from \eq{definig_eqn_N} that the \emph{dual mass}, i.e.\ the  integral of $N$ over $I^-\cong \mathbb{S}^2$ vanishes.
This is in accordance with the results in \cite{ashtekar}, that the dual mass has to vanish in a spacetime with a regular $\scri$ with topology $\mathbb{R}\times S^2$.

%In fact, it can be identified with the leading order term of the imaginary part of the spinor curvature component $\Psi_2$ \cite{pr} (while the mass aspect $M$ corresponds to its real part).

\subsection{Asymptotically Minkowski-like conformal Gauss gauge}

In Section~\ref{section_apriori} we have derived some  a priori restrictions on the gauge functions in order  to end up with a smooth cylinder
representation of spatial infinity.
However, it is useful and convenient to impose some weak additional restrictions on the asymptotic behavior of the gauge functions at $I^-$.

The  equations for Weyl connection, Schouten tensor etc. derived in Section~\ref{sec_evolution1} \& \ref{sec_evolution2}
involve terms which are quadratic in the unknowns.
This implies that the structure of the equations for the $n$th-order radial derivatives on the cylinder depends crucially on terms of $0$th-order
(in particular of connection and frame coefficients).

In the case of a smooth critical set $I^-$, the integration functions for the transport equations on the cylinder 
are determined at $I^-$ by the limit of the corresponding fields on $\scri^-$.
The initial data for the  $0$th-order equations,  \eq{ev_eqn_gauge_1}-\eq{ev_eqn_gauge_7}, which are non-linear,   are determined
by the asymptotic behavior of the gauge functions on $\scri^-$.
In order to study the transport equations for radial derivatives of order $m\geq 1$, a simple, explicit form for the terms of $0$th-order on $I$
is beneficial. We will therefore fix the leading order term of the asymptotic expansion of the gauge functions
(the condition  below on the next-to-leading order term for $\nu_{\tau}$ corresponds to a restriction on the leading-order of the divergence $\theta^+$).
Guided by the  representation \eq{Mink_cyl} of the Minkowski spacetime  we will restrict attention henceforth to gauge functions of the following form, for which, indeed,  \eq{ev_eqn_gauge_1}-\eq{ev_eqn_gauge_7} can be  solved explicitly, which will be accomplished in Section~\ref{sec_connection_schouten_I}.

\begin{definition}
\label{dfn_weak_asympt_gauge}
We call a conformal Gauss gauge ``weakly asymptotically Minkowski-like'' if the gauge functions are of the following form,
\begin{align}
\nu_{ \tau}=\frac{1}{r}+\frac{ \Theta^{(1,2)}}{2}+ \mathfrak{O}(r)
\,, \enspace 
 \nu_{\mathring A}=\mathfrak{O}(r)\,, 
%\quad  g_{\tau\tau}|_{\scri^-} =-1\,, 
\enspace
\Theta^{(1)}=2r +\Theta^{(1,2)}r^2 +  \mathfrak{O}(r^3) 
\,,\enspace
 \kappa=-\frac{2}{r} + \mathfrak{O}(r)
\,,
\label{main_gauge0B} 
\\
\theta^-= \mathfrak{O}(r^3)\,,\enspace 
g_{\mathring A\mathring B}|_{I^-}=s_{\mathring A\mathring B}
%\,, \quad  \mathring\nabla^{\mathring A}v^{(2)}_{\mathring A}=0
\,, \enspace 
  f_{r}|_{\scri^-} =\frac{1}{r} + \mathfrak{O}(1)
\,, \quad  f_{\mathring A}|_{\scri^-}=\mathfrak{O}(r) 
\label{main_gauge0B2} 
%\,,
%\\
%e_{0*} = \partial_{\tau}
%\,,
%\quad
%e_{1*} = \partial_{\tau} + r\partial_r
%\,,
%\quad
%e_{A*}= \mathring e^{\mathring A}{}_A \partial_{\mathring A}
\,.
%\label{main_gaugeB}
\end{align}
\end{definition}
%\tim{gauge data enter on the same level as the non-trivial physical data.... radiation field}
\begin{remark}
{\rm
In Section~\ref{section_apriori} we have assumed that $\scri^+$ is located at $\{\tau=+1\}$, so that \eq{conf_fac_relation} holds,  in order to motivate
\eq{expansion_conf_factor}. Since we are mainly interested in the behavior of the fields near $I^-$  we do not include \eq{conf_fac_relation}
in this definition so that all gauge functions are independent.
}
\end{remark}

It turns out
 that connection and frame coefficients  on $I$ do not depend on the physical, non-gauge data, while their 1st-order radial derivatives (and the restriction to $I$ of the rescaled Weyl tensor) depend on mass and dual mass aspect. Only the 2nd-order ones (the 1st-order ones of the rescaled  Weyl tensor) depend on the radiation field (and the angular momentum). 
Since we know that  \eq{Mink_cyl} provides a smooth 
representation of Minkowski, it therefore seems reasonable to expect that \eq{main_gauge0B}-\eq{main_gauge0B2}  do not impose 
restrictions on the non-gauge data to produce a spacetime which admits a smooth critical set $I^-$.

For later reference, we also  add the following
\begin{definition}
\label{dfn_asympt_Minik_gauge}
We call a conformal Gauss gauge ``asymptotically Minkowski-like at each  order'' if the gauge functions are of the following form,
\begin{align}
\nu_{ \tau}=\frac{1}{r}+\mathfrak{O}(r^{\infty}) 
\,, \enspace 
 \nu_{\mathring A}=\mathfrak{O}(r^{\infty}) \,, 
%\quad  g_{\tau\tau}|_{\scri^-} =-1\,, 
\enspace
\Theta^{(1)}=2r +  \mathfrak{O}(r^{\infty}) 
\,,\enspace
 \kappa=-\frac{2}{r} +\mathfrak{O}(r^{\infty}) 
\,,
\label{main_gauge0Bs} 
\\
\theta^-= \mathfrak{O}(r^{\infty}) \,,\enspace 
g_{\mathring A\mathring B}|_{I^-}=s_{\mathring A\mathring B}
\,, \enspace  \mcD^{\mathring A}v^{(2)}_{\mathring A}=0
\,, \enspace 
  f_{r}|_{\scri^-} =\frac{1}{r} +\mathfrak{O}(r^{\infty}) 
\,, \quad  f_{\mathring A}|_{\scri^-}=\mathfrak{O}(r^{\infty}) 
\label{main_gauge0B2s} 
%\,,
%\\
%e_{0*} = \partial_{\tau}
%\,,
%\quad
%e_{1*} = \partial_{\tau} + r\partial_r
%\,,
%\quad
%e_{A*}= \mathring e^{\mathring A}{}_A \partial_{\mathring A}
\,,
%\label{main_gaugeB}
\end{align}
i.e.\ if the gauge functions have the same expansions at $I^-$ as in \eq{Mink_gauge}-\eq{Mink_gauge2}.
\end{definition}

We will use this gauge in Section~\ref{section_aMlcGg} to establish sufficient conditions for the non-appearance of logarithmic terms at the critical sets.

\section{Appearance of log terms: Approaching $I^-$ from~$\scri^-$}
\label{section3}

Our goal is as follows: We assume we have been given asymptotic initial data, which will be the radiation field  on  $\scri^-$
supplemented by certain ``integration functions'' at $I^-$ such as the (ADM) mass aspect, cf.\ Appendix~\ref{app_ADM}.
Then we solve the characteristic constraint equations to determine all the relevant data for the evolution equations,
and analyze the appearance of logarithmic terms at $I^-$.

A related problem for an ordinary (i.e.\ non-asymptotic) characteristic inital value problem with one initial surface going all the way to null infinity has been anaylzed in \cite{ChPaetz2, ttp3,  CCTW}. There it turns out that, in an appropriate gauge, if the constraint equations do not produce logarithmic terms, the solution will be smooth, in particular higher-order transverse derivatives
will not pick up log terms when approaching null infinity.

When approaching spatial infinity the situation turns out to be  completely different, as   logarithmic terms can appear 
in transverse derivatives  of  arbitrary  high orders with all  lower orders being  smooth.
We  thus need to take higher order transverse derivatives into account as well, and analyze their behavior when approaching $I^-$, which makes the problem  significantly harder to deal with.
We will do this  by determining expansions of all the relevant fields  on $\scri^-$  (and transverse derivatives thereof) when approaching $I^-$.
Later on, we will study the same issue when approaching $I^-$ from the cylinder~$I$.

\subsection{Solution of the asymptotic constraint equations}
\label{sec_solution_constr_gen}

We assume  a weakly  asympotically Minkowski-like conformal Gauss gauge \eq{main_gauge0B}-\eq{main_gauge0B2}.
The constraint equations in adapted null coordinates are listed in Appendix~\ref{app_charact_constraints}.%
\footnote{Alternatively, one could analyze the constraints directly in a conformal Gauss gauge and the associated frame. Since 
the constraint equations in adapted null coordinates have  been derived in \cite{ttp1},  the  coordinates are adapted to the geometry of $\scri^-$, and since we also have identified the remaining gauge degrees of freedom using coordinates,   it seems convenient to start with them and determine
the behavior in  the conformal Gauss gauge afterwards.}
Recall that the data $\Xi_{AB}$ need to be of the form \eq{apriori_Xi}.
Those constraint equations \eq{constraint1}-\eq{constraint12}  which do not involve the radiation field can be straightforwardly solved,
\begin{align}
 \Sigma =& 2r^2+  \Sigma^{(4)}r^4 +  \Sigma^{(5)}r^5  + \mathfrak{O}(r^6) 
\,,
\\
\theta^+ =&  \theta^{+(1)} r + \theta^{+(2)}r^2 + \mathfrak{O}(r^3)\,, \quad \text{where} \quad \theta^{+(1)}=2(\kappa^{(1)}+\Sigma^{(4)})
\,,
\\
g_{\mathring A\mathring B}|_{\scri^-} =& \Big(1 +\frac{1}{2}\theta^{+(1)}r^2+\frac{1}{3} \theta^{+(2)} r^3\Big)s_{\mathring A\mathring B}+ \mathfrak{O}(r^4)
\,,
\\
\not \hspace{-0.2em}R =& 2-\frac{1}{2}(\Delta_s + 2)\theta^{+(1)}r^2 
-\frac{1}{3} (\Delta_s+2)\theta^{+(2)}r^3
+ \mathfrak{O}(r^4)
\,,
\\
 L_{rr} |_{\scri^-}
=& -\frac{3}{2} \theta^{+(1)} +  \mathfrak{O}(r)
\,,
\\
 \xi_{\mathring A} 
=&   \mathring\nabla_{\mathring A} \Sigma^{(4)} r^2 +\mathring\nabla_{\mathring A} \Sigma^{(5)} r^3 +  \mathfrak{O}(r^4) 
\,,
\\
 L_{r\mathring A} |_{\scri^-}
 =& -\frac{1}{2}\mathring\nabla_{\mathring A}\theta^{+(1)} r-\frac{1}{2}\mathring\nabla_{\mathring A}\theta^{+(2)} r^2  +  \mathfrak{O}(r^3)
\,,
\\
   g^{\mathring A\mathring B} L_{\mathring A\mathring B } |_{\scri^-}=& 1-\frac{1}{4}(\Delta_s + 2)\theta^{+(1)}r^2 +   \mathfrak{O}(r^3) 
\,,
\\
 L_r{}^r |_{\scri^-}
 =& -\frac{1}{2}+  \frac{1}{4}\Big(\theta^{-(3)}  + \Delta_s\Sigma^{(4)}+\frac{1}{2}(\Delta_s -4)\theta^{+(1)}\Big)r^2
+\frac{1}{4} \Big(2\theta^{-(4)} -4 \theta^{+(2)} 
\nonumber
\\
&
+\Delta_s \Sigma^{(5)} 
-3g^{rr(3)}\theta^{+(1)}
+ \frac{1}{3} (\Delta_s+2)\theta^{+(2)}
\Big)r^3
+ \mathfrak{O}(r^4)
\,.
\end{align}
Here $(\cdot)^{(n)}$ denotes the $n$th-order expansion coefficient at $r=0$.
The values for $\Sigma^{(4)}$,  $\Sigma^{(5)}$ and  $\theta^{+(2)}$ are determined by $\nu_{\tau}$, $\Theta^{(1)}$ and  $\kappa$; the precise relation is irrelevant here.
Integration of \eq{LAB_constraint} and  \eq{LAr_constraint} yields
\begin{align}
( L_{\mathring A\mathring B})_{\mathrm{tf}}|_{\scri^-}
=&
\frac{1}{2}\Xi^{(1)}_{\mathring A\mathring B}
+\mathfrak{O}(r^2)
\,,
\\
 L_{\mathring A}{}^r|_{\scri^-} 
=&
\frac{1}{2}v^{(1)}_{\mathring A} r
 +
\mathfrak{O}(r^2)
\,.
\end{align}
Then we employ \eq{expression_d1A1B}-\eq{expression_d1A10} to obtain
\begin{align}
W_{r\mathring Ar\mathring B}|_{\scri^-} 
 =&\mathfrak{O}(r^{-1})
\,,
\\
W_{r\mathring Ar}{}^r|_{\scri^-}
=& -\frac{1}{4r^2} v^{(1)}_{\mathring A} 
+ \mathfrak{O}(1)
\,,
\end{align}
which implies that the frame component (recall \eq{frame_field1}-\eq{frame_field3})
\begin{equation}
W_{010A}-W_{011A}|_{\scri^-}=
 \nu^{\tau}e^{\mathring A}{}_AW_{r\mathring A r}{}^r+ (\nu^{\tau})^2\nu^{\mathring B}e^{\mathring A}{}_AW_{r\mathring A r\mathring B}
= -\frac{1}{4r} v^{(1)}_{ A} + \mathfrak{O}(r)
\end{equation}
is unbounded at $I^-$ whenever $v^{(1)}_{ A} \ne 0$.  We deduce the regularity condition
\begin{equation}
v^{(1)}_{\mathring A} =0 \quad \Longleftrightarrow \quad \Xi^{(1)}_{\mathring A\mathring B} =0
\,.
\label{boundedness_Weyl}
\end{equation}
In the analysis the \emph{radiation field} $W_{r\mathring A r\mathring B}$ plays a distinguished role; in turns out the the expressions below take the most
compact form when expressed in term of $W_{r\mathring A r\mathring B}$ rather than $\Xi_{\mathring A\mathring B}$, which, tough,
does not comprise the integration functions $\Xi^{(1)}_{\mathring A\mathring B}$ and $\Xi^{(2)}_{\mathring A\mathring B}$.
It is convenient to make the following definitions,
\begin{equation}
%Q_{\mathring A\mathring B} :=W_{r\mathring A r\mathring B}|_{\scri^-} \,, \quad 
w_{\mathring A} := \rnabla^{\mathring B}W_{r\mathring A r\mathring B}|_{\scri^-}\,, \quad 
w^{(n)}_{\mathring A} := \mcD^{\mathring B}W^{(n)}_{r\mathring A r\mathring B}|_{\scri^-}
\,.
\end{equation}
Recall
 that $v^{(n)}_A =\mcD_A \ul v^{(n)} + \epsilon_A{}^B\mcD_B\ol v^{(n)}$, Definition~\ref{definig_eqn_N}, and that $\ul v^{(2)}$ may be regarded as a gauge function.
From  \eq{LAB_constraint},  \eq{LAr_constraint}-\eq{expression_d01AB}
we  obtain ($\epsilon_{\mathring A\mathring B}$ denotes the volume form of the round sphere)
\begin{align}
( L_{\mathring A\mathring B})_{\mathrm{tf}}|_{\scri^-}
=&
-\frac{1}{2}\Big( \Xi^{(3)}_{\mathring A\mathring B}-(\mcD_{\mathring A}\mcD_{\mathring B}\Sigma^{(4)})_\mathrm{tf}\Big)r^2
\nonumber
\\
&
-\frac{1}{2}\Big( 2\Xi^{(4)}_{\mathring A\mathring B}-\Sigma^{(4)}\Xi^{(2)}_{\mathring A\mathring B}
-(\mcD_{\mathring A}\mcD_{\mathring B}\Sigma^{(5)})_\mathrm{tf}
\Big)r^3
 + \mathfrak{O}(r^4)
\,,
\\
 L_{\mathring A}{}^r|_{\scri^-} 
=&
 \frac{1}{2}v^{(2)}_{\mathring A} r^2
+ \frac{1}{2}  v^{(3)}_{\mathring A}r^3 + \frac{1}{4}\mcD_{\mathring A}\Big(\theta^{-(3)} - \theta^{+(1)}  \Big)r^3
+\frac{1}{4} \Big(2  v^{(4)}_{\mathring A} +\mcD_{\mathring A}(\theta^{-(4)}-\theta^{+(2)}) 
\nonumber
\\
&
-\mcD^{\mathring B}\Sigma^{(4)}\Xi^{(2)}_{\mathring A\mathring B}
-\theta^{+(1)}v^{(2)}_{\mathring A}
-g^{rr(3)}\mcD_{\mathring A}\theta^{(+)1}\Big)r^4
+\mathfrak{O}(r^5)
\,,
\\
W_{r\mathring A r\mathring B}|_{\scri^-} 
 =&-\frac{1}{2} \Xi^{(3)}_{\mathring A\mathring B}r^{-1}
+\frac{1}{4}\Big(\mcD_{\mathring A}\mcD_{\mathring B}(\theta^{+(1)} +2\Sigma^{(4)})\Big)_{\mathrm{tf}}r^{-1}-\frac{3}{2} \Xi^{(4)}_{\mathring A\mathring B}
\nonumber
\\
&
+\frac{1}{4}\Big(\mcD_{\mathring A}\mcD_{\mathring B}(\theta^{+(2)} +3\Sigma^{(5)})\Big)_{\mathrm{tf}}
+\frac{3}{8}(\theta^{+(1)} + 2\Sigma^{(4)})\Xi^{(2)}_{\mathring A\mathring B}
+\mathfrak{O}(r)
\,,
\\
w_{\mathring A}
=&-\frac{1}{2}\Big(  v^{(3)}_{\mathring A} - \frac{1}{4}\mcD_{\mathring A}(\Delta_s + 2)(  2\Sigma^{(4)} +\theta^{+(1)} )\Big)r^{-1}
-\frac{3}{8}\Big[4  v^{(4)}_{\mathring A} 
\nonumber
\\
&
-\mcD^{\mathring B}\Big((2\Sigma^{(4)}+ \theta^{+(1)})\Xi^{(2)}_{\mathring A\mathring B}\Big)
 - \frac{1}{3}\mcD_{\mathring A} \Big( (\Delta_s+2)(\theta^{+(2)}+3 \Sigma^{(5)} )
\Big]
+\mathfrak{O}(r)
\,,
\\
W_{r\mathring Ar}{}^r|_{\scri^-}
%=&\underbrace{   \frac{1}{4} v^{(3)}_{\mathring A} - \frac{1}{16}\mathring\nabla_{\mathring A}(\Delta_s + 2)(  2\Sigma^{(4)} +\theta^{+(1)} )}_{=-\frac{1}{2}q^{(-1)}_{\mathring A}}
%\\
%&
%+\frac{1}{4}\Big[2  v^{(4)}_{\mathring A} 
%-\frac{1}{2}\mathring\nabla^{\mathring B}\Big((2\Sigma^{(4)}+ \theta^{+(1)})\Xi^{(2)}_{AB}\Big)
 %- \frac{1}{6}\rnabla_{A} \Big( (\Delta_s+2)(\theta^{+(2)}+3 \Sigma^{(5)} )
%\Big]r
%+\mathfrak{O}(r^2)
%\\
=&-\frac{1}{2}w^{(-1)}_{\mathring A} -\frac{1}{3}w^{(0)}_{\mathring A}r
+\mathfrak{O}(r^2)
\;,
\\
W_{\mathring A\mathring Br}{}^r|_{\scri^-}
% =& 
%\frac{1}{2}\rnabla_{[\mathring A} v^{(2)}_{\mathring B]}  +\underbrace{\frac{1}{2}\rnabla_{[\mathring A} v^{(3)}_{\mathring B]}  r}_{=-\mathring\nabla_{[\mathring A}q^{(-1)}_{\mathring B]}}
%+\mathfrak{O}(r^2)
%\\
 =& 
-2N \epsilon_{\mathring A\mathring B}  -\mcD_{[\mathring A}w^{(-1)}_{\mathring B]}r -\Big(\frac{1}{3}\mcD_{[\mathring A}w^{(0)}_{\mathring B]}
  +\frac{1}{2}\Xi^{(2)}_{[\mathring A}{}^{\mathring C} W^{(-1)}_{\mathring B]r\mathring Cr} -\frac{1}{4}N\theta^{+(1)}  \epsilon_{\mathring A\mathring B}\Big) r^2 +\mathfrak{O}(r^3)
\;.
\end{align}
%
%\begin{equation}
%\epsilon_{\mathring A\mathring B} 
%=\Big(1+\frac{1}{2} \theta^{+(1)}  r^2+\mathfrak{O}(r^3)\Big)\mathring\epsilon_{\mathring A\mathring B}
%\end{equation}
%
%\begin{eqnarray*}
%W_{ABr}{}^r|_{\scri^-} &=& 
%\frac{1}{2} \rnabla_{[A}v^{(2)}_{B]}
%+\frac{1}{8}\tau \rnabla_{[A}v^{(2)}_{B]}
%-\tau\frac{1}{4}\rnabla_{[A}  v^{(2)}_{B]}
%\end{eqnarray*}
%

The ODE \eq{adm_ode}  for $W_r{}^r{}_r{}^r$ ,
\begin{equation}
 \Big(\partial_r+\frac{3}{2} \theta^{+(1)}r+ \mathfrak{O}(r^2)\Big)W_r{}^r{}_r{}^r |_{\scri^-}
=
\frac{1}{2}\mcD^{\mathring A}w^{(-1)}_{\mathring A}
+\frac{1}{3} \mcD^{\mathring A}w^{(0)}_{\mathring A}r
+\frac{1}{2}\Xi^{(2)\mathring A\mathring B}W^{(-1)}_{r\mathring Ar\mathring B}r
+ \mathfrak{O}(r^2)
\,,
\end{equation}
does not produce log-terms. Its solution is of the form
%\tim{$\kappa^{(1)}$ vs $\theta^{+(1)}$}
%
\begin{equation}
W_r{}^r{}_r{}^r |_{\scri^-}
=
2M+ \frac{1}{2}\mcD^{\mathring A}w^{(-1)}_{\mathring A}r
+\frac{1}{2}\Big(\frac{1}{3} \mcD^{\mathring A}w^{(0)}_{\mathring A}
+\frac{1}{2}\Xi^{(2)\mathring A\mathring B}W^{(-1)}_{r\mathring Ar\mathring B}
-\frac{3}{2}M\theta^{+(1)}\Big)r^2
+ \mathfrak{O}(r^3)
\,.
\end{equation}
As explained in Appendix~\ref{app_alternative_data}, the integration function $M$ -- as the ones which appear below --
may be regarded as part of the freely prescribable initial data.
The function  $M$ can be identified with the ADM mass aspect, or rather the limit of the Bondi mass aspect at $I^-$.
There are results \cite{kn2} which show that for a certain class of data  it is this limit, whence  we will call it \emph{(ADM) mass aspect}.

We consider the ODE \eq{adm2_ode}
for $W_{\mathring A}{}^r{}_{r}{}^{r}$,
\begin{equation}
 \Big(\partial_{r}   -\frac{2}{r}+\mathfrak{O}(r)\Big) W_{\mathring A}{}^r{}_{r}{}^{r} |_{\scri^-}
=\mcD_{\mathring A}M+\epsilon_{\mathring A}{}^{\mathring B}\mcD_{\mathring B}N
+\frac{1}{2} \mcD_{\mathring A}\mcD^{\mathring B}w^{(-1)}_{\mathring B}r
-\frac{1}{4}(\Delta_s -1)w^{(-1)}_{\mathring A} r
 +\mathfrak{O}(r^2)
\,.
\end{equation}
The term of order $r$ on the right-hand side produces log-terms. We therefore need to impose the  \emph{no-logs-condition}
\begin{equation}
( \Delta_s - 1)w^{(-1)}_{\mathring A}
- 2\mcD_{\mathring A}\mcD^{\mathring B}w^{(-1)}_{\mathring B}
=0
\,.
\label{no-logs_first}
\end{equation}
Again, we use Hodge decomposition,
\begin{equation}
w^{(n)}_A=\mcD_A\ul w^{(n)} + \epsilon_A{}^B\mcD_B\ol w^{(n)}
\,.
\end{equation}
By taking the divergence and curl of \eq{no-logs_first} it follows that 
\begin{eqnarray}
\Delta_s \Delta_s \ul w^{(-1)} =0 &\Longrightarrow & \Delta_s \ul w^{(-1)} =\mathrm{const.}
\\
\Delta_s \Delta_s \ol w^{(-1)}  = 0  &\Longrightarrow &\Delta_s \ol w^{(-1)} =\mathrm{const.}
\end{eqnarray}
According to Gauss' theorem, on $S^2$ , a solution $\ul w^{(-1)}$ and $\ol w^{(-1)}$, respectively, exists  if and  only if the corresponding constant in the equation vanishes.
In that case  $\ul w^{(-1)} $ and $\ol w^{(-1)} $ need to be constant
and the  no-logs condition \eq{no-logs_first}  becomes
\begin{equation}
w^{(-1)}_{ \mathring A}=0 \quad \Longleftrightarrow \quad 
W^{(-1)}_{r\mathring  Ar\mathring B}=0
\,.
\label{no-logs_condition}
\end{equation}
At this stage it seems remarkable that all  gauge functions, which in principle provide contributions to this order,
 cancel out. In particular they cannot be employed to fulfill the no-logs conditions, at least not at this order
(of course, in principle  it is conceivable that the gauge functions of this order can be used to get rid of log terms which appear in higher orders).
We will return to this observation later on and in particular in Section~\ref{sec_gauge_ind}.

Assuming that \eq{no-logs_condition} holds, the ODE for $W_{\mathring A}{}^r{}_{r}{}^{r}$ takes the form
\begin{align*}
&\hspace{-2em} \Big(\partial_{r}   -2/r+(\kappa^{(1)}+\frac{1}{2}\theta^{+(1)})r+\mathfrak{O}(r^2)\Big) W_{\mathring A}{}^r{}_{r}{}^{r} |_{\scri^-}
\\
=&  
\mathfrak{M}_{\mathring A}
+\Big(-\frac{1}{12}(\Delta_s -3)w^{(0)}_{\mathring A} 
+\frac{1}{6}  \mcD_{\mathring A} \mcD^{\mathring B}w^{(0)}_{\mathring B}
-\frac{3}{4}\mcD_{\mathring A}(M\theta^{+(1)})
  -\frac{3}{4}  \mcD^{B}(N\theta^{+(1)})  \epsilon_{\mathring A\mathring B}
\\
&
-\frac{3}{2}\mcD_{\mathring A}\Sigma^{(4)} M
-\frac{3}{2}N\mcD^B\Sigma^{(4)}\epsilon_{\mathring A\mathring B}\Big)r^2
+\mathfrak{O}(r^3)
\,,
\end{align*}
where we have set 
$$\mathfrak{M}_{\mathring A}:=\mcD_{\mathring A}M+  \epsilon_{\mathring A\mathring B}\mcD^{\mathring B}N\,,
$$
whence, for some integration function $\mathring L_{\mathring A}$,
\begin{align*}
W_{\mathring A}{}^r{}_{r}{}^{r}|_{\scri^-}= &
-\mathfrak{M}_{\mathring A} r+\mathring L_{\mathring A}r^2
+\Big(
(\kappa^{(1)}-\frac{1}{4}\theta^{+(1)}) \mathfrak{M}_{\mathring A} 
-\frac{1}{12}(\Delta_s -3)w^{(0)}_{\mathring A} 
+\frac{1}{6}  \mcD_{\mathring A} \mcD^{\mathring B}w^{(0)}_{\mathring B}
\\
&
-\frac{3}{2}M\mcD_{\mathring A}(\theta^{+(1)}-\kappa^{(1)})
-\frac{3}{2}N\mcD^{\mathring B}(\theta^{+(1)}-\kappa^{(1)})\epsilon_{\mathring A\mathring B}
\Big)r^3
+\mathfrak{O}(r^4)
\,.
\end{align*}
Finally, we consider the ODE \eq{ODE_Wfinal}  which determines $( W_{\mathring A}{}^r{}_{\mathring B}{}^{r} )_{\mathrm{tf}}$,
\begin{align}
& \hspace{-2em}\Big(\partial_{r}-4/r +(2\kappa^{(1)} -\frac{1}{2}\theta^{+(1)}) r+\mathfrak{O}(r^2) \Big) ( W_{\mathring A}{}^r{}_{\mathring B}{}^{r} )_{\mathrm{tf}} |_{\scri^-}
\nonumber
\\
=& 
-(\mcD_{(\mathring A} \mathfrak{M}_{\mathring B)} )_{\mathrm{tf}}   r+\Big((\mcD_{(\mathring A}\mathring L_{\mathring B)} )_{\mathrm{tf}}  
+\frac{3}{2}M\Xi^{(2)}_{\mathring A\mathring B}
+\frac{3}{2}N\Xi^{(2)}_{(\mathring A}{}^{\mathring C}\epsilon_{\mathring B)\mathring C}\Big) r^2
\nonumber
\\
&
+\Big(
(\kappa^{(1)}-\frac{1}{4}\theta^{+(1)})  \mcD_{(\mathring A}\mathfrak{M}_{\mathring B)} 
-\frac{1}{12}\mcD_{(\mathring A}(\Delta_s -1)w^{(0)}_{\mathring B)} 
+\frac{1}{6} \mcD_{\mathring A}  \mcD_{\mathring B}\mcD^{\mathring C}w^{(0)}_{\mathring C}
\Big)_{\mathrm{tf}}r^3
+\mathfrak{O}(r^4)
\,.
\end{align}
The solution contains logarithmic term unless the following  \emph{no-logs-condition} holds,
\begin{equation}
\mcD_{(\mathring A}(\Delta_s -1)w^{(0)}_{\mathring B)} 
=2 \mcD_{\mathring A}  \mcD_{\mathring B}\mcD^{\mathring C}w^{(0)}_{\mathring C}
\,.
\label{no-logs_first2}
\end{equation}
Then, the solution takes the form (for some integration function $\mathring c^{(2,0)}_{\mathring A\mathring B}$)
\begin{equation}
 ( W_{\mathring A}{}^r{}_{\mathring B}{}^{r} )_{\mathrm{tf}} |_{\scri^-} =\frac{1}{2}(\mcD_{(\mathring A} \mathfrak{M}_{\mathring B)} )_{\mathrm{tf}}r^2
- \Big((\mcD_{(\mathring A} \mathring L_{\mathring B)} )_{\mathrm{tf}}  
+\frac{3}{2}M\Xi^{(2)}_{\mathring A\mathring B}
+\frac{3}{2}N\Xi^{(2)}_{(\mathring A}{}^{\mathring C}\epsilon_{\mathring B)\mathring C}\Big) r^3
+ \mathring c^{(2,0)}_{\mathring A\mathring B}r^4
+\mathfrak{O}(r^5)
\,.
\end{equation}

To analyze  \eq{no-logs_first2}
we decompose, as above,  $w^{(0)}_{ A}$ as $w^{(0)}_{ A} =\mcD_{ A} \ul w^{(0)} +  \epsilon_{ A}{}^{ B}\mcD_B\ol w^{(0)}$
and  apply $\mcD^{ A}\mcD^{ B}$,
\begin{equation*}
 \Delta_s\Delta_s(\Delta_s+2)\ul w^{(0)} =0
\,.
\end{equation*}
It  follows that $\ul w^{(0)}$ needs to be a linear combination of $\ell=0,1$-spherical harmonics.
Next we apply $\epsilon^{ A C}\mcD_{ C}\mcD^{ B}$ to
\eq{no-logs_first2} to obtain an identical equation for $\ol w^{(0)}$, $ \Delta_s \Delta_s(\Delta_s+2)\ol w^{(0)}=0$.
Consequently, $\ol w^{(0)}$ can be represented  by $\ell=0,1$-spherical harmonics as well.
Equivalently, $w^{(0)}_A$ is a conformal Killing 1-form on $\mathbb{S}^2$.

However, recall that $w^{(0)}_A$ was defined as the divergence of a symmetric trace-free tensor, $w^{(0)}_A=\mcD^Bw^{(0)}_{AB}$, whence, on $\mathbb{S} ^2$,
$\ul w^{(0)}$ and $\ol w^{(0)}$ are not allowed to have $\ell=1$-spherical harmonics in their decomposition, cf.\ 
\eq{tensor_decomposition2}.
The no-logs condition \eq{no-logs_first2} thus requires that also this expansion coefficient of the radiation field needs to vanish, again regardless of the choice of the gauge functions,
\begin{equation}
W^{(0)}_{r\mathring Ar\mathring B} =0
\,.
\label{no-logs_condition2}
\end{equation}

Altogether, assuming that the no-logs conditions \eq{no-logs_condition} and \eq{no-logs_condition2} as well as the  boundedness-condition \eq{boundedness_Weyl} hold, the restriction of the rescaled Weyl tensor to $\scri^-$  extends smoothly across $I^-$
and admits there an expansion  of the form,
\begin{align}
W_{r\mathring A r\mathring B}|_{\scri^-} 
 =& \mathfrak{O}(r)
\,,
\label{Weyl_expansion1}
\\
W_{r\mathring Ar}{}^r|_{\scri^-}
=&\mathfrak{O}(r^2)
\;,
\\
W_{\mathring A\mathring Br}{}^r|_{\scri^-}
 =& 
-2N \epsilon_{\mathring A\mathring B} +\frac{1}{2}N\theta^{+(1)}  \epsilon_{\mathring A\mathring B} r^2 +\mathfrak{O}(r^3)
\,,
\\
W_r{}^r{}_r{}^r |_{\scri^-}
=&
2M
-\frac{3}{2}M\theta^{+(1)}r^2
+ \mathfrak{O}(r^3)
\,,
\\
W_{\mathring A}{}^r{}_{r}{}^{r}|_{\scri^-}= &
-\mathfrak{M}_{\mathring A} r+\mathring L_{\mathring A}r^2
+\Big((\kappa^{(1)}-\frac{1}{4}\theta^{+(1)}) \mathfrak{M}_{\mathring A}
-\frac{3}{2}M\mcD_{\mathring A}(\theta^{+(1)}-\kappa^{(1)})
\nonumber
\\
&
-\frac{3}{2}N\mcD^{\mathring B}(\theta^{+(1)}-\kappa^{(1)})\epsilon_{\mathring A\mathring B}
\Big)r^3
+\mathfrak{O}(r^4)
\\
 ( W_{\mathring A}{}^r{}_{\mathring B}{}^{r} )_{\mathrm{tf}} |_{\scri^-} =&\frac{1}{2}(\mcD_{(\mathring A} \mathfrak{M}_{\mathring B)} )_{\mathrm{tf}}r^2
- \Big((\mcD_{(\mathring A} \mathring L_{\mathring B)} )_{\mathrm{tf}}  
+\frac{3}{2}M\Xi^{(2)}_{\mathring A\mathring B}
+\frac{3}{2}N\Xi^{(2)}_{(\mathring A}{}^{\mathring C}\epsilon_{\mathring B)\mathring C}\Big) r^3
+\mathring c^{(2,0)}_{\mathring A\mathring B}r^4
+\mathfrak{O}(r^5)
\;.
\label{Weyl_expansion6}
\end{align}

We have proved the following 
\begin{proposition}
\label{prop_smoothness_constraints}
Consider  asymptotic initial data $(W_{r\mathring Ar\mathring B}, \Xi^{(1)}_{\mathring A\mathring B}, M, N, \mathring L_{\mathring A}, \mathring c^{(2,0)}_{\mathring A\mathring B})$ for the GCFE on $\scri^-$ in a weakly asymptotically Minkowski-like conformal Gauss gauge
\eq{main_gauge0B}-\eq{main_gauge0B2}.%
\footnote{
To have a well-posed initial value problem the data on $\scri^-$ need to be supplemented by appropriate data on e.g.\ an incoming null hypersurface, cf.\ 
Appendix~\ref{app_ADM}.
}
Then the solution of the vacuum constraint equations for the GCFE  on $\scri^-$ admits a smooth expansion through $I^-$ if and only if the  boundedness condition  \eq{boundedness_Weyl} holds,  i.e.\  $\Xi^{(1)}_{\mathring A\mathring B}=0$, and
the data $W_{r\mathring Ar\mathring B}$  admit an expansion of the form $W_{r\mathring Ar\mathring B}=\mathfrak{O}(r)$ at $I^-$, i.e.\ if and only if the two leading order terms of the radiation field vanish.%
\footnote{
An $r^{-2}$-term would yield an unbounded frame component.
}
 In that case the expansion of  the  rescaled Weyl tensor
takes the form \eq{Weyl_expansion1}-\eq{Weyl_expansion6}.
 In the frame \eq{frame_field1}-\eq{frame_field3} its expansion is given by \eq{data_gen1}-\eq{data_gen6} below.
\end{proposition}

%\begin{remark}
%{\rm
%Since a non-vanishing $ \Xi^{(1)}_{\mathring A\mathring B}$ yields a rescaled Weyl tensor which is unbounded at $I^-$ we will henceforth only take\
%initial data sets  into account where  $ \Xi^{(1)}_{\mathring A\mathring B}=0$.
%}
%\end{remark}

\subsubsection{Frame coefficients}
\label{sec_fram_coeff}

So far we have solved the asymptotic constraint equations on $\scri^-$ in adapted null coordinates.
Now we want to compute the corresponding frame coefficients 
associated to the frame  \eq{frame_field1}-\eq{frame_field3} in our current  weakly asymptotically Minkowski-like conformal Gauss gauge, where
\begin{align}
e^{\tau}{}_{1}|_{\scri^-} =1
\,,
\quad
e^{r}{}_{1}|_{\scri^-} =\nu^{\tau}=r + \mathfrak{O}(r^2)
\,,
\quad
e^{\mathring A}{}_{1}|_{\scri^-} =0
\,,
\label{frame_scri_1}
\\
e^{\tau}{}_{A}|_{\scri^-}  = 0
\,,
\quad
e^{r}{}_{A}|_{\scri^-}  =   \mathfrak{O}(r^2)
\,,
\quad
e^{\mathring A}{}_{A}|_{\scri^-}  =   \mathring e^{\mathring A}{}_A+ \mathfrak{O}(r^2)
\,.
\label{frame_scri_2}
\end{align}
It follows from \eq{Weyl_connect1_1}-\eq{Weyl_connect1_5} that
\begin{align}
\widehat\Gamma_1{}^1{}_1 |_{\scri^-}  =&1 + f_1^{(1)} r + \mathfrak{O}(r^2)
\,,
\label{connection_scri_1}
\\
\widehat\Gamma_1{}^1{}_0 |_{\scri^-}  =&- f_1^{(1)} r +  \mathfrak{O}(r^2)
\,,
\\
\widehat\Gamma_1{}^A{}_0 |_{\scri^-}  =&- f^{A(1)} r +  \mathfrak{O}(r^2)
\,,
\\
\widehat\Gamma_1{}^A{}_1|_{\scri^-}  =&
\widehat\Gamma_1{}^A{}_0
\,,
\\
(\widehat\Gamma_1{}^A{}_B)_{\mathrm{tf}} |_{\scri^-}  =&
0
\,,
\end{align}
while \eq{Weyl_connectA_1}-\eq{Weyl_connectA_5} give (with $\nu_A^{(n)}:=\mathring e^{\mathring A}{}_A\nu_{\mathring A}^{(n)}$),
\begin{align}
 \widehat\Gamma_A{}^1{}_1   |_{\scri^-}=&f_A^{(1)} r + \mathfrak{O}(r^2)
\,,
\\
 \widehat\Gamma_A{}^1{}_0   |_{\scri^-}=&\Big(\mcD_A\nu_{\tau}^{(0)}-\nu_A^{(1)}\Big)r+ \mathfrak{O}(r^2)
\,,
\\
 \widehat\Gamma_A{}^B{}_0|_{\scri^-}  =& 
 \frac{1}{2}\Xi^{(2)}_{ A}{}^{ B}r+ \mathfrak{O}(r^2)
\,,
\\
\widehat\Gamma_A{}^B{}_1   |_{\scri^-}=& 
 \widehat\Gamma_A{}^B{}_0
+\Big(1+  f_1^{(1)}r + \mathfrak{O}(r^2)\Big)\delta^{B}{}_A
\,,
\\
\widehat\Gamma_A{}^C{}_B  |_{\scri^-} =&    \mathring\Gamma_A{}^C{}_B  + \Big(2\delta^C{}_{(A}f^{(1)}_{B)} - \eta_{AB} f^{(1)C}\Big)r
+ \mathfrak{O}(r^2)
\,.
\label{connection_scri_10}
\end{align}
From Section~\ref{sec_schouten_gen} we obtain
\begin{align}
\widehat L_{ 1 j}|_{\scri^-} 
=&  \mathfrak{O}(r^2)
\,,
\label{Schouten_scri1}
\\
\widehat L_{ A1}|_{\scri^-} =& \Big( \frac{1}{2}v_{ A}^{(2)} +f^{(1)}_{ A}-\mcD_{ A} f^{(0)}_r -\frac{1}{2}\mcD_{ A} g^{(3)rr}\Big)r
+ \mathfrak{O}(r^2)
\,,
\\
\widehat L_{ AB}|_{\scri^-} 
=&
\Big(-\mcD_{ A}f^{(1)}_{ B} -\frac{1}{2}\Xi^{(2)}_{ A B}-( f^{(0)}_{r} +\frac{1}{2} g^{(3)rr})\eta_{ A B}\Big)r
+ \mathfrak{O}(r^2)
\,,
\\
\widehat L_{ A 0}|_{\scri^-} 
=&\frac{1}{2}\Big(v_{ A}^{(2)} -\mcD_{ A} g^{(3)rr}    \Big)r
+ \mathfrak{O}(r^2)
\,.
\label{Schouten_scri4}
\end{align}
%%
%\begin{align*}
%\widehat L_{rr}|_{\scri^-} =& \mathfrak{O}(1)
%\,,
%\\
%\widehat L_{r\mathring A} |_{\scri^-}=&  \mathfrak{O}(r)
%\,,
%\\
%\widehat L_{\mathring A r}|_{\scri^-} =& f^{(1)}_{\mathring A}  -\mathring\nabla_{\mathring A} f^{(0)}_{r}  + \mathfrak{O}(r)
%\,,
%\\
%\widehat L_{r}{}^r|_{\scri^-} =&\mathfrak{O}(r^2)
%\,,
%\\
%g^{\mathring A\mathring B}\widehat L_{\mathring A\mathring B}|_{\scri^-} =& -\mathring\nabla_{\mathring A}f^{\mathring A(1)} r
%- 2f^{(0)}_{r} r
%- g^{(3)rr}r+  \mathfrak{O}(r^2)
%\,,
%\\
%(\widehat L_{\mathring A\mathring B})_{\mathrm{tf}}|_{\scri^-} =&  -(\mathring\nabla_{\mathring A}f^{(1)}_{\mathring B} )_{\mathrm{tf}}r
%-\frac{1}{2}\Xi^{(2)}_{\mathring A\mathring B}r+  \mathfrak{O}(r^2)
%\,,
%\\
%\widehat L_{\mathring A}{}^r|_{\scri^-} =& \Big( \frac{1}{2}v_{\mathring A}^{(2)} +f^{(1)}_{\mathring A}-\mathring\nabla_{\mathring A} f^{(0)}_r -\frac{1}{2}%\mathring\nabla_{\mathring A} g^{(3)rr}\Big)r^2
%+  \mathfrak{O}(r^3)
%\,.
%\end{align*}
%
For the rescaled Weyl tensor we find, redefining the integration functions $\mathring L_A$ and $\mathring c^{(2,0)}_{AB}$ (now denoted without $\mathring{\enspace}$),
%
%\begin{eqnarray*}
%W_{0101}|_{\scri^-}&=& %(\nu^{\tau})^2W_{\tau r \tau r }= 
%W_r{}^r{}_r{}^r +(\nu^{\tau})^2 \nu^{\mathring A}\nu^{\mathring B}W_{\mathring A r \mathring B r }- 2\nu^{\tau}\nu^{\mathring A}W_{\mathring A r  r}{}^r
%\\
%W_{01AB}|_{\scri^-} &=& -\re^{\mathring A}{}_A\re^{\mathring B}{}_BW_{ \mathring A\mathring Br}{}^r
%+2 (\nu^{\tau})^2\nu^{\mathring C} \re^{\mathring A}{}_A\re^{\mathring B}{}_B\nu_{[\mathring A}W_{ \mathring B] r \mathring C r }
%-4 \nu^{\tau}\re^{\mathring A}{}_A\re^{\mathring B}{}_B\nu_{[\mathring A}W_{ \mathring B] r r }{}^r
%\\
% W_{010A} &=&\nu^{\tau}\re^{\mathring A}{}_AW_{\tau r\tau \mathring A}- (\nu^{\tau})^2e^{\mathring A}{}_A\nu_{\mathring A}W_{\tau r\tau r}
%\\
% W_{011A} &=&\nu^{\tau} e^{\mathring A}{}_AW_{\tau r\tau \mathring A}+ (\nu^{\tau} )^2\re^{\mathring A}{}_AW_{\tau rr\mathring A}
%- (\nu^{\tau})^2 e^{\mathring A}{}_A\nu_{\mathring A} W_{\tau r\tau r}
%\\
%W_{A}^-|_{\scri^-} &=& 
%- \nu^{\tau} \re^{\mathring A}{}_AW^{r}{}_{ rr\mathring A}
%+(\nu^{\tau} )^2\re^{\mathring A}{}_A\nu^{\mathring B}W_{ r\mathring Ar\mathring B}
%\end{eqnarray*}
%
\begin{align}
W_{0101}|_{\scri^-}=& 
2M+ \mathfrak{O}(r^2)
\,,
\label{data_gen1}
\\
W_{01AB}|_{\scri^-} =& 2N\epsilon_{AB}
+ \mathfrak{O}(r^2)
\,,
\label{data_gen2}
\\
W_{A}^-|_{\scri^-} =& 
\mathfrak{O}(r^2)
\,,
\label{data_gen3}
\\
W_{A}^+|_{\scri^-}=& -2\mathfrak{M}_{ A} 
%+\Big(2L_{ A}-\nu^{(0)}_{\tau} \mathfrak{M}_A -3(M \nu^{(1)}_{ A} +N\nu^{(1)}_B\mathring\epsilon_{A}{}^{B})\Big)r
+2 L_A r
+ \mathfrak{O}(r^2)
\,,
\label{data_gen4}
\\
V^+_{AB}|_{\scri^-} =&
 \mathfrak{O}(r^2)
\,,
\label{data_gen5}
\\
V^-_{AB} |_{\scri^-} =&
(\mcD_{( A} \mathfrak{M}_{ B)} )_{\mathrm{tf}}
+\Big(-2(\mcD_{( A} L_{ B)} )_{\mathrm{tf}}  
-3M\Xi^{(2)}_{ A B}
-3N\Xi^{(2)}_{( A}{}^{ C}\epsilon_{ B) C}
\nonumber
\\
&
-2( \mathfrak{M}_{ (A}\mcD_{B)}\nu^{(0)}_{\tau} )_{\mathrm{tf}}
+2(\nu^{(1)}_{( A} \mathfrak{M}_{B)} )_{\mathrm{tf}}
\Big) r
+ c^{(2,0)}_{ A B}r^2
+ \mathfrak{O}(r^3)
\,.
\label{data_gen6}
\end{align}
%\begin{eqnarray*}
%W_{1A1B}  &=&e^{\mathring A}{}_Ae^{\mathring B}{}_BW_{\tau \mathring A\tau \mathring B}
%-2\nu^{\tau} e^{\mathring A}{}_Ae^{\mathring  B}{}_B\nu_{(\mathring A} W_{ \mathring B)\tau r\tau}
%+ (\nu^{\tau})^2 e^{\mathring A}{}_Ae^{\mathring B}{}_B\nu_{\mathring A}\nu_{\mathring B}W_{\tau r\tau r}
%\\
%&&
%-2\nu^{\tau} e^{\mathring A}{}_Ae^{\mathring B}{}_BW_{\tau (\mathring A \mathring B) r}
%-2(\nu^{\tau})^2 e^{\mathring A}{}_Ae^{\mathring B}{}_B\nu_{(\mathring A}W_{\mathring B) r r \tau}
%+ (\nu^{\tau})^2 e^{\mathring A}{}_Ae^{\mathring B}{}_BW_{r \mathring A r \mathring B}
%\\
% W_{0(AB)1} &=& -e^{\mathring A}{}_{A}e^{\mathring B}{}_{B}W_{\tau\mathring A\tau\mathring B}
%-2\nu^{\tau}  e^{\mathring A}{}_{(A}e^{\mathring B}{}_{B)}\nu_{(\mathring A}W_{ \mathring B)\tau\tau r}
%+(\nu^{\tau})^2  e^{\mathring A}{}_{(A}e^{\mathring B}{}_{B)}\nu_{\mathring A}\nu_{\mathring B}W_{\tau rr\tau}
%\\
%&&
%+  \nu^{\tau} e^{\mathring A}{}_{A}e^{\mathring B}{}_{B}W_{\tau (\mathring A\mathring B) r}
%+  (\nu^{\tau})^2 e^{\mathring A}{}_{A}e^{\mathring B}{}_{B}\nu_{(\mathring A}W_{ \mathring B) rr\tau }
%\end{eqnarray*}

\subsection{Higher-order derivatives: Structure of the equations and no-logs conditions}
\label{sec_structure_eqn_scri}

In the previous section we derived conditions which make sure that the restriction to $\scri^-$ of the  fields appearing in the GCFE admit
smooth extensions through $I^-$.
Here we devote attention to the issue under which conditions the statement remains true for transverse derivatives of these fields as well.
For this let us  assume that all the fields $\mathfrak{f}=(e^{\mu}{}_i, \widehat \Gamma_i{}^j{}_k, \widehat L_{ij}, W_{ijkl})$ have transverse derivatives $\partial^k_{\tau}\mathfrak{f}|_{\scri^-}$ which admit
smooth extensions across $I^-$  for $0\leq k\leq n-1$.
We aim to find conditions which guarantee that this also holds true for  $\partial^n_{\tau}\mathfrak{f}|_{\scri^-}$, $n\geq 1$.

Recall the evolution equations \eq{evolution1}-\eq{evolution7}.
We apply $\partial^{n-1}_{\tau}$, which yields algebraic equations
for $(\partial^n_{\tau}e^{\mu}{}_i, \partial^n_{\tau}\widehat \Gamma_i{}^j{}_k, \partial^n_{\tau}\widehat L_{ij})|_{\scri^-}$,
\begin{align}
\partial^n_{\tau}e^{\mu}{}_a|_{\scri^-} =\mathfrak{O}(1)
\,,
\quad
\partial^n_{\tau}\widehat\Gamma_{a}{}^i{}_j |_{\scri^-}
 =
\mathfrak{O}(1)
\,,
\quad
\partial^n_{\tau}\widehat L_{ai} |_{\scri^-}
=\mathfrak{O}(1)
\,,
\end{align}
whence the restrictions to $\scri^-$ of the $n$th-order $\tau$-derivatives of  frame field,  connection coefficients and  Schouten tensor 
are smooth at $I^-$, supposing that this is the case for all derivatives of $\mathfrak{f}$ up to and including order $n-1$.

Let us consider the evolution equations \eq{evolutionW1b}-\eq{evolutionW5b} for the rescaled Weyl tensor.
Again we apply $\partial^{n-1}_{\tau}$ and take the restriction to $\scri^-$.
We obtain a set of equations which determines all independent components (except $\partial^n_{\tau}V^-_{AB}|_{\scri^-}$)
 \emph{algebraically} in terms of $(\partial^n_{\tau}e^{\mu}{}_i, \partial^n_{\tau}\widehat \Gamma_i{}^j{}_k, \partial^n_{\tau}\widehat L_{ij})|_{\scri^-}$
and lower-order derivatives, which are already known at this stage.
In particular these components are smooth at $I^-$, as well,
\begin{equation}
\partial^n_{\tau}U_{AB}|_{\scri^-}=\mathfrak{O}(1)\,, \quad  \partial^n_{\tau}W^{\pm}_{A}|_{\scri^-}=\mathfrak{O}(1)\,, \quad \partial^n_{\tau}V^+_{AB}|_{\scri^-}=\mathfrak{O}(1)
\,.
\label{obvious_smooth_components}
\end{equation}
The  missing components  $\partial^n_{\tau}V^-_{AB}|_{\scri^-}$ are determined by \eq{evolutionW6b}.
We apply $\partial^n_{\tau}$   and take its restriction to $\scri^-$.
In this case it is not an algebraic equation but an ODE for  $\partial^n_{\tau}V^-_{AB}|_{\scri^-}$ along the null geodesic generators of $\scri^-$
(the $0th$-order  recovers the constraint \eq{ODE_Wfinal} for $( W_{\mathring A}{}^{r}{}_{\mathring B}{}^{r})_{\mathrm{tf}}$ in frame components),
\begin{align}
(\nu^{\tau} \partial_{r} +n\partial_{\tau}e^{\tau}{}_1 )\partial^n_{\tau}V^-_{AB} |_{\scri^-}
=&
 (\widehat\Gamma_1{}^0{}_0+ 
2 \widehat\Gamma_1{}^1{}_{0} 
  + \widehat\Gamma_C{}^C{}_0 
- \widehat\Gamma_C{}^C{}_1 ) \partial^n_{\tau}V^-_{AB} 
\nonumber
\\
&
- \Big( ( \widehat\Gamma_C{}^0{}_{(A} + \widehat\Gamma_C{}^1{}_{(A}- 2\widehat\Gamma_1{}^C{}_{(A})\partial^n_{\tau}V^-_{B)}{}^C
\Big)_{\mathrm{tf}}
+\mathfrak{O}(1)
\,,
\end{align}
where  $\mathfrak{O}(1)$ only involves terms  such as  \eq{obvious_smooth_components}  which are in principle  known at this stage, and known to be smooth at $I^-$.
Using \eq{Weyl_connectA_1}-\eq{Weyl_connect1_5} and taking  into account that by \eq{evolution7}  we have
\begin{equation}
\partial_{\tau}e^{\tau}{}_1|_{\scri^-}= (\partial_r +\kappa)\nu^{\tau}
\,,
\end{equation}
this can be written as
\begin{align}
\Big(\nu^{\tau} \partial_{r}   +\frac{1}{2}\theta^+ \nu^{\tau}+(n+2)(\partial_r +\kappa)\nu^{\tau} \Big)\partial^n_{\tau}V^-_{AB} |_{\scri^-}
=&
\mathfrak{O}(1)
\,,
\end{align}
%\begin{align}
%(\Theta^{(1)})^{-1}(\nu^{\tau})^{-n-2}\Big(\partial_{r}  +(n+3)\kappa \Big)\Theta^{(1)}(\nu^{\tau})^{n+3}\partial^n_{\tau}V^-_{AB} |_{\scri^-}
%=&
%\mathfrak{O}(1)
%\,.
%\end{align}
or,
\vspace{-.38em}
\begin{align}
r^{n+3}\big( \partial_{r}  +\mathfrak{O}(1)\big)(r^{-n-2}\partial^n_{\tau}V^-_{AB}) |_{\scri^-}
=&
\mathfrak{O}(1)
\,,
\label{derivs_V-}
\end{align}
%
%\begin{equation}
%r^{n+3}\partial_r(r^{-n-2}\partial^n_{\tau}V^-_{AB})
% =
%O(1)
%\,,
%\label{derivs_V-}
%\end{equation}
%
Equation \eq{derivs_V-} suggests that in general one should expect the appearance of logarithmic terms.
Under the premise that everything is smooth up to and including the $(n-1)$st-order, $n\geq 1$,  logarithmic terms
in the expansions in $r$ of the $n$th-order transverse derivatives can appear at most in the expansion  of $\partial^n_{\tau}V^-_{AB}|_{\scri^-}$.
To check whether this is indeed the case, one needs to compute the expansions in $r$ of all the other fields up to and including the order $n+2$: An $r^{n+2}$-contribution on the right-hand side of \eq{derivs_V-} produces log terms.
The observation that  log terms can in principle  appear at arbitrary  high orders makes the analysis  cumbersome.
In the following we will  analyze the mechanism how logarithmic terms arise via \eq{derivs_V-} in more detail.
In Section~\ref{section_aMlcGg}
 we will provide some more explicit calculations in an asymptotically Minkowski-like conformal Gauss gauge at each order,
where the asymptotic behavior of the gauge functions is fixed.

\begin{proposition}
\label{prop_n-1_no-logs}
Consider  asymptotic initial data $(W_{r \mathring Ar \mathring B}|_{\scri^-}, \Xi^{(1)}_{ A B}, M, N, L_{ A}, (c^{(n+2,n)}_{ A B})_{n\geq 0})$%
\footnote{The $c^{(n+2,n)}_{ A B}$'s are integration functions on $I^-$ which arise from the $\partial^n_{\tau}V^-_{AB}|_{\scri^-}$-equation, cf.\ Appendix~\ref{app_alternative_data}.
}
for the GCFE  on $\scri^-$ in a weakly asymptotically Minkowski-like conformal Gauss gauge
\eq{main_gauge0B}-\eq{main_gauge0B2}.
Then the  restrictions to $\scri^-$  of all the fields $(\partial^n_{\tau} e^{\mu}{}_i,\partial^n_{\tau} \widehat\Gamma_i{}^j{}_k,\partial^n_{\tau} \widehat L_{ij}, \partial^n_{\tau} W_{ijkl})$, $n\in\mathbb{N}$, admit smooth extensions through $I^-$ if and only if  
$\Xi^{(1)}_{AB}=0$,  the no-logs conditions \eq{no-logs_condition} and \eq{no-logs_condition2} are fulfilled by the radiation field, or, equivalently,
$W_{r \mathring Ar \mathring B}|_{\scri^-}=\mathfrak{O}(r)$, and 
 \eq{derivs_V-} does not produce log terms $\forall \,n\geq 1$.
\end{proposition}

\subsection{No-logs condition for $V^-_{AB}$}
\label{sec_no-logs_V-scri_gen}

The no-logs condition  \eq{no-logs_condition2}  for $V^-_{AB}|_{\scri^-}$
arises as Laplace-like  equation on the  expansion coefficient $W^{(0)}_{r \mathring Ar \mathring B}|_{\scri^-}$ (equivalently $\Xi^{(4)}_{AB}$) of the radiation field.
This leads to the question whether also in higher orders the no-logs condition for  $\partial^n_{\tau}V^-_{AB}|_{\scri^-}$
can be read as a Laplace equation for $W^{(n)}_{r \mathring Ar \mathring B}|_{\scri^-}$, or, alternatively, $\Xi^{(n+4)}_{AB}$.
To get some insights, set
%
%\vspace{-0.4em}
\begin{equation}
f^{(m,n)} := \frac{1}{m!\,n!}\partial^m_r\partial^n_{\tau}f|_{I^-}
\,.
\label{dfn_deriv}
\end{equation}
Moreover, we write
\begin{equation}
f= O_{\Xi}(n)
\end{equation}
if the function $f$ is  smooth at  $I^-$ and  depends only on $\Xi^{(k)}_{AB}$ with $k\leq n$ and possibly the gauge functions and the integration functions $M$, $N$,  $L_A$ and $c^{(k+2,k)}_{AB}$ \emph{but not}
on  $\Xi^{(k)}_{AB}$ with $k\geq n+1$
In this section it  is  convenient   to express everything  in terms of  $\Xi_{AB}$  rather than $W_{r\mathring Ar\mathring B}$.

From the constraint equations derived in Appendix~\ref{app_charact_constraints} we deduce that only the following coordinate components depend on  $\Xi_{\mathring A\mathring B}$
\begin{align}
(L_{\mathring A\mathring B})^{(m,0)}_{\mathrm{tf}} =& -\frac{m-1}{2}\Xi^{(m+1)}_{\mathring A\mathring B} + O_{\Xi}(m)
\,,
\\
L_{\mathring A}{}^r{}^{(m,0)} =& \frac{1}{2}v^{(m)}_{\mathring A}+ O_{\Xi}(m-1)
\,,
\\
W_{r\mathring A r\mathring B}{}^{(m,0)}=& -\frac{(m+2)(m+3)}{4}\Xi^{(m+4)}_{\mathring A\mathring B}+ O_{\Xi}(m+3)
\,,
\\
W_{r\mathring Ar}{}^r{}^{(m,0)}
=&\frac{m+1}{4}v^{(m+3)}_{\mathring A}+ O_{\Xi}(m+2)
\;,
\\
W_{\mathring A\mathring Br}{}^r{}^{(m,0)} =& 
\frac{1}{2}\mcD_{[\mathring A}v^{(m+2)}_{\mathring B]}+ O_{\Xi}(m+1)
\;,
\\
W_r{}^r{}_r{}^r{}^{(m,0)} 
=&
- \frac{1}{4}\mcD^Av^{(m+2)}_{\mathring A}+ O_{\Xi}(m+1)
\,,
\\
 W_{\mathring A}{}^r{}_{r}{}^{r}{}^{(m,0)} 
=& 
\frac{1}{8(m-2)}\Big( (\Delta_s-(m-1)(m-2)-1)v^{(m+1)}_{\mathring A}
-2\mcD_{\mathring A}\mcD^Bv^{(m+1)}_{\mathring B}\Big)+ O_{\Xi}(m)
\,,
\\
( W_{\mathring A}{}^r{}_{\mathring B}{}^{r} )_{\mathrm{tf}}^{(m,0)}
=& 
\frac{1}{8(m-3)(m-4)} \Big(\mcD_{(\mathring A} ( \Delta_s-1)v^{(m)}_{\mathring B)}
-2 \mcD_{\mathring A} \mcD_{\mathring B}\mcD^Cv^{(m)}_{\mathring C}
\Big)_{\mathrm{tf}}  
\nonumber
\\
&
 -\frac{(m-2)(m-1)}{16}\Xi^{(m)}_{\mathring A\mathring B}
 + O_{\Xi}(m-1)
\,.
\label{exp_coeffs_Weyl}
\end{align}
Terms with  vanishing denominator, such as the first one  on the right-hand side in \eq{exp_coeffs_Weyl} for $m=3,4$, are defined to be  zero.
For the frame components we then obtain (cf.\ the formulas in Section~\ref{sec_connection_gen} \& \ref{sec_schouten_gen}),
\begin{align}
\widehat L_{ij}^{(m,0)} =& O_{\Xi}(m+1)
\,,
\label{initial_data_lead_scri1}
\\
\widehat\Gamma_i{}^j{}_k{}^{(m,0)} =&  O_{\Xi}(m+1)
\,,
\\
e^{\mu}{}_i{}^{(m,0)} =& O_{\Xi}(m+1)
\,,
\\
V^{+(m,0)}_{AB}  =& -\frac{m(m+1)}{8}\Xi^{(m+2)}_{ A  B}+  O_{\Xi}(m+1)
\,,
\\
W^{-(m,0)}_A =&  \frac{m}{4}v^{(m+2)}_A+  O_{\Xi}(m+1)
\,,
\\
W^{(m,0)}_{0101}
=&  -\frac{1}{4}\mcD^{ A}v^{(m+2)}_{ A}+O_{\Xi}(m+1)
\,,
\\
W^{(m,0)}_{01AB}
 =& - \frac{1}{2}\mcD_{[ A}v^{(m+2)}_{ B]}+O_{\Xi}(m+1)
\,,
\\
W^{+(m,0)}_A=&\frac{1}{4(m-1)}( \Delta_s-1)v^{(m+2)}_{ A}
-\frac{1}{2(m-1)} \mcD_{ A}\mcD^Bv^{(m+2)}_{ B}
+  O_{\Xi}(m+1)
\,,
\label{exp_W+_scri}
\\
V^{-(m,0)}_{AB}=&\frac{1}{8(m-1)(m-2)} \Big(\mcD_{( A} ( \Delta_s-1)v^{(m+2)}_{ B)}
-2 \mcD_{ A} \mcD_{ B}\mcD^Cv^{(m+2)}_{ C}
\Big)_{\mathrm{tf}}  + O_{\Xi}(m+1)
\,.
\label{initial_data_lead_scri9}
\end{align}
From the evolution equations \eq{evolution1}-\eq{evolution7} we deduce by induction over $n$, and assuming that the solution is
smooth up to and including the order $n-1$,
\begin{equation}
e^{\mu}{}_{i}{}^{(m,n)}  =O_{\Xi}(m+1)
\,,
\quad
\widehat\Gamma_i{}^j{}_k{}^{(m,n)}
= O_{\Xi}(m+1)
\,,
\quad
\widehat L_{ ij}^{(m,n)}
=
O_{\Xi}(m+1)
\,.
\end{equation}
Similarly, taking  the $(n-1)$st-order $\tau$-derivatives of \eq{evolutionW1b}-\eq{evolutionW5b},  we deduce that
$W_{ijkl}^{(m,n)}=O_{\Xi}(m+2)$, except possibly for   $V^{-(m,n)}_{AB}$, which so far does not even need to  exist if logarithmic terms appear.
We want to work out  the dependence on $\Xi^{(m+2)}_{AB}$  explicitly.
For this the following observation is important:
%Assuming that the data generate a solution which admits a sufficiently regular cylinder $I$, 
 It follows from \eq{frameI_1}-\eq{Schouten_I}  below that in our current weakly asymptotically Minkowski gauge the  Schouten tensor as well as frame and connection coefficients
are $\tau$-independent on $I$, except for $e^{\tau}{}_1|_I=-\tau$.
In particular a regular  $I^-$
requires  the following relations on $\scri^-$ for $p\leq n-1$,
\begin{align}
\partial_{\tau}^p\widehat\Gamma_i{}^j{}_k|_{\scri^-}=&O(r)\,, \quad \text{for $p\geq 1$}\,,
\\
\partial^p_{\tau}e^{\mu}{}_i|_{\scri^-}=&O(r)\,, \quad\text{for $p\geq 2$}\,,
\\
\partial^p\partial_{\tau}e^{\mu}{}_1|_{\scri^-}=& - \delta^{\mu}{}_0+ O(r)
\,,
\\
\partial^p\partial_{\tau}e^{\mu}{}_A|_{\scri^-}=& O(r)
\,,
\end{align}
whence most terms in the Bianchi equation do not contribute to a $\Xi^{(m+2)}_{AB}$-term.
We evaluate the $(n-1)$st-order $\tau$-derivative of \eq{evolutionW1b}-\eq{evolutionW2b} and \eq{evolutionW5b},
and the $n$th-order $\tau$-derivative of \eq{evolution_WA-_2AB} for $m,n\geq 1$,
\begin{align}
nW_{0101} {}^{(m,n)}
 =& 
-\frac{1}{2}\mcD^A W^{+}_A{}^{(m,n-1)}
+\frac{1}{2}\mcD^AW^{-}_A{}^{(m,n-1)}
+  O_{\Xi}(m+1)
\,,
\label{double_exp_W0101}
\\
nW_{01AB}{}^{(m,n)}     =&
\mcD_{[A}  W^{+}_{B]}{}^{(m,n-1)}
+\mcD_{[A} W^{-}_{B]}{}^{(m,n-1)}
+  O_{\Xi}(m+1)
\,,
%\\
%nW^-_A{}^{(m,n)}
% =&  \mcD^BV^+_{AB}{}^{(m,n-1)}
% +\frac{1}{2}\mcD^B U_{BA}{}^{(m,n-1)}
%+W^-_{A} {}^{(m,n-1)}
%+  O_{\Xi}(m+1)
\\
n  V^+_{AB}{}^{(m,n)}
 =&
\frac{1}{2}(n-m+1)V^+_{AB}{}^{(m,n-1)}
+\frac{1}{2} (  \mcD_{(A} W^-_{B)}{}^{(m,n-1)}
)_{\mathrm{tf}}+  O_{\Xi}(m+1)
\,,
\\
(n-m-1)W^-_{A}{}^{(m,n)} =&2 \mcD^B V^+_{AB}{}^{(m,n)} + O_{\Xi}(m+1)
\label{double_exp_W-}
\,.
\end{align}
We further evaluate the $n$th-order $\tau$-derivative of  \eq{evolution9} and the ($n-1)$st-order $\tau$-derivative of \eq{evolution11} below (which arise as linear combinations of \eq{evolutionW1b}-\eq{evolution_WA-_2}) 
\begin{align}
(n-m+2)  V^-_{AB} {}^{(m,n)}
=&
-(\mcD_{(A}W^+_{B)}{}^{(m,n)})_{\mathrm{tf}}+  O_{\Xi}(m+1)
\,,
\label{double_exp_V-}
\\
nW^+_{A} {}^{(m,n)}
 =&  - \mcD^BV^-_{AB}{}^{(m,n-1)}
+\frac{1}{2}(n-m-2)W^+_{A}{}^{(m,n-1)}+  O_{\Xi}(m+1)
\,.
\label{double_exp_W+}
\end{align}
These equations can be decoupled to provide recursive formulas for various components of the  transverse derivatives of the rescaled Weyl tensor at $I^-$,
\begin{align}
(n-m-1)W^-_{A}{}^{(m,n)} 
 =&
\frac{1}{2n} \Big(\Delta_s+(m-n) (m-n+1)-1\Big)  W^-_{A}{}^{(m,n-1)}
+  O_{\Xi}(m+1)
\,,
\label{recursion1}
\\
(n-m-2) \mcD^B V^+_{AB}{}^{(m,n)} 
 =&
\frac{1}{2n} \Big(\Delta_s+(m-n) (m-n+1)-1\Big)  \mcD^B V^+_{AB}{}^{(m,n-1)}  +  O_{\Xi}(m+1)
\,,
\label{recursion2}
\\
(m-n-1) W^+_{A} {}^{(m,n)}
 =&  -\frac{1}{2n}\Big(\Delta_s 
+(m-n) (m-n+1)-1\Big)W^+_{A}{}^{(m,n-1)}+  O_{\Xi}(m+1)
\,,
\label{recursion3}
\\
(m-n-2) \mcD^BV^-_{AB} {}^{(m,n)}
=&
-\frac{1}{2n}\Big(\Delta_s 
+(m-n) (m-n+1)-1\Big)\mcD^BV^-_{AB} {}^{(m,n-1)}+  O_{\Xi}(m+1)
\,.
\label{recursion4}
\end{align}

The no-logs condition for $\partial^n_{\tau}V^-_{AB}$ takes the form
\begin{equation}
(\mcD_{(A}W^+_{B)}{}^{(n+2,n)})_{\mathrm{tf}}=  O_{\Xi}(n+3)
\,,
\end{equation}
its divergence reads
\begin{equation}
(\Delta_s +1)W^+_{A}{}^{(n+2,n)}=  O_{\Xi}(n+3)
\,.
\label{no_logs_W+}
\end{equation}
%
%Combining \eq{double_exp_V-} and  \eq{double_exp_W+} yields
%%
%\begin{equation}
%(m-n-1) nW^+_{A} {}^{(m,n)}
% =  -\frac{1}{2}\Big(\Delta_s 
%+(m-n) (m-n+1)-1\Big)W^+_{A}{}^{(m,n-1)}+  O_{\Xi}(m+1)
%\,.
%\end{equation}
%%
In the special case where $m=n+2$ we express $W^+_{A} {}^{(n+2,n)}$ in terms of the initial data on $\scri^-$. 
Using \eq{recursion3} we obtain
\begin{align*}
W^+_{A} {}^{(n+2,n)}
 =&  -\frac{1}{2n}(\Delta_s +5)W^+_{A}{}^{(n+2,n-1)}+  O_{\Xi}(n+3)
\\
 =&  \frac{1}{8n(n-1)}(\Delta_s +5)(\Delta_s +11)W^+_{A}{}^{(n+2,n-2)}+  O_{\Xi}(n+3)
\\
=& \dots
\\ 
=&
\frac{(-1)^n}{2^n(n!)^2}\prod_{\ell =2}^{n+1}\Big(\Delta_s+\ell (\ell+1)-1\Big)W_A^{+(n+2,0)}+  O_{\Xi}(n+3)
\\
\overset{\eq{exp_W+_scri}}{=}&
\frac{(-1)^n}{2^{n+2}n!(n+1)!}\prod_{\ell =2}^{n+1}\Big(\Delta_s+\ell (\ell+1)-1\Big)\Big(( \Delta_s-1)v^{(n+4)}_{ A}
-2 \mcD_{ A}\mcD^Bv^{(n+4)}_{ B}\Big)
+  O_{\Xi}(n+3)
\,.
\end{align*}
The no-logs condition \eq{no_logs_W+} therefore adopts the form
\begin{equation*}
\prod_{\ell =1}^{n+1}\Big(\Delta_s+\ell (\ell+1)-1\Big)\Big(( \Delta_s-1)v^{(n+4)}_{\mathring A}
-2 \mcD_{ A}\mcD^Bv^{(n+4)}_{ B}\Big)= O_{\Xi}(n+3)
\,.
\label{some_no_logs_cond}
\end{equation*}
Curl and divergence read
\begin{align}
\prod_{\ell =0}^{n+1}\Big(\Delta_s+\ell (\ell+1)\Big)\mcD^Av^{(n+4)}_A
=&   O_{\Xi}(n+3)
\,,
\label{no_log_rad1}
\\
\prod_{\ell =0}^{n+1}\Big(\Delta_s+\ell (\ell+1)\Big)( \epsilon^{AB}\mcD_{A}v^{(n+4)}_{B})
=&   O_{\Xi}(n+3)
\,.
\label{no_log_rad2}
\end{align}
We  may regard the no-logs condition on $\partial_{\tau}^nV^-_{AB}$ as a condition on the $(n+4)$th-order expansion coefficient
of $\Xi_{AB}$ (equivalently, on the $n$th-order expansion coefficient of the radiation field $W_{r\mathring Ar\mathring B}|_{\scri^-}$).
In general, though, this Laplace-like equation does not need to admit a solution:
By construction from a symmetric trace-free tensor, the right-hand sides do not contain $\ell=0,1$-spherical harmonics in their decomposition (cf.\ \eq{tensor_decomposition}-\eq{tensor_decomposition2}).
Nonetheless, they may contain spherical harmonics with $2\leq \ell \leq n+1$ which straightaway suppress the existence of a solution.

If a solution exists, i.e.\ if and only if no such spherical harmonics arise, there is the freedom to choose the spherical harmonics in the harmonics decomposition of the Hodge decomposition functions $\ul{\Xi}^{(n+4)}$ and $\ol \Xi^{(n+4)}$ of $\Xi^{(n+4)}_{AB}$ up to and including the order $\ell=n+1$.

In Section~\ref{section6} we will show that for a more restricted class of gauge functions  a radiation field which vanishes at any order at $I^-$ satisfies the no-logs conditions  \eq{no_log_rad1}-\eq{no_log_rad2} for any  $n$.
However, it is not clear to us whether a radiation field with a non-trivial expansion at $I^-$ exists which fulfills the no-logs conditions,
i.e.\ where its asymptotic expansion is adjusted 
 in such a way that the right-hand sides of \eq{no_log_rad1}-\eq{no_log_rad2} do not contain spherical harmonics with $2\leq \ell\leq n+1$.
One might expect this to be very restrictive, and, if possible at all, might impose restriction not only on $\Xi_{AB}$ but also  on the integration functions such as the mass aspect $M$ etc.
 We will discuss this  more detailed in Section~\ref{sec_const_mass}, where it is shown that the radiation does need to have a trivial expansion at least for
constant $M$ and vanishing $N$.

So far we have not analyzed the impact of the gauge functions, which one might think could be employed  to get rid of the disturbing spherical harmonics.
However, when computing the 0th- 1st-order transverse derivatives we have seen that the gauge function drop out, 
and in Section~\ref{sec_gauge_ind} we will show that logarithmic terms cannot be eliminated by appropriately adjusted gauge functions.

%so one needs to control more orders to deduce in which way they appear on the right-hand sides of \eq{no_log_rad1}-\eq{no_log_rad2}.
%\tim{might need rewording}

%For completeness, a corresponding analysis for the constraint equations \eq{evolution_WA-_2A}-\eq{evolution_WA-_2} gives
%\tim{also discuss no-logs condition... needs to be extended}
%%
%\begin{align}
%(n-m)W_{0101}{}^{(m,n)}
% =& 
%\mcD^AW^{-}_A{}^{(m,n)}+  O_{\Xi}(m+1)
%\,,
%\label{constr_expansion1}
%\\
%(n-m)W_{01AB}{}^{(m,n)}   =&
%2\mcD_{[A} W^{-}_{B]}{}^{(m,n)}+  O_{\Xi}(m+1)
%\,,
%\\
%(n-m-1)W^-_{A} {}^{(m,n)}
% =& 2 \mcD^BV^+_{AB}{}^{(m,n)}+  O_{\Xi}(m+1)
%\,,
%\\
%(n-m+1)W^+_{A}{}^{(m,n)}   =&
% -\mcD^B U_{AB}{}^{(m,n)}
%+  O_{\Xi}(m+1)
%\,.
%\label{constr_expansion4}
%\end{align}

\section{Appearance of log terms: Approaching $I^-$ from~$I$}
\label{sec_appro_I-I}
\label{section4}

In the previous sections we have analyzed the appearance of logarithmic terms when approaching the critical set $I^-$ from $\scri^-$.
The aim of this section is to carry out a corresponding analysis when approaching $I^-$ from the cylinder $I$.

Our goal is as follows: We assume that we have been given,  in a weakly  asymptotically 
Minkowski-like  conformal Gauss gauge,   a smooth solution of the GCFE which admits a smooth $\scri^-$ and  a smooth 
spatial infinity $I$.
We have already seen above that, in general, solutions cannot expected to be smooth at the critical set  $I^-$ where $I$ and $\scri^-$ intersect.
due to the appearance of logarithmic terms.
We therefore aim to extract conditions on the initial data
which are compatible with  smoothness at $I^-$ of all the relevant fields (and radial derivatives thereof) when approaching $I^-$ from  $I$,
under the assumption that the initial data for the transport equations on $I$  are induced by the limit of the corresponding fields given on $\scri^-$
to $I^-$.

\subsection{Solution of the inner equations on $I$ for %$\widehat\Gamma_i{}^j{}_k$, $\widehat L_{ij}$ and $e^{\mu}{}_i$
 connection coefficients, Schouten tensor and frame field}
\label{sec_connection_schouten_I}

We want to solve the   transport equations \eq{ev_eqn_gauge_1}-\eq{ev_eqn_gauge_7}  for connection coefficients, Schouten tensor and frame field
on the cylinder $I$ (in our setting the cylinder ``touches'' $\scri^-$  at  $I^-=\{\tau=-1,r=0\}$).
By assumption, the initial data for the transport equations are determined by taking the limit of the corresponding fields on $\scri^-$
to $I^-$.
It follows from Section~\ref{sec_fram_coeff} that
\begin{align}
e^{\tau}{}_{1}|_{I^-} =&1
\,,
\quad
e^{r}{}_{1}|_{I^-} =0
\,,
\quad
e^{\mathring A}{}_{1}|_{I^-} =0
\,,
\\
e^{\tau}{}_{A}|_{I^-}  =& 0
\,,
\quad
e^{r}{}_{A}|_{I^-}  =   0
\,,
\quad
e^{\mathring A}{}_{A}|_{I^-}  =   \mathring e^{\mathring A}{}_A
\,,
\\
\widehat\Gamma_1{}^i{}_j |_{I^-}  =&\delta^i{}_j
\,,
\quad
\widehat\Gamma_a{}^b{}_0 |_{I^-}  =0
\,,
\quad
 \widehat\Gamma_A{}^1{}_1   |_{I^-}=0
\,,
\quad
\widehat\Gamma_A{}^B{}_1   |_{I^-}= \delta^B{}_A
\,,
\quad
\widehat\Gamma_A{}^C{}_B  |_{I^-} =    \mathring\Gamma_A{}^C{}_B 
\,,
\\
\widehat L_{ ij}|_{I^-} 
=&0
\,.
\end{align}
Although the equations \eq{ev_eqn_gauge_1}-\eq{ev_eqn_gauge_7} are not linear, they can be solved explicitly due to the fact
that the initial data are almost trivial:
We observe that \eq{ev_eqn_gauge_2} and \eq{ev_eqn_gauge_3},
are decoupled from the other ones.
Since the initial data for these equations vanish we conclude  $\widehat\Gamma_{a}{}^b{}_0$ and $\widehat L_{ab}$
vanish on $I$.
It then follows from  \eq{ev_eqn_gauge_1} that $\widehat L_{a0}$ vanishes as well.
The remaining equations,
\begin{align}
\partial_{\tau}\widehat\Gamma_{a}{}^1{}_b |_I
 &=
0
\,,
\quad
\partial_{\tau}\widehat\Gamma_{1}{}^A{}_B |_I
 =
0
\,,
\quad
\partial_{\tau}\widehat\Gamma_{A}{}^B{}_C |_I
 =
0
\,,
\quad
\partial_{\tau}e^{\mu}{}_a|_I
= -\widehat\Gamma_{a}{}^0{}_{0} \delta^{\mu}{}_0
\,,
\end{align}
can then be straightforwardly integrated.
Altogether, we obtain the  following solution
\begin{align}
e^{\tau}{}_{1}|_{I} =&- \tau 
\,, 
\label{frameI_1}
\quad
e^{r}{}_{1}|_{I} =0
\,,
\quad
e^{\mathring A}{}_{1}|_{I} =0
\,,
\\
e^{\tau}{}_{A}|_{I}  =& 0
\,,
\quad
e^{r}{}_{A}|_{I}  =   0
\,,
\quad
e^{\mathring A}{}_{A}|_{I}  =   \mathring e^{\mathring A}{}_A
\,,
\\
\widehat\Gamma_1{}^i{}_j |_{I}  =& \delta^i{}_j
\,,
\quad
\widehat\Gamma_a{}^b{}_0 |_{I}  =0
\,,
\quad
 \widehat\Gamma_A{}^1{}_1   |_{I}=0
\,,
\quad
\widehat\Gamma_A{}^B{}_1   |_{I}= \delta^B{}_A
\,,
\quad
\widehat\Gamma_A{}^C{}_B  |_{I} =    \mathring\Gamma_A{}^C{}_B 
\,,
\\
\widehat L_{ ij}|_{I} 
=& 0
\,.
\label{Schouten_I}
\end{align}
Recall  that in a conformal Gauss gauge the  relations
$\widehat\Gamma_{0}{}^i{}_j =0$, $\widehat L_{0i} =0$ and
$e^{\mu}{}_0 = \delta^{\mu}{}_{\tau}$  hold globally.

\subsection{Solution of the Bianchi equation on $I$}
\label{sec_solve_Bianchi_I}

Next, let us analyze the Bianchi equation for the rescaled Weyl tensor  on the cylinder $I$ at spatial infinity.
Although this is not necessary for our purposes, let us analyze the full system.
Evaluation of \eq{evolutionW1b}-\eq{evolution_WA-_2}
on $I$ using \eq{frameI_1}-\eq{Schouten_I}  yields equations which can be written as
\begin{align}
(1-\tau) \partial_{\tau} W^+_{A}  |_{I}         =&  - W^+_{A}   -2\mcD^BV^-_{AB} 
\,, 
\label{eveq_I_1}
\\
(1+\tau)\partial_{\tau}W^-_{A} |_{I}         =& W^-_A +2 \mcD^B V^+_{AB} 
\label{eveq_I_2}
\,,
\\
(1-\tau^2) \partial_{\tau } W_{0101}|_{I}       =& 
-\frac{1}{2}(1+\tau) \mcD^AW^+_A
+\frac{1}{2}(1-\tau)\mcD^AW^-_A
\label{eveq_I_3}
\,,
\\
(1-\tau^2)  \partial_{\tau} W_{01AB}  |_{I}      =& 
 (1+\tau)\mcD_{ [A} W^+_{B]}+ (1-\tau)\mcD_{ [A} W^-_{B]}
\,,
\label{eveq_I_4}
\\
\partial_{\tau}[(1-\tau)^2V^+_{AB}  ] |_{I} 
 =&
(1-\tau) (\mcD_{(A}W^-_{B)})_{\mathrm{tf}}
\,,
\label{eveq_I_5}
\\
\partial_{\tau}[(1+\tau)^2V^-_{AB} ]|_{I} 
 =&
-  (1+\tau)(\mcD_{(A} W^+_{B)})_{\mathrm{tf}}
\,,
\label{eveq_I_6}
\end{align}
and
\begin{align}
(1+\tau) \mcD^{ A} W^+_A|_{I}     =&-(1-\tau) \mcD^{ A} W^-_A
\,,
\label{constr_I_1}
\\
(1+\tau)\mcD_{ [A}W^+_{B]} |_{I}        =&(1-\tau)\mcD_{ [A}W^-_{B]}
\,,
\label{constr_I_2}
\\
\tau W^+_A|_{I}   =& 
\frac{1}{2}(1   - \tau)\mcD^BU_{AB} - (1+\tau )\mcD^B V^-_{AB}
\,,
\\
\tau W^-_A|_{I} 
      =&  
- \frac{1}{2}(1+  \tau)\mcD^BU_{BA}
+(1-\tau)\mcD^B V^+_{AB}
\,.
\label{constr_I_4}
\end{align}
The subsystem \eq{eveq_I_1}-\eq{eveq_I_6} provides transport equations for all independent components of the  rescaled Weyl tensor on $I$.
The  subsystem \eq{constr_I_1}-\eq{constr_I_4} can be regarded as the constraint part of the Bianchi system on $I$.
%We note that  the  constraint equation of the Bianchi system by which we mean the equations $ \widehat\nabla_a W^a{}_{0kl}  = \frac{1}{4}\widehat\Gamma_b{}^p{}_pW^b{}_{0kl} $ does not provide any ``new'' equations on $I$ as compared to  \eq{eveq_I_1}-\eq{constr_I_4}.
A straightforward computation shows  that the constraint equations are preserved under the  evolution  of \eq{eveq_I_1}-\eq{eveq_I_6},
and therefore merely need to be satisfied initially at $I^-$.

We want to decouple the evolution equations.
%\tim{eventually one needs to make sure that full equations are satisfied}
Differentiation of the equations for $W^{\pm}_A$ by $\tau$ yields with   \eq{eveq_I_3}-\eq{eveq_I_4} and  \eq{constr_I_1}-\eq{constr_I_4}
\begin{equation}
((1-\tau^2) \partial^2_{\tau} -\Delta_s
\pm 2(1\mp\tau) \partial_{\tau}+1) W^{\pm}_{A}   |_{I}         =
0
\,.
\end{equation}
Let us also take into account that the constraints merely need to be satisfied at $I^-$ where they read
\begin{align}
 \mcD^{ A} W^-_A|_{I^-}      = 0
\,,
\label{cond_I-_1_int}
\quad
\mcD_{ [A}W^-_{B]}|_{I^-}          = 0
\,,
\quad
 W^-_A|_{I^-} 
      =
-2\mcD^B V^+_{AB}
\,,
\quad
W^+_A|_{I^-}   =
-\mcD^BU_{AB} 
\,.
\end{align}
On $I^-\cong{S}^2$ this can only be satisfied if $W^-_A$ and $V^+_{AB}$ vanish there altogether.
We conclude that the  system    \eq{eveq_I_1}-\eq{constr_I_4} is equivalent to the following one,
\begin{align}
((1-\tau^2) \partial^2_{\tau} -\Delta_s
+ 2(1-\tau)\partial_{\tau}+1) W^+_{A}   |_{I}         =& 
0
\,,
\label{decoupled_k_A}
\\
((1-\tau^2)\partial^2_{\tau}-\Delta_s
-2(1+\tau)\partial_{\tau}+1)W^-_{A}   |_{I}   =&
0
\,,
\label{decoupled_h_A}
\\
(1-\tau^2) \partial_{\tau } W_{0101}|_{I}       =& 
-\frac{1}{2}(1+\tau)\mcD^AW^+_A
+\frac{1}{2}(1-\tau) \mcD^AW^-_A
\,,
\label{ev_W0101}
\\
(1-\tau^2)  \partial_{\tau} W_{01AB}  |_{I}      =& 
 (1+\tau)\mcD_{ [A} W^+_{B]}+ (1-\tau)\mcD_{ [A} W^-_{B]}
\,,
\\
\partial_{\tau}[(1-\tau)^2V^+_{AB}  ] |_{I} 
 =&
(1-\tau) (\mcD_{(A}W^-_{B)})_{\mathrm{tf}}
\,,
\\
\partial_{\tau}[(1+\tau)^2V^-_{AB} ]|_{I} 
 =&
-  (1+\tau)(\mcD_{(A} W^+_{B)})_{\mathrm{tf}}
\,,
\label{ev_kAB}
\end{align}
and
\begin{align}
W^-_A|_{I^-}      = 0
\,, \quad
\partial_{\tau} W^-_A|_{I^-}
      =
 \frac{1}{2}\mcD^BU_{BA}=\ol{\mathfrak{M}}_A
\label{vanish_h_I-}
\\
W^+_A|_{I^-}   =
-\mcD^BU_{AB} =-2\mathfrak{M}_A
\,, \quad
\lim_{\tau\rightarrow -1}[(1+\tau)W^+_A|_{I^-}= 0
\,,
\label{cond_I-_4}
\\
 V^+_{AB} |_{I^-}  
      =  0
\,,
\label{cond_I-_6}
\end{align}
with
\begin{equation}
\mathfrak{M}_A\equiv \mcD_A M + \epsilon_A{}^B\mcD_BN\,, \quad \ol {\mathfrak{M}}_A\equiv \mcD_A  M - \epsilon_A{}^B\mcD_BN
\,.
\end{equation}
The data at $I^-$ can be computed from \eq{data_gen1}-\eq{data_gen6} by continuity at $I^-$.
The additional conditions in \eq{vanish_h_I-}-\eq{cond_I-_4} are needed since we have replaced the first-order equations for $W^{\pm}_A$ by second-order ones.
The analysis in Section~\ref{sec_smoothness_Bianchi_I} below
shows  that these are indeed the ``right'', i.e.\ freely prescribable,  data.
%%
%\begin{equation*}
%W^+_A|_{I^-}\,, \quad \partial_{\tau}W^+_A|_{I^-}\,, \quad W^-_A|_{I^-}\,,\quad \lim_{\tau\rightarrow -1}[(1+\tau)W^-_A|_{I^-}
%\,.
%\end{equation*}
%
%The latter datum needs to vanish for the solution to be bounded at $I^-$, the remaining data are given by \eq{vanish_h_I-}, \eq{vanish_k_I-} and \eq{cond_I-_5}. 
It is
further shown there that the solutions are regular at $I^-$.
They admit the following expansions,
\begin{align}
W^-_A|_I =&\ol{\mathfrak{M}}_A(1+\tau) +\frac{1}{4}(\Delta_s+1)\ol{\mathfrak{M}}_A (1+\tau)^2+ \mathfrak{O}(1+\tau)^3
\,,
\label{Weyl_I_1}
\\
W^+_A|_I =&-2\mathfrak{M}_A -\frac{1}{2} (\Delta_s-1)\mathfrak{M}_A(1+\tau) + \mathfrak{O}(1+\tau)^2
\,.
\end{align}
We  observe that once $W^{\pm}_A|_I$  are known, the remaining evolution equations
are merely ODEs (some of them of Fuchsian type) which can be straightforwardly integrated.
No logarithmic terms arise when integrating \eq{ev_W0101}-\eq{ev_kAB}.
We obtain the following expansions
\begin{align}
W_{0101}|_{I}       =& 
M+ \frac{1}{2}\Delta_s M(1+\tau)
+ \mathfrak{O}(1+\tau)^2
\,,
\\
W_{01AB}  |_{I}      =& N\epsilon_{AB}
 + \frac{1}{2}\Delta_sN\epsilon_{AB}(1+\tau)
+ \mathfrak{O}(1+\tau)^2
\,,
\\
V^+_{AB} |_{I} 
 =&
 \frac{1}{8}(\mcD_{(A}\ol{\mathfrak{M}}_{B)})_{\mathrm{tf}}(1+\tau)^2
+ \mathfrak{O}(1+\tau)^3
\,,
\\
V^-_{AB} |_{I} 
 =&
\frac{1}{2}( \mcD_{(A}\mathfrak{M}_{B)})_{\mathrm{tf}}
+ \mathfrak{O}(1+\tau)
\,.
\label{Weyl_I_6}
\end{align}
The expansions are compatible at $I^-$ with the corresponding ones computed on $\scri^-$.
In general, a solution $V^-_{AB}$ to \eq{ev_kAB} will be unbounded at $I^-$ whence there is no freedom to choose initial data
if one requires the solution to be bounded.
% (and to approach the value induced by $V^-_{AB}|_{\scri^-}$.

\subsection{Rewriting the Bianchi equation}

For the analysis on the cylinder it turns out that it is convenient to use a different subsystem of the Bianchi equation as compared to our analysis on $\scri^-$
to evolve the  independent components of the  radial derivatives of the  rescaled Weyl tensor (in this paper  we do not care whether the subsystem used to determine higher order derivatives forms  a symmetric hyperbolic system in spacetime).
The following system  is obtained by taking appropriate linear combinations of \eq{evolutionW1b}-\eq{evolution_WA-_2},
\begin{align}
(e^{\mu}{}_1\partial_{\mu} + \partial_{\tau})  V^+_{AB}
 =&
%\Big( 3\widehat\Gamma_1{}^0{}_{0}
%- 2\widehat\Gamma_1{}^1{}_{0} 
%- \frac{1}{2}\widehat\Gamma_C{}^C{}_0
%- \frac{1}{2}\widehat\Gamma_C{}^C{}_1 \Big)V^+_{AB}
%+ \Big( 
%- \frac{3}{2}(\widehat\Gamma_C{}^0{}_{(A} + \widehat\Gamma_C{}^1{}_{(A}) U^C{}_{B)} 
%\nonumber
%\\
%&
%+ (\check\nabla_{(A} 
%+ 2\widehat\Gamma_1{}^0{}_{(A}
%-2 \widehat\Gamma_{(A}{}^0{}_{|0|}
%+ \widehat\Gamma_{(A}{}^0{}_{|1|}
%+ 2\widehat\Gamma_1{}^1{}_{(A} )W^-_{B)}
%\Big)_{\mathrm{tf}}
( \widehat\Gamma_1{}^0{}_{0}
- 2\widehat\Gamma_1{}^1{}_{0} 
- \widehat\Gamma_C{}^C{}_0
- \widehat\Gamma_C{}^C{}_1 )V^+_{AB}
\nonumber
\\
&
+ \Big( ( \widehat\Gamma_C{}^0{}_{(A} - \widehat\Gamma_C{}^1{}_{(A}+2 \widehat\Gamma_{1C(A} )V^+_{B)}{}^{C}
- \frac{3}{2}(\widehat\Gamma_C{}^0{}_{(A} + \widehat\Gamma_C{}^1{}_{(A}) U^C{}_{B)} 
\nonumber
\\
&
+ (\check\nabla_{(A} 
+ 2\widehat\Gamma_1{}^0{}_{(A}
-2 \widehat\Gamma_{(A}{}^0{}_{|0|}
+ \widehat\Gamma_{(A}{}^0{}_{|1|}
+ 2\widehat\Gamma_1{}^1{}_{(A} )W^-_{B)}
\Big)_{\mathrm{tf}}
\,,
\label{evolution8}
\\
( e^{\mu}{}_1\partial_{\mu}-\partial_{\tau})
 V^-_{AB} 
=&
% (3\widehat\Gamma_1{}^0{}_0+ 
%2 \widehat\Gamma_1{}^1{}_{0} 
%  + \frac{1}{2}\widehat\Gamma_C{}^C{}_0 
%- \frac{1}{2}\widehat\Gamma_C{}^C{}_1 ) V^-_{AB} 
%+ \Big(  \frac{3}{2}( \widehat\Gamma_C{}^0{}_{(A} 
%-  \widehat\Gamma_C{}^1{}_{(A} )U_{B)}{}^C
%\nonumber
%\\
%&
%+\Big(\check\nabla_{(A}
%- 2\widehat\Gamma_1{}^0{}_{(A} 
%- 2 \widehat\Gamma_{(A}{}^0{}_{|0|}
%- \widehat\Gamma_{(A}{}^0{}_{|1|}
%+2 \widehat\Gamma_1{}^1{}_{(A} \Big)W^+_{B)}
%\Big)_{\mathrm{tf}}
-e^{\alpha}{}_1\partial_{\alpha} V^-_{AB} 
+ (\widehat\Gamma_1{}^0{}_0+ 
2 \widehat\Gamma_1{}^1{}_{0} 
  + \widehat\Gamma_C{}^C{}_0 
- \widehat\Gamma_C{}^C{}_1 ) V^-_{AB} 
\nonumber
\\
&
+ \Big( -( \widehat\Gamma_C{}^0{}_{(A} + \widehat\Gamma_C{}^1{}_{(A}- 2\widehat\Gamma_{1C(A})V^-_{B)}{}^C
+ \frac{3}{2}( \widehat\Gamma_C{}^0{}_{(A} 
-  \widehat\Gamma_C{}^1{}_{(A} )U_{B)}{}^C
\nonumber
\\
&
+(\check\nabla_{(A}
- 2\widehat\Gamma_1{}^0{}_{(A} 
- 2 \widehat\Gamma_{(A}{}^0{}_0
- \widehat\Gamma_{(A}{}^0{}_{|1|}
+2 \widehat\Gamma_1{}^1{}_{(A} )W^+_{B)}
\Big)_{\mathrm{tf}}
\label{evolution9}
\,,
\\
(\partial_{\tau}-e^{\mu}{}_1\partial_{\mu} )W^-_{A}   
 =& \Big( 2\check\nabla^B
-6 \widehat\Gamma^{B0}{}_{0}
+4\widehat\Gamma^{B0}{}_1-\widehat\Gamma_1{}^B{}_{0} - \widehat\Gamma_1{}^B{}_{1}\Big)V^+_{AB}
-\frac{3}{2}( \widehat\Gamma_1{}^B{}_{0} - \widehat\Gamma_1{}^B{}_{1}  )U_{BA}
\nonumber
\\
&
+ \Big( 4\widehat\Gamma_{[A}{}^0{}_{B]} 
+4\widehat\Gamma_{[A}{}^{1}{}_{B]} 
 -(\widehat\Gamma_{1BA})_{\mathrm{tf}}\Big) W^{-B}
-\Big(3 \widehat\Gamma_1{}^0{}_{0} -2\widehat\Gamma_B{}^B{}_1 +2 \widehat\Gamma_B{}^B{}_0 -\widehat\Gamma_1{}^1{}_0\Big)W^-_{A}
\,,
\label{evolution10}
\\
(\partial_{\tau}+ e^{\mu}{}_1\partial_{\mu} )W^+_{A} 
 =& -\Big(  2 \check\nabla^B
-6 \widehat\Gamma^{B0}{}_{0}
-4\widehat\Gamma^{B0}{}_1
+ \widehat\Gamma_1{}^B{}_{0} -\widehat\Gamma_1{}^B{}_{1}\Big) V^-_{AB}
-\frac{3}{2}(  \widehat\Gamma_1{}^B{}_{0} + \widehat\Gamma_1{}^B{}_{1}  )U_{AB}
\nonumber
\\
&
+ \Big(4 \widehat\Gamma_{[A}{}^0{}_{B]}
-4\widehat\Gamma_{[A}{}^{1}{}_{B]} 
+ (\widehat\Gamma_{1BA})_{\mathrm{tf}}\Big) W^{+B}
+\Big(  3\widehat\Gamma_1{}^0{}_{0} - 2\widehat\Gamma_B{}^B{}_1  -2\widehat\Gamma_B{}^B{}_0  + \widehat\Gamma_1{}^1{}_0\Big) W^+_{A}
\,,
\label{evolution11}
\\
(\tau\partial_{\tau} +e^{\mu}{}_1\partial_{\mu}) W_{0101}
 =& 
(1+\tau)\Big(-\frac{1}{2}\check\nabla^A
+2 \widehat\Gamma^{A0}{}_{0}
+\frac{1}{2} \widehat\Gamma^{A0}{}_1\Big) W^{+}_A
-(1-\tau)\Big(\frac{1}{2}\check\nabla^A
-2 \widehat\Gamma^{A0}{}_{0}
+\frac{1}{2} \widehat\Gamma^{A0}{}_1\Big)W^{-}_A
\nonumber
\\
&
+ \frac{1}{2}(1-\tau)\Big(\widehat\Gamma^{AB}{}_{1}+ \widehat\Gamma^{AB}{}_{0}\Big)V^+_{AB}
+ \frac{1}{2}(1+\tau)\Big(\widehat\Gamma^{AB}{}_{1}-\widehat\Gamma^{AB}{}_{0}\Big)V^-_{AB}
\nonumber
\\
&
+3\Big(  \widehat\Gamma_1{}^0{}_0 -\frac{1}{2}\widehat\Gamma_A{}^A{}_{1}-\frac{1}{2}\tau \widehat\Gamma_A{}^A{}_0 \Big)W_{0101}
-\frac{3}{2}( \widehat\Gamma_A{}^B{}_{0}+\tau  \widehat\Gamma_A{}^B{}_{1}) W^A{}_{B01}
\nonumber
\\
&
+\Big(\widehat\Gamma_1{}^A{}_{1} - \widehat\Gamma_1{}^A{}_0 \Big)W^{+}_{A}
+\Big( \widehat\Gamma_1{}^A{}_0 
+\widehat\Gamma_1{}^A{}_{1} \Big) W^{-}_{A}
\,,
\label{evolution12}
\\
(\tau\partial_{\tau}  +  e^{\mu}{}_1\partial_{\mu} )W_{01AB}    =&
(1+\tau)\Big(\check\nabla_{[A}  
-  2\widehat\Gamma_{[A}{}^0{}_{|0|} 
- \widehat\Gamma_{[A}{}^0{}_{|1|} \Big)W^{+}_{B]}
-(1-\tau)\Big(\check\nabla_{[A} 
-  2\widehat\Gamma_{[A}{}^0{}_{|0|} 
+ \widehat\Gamma_{[A}{}^0{}_{|1|}\Big) W^{-}_{B]}
\nonumber
\\
&
-(1-\tau)( \widehat\Gamma^{C1}{}_{[A}- \widehat\Gamma^{C0}{}_{[A})V^+_{B]C}
+(1+\tau)( \widehat\Gamma^{C1}{}_{[A}+ \widehat\Gamma^{C0}{}_{[A}) V^-_{B]C}
\nonumber
\\
&
+3(\widehat\Gamma_{[A}{}^0{}_{B]}  -\tau\widehat\Gamma_{[A}{}^{1}{}_{B]} )W_{0101}
+ 3\Big(\widehat\Gamma_1{}^0{}_{0} -\frac{1}{2}\widehat\Gamma_C{}^C{}_1-\frac{1}{2}\tau\widehat\Gamma_C{}^C{}_0 \Big)W_{01AB} 
\nonumber
\\
&
+2\Big(  \widehat\Gamma_1{}^0{}_{[A} 
+ \widehat\Gamma_1{}^1{}_{[A} \Big)W^{+}_{B]}
+2\Big( \widehat\Gamma_1{}^0{}_{[A} -\widehat\Gamma_1{}^1{}_{[A}\Big)  W^{-}_{B]}
\,.
\label{evolution13}
\end{align}
Taking radial derivatives of \eq{evolution8}-\eq{evolution13} and evaluating them on $I$ gives the desired equations. Note that the equations for 
$\partial_r^n W_{0101}|_I$ and $\partial^n_rW_{01AB}|_I$ are algebraic  (cf.\ \eq{frameI_1}), supposing that all lower order derivatives are known and supposing that $n \geq 1$.
For $n=0$ \eq{evolution12}-\eq{evolution13} need to be replaced by e.g.\ \eq{evolutionW1b}-\eq{evolutionW2b}, and this case has already been treated  in the previous section.

\subsection{First-order radial derivatives}
\label{sec_radial_equations}

Let us   consider the case $n=1$ for the first-oder radial derivatives on $I$ explicitly. 
For this, we differentiate the evolution equations \eq{evolution1}-\eq{evolution7} and  \eq{evolution8}-\eq{evolution13}  by $r$. Taking their restrictions to the cylinder
and using the results of Section~\ref{sec_connection_schouten_I} \& \ref{sec_solve_Bianchi_I}
we obtain 
transport equations on $I$  for $(e^{\mu}{}_i, \widehat \Gamma_i{}^j{}_k, \widehat L_{ij},W_{ijkl})$.

First of all note that with regard to \eq{evolution1}-\eq{evolution7}
\begin{align}
\partial_r\Theta|_I =& 1-\tau^2\,,
\quad
b_i|_I =0
\,,
\\
\partial_r b_0|_I =&-2\tau
\,, \quad
\partial_r b_1|_I = 2
\,, \quad
\partial_r b_A|_I = 0
\,.
\end{align}
For the Schouten tensor we  obtain the following set of equations,
\begin{align}
\partial_{\tau}\partial_r\widehat L_{10} |_I
&=-4M + \mathfrak{O}(1+\tau)
\,,
\\
\partial_{\tau}\partial_r\widehat L_{11}  |_I
&= -4 M+ \mathfrak{O}(1+\tau)
\,,
\\
\partial_{\tau}\partial_r\widehat L_{1A}  |_I
&= \mathfrak{O}(1+\tau)
\,,
\\
\partial_{\tau}\partial_r(\widehat L_{A0}-\widehat L_{A1}) |_I
&= \mathfrak{O}(1+\tau)
\,,
\\
\partial_{\tau}\partial_r(\widehat L_{A0}+\widehat L_{A1}) |_I
&=4\mathfrak{M}_A + \mathfrak{O}(1+\tau)
\,,
\\
\partial_{\tau}\partial_r\widehat L_{AB}  |_I
&=  2M\eta_{AB} +2 N\epsilon_{AB} +\mathfrak{O}(1+\tau)
\,.
\end{align}
Integration yields (the integration functions are determined by \eq{Schouten_scri1}-\eq{Schouten_scri4})
\begin{align}
\partial_r\widehat L_{10} |_I
&=-4M (1+\tau)+ \mathfrak{O}(1+\tau)^2
\,,
\\
\partial_r\widehat L_{11}  |_I
&= -4 M(1+\tau)+ \mathfrak{O}(1+\tau)^2
\,,
\\
\partial_r\widehat L_{1A}  |_I
&= \mathfrak{O}(1+\tau)^2
\,,
\\
\partial_r(\widehat L_{A0}-\widehat L_{A1}) |_I
&=\mcD_{ A} \nu_{\tau}^{(0)} +  \mcD_{ A} f^{(1)}_1 
-f^{(1)}_{ A}
+ \mathfrak{O}(1+\tau)^2
\,,
\\
\partial_r(\widehat L_{A0}+\widehat L_{A1}) |_I
&=v_{ A}^{(2)} +f^{(1)}_{ A}
-\mcD_{ A} f^{(1)}_1 +\mcD_{ A} \nu_{\tau}^{(0)}
+4\mathfrak{M}_A(1+\tau) + \mathfrak{O}(1+\tau)^2
\,,
\\
\partial_r\widehat L_{AB}  |_I
&= -\mcD_{ A}f^{(1)}_{ B} -\frac{1}{2}\Xi^{(2)}_{ A B}-f^{(1)}_1\eta_{ A B}
+2 (M\eta_{AB} + N\epsilon_{AB})(1+\tau) +\mathfrak{O}(1+\tau)^2
\,.
\end{align}
For the connection coefficients we end up with the following equations
\begin{align}
\partial_{\tau}\partial_r\widehat\Gamma_{1}{}^0{}_1  |_I
 &=
\mathfrak{O}(1+\tau)
\,,
\\
\partial_{\tau}\partial_r\widehat\Gamma_{1}{}^0{}_A  |_I
 &=
\mathfrak{O}(1+\tau)
\,,
\\
\partial_{\tau}\partial_r\widehat\Gamma_{A}{}^0{}_1  |_I
 &=
 \partial_r\widehat L_{A1}|_{I^-}  +\mathfrak{O}(1+\tau)
\,,
\\
\partial_{\tau}\partial_r\widehat\Gamma_{A}{}^0{}_B  |_I
 &=
 \partial_r\widehat L_{AB} |_{I^-}  +\mathfrak{O}(1+\tau)
\,,
\\
\partial_{\tau}\partial_r\widehat\Gamma_{1}{}^1{}_1  |_I
 &=
- \partial_r\widehat\Gamma_{1}{}^1{}_{0}|_{I^-}
    +\mathfrak{O}(1+\tau)
\,,
\\
\partial_{\tau}\partial_r\widehat\Gamma_{1}{}^1{}_A  |_I
 &=
\partial_r\widehat\Gamma_{1}{}^0{}_{A}|_{I^-}
    +\mathfrak{O}(1+\tau)
\,,
\\
\partial_{\tau}\partial_r\widehat\Gamma_{A}{}^1{}_1  |_I
 &=
-  \partial_r\widehat\Gamma_{A}{}^1{}_{0}|_{I^-}
 +\partial_r\widehat L_{A0} |_{I^-}   +\mathfrak{O}(1+\tau)
\,,
\\
\partial_{\tau}\partial_r\widehat\Gamma_{A}{}^1{}_B  |_I
 &=
\partial_r\widehat\Gamma_{A}{}^0{}_{B}|_{I^-}
   +\mathfrak{O}(1+\tau)
\,,
\\
\partial_{\tau}\partial_r\widehat\Gamma_{1}{}^A{}_B  |_I
 &=
- \delta^A{}_B\partial_r\widehat\Gamma_{1}{}^1{}_{0}|_{I^-}  
-  \mathring\Gamma_C{}^A{}_B\partial_r\widehat\Gamma_{1}{}^C{}_{0}  |_{I^-}
+\mathfrak{O}(1+\tau)
\,,
\\
\partial_{\tau}\partial_r\widehat\Gamma_{A}{}^B{}_C  |_I
 &=
-  \mathring\Gamma_D{}^B{}_C\partial_r\widehat\Gamma_{A}{}^D{}_{0} |_{I^-}
+ \delta^B{}_{C}\partial_r\widehat L_{A0}|_{I^-}   +\mathfrak{O}(1+\tau)
\,.
\end{align}
The solutions have the expansions (the integration functions follow from \eq{connection_scri_1}-\eq{connection_scri_10})
\begin{align}
\partial_r\widehat\Gamma_{1}{}^0{}_1  |_I
 &=
- f_1^{(1)}+ \mathfrak{O}(1+\tau)^2
\,,
\\
\partial_r\widehat\Gamma_{1}{}^0{}_A  |_I
 &=
- f^{(1)}_A + \mathfrak{O}(1+\tau)^2
\,,
\\
\partial_r\widehat\Gamma_{A}{}^0{}_1  |_I
 &=
\mcD_A\nu_{\tau}^{(0)}-\nu_A^{(1)} + \Big( \frac{1}{2}v_{ A}^{(2)} +f^{(1)}_{ A}-\mcD_{ A} f^{(1)}_1\Big)(1+\tau)  +\mathfrak{O}(1+\tau)^2
\,,
\\
\partial_r\widehat\Gamma_{A}{}^0{}_B  |_I
 &=
 \frac{1}{2}\Xi^{(2)}_{ AB}- \Big(\mcD_{ A}f^{(1)}_{ B} +\frac{1}{2}\Xi^{(2)}_{ A B}+f^{(1)}_1\eta_{ A B}\Big)(1+\tau)   +\mathfrak{O}(1+\tau)^2
\,,
\\
\partial_r\widehat\Gamma_{1}{}^1{}_1  |_I
 &=
f_1^{(1)} +  f_1^{(1)}(1+\tau)
    +\mathfrak{O}(1+\tau)^2
\,,
\\
\partial_r\widehat\Gamma_{1}{}^1{}_A  |_I
 &=f^{(1)}_A
- f^{(1)}_A(1+\tau)
    +\mathfrak{O}(1+\tau)^2
\,,
\\
\partial_r\widehat\Gamma_{A}{}^1{}_1  |_I
 &=
f_A^{(1)} 
 +\Big(\frac{1}{2} v_{ A}^{(2)}+\nu_A^{(1)}   \Big)(1+\tau)    +\mathfrak{O}(1+\tau)^2
\,,
\\
\partial_r\widehat\Gamma_{A}{}^1{}_B  |_I
 &=
 - \frac{1}{2}\Xi^{(2)}_{ AB}
-  f_1^{(1)}\eta_{AB}
+ \frac{1}{2}\Xi^{(2)}_{ A B}(1+\tau)
   +\mathfrak{O}(1+\tau)^2
\,,
\\
\partial_r(\widehat\Gamma_{1}{}^A{}_B)_{\mathrm{tf}}  |_I
 &=
  (\mathring\Gamma_C{}^A{}_Bf^{C(1)} )_{\mathrm{tf}}(1+\tau)
+\mathfrak{O}(1+\tau)^2
\,,
\\
\partial_r\widehat\Gamma_{A}{}^B{}_C  |_I
 &=
2\delta^B{}_{(A}f^{(1)}_{C)} - \eta_{AC} f^{(1)B}
-  \frac{1}{2}\mathring\Gamma_D{}^B{}_C \Xi^{(2)}_{ A}{}^{ D}(1+\tau)
\nonumber
\\
&
+ \delta^B{}_{C}\Big(\frac{1}{2} v_{ A}^{(2)} +\mcD_{ A}\nu_{\tau}^{(0)}   \Big) (1+\tau)  +\mathfrak{O}(1+\tau)^2
\,.
\end{align}
In fact, we have
\begin{equation}
\partial_{\tau}\partial_r(\widehat\Gamma_{1}{}^0{}_A +\widehat\Gamma_{1}{}^1{}_A )|_I
 =
-f_A^{(1)}   +\mathfrak{O}(1+\tau)^2
\,,
\end{equation}
whence
\begin{equation}
\partial_r(\widehat\Gamma_{1}{}^0{}_A +\widehat\Gamma_{1}{}^1{}_A )|_I
 =
- f^{(1)}_A (1+\tau) +\mathfrak{O}(1+\tau)^3
\,,
\label{better_decay1}
\end{equation}
which will be relevant below.
The equations for the frame coefficients read
\begin{align}
\partial_{\tau}\partial_re^{\tau}{}_1  |_I & = -  f_1^{(1)}(1+\tau)
    +\mathfrak{O}(1+\tau)^2
\,,
\\
\partial_{\tau}\partial_re^{r}{}_1  |_I & =
0
\,,
\\
\partial_re^{\mathring A}{}_1  |_I & = 
 f^{\mathring A(1)}(1+\tau) +\mathfrak{O}(1+\tau)^2
\,,
\\
\partial_re^{\tau}{}_A  |_I & = \Big(\nu_A^{(1)}-\mcD_A\nu_{\tau}^{(0)} -f^{(1)}_A\Big)(1+\tau)
 +\mathfrak{O}(1+\tau)^2
\,,
\\
\partial_{\tau}\partial_re^{r}{}_A  |_I & = 
0
\,,
\\
\partial_{\tau}\partial_re^{\mathring A}{}_A  |_I & =
 -\frac{1}{2}\Xi^{(2)}_A{}^B\mathring e^{\mathring A}{}_B +\mathfrak{O}(1+\tau)
\,,
\end{align}
from which we obtain the solutions (with data induced by \eq{frame_scri_1}-\eq{frame_scri_2})
\begin{align}
\partial_re^{\tau}{}_1  |_I & =-  f_1^{(1)}(1+\tau)^2 + \mathfrak{O}(1+\tau)^3
\,,
\\
\partial_re^{r}{}_1  |_I & =
1
\,,
\\
\partial_re^{\mathring A}{}_1  |_I & = 
 f^{\mathring A(1)} (1+\tau)+\mathfrak{O}(1+\tau)^2
\,,
\\
\partial_re^{\tau}{}_A  |_I & = \Big(\nu_A^{(1)}-\mcD_A\nu_{\tau}^{(0)} -f^{(1)}_A\Big)(1+\tau)
 +\mathfrak{O}(1+\tau)^2
\,,
\\
\partial_re^{r}{}_A  |_I & = 
0
\,,
\\
\partial_re^{\mathring A}{}_A  |_I & =
 -\frac{1}{2}\Xi^{(2)}_A{}^B\mathring e^{\mathring A}{}_B(1+\tau) +\mathfrak{O}(1+\tau)^2
\,.
\end{align}
From the equations \eq{evolution8}-\eq{evolution11} for the rescaled Weyl tensor we obtain
\begin{align}
\partial_{\tau}[(1-\tau) \partial_rV^+_{AB}] |_I 
 =&
\Big( 3(M \mcD_{ (A}f^{(1)}_{ B)} +  N\mcD_{ C}f^{(1)}_{ (A} \epsilon^C{}_{B)})_{\mathrm{tf}} 
-5(f^{(1)}_{(A}\ol{\mathfrak{M}}_{B)})_{\mathrm{tf}}\Big)(1+\tau)
\nonumber
\\
&
+(\mcD_{(A} \partial_r W^-_{B)})_{\mathrm{tf}}
+\mathfrak{O}(1+\tau)^2
\,,
\label{eq_trans_derv_V+}
\\
\partial_{\tau}[(1+\tau)\partial_r V^-_{AB} ]|_I 
=&
-(\mcD_{(A} \partial_rW^+_{B)})_{\mathrm{tf}}+\mathfrak{O}(1)
\,,
\\
(1-\tau)\partial_{\tau}\partial_rW^+_{A} |_I 
 =&
-2\partial_{r}W^+_{A} 
-2 \mcD^B\partial_rV^-_{AB} +\mathfrak{O}(1)
\,,
\\
(1+\tau )\partial_{\tau} \partial_r W^-_{A}  |_I
 =&
2\partial_{r} W^-_{A}
 +2\mcD^B\partial_rV^+_{AB}
+3f^B(M\eta_{AB}-N\epsilon_{AB})(1+\tau)
\nonumber
\\
&
+2\Big(
f^B\Delta_s( M\eta_{AB}- N\epsilon_{AB})
 -2\ol{\mathfrak{M}}^B\mcD_{[A}f^{(1)}_{B]}
+\ol{\mathfrak{M}}_A\mcD^Bf^{(1)}_B
\Big)(1+\tau)^2
\nonumber
\\
&
+\mathfrak{O}(1+\tau)^3
\,.
\end{align}
Taking the divergence of the first two equations and inserting  the result into the latter two yields
decoupled equations for $\partial_rW^{\pm}_A|_I$,
\begin{align}
&\hspace{-2em}
[(1-\tau^2)\partial^2_{\tau}
 -2 \partial_{\tau}   
-(\Delta_s-1)]\partial_rW^-_A|_I
\nonumber
\\
 =&
6 f^{(1)B}(M\eta_{AB}-N\epsilon_{AB})+
 3 (\Delta_s-1)\Big( f^{(1)B}(M\eta_{AB}-N\epsilon_{AB})\Big)(1+\tau)
+\mathfrak{O}(1+\tau)^2
\,,
\label{sing_wave_rad_W-}
\\
&\hspace{-2em}
[(1-\tau^2)\partial^2_{\tau}
 +2 \partial_{\tau}   
-(\Delta_s-1)]\partial_rW^+_A|_I
=\mathfrak{O}(1)
\,.
\label{sing_wave_rad_W+}
\end{align}
The analysis in Section~\ref{sec_smoothness_Bianchi_I}
shows that the   solutions are regular at $I^-$ for any  data 
$\partial_rW^-_A|_{I^-} $, $\partial^2_{\tau}\partial_rW^-_A|_{I^-}$
and $\partial_rW^+_A|_{I^-} $ (the second integration function for $W^+_A$ comes along with a log term and therefore needs to vanish).
We have (cf.\ \eq{data_gen3}-\eq{data_gen4})
\begin{equation}
\partial_rW^-_A|_{I^-} =0
\,, \quad 
\partial_rW^+_A|_{I^-} =2L_A
\,,
\end{equation}
the  datum $\partial^2_{\tau}\partial_rW^-_A|_{I^-}$ is irrelevant for our purposes.
Moreover, the solutions  admit an expansion of the form
\begin{align}
\partial_rW^-_A|_I =& -3 f^{(1)B}(M\eta_{AB}-N\epsilon_{AB}) (1+\tau)
+\mathfrak{O}(1+\tau)^2
\,,
\\
\partial_rW^+_A|_I =& 2L_A +\mathfrak{O}(1+\tau)
\,.
\end{align}
From this we compute $\partial_rV^{\pm}_{AB}|_I$. The integration function for $\partial_rV^{+}_{AB}|_I$
is determined by continuity from \eq{data_gen5}, while that for  $\partial_rV^{-}_{AB}|_I$ needs to vanish in order to get a bounded solution.
\begin{align}
\partial_rV^+_{AB} |_I 
 =& \mathfrak{O}(1+\tau)^2
\,,
\label{V+AB_I}
\\
\partial_rV^-_{AB} |_I
=&
 \Big(
-2\mcD_{(A}L_{ B)} 
- 3\Xi^{(2)}_{(A}{}^C (M\eta_{B)C} + N\epsilon_{B)C})
+2\nu^{(1)}_{(A}\mathfrak{M}_{B)}-2\mathfrak{M}_{(A}\mcD_{B)} \nu^{(0)}_{\tau}
\Big)_{\mathrm{tf}}
+ \mathfrak{O}(1+\tau )
\,.
\end{align}
Finally, the radial derivative of \eq{evolution12}-\eq{evolution13} yields,
\begin{equation}
\partial_{r}U_{AB}|_I
 = 
\mathfrak{O}(1+\tau )
\,.
\end{equation}
Again, one checks that all values at $I^-$ are  in accordance with \eq{data_gen1}-\eq{data_gen6}.
We conclude:

\begin{lemma}
Under the same hypotheses as in Proposition~\ref{prop_smoothness_constraints},
no additional restrictions need to be imposed on the data, to get regular expansions of $(e^{\mu}{}_i, \widehat\Gamma_i{}^j{}_k, L_{ij},W_{ijkl})|_I$ and
their first order radial derivatives at $I^-$.
\end{lemma}

\subsection{Higher-order derivatives: Structure of the equations and no-logs conditions}
\label{sec_structure_eqn_I}

As in Section~\ref{sec_structure_eqn_scri} we want to derive the overall structure of the transport equations on the cylinder 
for radial derivatives of any order of the fields involved, in particular concerning the appearance of logarithmic terms.
For this we assume that appropriate data have been prescribed on $\scri^-$ (and some incoming null hypersurface) which generate a smooth solution $\mathfrak{f}=(e^{\mu}{}_i, \widehat \Gamma_i{}^j{}_k, \widehat L_{ij}, W_{ijkl})$ to the GCFE in a weakly asymptotically Minkowski-like conformal Gauss gauge which
admits a smooth $\scri^-$ and a smooth cylinder $I$ .
We further assume that all transverse derivatives $\partial^n_{\tau}\mathfrak{f}|_{\scri^-}$ are smooth at $I^-$, and induce there the data for
the transport equations on $I$.
Our goal is to analyze the smoothness of $\partial^n_r\mathfrak{f}|_{I}$  at $I^-$.

Previously we have shown that no additional restrictions apart from those in 
Proposition~\ref{prop_smoothness_constraints}
are needed for the smoothness of $\mathfrak{f}|_I$ and $\partial_r\mathfrak{f}|_I$. Let us assume now that $\partial^k_r\mathfrak{f}|_{I}$ is smooth at $I^-$ for $0\leq k\leq n-1$.

We consider  the evolution equations \eq{evolution1}-\eq{evolution7},
and apply $\partial^{n}_{r}$. 
With \eq{frameI_1}-\eq{Schouten_I}
that  yields ODEs
for $(\partial^n_{r}e^{\mu}{}_i, \partial^n_{r}\widehat \Gamma_i{}^j{}_k, \partial^n_{r}\widehat L_{ij})|_{I}$,
\begin{align}
\partial_{\tau}\partial^n_r\widehat L_{a0} |_I
&=\mathfrak{O}(1)
\,,
\\
\partial_{\tau}\partial^n_r\widehat L_{ab}  |_I
&= \mathfrak{O}(1)
\,,
\\
\partial_{\tau}\partial^n_r\widehat\Gamma_{a}{}^0{}_b  |_I
 &=
 \partial^n_r\widehat L_{ab}  
+\mathfrak{O}(1)
\,,
\\
\partial_{\tau}\partial^n_r\widehat\Gamma_{a}{}^1{}_b  |_I
 &=
-  \widehat\Gamma_c{}^1{}_b\partial^n_r\widehat\Gamma_{a}{}^c{}_{0} +\delta^1{}_{b}\partial^n_r\widehat L_{a0}  
+\mathfrak{O}(1)=\mathfrak{O}(1)
\,,
\\
\partial_{\tau}\partial^n_r\widehat\Gamma_{1}{}^A{}_B  |_I
 &=
-  \widehat\Gamma_c{}^A{}_B\partial^n_r\widehat\Gamma_{1}{}^c{}_{0} + \delta^A{}_{B}\partial^n_r\widehat L_{10}   
+\mathfrak{O}(1)=\mathfrak{O}(1)
\,,
\\
\partial_{\tau}\partial^n_r\widehat\Gamma_{A}{}^B{}_C  |_I
 &=
-  \widehat\Gamma_D{}^B{}_C\partial^n_r\widehat\Gamma_{A}{}^D{}_{0} + \delta^B{}_{C}\partial^n_r\widehat L_{A0}  
+\mathfrak{O}(1)=\mathfrak{O}(1)
\,,
\\
\partial_{\tau}\partial^n_re^{\mu}{}_a  |_I &= -\partial^n_r\widehat\Gamma_{a}{}^0{}_{0} \delta^{\mu}{}_0
 -\partial^n_r\widehat\Gamma_{a}{}^b{}_{0} e^{\mu}{}_b
+\mathfrak{O}(1)=\mathfrak{O}(1)
\,.
\end{align}
The ODEs can  be straightforwardly integrated
with  initial data computed  from the data on $\scri^-$, or rather their limit to $I^-$.
It follows that the restrictions to $I$ of the $n$-th-order $r$-derivatives of  frame field,  connection coefficients and  Schouten tensor 
are smooth at $I^-$, supposing that this is the case for all derivatives of $\mathfrak{f}$ up to and including order $n-1$,
\begin{align}
\partial^n_r\widehat L_{ai} |_I
&=\mathfrak{O}(1)
\,,
\\
\partial^n_r\widehat\Gamma_{a}{}^i{}_j  |_I
 &=
\mathfrak{O}(1)
\,,
\\
\partial^n_re^{\mu}{}_a  |_I &= \mathfrak{O}(1)
\,.
\end{align}
We also apply $\partial^n_r$  to the equations \eq{evolution8}-\eq{evolution13} for the rescaled Weyl tensor, and take their restriction to $I$,
\begin{align}
 \partial_{\tau}[(1-\tau)^{2-n}  \partial^n_r V^+_{AB}]|_I
 =& (1-\tau)^{1-n} (\mcD_{(A}  \partial^n_rW^-_{B)})_{\mathrm{tf}}+\mathfrak{O}(1)
\,,
\label{Wely_general_I1}
\\
\partial_{\tau}[(1+\tau)^{2-n}\partial^n_rV^-_{AB} ]|_I
=&
 -(1+\tau)^{1-n}( \mcD_{(A} \partial^n_rW^+_{B)})_{\mathrm{tf}}+\mathfrak{O}(1+\tau)^{1-n}
\,,
\label{Wely_general_I2}
\\
[(1+\tau )\partial_{\tau}-(n+1)]\partial^n_rW^-_{A}|_I
 =&2  \mcD^B \partial^n_rV^+_{AB}
+\mathfrak{O}(1)
\,,
\label{Wely_general_I3}
\\
[(1-\tau)\partial_{\tau} +  (n+1)]\partial^n_rW^+_{A}|_I
 =& -  2\mcD^B \partial^n_rV^-_{AB}
+\mathfrak{O}(1)
\,,
\label{Wely_general_I4}
\\
n\partial^n_{r} W_{0101}|_I
 =& 
- \frac{1}{2}(1+\tau)\mcD^A \partial^n_r  W^{+}_{A}
- \frac{1}{2}(1-\tau)\mcD^A \partial^n_r  W^{-}_{A}
+\mathfrak{O}(1)
\,,
\label{Wely_general_I5}
\\
n\partial^n_{r} W_{01AB}|_I    =&
(1+\tau) \mcD_{[A}  \partial^n_rW^{+}_{B]}
-(1-\tau)\mcD_{[A}  \partial^n_r W^{-}_{B]}
+\mathfrak{O}(1)
\,.
\label{Wely_general_I6}
\end{align}
We take the divergence of \eq{Wely_general_I1} and \eq{Wely_general_I2}, and insert them  into
\eq{Wely_general_I3} and \eq{Wely_general_I4}, respectively, to get decoupled equations
\begin{align}
\Big((1-\tau^2)\partial^2_{\tau}
+2[(n-1)\tau-1] \partial_{\tau}
-(\Delta_s + n^2-n-1)\Big)\partial^n_rW^-_{A}
 =&\mathfrak{O}(1)
\,,
\label{sing_wave1}
\\
\Big((1-\tau^2) \partial^2_{\tau} 
+2[(n-1)\tau+1]\partial_{\tau}
- (\Delta_s +n^2-n-1)\Big)\partial^n_rW^+_{A}
 =& \mathfrak{O}(1)
\,.
\label{sing_wave2}
\end{align}
The regularity of solutions to this equation at $I^-$ is discussed in Section~\ref{sec_smoothness_Bianchi_I}.
For the time being, let us assume  that the data are such that the  solutions are smooth at $I^-$.
Then \eq{Wely_general_I1} can be integrated for initial data induced by $V^+_{AB}|_{\scri^-}$. The solution $  \partial^n_r V^+_{AB}|_I$ will be smooth at $I^-$.
The equations \eq{Wely_general_I5}-\eq{Wely_general_I6} determine $\partial^n_{r} W_{0101}|_I$ and  $\partial^n_{r} W_{01AB}|_I$
algebraically and the components will be smooth at $I^-$, as well.

It remains to compute $  \partial^n_r V^-_{AB}|_I$.
We observe that, in contrast to $n=0,1$, for $n\geq 2$ the solution to  \eq{Wely_general_I2},
\begin{equation}
\partial_{\tau}[(1+\tau)^{2-n}\partial^n_rV^-_{AB} ]|_I
=
\mathfrak{O}(1+\tau)^{1-n}
\label{no-logs-cond-r_n_V-}
\end{equation}
 will be bounded at $I^-$ for any choice of the initial data, which are given by the integration functions $\partial^{n-2}_{\tau}\partial^n_rV^-_{AB} |_{I^-}=c^{(n,n-2)}_{\mathring A\mathring B}$, $n\geq 2$, which can be regarded as part of the freely prescribable data, cf.\ Appendix~\ref{app_alternative_data}.

We further observe that, for $n\geq 2$, the solution will  develop logarithmic terms at $I^-$  unless the right-hand side does not have a term of order $(1+\tau)^{-1}$
in its expansion at $I^-$. This is another no-logs condition which needs to be imposed.

Comparing this with \eq{derivs_V-} we observe that \eq{no-logs-cond-r_n_V-} is very similar to the corresponding one $\scri^-$
(cf.\  Section~\ref{sec_comparison_Iscri}),
\begin{equation*}
\Big(\partial_r+\mathfrak{O}(1)\Big)(r^{-n-2}\partial^n_{\tau}V^-_{AB})|_{\scri^-}
 =
\mathfrak{O}(r^{-n-3})
\,.
\end{equation*}
In both cases $\partial^{k+2}_r\partial_{\tau}^kV^-_{AB}$ diverges at $I^-$ for some $k$ if logarithmic terms appear.
% while lower-order derivatives are smooth.

\subsection{Analysis of the singular wave equation on $I$}
\label{sec_smoothness_Bianchi_I}

We want to analyze \eq{sing_wave1}-\eq{sing_wave2}, as well as \eq{decoupled_k_A}-\eq{decoupled_h_A} and \eq{sing_wave_rad_W-}-\eq{sing_wave_rad_W+}.
To deal with scalar equations we take curl and divergence.
Let
\begin{equation}
\phi^{\pm}_n  \in\{ \mcD^A\partial^n_rW^{\pm}_A , \epsilon^{ A B}\mcD_A\partial^n_r W^{\pm}_B\}
\,,
\end{equation}
then we are led to study the following  linear PDE of Fuchsian type
\begin{equation}
\Big( (1-\tau^2)\partial^2_{\tau}
+ 2 [ (n-1)\tau \pm1]\partial_{\tau}
 -[ \Delta_s
+n(n-1)]
\Big)\phi_{n}^{\pm}
=q^{\pm}_n
\,.
\label{Fuchsian_PDE}
\end{equation}
on $[-1,1]\times S^2$ for a given smooth source $q^{\pm}_n$,
and
with  $s = \mathrm{d}\theta^2 + \sin^2\theta\mathrm{d}\varphi^2$.
Eventually we are interested in smooth solutions  which allow a decomposition into spherical harmonics.
 Since, by construction, $\phi_n^{\pm}$ and $W^{\pm}_n$  are divergence or curl of a 1-form, their harmonic decompositions  will
not contain $\ell=0$-spherical harmonics,
\begin{align}
\phi^{\pm}_n (\tau, \theta,\varphi)=&\sum_{\ell =1}^{\infty}\sum_{m=-\ell}^{+\ell}\phi^{\pm}_{n\ell m}(\tau)Y_{\ell m}(\theta,\varphi)
\,,
\\
q^{\pm}_n (\tau, \theta,\varphi)=&\sum_{\ell =1}^{\infty}\sum_{m=-\ell}^{+\ell}q^{\pm}_{n\ell m}(\tau)Y_{\ell m}(\theta,\varphi)
\,,
\end{align} 
where
\begin{equation}
\Delta_s Y_{\ell m}= -\ell(\ell + 1)Y_{\ell m}
\,.
\end{equation}
That yields ODEs for the expansion coefficients,
\begin{equation}
 \Big(  (1-\tau^2)\partial^2_{\tau}  
+2  [(n-1)\tau \pm1]\partial_{\tau}
+(\ell +n)(\ell-n +1)
 \Big)\phi^{\pm}_{n\ell m}(\tau)
=
q^{\pm}_{n\ell m}(\tau)
\,.
\label{hypergeom_equations}
\end{equation}
Set 
\begin{equation*}
z :=\frac{1+\tau}{2}
\,,
\quad
a_{n\ell }:= -(\ell+n)
\,,
\quad
b_{n\ell }:= \ell-n+1
\,,
\quad
c^{\pm}_n:= n-1\mp 1
\,,
\end{equation*}
then \eq{Fuchsian_PDE} becomes,
\begin{equation}
 \Big(  z(1-z)\partial^2_{z}  -(c^{\pm}_n+(a_{n\ell }+b_{n\ell }+1)z)\partial_{z}-a_{n\ell }b_{n\ell }
 \Big)\phi^{\pm}_{n\ell m}(z)
=
q^{\pm}_{n\ell m}(z)
\,,
\end{equation}
which is a \emph{hypergeometric equation} with source term which can be solved using e.g.\  Frobenius method, cf.\ e.g.\ \cite{david} and the references given there.
Such an equation already appears in \cite{F_i0}, cf.\ also \cite{kroon},
 where, among oher things,  the general solutions to its homogeneous  counterpart is constructed in terms of generalized Jacobi polynomials.
The coefficients  there differ slightly from \eq{hypergeom_equations}, because the equations  are  expressed  in terms of different components
(equations for divergence and curl of $\mcD^B\partial_r^nV^{\pm}_{AB}|_I$ would yield identical equations).
 In particular it becomes clear from the discussion there, that non-smoothness of the solutions is actually
due to the appearance of\emph{ logarithmic terms}, which will not be immediate from our subsequent discussion.

In the following let us first  focus  on the case $n\geq 2$. 
We are interested in the behavior at $\tau=-1$, i.e.\ at $z=0$.
By assumption, the source term is smooth in $z$ and therefore admits an expansion of the form
\begin{equation*}
q^{\pm}_{n\ell m}\sim\sum_{k=0}^{\infty}q^{\pm}_{k n  \ell m}z^k
\,.
\end{equation*}
Any  smooth solution $\phi^{\pm}_{n\ell m}$  admits an expansion  at $z=0$ of the form,
\begin{equation*}
\phi^{\pm}_{n\ell m}\sim\sum_{k=0}^{\infty}\phi^{\pm}_{k n  \ell m}z^k
\,.
\end{equation*}
We plug it in
\begin{equation}
(k+1) (k-c^{\pm}_n)\phi^{\pm}_{(k+1) n  \ell m}
-\Big( k(k-1)
+k(a_{n\ell } +b_{n\ell }+1)
+a_{n\ell }b_{n\ell }\Big)\phi^{\pm}_{kn  \ell m}
=
q^{\pm}_{kn  \ell m}
\,.
\label{compoute_coeffs}
\end{equation}
If the solution is smooth at $I^-$ this system needs to admit a solution. The solution  is determined by  regarding   the system as a hierarchy 
of equations for $\{\phi^{\pm}_{kn  \ell m}\}_{k\in\mathbb{N}}$.
This imposes the  restriction,
\begin{equation}
%\Big(  c^{\pm}_n (n\pm 1) -a_{n\ell }b_{n\ell }\Big)\phi^{\pm}_{ c^{\pm}_n n  \ell m}= q^{\pm}_{c^{\pm}_n n  \ell m}
\Big(\ell(\ell+1)-1\mp1\Big)\phi^{\pm}_{ c^{\pm}_n n  \ell m}= q^{\pm}_{c^{\pm}_n n  \ell m}
\,.
\label{smoothness_condition}
\end{equation}
Supposing that \eq{smoothness_condition} holds, the solution to \eq{Fuchsian_PDE} will be  of the form
$\phi^{\pm}_n=\mathfrak{O}(1)$.

Let us analyze this condition in detail. We begin with $\phi^-$. Note that $c^-_n=n$, and that \eq{smoothness_condition} can be written as
\begin{equation}
\Delta_s\partial^n_{\tau}\phi^{-}_{ n}= -\partial^n_{\tau}q^{-}_{n }
\,.
\label{smoothness_conditionA}
\end{equation}
The first factor in \eq{smoothness_condition} is nonzero for all $\ell\geq 1$.
For fixed $n$, $\ell$ and $m$, $\phi^{-}_{n n  \ell m}$ is determined from  \eq{compoute_coeffs} by solving a hierarchical system for $\phi_{kn  \ell m}$,
$1\leq k\leq n$,  in terms 
of  the initial data $\phi^{-}_{0 n\ell m}$ and the source $q^{-}_{k n  \ell m}$, $k\leq n-1$. 
Then the coefficients  $\phi^{-}_{kn  \ell m}$,
$k\geq n+1$, are determined in terms of the data  $\phi^{\pm}_{(n+ 1) n\ell m}$   and the source $q^{\pm}_{k n  \ell m}$, $k\geq  n+ 1$. 
The ``integration functions''  
\begin{equation}
\phi^{-}_{0n  \ell m} \,, \quad \phi^{-}_{ (n+ 1)n  \ell m}
\,,
\end{equation}
are determined from the data on $\scri^-$, i.e.\ from the solution of the $W^-_A$-constraint and its $(n+1)$st-order transverse derivative.

Next we consider condition \eq{smoothness_condition} for $\phi^+$. Note that $c^+_n=n-2$, and that \eq{smoothness_condition} can be written as
\begin{equation}
(\Delta_s +2)\partial^{n-2}_{\tau}\phi^{+}_{ n }= -\partial^{n-2}_{\tau} q^{+}_{ n }
\,.
\label{smoothness_conditionB}
\end{equation}

%For $n=0,1$ the smoothness condition \eq{smoothness_condition}  becomes trivial and no restrictions need to be imposed.
For $n\geq 2$ the first factor in \eq{smoothness_condition} is non-zero for $\ell \geq 2$. 
In that case, again,  $\phi^{+}_{(n-2) n  \ell m}$ is determined from  \eq{compoute_coeffs} by solving a hierarchical system for $\phi_{kn  \ell m}$,
$1\leq k\leq n-2$,  in terms 
of  the initial data $\phi^{+}_{0 n\ell m}$ and the source $q^{+}_{k n  \ell m}$, $k\leq n-3$. 
Then the coefficients  $\phi^{+}_{kn  \ell m}$,
$k\geq n-1$, are determined in terms of the data  $\phi^{+}_{(n- 1) n\ell m}$   and the source $q^{+}_{k n  \ell m}$, $k\geq  n- 1$. 
The ``integration functions''  
\begin{equation}
\phi^{+}_{0n  \ell m} \,, \quad \phi^{+}_{ (n- 1)n  \ell m}
\,.
\end{equation}
 are determined by $W^+_A|_{\scri^-}$ and its $(n-1)$st-order transverse derivative.
For $n \geq 2$ and $\ell=1$ \eq{smoothness_condition} becomes a condition on the source term, 
\begin{equation}
q^{+}_{(n-2) n  1 m}=0\,.
\label{source_condition}
\end{equation}

Some  consequences of the above considerations are provided by the following lemma:
\begin{lemma}
\label{lem_sing_wave}
Take $n\geq2$.
\begin{enumerate}
\item[(i)]
Assume that  the initial data at $I^-$ satisfy
$$\phi^{\pm}_n|_{I^-}=\sum_{\ell=1}^{n- 1}\sum_{m=-\ell}^{+\ell}\phi^{\pm}_{0n\ell m}Y_{\ell m}(\theta,\phi)
\,,
$$
and that the source term satisfies
$$
q^{\pm}_n|_{I}=\sum_{\ell=1}^{n- 1}\sum_{m=-\ell}^{+\ell}q^{\pm}_{n\ell m}(\tau)Y_{\ell m}(\theta,\phi)
+\mathfrak{O}(1+\tau)^{n\mp 1} 
\,.
$$
 Then the solution is smooth at $I^-$ if $q^{\pm}_{n\ell m}(\tau)=\mathcal{P}^{n-\ell-2}$,
where $\mathcal{P}^k$ denotes a polynomial in $(1+\tau)$ of degree $\leq k$.
%If
%$$
%q^{\pm}_{n\ell m}(\tau)=\mathcal{P}^{n-\ell-2}+ \sum_{k=n-\ell-1}^{n-\ell+2}q^{\pm}_{kn\ell m}(1+\tau)^k
%\,,
%$$
%\tim{add comment why 3}
%and, to avoid a case distinction, $q^{\pm}_{kn\ell m} =0$ for $\ell <3\pm 1$,
%\tim{add comment.. $\ell=1$}
% then the solution is smooth at $I^-$ if and only if
%$$
%2(1-2\ell)\Big(\frac{ -2\ell q^{\pm}_{k_*n  \ell m} }{(k_*+1) (k_*-c^{\pm}_n)}+q^{\pm}_{(k_*+1)n  \ell m}\Big)
%+q^{\pm}_{(k_*+2)n  \ell m} (k_*+2) (k_*+1-c^{\pm}_n)=0
%\,.
%$$
\item[(ii)]
Assume that $\phi^{\pm}_n|_{I^-}=0$, $q^{\pm}_n|_{I^-}= O(1+\tau)^{c_n^{\pm}} $ with $\partial^{c_n^{\pm}}_{\tau}q_n|_{I^-}\ne0 $. Then the solution cannot be smooth at $I^-$.
\item[(iii)]
Assume that $\langle \phi^{\pm}_n|_{I^-},Y_{\ell* m_*}\rangle\ne 0$ for some $\ell_* \geq n$ and some $-\ell_*\leq m_*\leq \ell_*$, and $\langle q^{\pm}_n|_{I^-},Y_{\ell_* m_*}\rangle=O(1+\tau)^{n\mp 1} $.
Then the solution cannot be smooth at $I^-$.
\item[(iv)]
Assume that $\langle \partial_{\tau}^{n-2}q^{+}_{ n}|_{I^-}, Y_{1 m^*}\rangle \ne 0$ for $m^*\in \{-1,0,1\}$.
Then the solution cannot be smooth at $I^-$.
\end{enumerate}
\end{lemma}

\begin{proof}
(i): We need to check whether  \eq{smoothness_condition} holds. The data and the source term has been chosen in such a way, that
$\phi^{\pm}_{kn\ell m}=0$ and $q^{\pm}_{kn\ell m}=0$ for $\ell \geq  n$ and $k\leq c^{\pm}_n$, in particular   \eq{smoothness_condition}
holds for $\ell \geq n$ (and $\ell=0$).
To deal with  $1\leq \ell \leq  n- 1$, we consider  \eq{compoute_coeffs} from which we obtain
\begin{equation}
\label{recursion_relation}
\phi^{\pm}_{(k+1) n  \ell m}
=\frac{ k(k-1)
+k(a_{n\ell } +b_{n\ell }+1)
+a_{n\ell }b_{n\ell }}{(k+1) (k-c^{\pm}_n)}
\phi^{\pm}_{kn  \ell m}
+ \frac{1}{(k+1) (k-c^{\pm}_n)}q^{\pm}_{kn  \ell m}
\,.
\end{equation}
The zeros of the  numerator
\begin{align*}
k(k-1)
+k(a_{n\ell } +b_{n\ell }+1)
+a_{n\ell }b_{n\ell }
=&
k(k-1)
-2k(n-1)
-(\ell+n)(\ell-n+1)
\\
=&
k^2+k-2kn-\ell^2-\ell+n^2-n
\end{align*}
 are given by $k=n+\ell $ and $k=n-\ell -1$. We are only interested in those zeros where $k$ is an integer in the interval $[0,n]$ for $\phi^-$ and
 where $k$ is an integer in the interval $[0,n-2]$ for $\phi^+$.
Recall that $\ell\geq 1$.
Zeros of $k$ in the desired range appear if and only if  $\ell$ satisfies $1\leq \ell\leq n-1$, 
$$
k_*=n-\ell -1
\,.
$$
In particular,
\begin{equation*}
\phi^{\pm}_{(k_*+1) n  \ell m}
= \frac{1}{(k_*+1) (k_*-c^{\pm}_n)}q^{\pm}_{k_*n  \ell m}
\,,
\end{equation*}
and  $\phi^{\pm}_{k n  \ell m}$ with $k_*+1\leq k\leq c^{\pm}_n$ merely depends on the source $q_n$ but not on the second initial datum.
Now if the source has been chosen in such a way that $q^{\pm}_{k n\ell m}=0$ for $k_*\leq k\leq c^{\pm}_n$,
it follows that $\phi^{\pm}_{k n  \ell m}=0$ for $k_*+1\leq k\leq c^{\pm}_n$, and \eq{smoothness_condition} is fulfilled
(for $\phi^+$ and $\ell=1$, where $k_*+1> c^{+}_n$ we have \eq{source_condition}, and  \eq{smoothness_condition} holds as well).

%If the source is non-trivial (for say 3 orders as in the assertion), the recursion formula necessarily needs to map $\phi^{\pm}_{k n  \ell m}$ to zero 
%which is then only possible for $k_*+1\leq k\leq k_*+3$,
%%
%\begin{align*}
%&
%0\overset{!}{=}\phi^{\pm}_{(k_*+3) n  \ell m}
%=\frac{2(1-2\ell)\phi^{\pm}_{(k_*+2)n  \ell m}+q^{\pm}_{(k_*+2)n  \ell m}}{(k_*+3) (k_*+2-c^{\pm}_n)}
%\\
%\Longleftrightarrow \quad &
%0
%=2(1-2\ell)\Big(\frac{ -2\ell q^{\pm}_{k_*n  \ell m} }{(k_*+1) (k_*-c^{\pm}_n)}+q^{\pm}_{(k_*+1)n  \ell m}\Big)
%+q^{\pm}_{(k_*+2)n  \ell m} (k_*+2) (k_*+1-c^{\pm}_n)
%\,,
%\end{align*}
%%
%for $k_*+3\leq c_n^{\pm}$, i.e.\ for $\ell\geq 3\pm 1$.

(ii): This is straightforward.

(iii): By assumption there exist $\ell_* \geq n$ and $m_*$ such that $\phi^{\pm}_{0n  \ell_* m_*}\ne 0$.
Taking also the assumption on the source term into account we deduce from \eq{compoute_coeffs}
\begin{equation}
\phi^{\pm}_{(k+1) n  \ell_* m_*}
=\frac{ k(k-1)
+k(a_{n\ell_* } +b_{n\ell_* }+1)
+a_{n\ell_* }b_{n\ell_* }}{(k+1) (k-c^{\pm}_n)}
\phi^{\pm}_{kn  \ell_* m_*}
\,,
\end{equation}
and the consderations above show that the numerator does not have integer zeros in $[0, c^{\pm}_n]$.
It follows that  $\phi^{\pm}_{k n  \ell_* m_*}\ne 0$  for $0\leq k\leq  c^{\pm}_n$ and the smoothness condition \eq{smoothness_condition} is violated.

(iv): That is \eq{source_condition}.
\qed
\end{proof}

Let us now  consider the remaining  cases where $n=0,1$.
For $n=0$ there is no source term, cf.\ \eq{decoupled_k_A}-\eq{decoupled_h_A}, 
and   \eq{compoute_coeffs} becomes
\begin{align}
(k+1) (k+2)\phi^{+}_{(k+1) 0  \ell m}
-\Big( k(k-1)
+2k
-\ell(\ell+1)\Big)\phi^{+}_{k0  \ell m}
=&
0
\,,\quad k\geq -1
\,,
\\
(k+1) k\phi^{-}_{(k+1) 0  \ell m}
-\Big( k(k-1)
+2k
-\ell(\ell+1)\Big)\phi^{-}_{k0  \ell m}
=&
0
\,, \quad k\geq 0
\,.
\label{--case_smoothness} 
\end{align}
We observe that $\phi^+_{-10\ell m}$ is an integration function which produces a divergence term at $I^-$. Also in the second case one integration function is lost
by the smoothness requirement as can be seen by evaluating \eq{--case_smoothness} for $k=0$.

In both cases, to get a smooth solution at $\tau=-1$ there is only one freely prescribable datum, while the other one needs to vanish,
\begin{equation}
\phi^+_{-10\ell m}=0\,, \quad  \phi^+_{00\ell m}
\,,  \quad  \phi^-_{00\ell m}=0, \quad  \phi^-_{10\ell m}
\,.
\end{equation}
Note that the data \eq{vanish_h_I-}-\eq{cond_I-_4} are indeed of this form. 

For $n=1$ the hypergeometric equation is of the form, cf.\eq{sing_wave_rad_W-}-\eq{sing_wave_rad_W+},
\begin{align}
 \Big(  (1-\tau^2)\partial^2_{\tau}  
+ 2  \partial_{\tau}
+\ell(\ell +1)
 \Big)\phi^{+}_{1\ell m}(\tau)
=&
 \mathfrak{O}(1)
\,,
\\
 \Big(  (1-\tau^2)\partial^2_{\tau}  
- 2  \partial_{\tau}
+\ell(\ell +1)
 \Big)\phi^{-}_{1\ell m}(\tau)
=&
\mathfrak{O}(1+\tau)^2
\,.
\end{align}
In this case we are led to the system
\begin{equation}
(k+1) (k\pm 1)\phi^{\pm}_{(k+1) 1  \ell m}
-\Big( k(k-1)
-\ell(\ell+1)\Big)\phi^{\pm}_{k1  \ell m}
=
q^{\pm}_{k1  \ell m}
\,,
\end{equation}
%
%where $2q^{+}_{11  \ell m}=\ell(\ell+1) q^{+}_{01  \ell m}$ and $q^{-}_{01  \ell m}=0=q^{-}_{11  \ell m}$.
In the ``$-$''-case the free data are $\phi^{-}_{0 1  \ell m}$ and  $\phi^{-}_{21  \ell m}$ and  the solution is smooth at $\tau=-1$ since the source is
of order $(1+\tau)^2$.
In the  ``$+$''-case there is only one free datum (the second one is not visible as it comes along with a logarithmic   term and  therefore needs to vanish), namely   $\phi^{+}_{0 1  \ell m}$,
 which generates a smooth solution.

%general solution reads
%%
%\begin{align*}
%\phi^{+}_{0\ell m}(\tau)
%%=&c_2 G_{2,2}^{2,0}\left(\frac{1-\tau}{2}|
%%\begin{array}{c}
%% -l,l+1 \\
%% 0,1 \\
%%\end{array}
%%\right)+ c_1 (1-\tau ) \, _2F_1\left(1-l,l+2;2;\frac{1-\tau}{2}\right)
%%\\
%=&c_2 \int\frac{\Gamma(-s)\Gamma(1-s)\Gamma(1+\ell+s)\Gamma(-\ell+s)}{\Gamma(1+s)\Gamma(s)\Gamma(-\ell-s)\Gamma(\ell+1-s)}\frac{(1-\tau)^s}{2^s}\mathrm{d}s+ c_1\sum_{k=0}^{\infty}\frac{\Gamma(1-\ell+k)\Gamma(\ell+2+k)}{\Gamma(1-\ell)\Gamma(\ell+2)\Gamma(2+k)}\frac{(1-\tau)^{k+1}}{2^kk!}
%\,,
%\\
%\phi^{-}_{0\ell m}(\tau)
%%=&c_2 G_{2,2}^{2,0}\left(\frac{1-\tau}{2}|
%%\begin{array}{c}
%% -l,l+1 \\
%% -1,0 \\
%%\end{array}
%%\right)+c_1 \, _2F_1\left(-l,l+1;2;\frac{1-\tau}{2}\right)
%%\\
%=&c_2 G_{2,2}^{2,0}\left(\frac{1-\tau}{2}|
%\begin{array}{c}
% -l,l+1 \\
% -1,0 \\
%\end{array}
%\right)+c_1 \sum_{k=0}^{\infty}\frac{\Gamma(-\ell+k)\Gamma(\ell+1+k)}{\Gamma(-\ell)\Gamma(\ell+2)\Gamma(2+k)}\frac{(1-\tau)^{k}}{2^kk!}
%\end{align*}

\subsection{Comparison: Approaching $I^-$ from $\scri^-$ and $I$}
\label{sec_comparison_Iscri}

We have analyzed the appearance of logarithmic terms when approaching the critical set $I^-$ from both $\scri^-$ and $I$.
Let us assume that data have been constructed such that all relevant fields and all of their transverse derivatives remain smooth at $I^-$  
when taking their limit from $\scri^-$.
Then the question arises whether this already implies the existence of a smooth cylinder $I$ and a smooth critical set $I^-$.
We do not attempt here to solve the evolution problem. On the level of constraint/transport equations, though,  this leads to the question
whether the no-logs conditions on $I^-$ viewed from $I$ do  impose additional restrictions on the data if $I^-$ is known to be  smooth when approached  from $\scri^-$.

When analyzing the appearance of log terms at $I^-$ from $\scri^-$ and from the cylinder $I$ it was convenient to use different subsystems of the Bianchi equations
whence we have obtained ``more'' no-logs conditions when approaching $I^-$ from  $I$. However, we could have derived an analog to the singular wave equation we just discussed
on $\scri^-$  as well. And, we will see this more explicitly in the case of the spin-2 equation discussed in Section~\ref{sec_spin2}, due to the constraint propagation,
one should expect  the no-logs conditions  for the $V^-_{AB}$ and the $W^+_A$-equation to be  equivalent,
though this does not follow from  the considerations  we made here.

We have shown that the no-logs condition arising from the $V^-_{AB}$-equation evaluated on $\scri^-$ and on $I$ adopt a very similar form, \eq{derivs_V-} and \eq{no-logs-cond-r_n_V-}.
Both equations are obtained by differentiating \eq{evolutionW6b} by $r$ and $\tau$, and in both cases the appearance of a log term becomes evident by the divergence of $\partial_r^{n+2}\partial_{\tau}^nV^-_{AB}$ at $I^-$ for some $n$, which one may also regard as an indication that the no-logs condition on $\scri^-$ implies
that on $I$.
%Assuming that lower order derivatives are continuous at $I^-$ the source terms, evaluated on $I^-$ coincide so that  \eq{no-logs-cond-r_n_V-} induces the ``right'' value at $I^-$, compatible with the corresponding one computed from $\scri^-$.
% but we have seen that all the no-logs conditions adopt  the same form when viewed as  Laplace-like equations on the expansions coefficients of the radiation field.
In Section~\ref{section8}  below we  will show that for $M=\mathrm{const.}$ and $N=0$ the  radiation field necessarily needs to vanish asymptotically at $I^-$ at any order to have a
regular $I^-$ viewed from $\scri^-$, and we will see in Section~\ref{sec_suff_cond_I}  that in that case the expansions coming from $I$ do not produce log-terms either.

To get a satisfactory answer to this question one needs to solve the evolution problem through the critical set $I^-$. 
However, by the above  considerations one might be led to the expectation that the no-logs condition  \eq{derivs_V-} \emph{characterizes} data which generate a spacetime with a smooth cylinder $I$ and a smooth critical set $I^-$.

\newpage

\section{Gauge independence of the no-logs conditions}
\label{sec_gauge_ind}
In all the previous considerations we have assumed that the gauge functions at $\scri^-$  (or rather their expansions at $I^-$) have been
given, and tried to construct data which do not produce logarithmic terms by deriving  ``no-logs conditions''.

A priori one might expect that the gauge functions such as $f_{\alpha}$, $\nu_{\mathring A}$, $\nu_{\tau}$ etc.\ appear in the right-hand sides of
\eq{no_log_rad1}-\eq{no_log_rad2}.
In that case  the appearance  of logarithmic terms would depend   on the gauge. Conversely, one should get rid of many logarithmic
terms by an appropriately adjusted gauge.
However, if one  computes in which way  the expansion coefficients  $f^{(n+2)}_{\alpha}$, $\nu^{(n+2)}_{\mathring A}$, $\nu^{(n+2)}_{\tau}$ etc.\
enter the no-logs conditions \eq{no_log_rad1}-\eq{no_log_rad2} for $\partial^{n}V^-_{AB}|_{\scri^-}$,  one finds that they cancel out.%
\footnote{We have seen this explicitly for the lower orders in Section~\ref{sec_solution_constr_gen}, cf.\ also Section~\ref{sec_radial_equations}.
As the computations are not very illuminating we forgo to present the  general case.
}
Unfortunately, on the level of formal expansions, there seems to be no chance to get insights in which way  and whether at all they enter the
no-logs conditions for $\partial^{k}_{\tau}V^-_{AB}|_{\scri^-}$ for $k \geq n+1$.
To analyze the gauge-dependence of the appearance of logarithmic terms at the critical sets of spatial infinity we will therefore 
consider the behavior of smooth solutions to the GCFE under coordinate transformation which are associated with 
changes in the gauge data at $\scri^-$.

Let us assume we have been given a smooth  solution $(e^{\mu}{}_i,  \widehat \Gamma_i{}^j{}_k, \widehat L_{ij}, W_{ijkl})$
to the GCFE in a weakly asymptotically Minkowski-like conformal Gauss gauge (in fact some steps rely on \eq{main_gauge0B} -\eq{main_gauge0B2})  with gauge functions
\begin{equation}
\nu_{\tau}\,, \quad \nu_{\mathring A}\,, \quad \Theta^{(1)}\,, \quad \kappa\,, \quad \theta^-\,, \quad f_r\,, \quad f_{\mathring A}\,,
\label{old_gauge data}
\end{equation}
  which admits a smooth representation of $\scri^-\cup I^-\cup I$.
Let us  consider another weakly asymptotically Minkowski-like conformal Gauss gauge, given by the gauge functions 
\begin{equation}
\nu^{\mathrm{new}}_{\tau}\,, \quad \nu^{\mathrm{new}}_{\mathring A}\,, \quad \Theta^{(1)}_{\mathrm{new}}\,, \quad \kappa^{\mathrm{new}}\,, \quad \theta^-_{\mathrm{new}}\,, \quad f^{\mathrm{new}}_r\,, \quad f^{\mathrm{new}}_{\mathring A}\,.
\label{new_gauge data}
\end{equation}
To transform into the new gauge, we first apply a combintation  of a conformal and a coordinate transformation of the form (cf.\ \eq{kappa_trafo1})
\begin{align*}
r &\mapsto r_{\mathrm{new}}= r^{(1)}_{\mathrm{new}}(r,x^{\mathring A}) 
\,,
\\
1+\tau &\mapsto 1+ \tau_{\mathrm{new}} = h(r, x^{\mathring A})(1+\tau)
\,,
\\
\Theta &\mapsto \Theta_{\mathrm{new}} = \psi(r, x^{\mathring A}) \Theta
\,.
\end{align*}
The function $r^{(1)}_{\mathrm{new}}$ is
given by \eq{kappa_trafo3},
\begin{equation}
\frac{{\partial^2 r^{(1)}_\mathrm{new}}}{\partial r^2}
=\frac{\partial r^{(1)}_\mathrm{new}}{\partial r}[\kappa(r) - 2\partial_r\log\psi^{(1)}(r)]
 -\Big( \frac{\partial r^{(1)}_\mathrm{new}}{\partial r}\Big)^2\kappa^{\mathrm{new}}( r^{(1)}_{\mathrm{new}}(r))
\,,
\end{equation}
with, cf.\ \eq{kappa_trafo4},
\begin{equation}
\psi( r) =h( r)\frac{\Theta^{(1)}(r^{(1)}_{\mathrm{new}}( r))}{ \Theta^{(1)}( r)}
\,, \quad 
h( r)=\frac{\partial r}{\partial r^{(1)}_{\mathrm{new}}} \frac{\nu^{\tau}_{\mathrm{new}}(r^{(1)}_{\mathrm{new}}( r))}{ \nu^{\tau}( r)}
\,.
\end{equation}
In a weakly asymptotically Minkowski-like conformal Gauss gauge this equation is of the form (cf.\ \eq{some_trafo1})
\begin{equation}
\frac{{\partial^2 r^{(1)}_\mathrm{new}}}{\partial r^2}
=
 \Big( \frac{\partial r^{(1)}_\mathrm{new}}{\partial r}\Big)^2\Big( \frac{2}{r^{(1)}_\mathrm{new}(r)} + \mathfrak{O}(r^{(1)}_\mathrm{new}(r))\Big)
-\frac{\partial r^{(1)}_\mathrm{new}}{\partial r}\Big(\frac{2}{r}+\mathfrak{O}(r)\Big)
\,.
\end{equation}
With  $u:=\partial_{ r}\log( r/ r^{(1)}_{\mathrm{new}})$ and  $v:=r^{(1)}_{\mathrm{new}}/ r$ this  singular ODE becomes a regular first-order system
 (cf.\ \eq{some_trafo2}-\eq{some_trafo3}),
\begin{align}
 \partial_{ r} u
=&-u^2- v^2(1-u r)^2\mathfrak{O}((v r)^0) +(1-u r) \mathfrak{O}(r^0)
\,,
\\
\partial_{ r}v =&-uv
\,.
\end{align}
In fact to get this regular system it is crucial that $\Theta^{(1)}$, $\kappa$ and $\nu_{\tau}$ have an asymptotic behavior at $I^-$ as
required by the weakly asymptotically Minkowski-like gauge condition.
The solution is of the form (with $p(x^{\mathring A})>0$)
$$
r^{(1)}_{\mathrm{new}}(r, x^{\mathring A})=p(x^{\mathring A}) r + \mathfrak{O}(r^2)
\,,
\quad
r(r^{(1)}_{\mathrm{new}}, x^{\mathring A})= \frac{1}{p(x^{\mathring A})}r^{(1)}_{\mathrm{new}} + \mathfrak{O}(r^{(1)}_{\mathrm{new}})^2
\,.
$$
The function $p$ is determined as described in Section~\ref{sec_yet_another} (as the second datum  $q$ which we do not need to consider here explicitly).
In particular $\psi(r(r^{(1)}_{\mathrm{new}}), x^{\mathring A}) $ and $h(r(r^{(1)}_{\mathrm{new}}), x^{\mathring A}) $  will be smooth,
\begin{align}
\psi(r(r^{(1)}_{\mathrm{new}}), x^{\mathring A})  =& p(x^{\mathring A})+\mathfrak{O}(r^{(1)}_{\mathrm{new}})
\,,
\\
h(r(r^{(1)}_{\mathrm{new}}), x^{\mathring A})  =& 1+\mathfrak{O}(r^{(1)}_{\mathrm{new}})
\,.
\end{align}

We then consider the coordinate transformation \eq{coord_trafo2}
(for this let us denote the just obtained $r_{\mathrm{new}}$, $\tau_{\mathrm{new}}$ and $\Theta_{\mathrm{new}}$ by  $r$, $\tau$ and $\Theta$),
\begin{align*}
r &\mapsto r_{\mathrm{new}}:=r +  r^{(2)}_{\mathrm{new}}(r,x^{\mathring A}) (1+\tau)
\,,
\\
x^{\mathring A}&\mapsto x^{\mathring A}_{\mathrm{new}}:=x^{\mathring A} +  h^{\mathring A}(r, x^{\mathring B})(1+\tau)
\,.
\end{align*}
The functions $ r^{(2)}_{\mathrm{new}}$ and $ h^{\mathring A}$
are chosen in such a way that      $\nu^{\mathrm{new}}_{\mathring A}$ is  realized
and  $g_{\tau\tau}|_{\scri^-}=-1$,
\begin{align*}
h^{\mathring A} =& g^{\mathring A \mathring B}(\nu_{\mathring B}-\nu^{\mathrm{new}}_{\mathring B})= \mathfrak{O}(r)
\,, \\ 
 r^{(2)}_{\mathrm{new}}=& -\nu^{\tau}_{\mathrm{new}}(h^{\mathring A}\nu_{\mathring A}^{\mathrm{new}}
- \frac{1}{2}h^{\mathring A}h^{\mathring B}g_{\mathring A\mathring B})= \mathfrak{O}(r^3)
\,.
\end{align*}
A conformal  transformation \eq{coord_trafo3}  yields the desired value $\theta^-_{\mathrm{new}}$, (cf. \eq{behave_-expansion}),
\begin{equation}
\Theta\mapsto [1+ \phi(r,x^{ \mathring A})(1+\tau)]\Theta
\,,
\end{equation}
where
\begin{equation}
 \phi(r,x^{ \mathring A})= \frac{1}{4} \nu_{\tau}( \theta^-- \theta^-_{\mathrm{new}})=\mathfrak{O}(r^2)
\,,
\end{equation}
because of \eq{main_gauge0B2}.
We deduce that the combination of conformal and coordinate transformations which realize \eq{new_gauge data} is smooth at $I^-$.

Before we proceed, let us direct attention to some consequences of our smoothness assumption on $(e^{\mu}{}_i,  \widehat \Gamma_i{}^j{}_k, \widehat L_{ij}, W_{ijkl})$. 
It implies that also   the fields $(e^{\mu}{}_i,   \Gamma_i{}^j{}_k,  L_{ij}, W_{ijkl})$ are smooth.
In the next step  the conformal geodesic equations need to be  solved with initial data $(\dot x^{\mu})|_{\scri^-}=(1,0,0,0)$ and $(f_{\mu})|_{\scri^-}=(0,f^{\mathrm{new}}_r, f^{\mathrm{new}}_{\mathring A})$.
We analyze the  conformal geodesic equations  in a frame, as the frame components are regular at $I^-$.
The initial data then read
\begin{align*}
 (\dot x^{i})|_{\scri^-}=&(1,0,0,0)
\,,
\\
(f_{i})|_{\scri^-}=&(0,\nu^{\tau}f^{\mathrm{new}}_r,\mathring e^{\mathring A}{}_A( f^{\mathrm{new}}_{\mathring A}-\nu^{\tau}\nu_{\mathring A}f^{\mathrm{new}}_r)=\mathfrak{O}(r^0)
\,.
\end{align*}
In frame components the  conformal geodesic equations read
\begin{align*}
\dot x^{j}e^{\mu}{}_j\partial_{\mu} \dot x^{i} +\Gamma_j{}^i{}_k \dot x^{j} \dot x^{k}   &= -2  \dot x^{j} f_{j} \dot x^{i}+ 
\dot x_j \dot x^{j} f^i
\,,
\\
\dot x^{j}e^{\mu}{}_j\partial_{\mu}  f_{i} -\Gamma_j{}^k{}_i  x^{j}f_{k}  &=  \dot x^{j} f_{j} f _{i}
-\frac{1}{2}f_{j}f^j\dot x_i +  \dot x^{j} L_{ij}
\,.
\end{align*}
This is a regular symmetric hyperbolic system
which gives a smooth solution $(\dot x^i, f_ i)$ in some neighborhood  of the initial surface, including some neighborhood
of $I^-$.
In particular $\dot x^{\mu}= e^{\mu}{}_i x^i$ is smooth.
Next, we apply a coordinate transformation $x^{\mu}\mapsto \widehat x^{\mu}$ which transforms  $\dot x^{\mu}$ to $\partial_{\widehat \tau}$.
\begin{equation}
1 = \frac{\partial\widehat x^{\tau}}{\partial x^{\mu}}\dot x^{\mu}
\,,
\quad
0 = \frac{\partial\widehat x^{r}}{\partial x^{\mu}}\dot x^{\mu}
\,,
\quad
0 = \frac{\partial\widehat x^{\mathring A}}{\partial x^{\mu}}\dot x^{\mu}
\,.
\label{global_coord_trafo}
\end{equation}
As initial data we take $\widehat x^{\mu}|_{\scri^-}=x^{\mu}$.
Note that, near $\scri^-$ the transformation is of the form
$x^{\mu}\mapsto x^{\mu} + \mathfrak{O}(1+\tau)^2$, 
so that the initial gauge conditions realized above are preserved.

It is instructive to evaluate the conformal geodesics equations on $I$.
For this note that \eq{frameI_1}-\eq{Schouten_I}  hold, and 
one checks that $(\dot x^i)|_I = (1,0,0,0)$ and $(f_ i)|_I= (0,1,0,0)$.
The coordinate transformation \eq{global_coord_trafo} therefore reduces to the identity on $I$,
whence the leading-order behavior of all fields is unaffected at $I$.
Of course this is to be expected as the fields  acquire their ``weakly asymptotically Minkowski-like conformal Gauss gauge'' values there.

The final gauge is obtained by another conformal transformation $g\mapsto \Psi^2 g$, $\Theta\mapsto \Psi\Theta$, which is determined by the equation
\begin{equation}
\nabla_{\dot x}\Psi = \Psi\langle \dot x, f\rangle
\quad \Longleftrightarrow \quad \partial_{\widehat \tau}\Psi = \Psi f_{\tau}
\,,
\end{equation}
with initial data $\Psi|_{\scri^-}=1$. Since $f_{\tau}|_{\scri^-}=0$ we  have, near $\scri^-$, $\Psi = 1+ \mathfrak{O}(1+\tau)^2$.
Near $I$ we have $\Psi= 1+\mathfrak{O(r)}$.
Note that $g(\dot x, \dot x)|_{\scri^-}=-1$, so by \eq{norm_conf_geod} $\dot x$ is globally normalized to $-1$.

The conformal Gauss coordinates underlying the  weakly asymptotically Minkowski-like conformal Gauss gauge as determined by the gauge data \eq{new_gauge data}
is obtained from the original one by a combination of a conformal transformation and a coordinate transformation both of which are smooth
near $I^-$. The transformed fields  $(e^{\mu}{}_i,  \widehat \Gamma_i{}^j{}_k, \widehat L_{ij}, W_{ijkl})$ which appear in the GCFE and
which are determined  by $g$ as well as $\Theta$ and $f$ are therefore smooth as well.
This is as one should expect  since, by choice of the gauge data \eq{old_gauge data}, the congruence of conformal geodesics on which this gauge is based does not have conjugate point near $\scri^-\cup I^-$.

\begin{lemma}
\label{lem_inv}
Consider a solution  $(e^{\mu}{}_i,  \widehat \Gamma_i{}^j{}_k, \widehat L_{ij}, W_{ijkl})$ of  the GCFE in a
weakly asymptotically Minkowski-like conformal Gauss gauge which is smooth at $\scri^-$, $I$ and $I^-$.
Then the validity of all the  no-logs conditions obtained in Section~\ref{section3} \& \ref{section4}
 are preserved  under gauge transformations which transform into any other 
weakly asymptotically Minkowski-like conformal Gauss gauge.
\end{lemma}

\section{Asympt.\ Minkowski-like conformal Gauss gauge}
\label{section_aMlcGg}
\label{section6}

\subsection{Solution of the constraint equations}
\label{sec_exp_asymp_Mink}

In Section~\ref{sec_solution_constr_gen} we have studied  the constraint equations  listed in Appendix~\ref{app_charact_constraints}
in a weakly asymptotically Minkowski-like gauge. 
For a further analysis concerning the appearance of logarithmic terms at $I^-$ it is convenient, and by Lemma~\ref{lem_inv} without restriction, to  assume an asymptotically
Minkowski-like conformal Gauss gauge  at each  order (cf.\ Definition~\ref{dfn_asympt_Minik_gauge}), where the asymptotic expansions of the gauge functions
are completely fixed. Then the computations are much simpler.

As  ``physical'' initial data on $\scri^-$ we regard $\Xi_{\mathring A\mathring B}$ rather than the radiation field $W_{r\mathring Ar\mathring B}$.  
By  \eq{constraint_Weyl1} below they are -- apart from the integration functions $\Xi^{(1)}_{\mathring A\mathring B}$ and  $\Xi^{(2)}_{\mathring A\mathring B}$ -- in  one-to-one correspondence, and also their expansion coefficients
% $W_{r\mathring Ar\mathring B}^{(n)}$ and $\Xi^{(n+4)}_{\mathring A\mathring B}$ 
are, cf.\ \eq{Weyl_coord_scri1}.

From the constraint equations derived in Appendix~\ref{app_charact_constraints} we obtain
\begin{align}
 g_{\mathring A\mathring B}|_{\scri^-}=& s_{\mathring A\mathring B}+\mathfrak{O}(r^{\infty})\,, \quad \theta^+=\mathfrak{O}(r^{\infty})\,, 
\quad \xi_{\mathring A} =\mathfrak{O}(r^{\infty})\,,
\label{gauge_scri1}
\\
 L_{rr}|_{\scri^-}=& \mathfrak{O}(r^{\infty})
\,,\quad 
 L_{r\mathring A}|_{\scri^-}= \mathfrak{O}(r^{\infty})
\,,
\\ 
( L_{\mathring A\mathring B})_{\mathrm{tf}}|_{\scri^-} =& -\frac{1}{2}\Big(\partial_{r} -\frac{2}{r}\Big)\Xi_{\mathring A\mathring B}+ \mathfrak{O}(r^{\infty})
\,, \quad
s^{\mathring A\mathring B}  L_{\mathring A\mathring B}|_{\scri^-}=  1+ \mathfrak{O}(r^{\infty})
\,,
\\
 L_{r}{}^{r}|_{\scri^-}=&-\frac{1}{2} + \mathfrak{O}(r^{\infty})
\,,
\quad
 L_{\mathring A}{}^{r} |_{\scri^-} = \frac{1}{2}v_{\mathring A}+ \mathfrak{O}(r^{\infty})
\,,
\label{gauge_scri6}
\\
 W_{r\mathring A r\mathring B} |_{\scri^-} 
 =&-\frac{1}{4r^2}
\partial_r\Big(\partial_{r}-\frac{2}{r}\Big)\Xi_{\mathring A\mathring B}+ \mathfrak{O}(r^{\infty})
\,,
\label{constraint_Weyl1}
\\
  W_{r\mathring A r}{}^r|_{\scri^-} 
=& \frac{1}{4r^2}\Big(\partial_{r}-\frac{2}{r}\Big)v_{\mathring A}+ \mathfrak{O}(r^{\infty})
\;,
\label{expression_d1A10}
\\
  W_{\mathring A\mathring B r}{}^r|_{\scri^-}  =& 
\frac{1}{2r^2}\Big(\mcD_{[\mathring A}v_{\mathring B]}
  +\frac{1}{2}\Xi_{[\mathring A}{}^{\mathring C}\partial_{r}\Xi_{\mathring B]\mathring C} \Big)+ \mathfrak{O}(r^{\infty})
\;,
\\
(\partial_r+ \mathfrak{O}(r^{\infty})) W_r{}^r{}_r{}^r |_{\scri^-} 
=&
- \frac{1}{4r^2}\Big(\partial_{r}-\frac{2}{r}\Big)\mcD_{\mathring A}v^{\mathring A}
+\frac{1}{2}\Xi^{\mathring A\mathring B}  W_{r\mathring A r\mathring B}+ \mathfrak{O}(r^{\infty})
\,,
\label{adm_ode}
\\
 \Big(\partial_{r}   -\frac{2}{r}+ \mathfrak{O}(r^{\infty})\Big)  W_{\mathring A}{}^r{}_{r}{}^{r}|_{\scri^-} 
=& 
  -\frac{1}{8} \Big(\partial_r-\frac{2}{r}\Big)\Big(\partial_{r}-\frac{2}{r}\Big)v_{\mathring A}
-\frac{1}{2}  \mathring \nabla^{\mathring B} W_{\mathring A\mathring B r}{}^{r}
\nonumber
\\
&
+\frac{1}{2} \mcD_{\mathring A}  W_{r}{}^r{}_{r}{}^r
  -\Xi_{\mathring A}{}^{\mathring B} W_{r\mathring B r}{}^r  + \mathfrak{O}(r^{\infty})
\,,
\end{align}
and
\begin{align}
\Big(\partial_{r} -\frac{4}{r}+\mathfrak{O}(r^{\infty}) \Big)\Big ( (  W_{\mathring A}{}^r{}_{\mathring B}{}^{r} )_{\mathrm{tf}}- \frac{r^4}{4}  W_{r\mathring A r \mathring B} \Big)\Big|_{\scri^-} 
  =& 
 \Big(\mcD_{(\mathring A} W_{\mathring B)}{}^{r}{}_{r}{}^r -\frac{r^2}{2}\mcD_{(\mathring A} W_{\mathring B)rr}{}^r
\Big)_{\mathrm{tf}}  
\nonumber
\\
&
+\frac{3}{4}\Xi_{\mathring A\mathring B} W_{r}{}^{r}{}_{r}{}^r
-\frac{3}{4}\Xi_{(\mathring A}{}^{\mathring C} W_{\mathring B)\mathring C r}{}^r+ \mathfrak{O}(r^{\infty})
\,.
\label{hardest_constraint_Weyl}
\end{align}
Let us  compute the relevant frame components. Note that in an asymptotically Minkowski-like conformal Gauss gauge at each order, on $\scri^-$, 
\begin{equation}
e_{0}|_{\scri^-} = \partial_{\tau}
\,,
\label{gauge_frame1}
\quad
e_{1}|_{\scri^-}  = \partial_{\tau} + r\partial_r+ \mathfrak{O}(r^{\infty})
\,,
\quad
e_{A}|_{\scri^-}  = \mathring e^{\mathring A}{}_A \partial_{\mathring A}+ \mathfrak{O}(r^{\infty})
\,.
\end{equation}
Using the formulas derived in Section~\ref{sec_connection_gen}
we find for the connection coefficients 
\begin{align}
\widehat\Gamma_1{}^j{}_i|_{\scri^-} 
=&\delta^j{}_i+ \mathfrak{O}(r^{\infty})
\label{connection_scri1}
\,,
\quad
\widehat\Gamma_A{}^1{}_0 |_{\scri^-} = \mathfrak{O}(r^{\infty})
\,,
\quad
\widehat\Gamma_A{}^0{}_0 |_{\scri^-} = \mathfrak{O}(r^{\infty})
\,,
\\
\widehat\Gamma_A{}^B{}_0 |_{\scri^-} =&
 \frac{1}{2r}\Xi_{ A}{}^{ B} + \mathfrak{O}(r^{\infty})
\label{connection_scri7}
\,,
\\
\widehat\Gamma_A{}^B{}_1  |_{\scri^-}=& 
 \frac{1}{2r}  \Xi_{ A}{}^{ B}+\delta^{ B}{}_{ A}+ \mathfrak{O}(r^{\infty})
\,,
\\
\widehat\Gamma_A{}^C{}_B |_{\scri^-} 
=&  
\mathring\Gamma_A{}^C{}_B+ \mathfrak{O}(r^{\infty})
\,.
\label{connection_scri9}
\end{align}
For the components of the Schouten tensor the results of Section~\ref{sec_schouten_gen} yield
\begin{align}
\widehat L_{ 1 i}|_{\scri^-} 
=&
\mathfrak{O}(r^{\infty})
\,,
\label{schouten_scri3}
\\
\widehat L_{ A 1}|_{\scri^-}
=&
\frac{1}{2r} v_A+ \mathfrak{O}(r^{\infty})
\,,
\\
\widehat L_{ AB}|_{\scri^-}
 =&  -\frac{1}{2} 
\Big(\partial_{r} -\frac{1}{r}\Big)\Xi_{ A B} + \mathfrak{O}(r^{\infty})
\,,
\label{schouten_scri5}
\\
\widehat L_{ A0}|_{\scri^-}=&
\frac{1}{2r}  v_A+ \mathfrak{O}(r^{\infty})
\,.
\label{schouten_scri6}
\end{align}
%
%These equations determine  (most of) the initial data for the transport equations \eq{ev_eqn_gauge_1}-\eq{ev_eqn_gauge_7} on the cylinder at spacelike infinity, including their radial derivatives.

For the solutions of  the constraint equations for the rescaled Weyl tensor to be  smooth at $I^-$  $\Xi_{AB}$ necessarily needs to admit  an expansion of the form
\eq{apriori_Xi} 
\begin{equation}
\Xi_{ A B}\sim\sum_{m=1}^{\infty} \Xi^{(m)}_{ A B} r^m
\,,
\label{expansion_phys_data}
\end{equation}
where the  $\Xi^{(m)}_{ A B}$'s denote trace-free tensors on the round 2-sphere.
Recall that 
\begin{equation}
v_{ A} \equiv\rnabla_{ B}\Xi_{ A}{}^{ B} %\overset{\eq{gauge_scri1}}{=} \mcD_B\Xi_{ A}{}^{ B}
\,, \quad v^{(m)}_{ A} \equiv  \mcD_{ B}\Xi^{(m)}_{ A}{}^{ B}
\,.
\end{equation}
We assume that  all smoothness conditions in Proposition~\ref{prop_smoothness_constraints} are satisfied, i.e.\
\begin{equation}
\Xi^{(1)}_{AB}=\Xi^{(3)}_{AB}=\Xi^{(4)}_{AB} =0\,.
\label{no-logs_spec_gauge}
\end{equation}
Then the restriction of the rescaled Weyl tensor to $\scri^-$  extends smoothly across $I^-$.
We determine its  expansion coefficients (terms with  vanishing denominator are defined to be  zero),
\begin{align}
W^{(m)}_{r\mathring Ar\mathring B}=& -\frac{(m+3)(m+2)}{4}\Xi^{(m+4)}_{\mathring A \mathring B}
\,,
\label{Weyl_coord_scri1}
\\
W^{(m)}_{r\mathring Ar}{}^{r}=&
\frac{m+1}{4}v^{(m+3)}_{\mathring A}
\,,
\\
W^{(m)}_{\mathring A\mathring B r}{}^{r}=& \frac{1}{2}\mcD_{[\mathring A}v^{(m+2)}_{\mathring B]}
+ \frac{1}{4}\sum_kk\Xi^{(m-k+3)}{}_{[\mathring A}{}^{\mathring C}\Xi^{(k)}_{\mathring B]\mathring C}
\,,
\\
 W^{(m)}_{r}{}^{r}{}_{r}{}^{r} =&  2 M\delta^m{}_0 -\frac{1}{4}\mcD^{\mathring A}v^{(m+2)}_{\mathring A}
- \frac{1}{8n}\sum_k(k+3)(k+2)\Xi^{(m-1-k)\mathring A \mathring B} \Xi^{(k+4)}_{\mathring A \mathring B}
\,,
\\
W^{(m)}_{\mathring A}{}^{r}{}_{r}{}^{r} =&\delta^m{}_2 L_{\mathring A}
-\frac{m-1}{8} v^{(m+1)}_{\mathring A}
-\frac{1}{2(m-2)}\mcD^{\mathring B} W^{(m-1)}_{\mathring A\mathring B r}{}^{r}
+\frac{1}{2(m-2)}  \mcD_{\mathring A}  W^{(m-1)}_{r}{}^r{}_{r}{}^r
\nonumber
\\
&
  -\frac{1}{m-2}\sum_k\Xi^{(m-k-1)}{}_{\mathring A}{}^{\mathring B} W^{(k)}_{r\mathring B r}{}^r 
\,,
\\
 ( W^{(m)}_{\mathring A}{}^{r}{}_{\mathring B}{}^{r})_{\mathrm{tf}}
=&  \delta^m{}_4c^{(2,0)}_{\mathring A\mathring B}+ \frac{1}{4}W^{(m-4)}_{r \mathring A r\mathring B}
+ \frac{1}{m-4}\Big(\mcD_{(\mathring A} W^{(m-1)}_{\mathring B)}{}^{r}{}_{r}{}^r -\frac{1}{2}\mcD_{(\mathring A} W^{(m-3)}_{\mathring B)rr}{}^r
\Big)_{\mathrm{tf}}  
\nonumber
\\
&
+\frac{3}{4(m-4)}\sum_k\Big( \Xi^{(m-k-1)}_{\mathring A\mathring B} W^{(k)}_{r}{}^{r}{}_{r}{}^r
-\Xi^{(m-k-1)}_{(\mathring A}{}^{\mathring C} W^{(k)}_{\mathring B)\mathring C r}{}^r\Big)
\,.
\end{align}
Recall that $M$, $L_{\mathring A}$, and $c^{(2,0)}_{\mathring A \mathring B}$ arise as integration functions.
For the frame components we have
\begin{align}
V^{+}_{AB}|_{\scri^-}  =& \frac{ r^2}{2} \mathring e^{\mathring A}{}_{A} \mathring e^{\mathring B}{}_{B} W_{r\mathring A r\mathring B} 
\,,
\\
V^{-}_{AB}|_{\scri^-} =& -\mathring e^{\mathring A}{}_{A} \mathring e^{\mathring B}{}_{B}\Big(\frac{ r^2}{2}  W_{r\mathring A r\mathring B} 
  - \frac{2}{r^2} (W_{\mathring A}{}^{r}{}_{\mathring B}{}^{r})_{\mathrm{tf}}\Big)
\,,
\\
W^{+}_A |_{\scri^-} =& r \mathring e^{\mathring A}{}_A W_{r\mathring Ar}{}^{r}
+ 2r^{-1}\mathring e^{\mathring A}{}_A W_{\mathring A}{}^{r}{}_{r}{}^{r}
\,,
\\
W^{-}_A |_{\scri^-} =& r \mathring e^{\mathring A}{}_A W_{r\mathring Ar}{}^{r}
\,,
\\
W_{0101}|_{\scri^-} 
=& W_{r}{}^{r}{}_{r}{}^{r}
\,,
\\
W_{01AB}|_{\scri^-} 
 =& -\mathring e^{\mathring A}{}_A \mathring e^{\mathring B}{}_B W_{\mathring A\mathring Br}{}^{r}
\,,
\end{align}
whence  (again, terms with  vanishing denominator are defined to be  zero)
\begin{align}
V^{+(m)}_{AB}  =& -\frac{m(m+1)}{8}\Xi^{(m+2)}_{ A  B}
\,,
\label{constr_spec_gauge1}
\\
W^{-(m)}_A =&  \frac{m}{4}v^{(m+2)}_A
\,,
\\
W^{(m)}_{0101}
=&2M  \delta^m{}_0  -\frac{1}{4}\mcD^{ A}v^{(m+2)}_{ A}
- \frac{1}{8m}\sum_k(k+3)(k+2)\Xi^{(m-1-k) A  B} \Xi^{(k+4)}_{ A  B}
\,,
\\
W^{(m)}_{01AB}
 =& - \frac{1}{2}\mcD_{[ A}v^{(m+2)}_{ B]}
- \frac{1}{4}\sum_kk\Xi^{(m-k+3)}{}_{[ A}{}^C\Xi^{(k)}_{ B] C}
\,,
\\
W^{+(m)}_A=& 2\delta^m{}_1 L_{ A}
+\frac{1}{m-1} \Big( \mcD^{ B} U^{(m)}_{AB}
  -\frac{1}{2}\sum_k k\Xi^{(m-k+1)}{}_{ A}{}^{ B}v^{(k+2)}_B\Big)
\,,
\\
V^{-(m)}_{AB}=& 2 \delta^m{}_2c^{(2,0)}_{ A B}
+ \frac{1}{m-2}\Big( (\mcD_{( A} W^{+(m)}_{ B)})_{\mathrm{tf}}  
+\frac{3}{2}\sum_k\Xi^{(m-k+1)}_{( A}{}^{ C} U^{(k)}_{ B) C}\Big)
\,.
\label{constr_spec_gauge6}
\end{align}
In particular, we find for the leading orders
(recall that $N\equiv -\frac{1}{8}\epsilon^{AB}\mcD_Av^{(2)}_B$)
\begin{align}
W_{0101}|_{\scri^-} 
=& 2M
-\frac{r^3}{4}\mcD^{ A}v^{(5)}_{ A} 
+  \mathfrak{O}(r^4)
\,,
\label{Weyl_frame_scri1}
\\
W^{-}_A|_{\scri^-} 
=& \frac{3}{4}r^3 v^{(5)}_{ A} 
+\mathfrak{O}(r^4)
\,,
\label{Weyl_frame_scri2}
\\
W^{+}_A|_{\scri^-} 
=&-2\mathcal{M}_A
+2L_{ A} r 
+\frac{r^3}{8}\Big(
( \Delta_s-1)v^{(5)}_{ A}
-2\mcD_{ A}\mcD^{ B} v^{(5)}_{ B}
\Big)
+\mathfrak{O}(r^4)
\,,
\label{Weyl_frame_scri3}
\\
W_{01AB}|_{\scri^-} 
 =&2 N\epsilon_{AB}
-\frac{r^3}{2}\mcD_{[ A}v^{(5)}_{ B]} 
+\mathfrak{O}(r^4)
\,,
\\
V^+_{AB} |_{\scri^-} 
=&
-\frac{3}{2}r^3  \Xi^{(5)}_{ A B} 
+\mathfrak{O}(r^4)
\,,
\\
V^-_{AB}|_{\scri^-} 
=&
(\mcD_{(A}\mathcal{M}_{B)})_{\mathrm{tf}}
-\Big(2(\mcD_{(A}L_{B)})_{\mathrm{tf}}+3(M\Xi^{(2)}_{AB}+N\epsilon_{(A}{}^C\Xi^{(2)}_{B)C})
\Big)r
\nonumber
\\
&
+ c^{(2,0)}_{ A B}r^2
+ \frac{r^3}{8}\Big(
(\Delta_s -4)\mcD_{( A} v^{(5)}_{ B)}
-2\mcD_{ A}\mcD_{ B} \mcD^{ C} v^{(5)}_{ C}
\Big)_{\mathrm{tf}}
+  \mathfrak{O}(r^4)
\,.
\label{Weyl_frame_scri6}
\end{align}

\subsection{First-order transverse derivatives on $\scri^-$}
\label{sec_confG_firstoder}

To get a better idea what is going on let us consider the 1st-order derivatives (i.e.\ the $n=1$ case) explicitly, as well.
In particular we want to determine the source terms in \eq{no_log_rad1}-\eq{no_log_rad2}.
Evaluation of \eq{evolution1}-\eq{evolution7} on $\scri^-$ by using all the expressions we have derived in the previous section
 gives the following relations for connection and  frame coefficients,  and Schouten tensor
(observe  that $\Xi_A{}^C\Xi_{BC} =\frac{1}{2} |\Xi|^2\eta_{AB}$),
\begin{align}
\partial_{\tau}(\widehat L_{10}-\widehat L_{11}) |_{\scri^-} 
&=  \mathfrak{O}(r^{\infty})
\,,
\label{frame_1storder1}
\\
\partial_{\tau}(\widehat L_{10}+\widehat L_{11}) |_{\scri^-} 
&= -4r W_{0101}+ \mathfrak{O}(r^{\infty})
\,,
\\
\partial_{\tau}(\widehat L_{A0}-\widehat L_{A1}) |_{\scri^-} 
&= \mathfrak{O}(r^{\infty})
\,,
\\
\partial_{\tau}(\widehat L_{A0} +\widehat L_{A1})|_{\scri^-} 
&= -2r W^-_{A} -2r W^+_{A}-  \frac{1}{2r^2}\Xi_{ A}{}^{ B}v_B  + \mathfrak{O}(r^{\infty})
\,,
\\
\partial_{\tau}\widehat L_{1A} |_{\scri^-} 
&= - 2r W^-_{A} + \mathfrak{O}(r^{\infty})
\,,
\\
\partial_{\tau}\widehat L_{AB} |_{\scri^-} 
&= -2rV^+_{AB} +rU_{AB}
+  \frac{1}{4r}\Xi_{ A}{}^{ C} \Big(\partial_{r} -\frac{1}{r}\Big)\Xi_{ BC} + \mathfrak{O}(r^{\infty})
\,,
\\
\partial_{\tau}\widehat\Gamma_{1}{}^j{}_k |_{\scri^-} 
 &=
\mathfrak{O}(r^{\infty})
\,,
\\
\partial_{\tau}(\widehat\Gamma_{A}{}^0{}_1 -\widehat\Gamma_{A}{}^1{}_1 )|_{\scri^-} 
 &=
\mathfrak{O}(r^{\infty})
\,,
\\
\partial_{\tau}(\widehat\Gamma_{A}{}^0{}_1+\widehat\Gamma_{A}{}^1{}_1 ) |_{\scri^-} 
 &=
\frac{1}{r} v_A+ \mathfrak{O}(r^{\infty})
\,,
\\
\partial_{\tau}(\widehat\Gamma_{A}{}^0{}_B-\widehat\Gamma_{A}{}^1{}_B) |_{\scri^-} 
 &=
-   \frac{1}{4r^2}|\Xi|^2\delta_{AB}  -\frac{1}{2} \partial_{r}\Xi_{ A B}   + \mathfrak{O}(r^{\infty})
\,,
\\
\partial_{\tau}(\widehat\Gamma_{A}{}^0{}_B+\widehat\Gamma_{A}{}^1{}_B ) |_{\scri^-} 
 &=
  -\frac{1}{2} \Big(\partial_{r} -\frac{2}{r}\Big)\Xi_{ A B}   + \mathfrak{O}(r^{\infty})
\,,
\\
\partial_{\tau}\widehat\Gamma_{A}{}^B{}_C |_{\scri^-} 
 &=
\frac{1}{2r}  \delta^B{}_{C} v_A-  \frac{1}{2r} \mathring \Gamma_D{}^B{}_C\Xi_{ A}{}^{ D}    + \mathfrak{O}(r^{\infty})
\,,
\\
\partial_{\tau}e^{\mu}{}_1|_{\scri^-} &= -\delta^{\mu}{}_0+ \mathfrak{O}(r^{\infty})
\,,
\label{frame_1storder13}
\\
\partial_{\tau}e^{\mu}{}_A|_{\scri^-} &= - \frac{1}{2r}\Xi_{ A}{}^{ B}  e^{\mu}{}_B+ \mathfrak{O}(r^{\infty})
\,,
\label{frame_1storder14}
\end{align}
where we have used that in an asymptotically Minkowski-like conformal Gauss gauge (cf.\ Section~\ref{sec_b_theta})
\begin{equation}
\Theta = r(1-\tau^2)+ \mathfrak{O}(r^{\infty})
\,, \quad
b_0 = -2r \tau+ \mathfrak{O}(r^{\infty})
\,, \quad
b_1= 2r+ \mathfrak{O}(r^{\infty})
\,, \quad
b_A=\mathfrak{O}(r^{\infty})\,.
\label{global_functions_confG}
\end{equation}
For the relevant transverse derivatives of the frame components of the rescaled Weyl tensor we obtain from \eq{evolutionW1b}-\eq{evolutionW5b} 
(observe  that $(\Xi_{(A}{}^C V^{\pm}_{B)C})_{\mathrm{tf}}=0$)
\begin{align}
\partial_{\tau}W_{0101}  |_{\scri^-}   
 =& 
-\frac{1}{2}\mcD^A(W^{+}_A-W^{-}_A)
-   \frac{1}{2r}\Xi^{ AB}V^+_{AB}+ \mathfrak{O}(r^{\infty})
\,,
\label{trans_Weyl_spec1}
\\
\partial_{\tau}W_{01AB}  |_{\scri^-}       =&
\mcD_{[A}  (W^{+}_{B]}+  W^{-}_{B]})
- \frac{1}{r}\Xi_{ [A}{}^{ C}V^+_{B]C}+ \mathfrak{O}(r^{\infty})
\,,
\\
\partial_{\tau}W^-_A |_{\scri^-}   
 =&  \mcD^BV^+_{AB}
 +\frac{1}{2}\mcD^B U_{BA}
+W^-_{A} + \mathfrak{O}(r^{\infty})
\,,
\\
\partial_{\tau}W^+_A |_{\scri^-}   
 =& -  \mcD^BV^-_{AB}
 -\frac{1}{2}\mcD^B U_{AB}
 - W^+_A
+\frac{1}{r}\Xi_{ A}{}^{ B} W^-_B+ \mathfrak{O}(r^{\infty})
\,,
\\
\partial_{\tau}  V^+_{AB} |_{\scri^-}   
 =&
-\frac{1}{2}( r\partial_{r}-2) V^+_{AB}
+\frac{1}{2} ( \mcD_{(A} W^-_{B)})_{\mathrm{tf}}
+ \mathfrak{O}(r^{\infty})
\,.
\label{trans_Weyl_spec5}
\end{align}
Finally, evaluation of  the $\tau$-derivative of  \eq{evolutionW6b} on $\scri^-$ yields
\begin{align}
& \hspace{-2em}r^4 \partial_{r}(r^{-3}\partial_{\tau}V^-_{AB})|_{\scri^-}
\nonumber
\\
=&
 \Big(
\mcD_{(A} \partial_{\tau}  W^+_{B)}
-\frac{1}{2r}(\Xi_{(A}{}^C\mcD_{ C}+4v_{(A}) W^+_{B)}
-\frac{3}{4}\partial_r\Xi_{(A}{}^{C} U_{B)C}
+\frac{3}{2r}\Xi_{(A}{}^C \partial_{\tau}U_{B)C}
\Big)_{\mathrm{tf}}+ \mathfrak{O}(r^{\infty})
\nonumber
\\
=&
 \Big(
 -\mcD_{(A}  \mcD^CV^-_{B)C}
 -\frac{1}{2}\mcD_{(A}\mcD^C U_{B)C}
-\frac{3}{4}\partial_r\Xi_{(A}{}^{C} U_{B)C}
 -\mcD_{(A} W^+_{B)}
-\frac{5}{4r}\Xi_{(A}{}^C\mcD_{ |C|}W^+_{B)}
\nonumber
\\
&
+\frac{3}{4r}\Xi_{(A}{}^C\mcD_{B)}  W^{+}_{C}
-\frac{2}{r}v_{(A} W^+_{B)}
+\frac{1}{r}W^-_C\mcD_{(A}\Xi_{ B)}{}^{ C} 
-\frac{3}{4r}\Xi_{(A}{}^C\mcD_{C}   W^{-}_{B)}
+\frac{7}{4r}\Xi_{ (A}{}^{ C} \mcD_{B)}W^-_C
\Big)_{\mathrm{tf}}
\nonumber
\\
&
 + \frac{3}{4r}\Xi_{AB}\mcD^CW^{-}_C
-\frac{3}{4r}\Xi_{AB}\mcD^CW^{+}_C
+ \frac{3}{4r^2}|\Xi|^2  V^+_{AB}
-  \frac{3}{2r^2}\Xi_{AB} \Xi^{CD} V^+_{CD}+ \mathfrak{O}(r^{\infty})
\,.
\label{V-_eqn_scri_exp}
\end{align}
We obtain  a smooth solution $V^-_{AB}$ whenever the term of order $r^3$ on the right-hand side of \eq{V-_eqn_scri_exp} vanishes.
Taking the expansions computed in Section~\ref{sec_exp_asymp_Mink}
%and the no-logs-conditions xx
into account we find that this will be the case whenever
\begin{align}
0
=&
 \Big(
\mcD_{(A} \mcD^CV^{-(3)}_{B)C}
 +\frac{1}{2}\mcD_{(A}\mcD^C U^{(3)}_{B)C}
%+3\Xi^{(4)}_{(A}{}^{C} U^{(0)}_{B)C}
+\frac{3}{2}\Xi^{(2)}_{(A}{}^{C} U^{(2)}_{B)C}
 +\mcD_{(A} W^{+(3)}_{B)}
+2v^{(2)}_{(A} W^{+(2)}_{B)}
\nonumber
\\
&
+2\Xi^{(2)}_{(A}{}^C\mcD_{ |C|}W^{+(2)}_{B)}
%+2\Xi^{(4)}_{(A}{}^C\mathring\nabla_{ |C|}W^{+(0)}_{B)}
%-\frac{3}{2}\Xi^{(4)}_{(A}{}^C(\mathring\nabla_{(B)}  W^{+(0)}_{C)})_{\mathrm{tf}}
%+2v^{(4)}_{(A} W^{+(0)}_{B)}
%\nonumber
%\\
%&
-W^{-(2)}_C\mcD_{(A}\Xi^{(2)}_{ B)}{}^{ C} 
-\frac{5}{2}\Xi^{(2)}_{ (A}{}^{ C} \mcD_{B)}W^{-(2)}_C
\Big)_{\mathrm{tf}}
\nonumber
%\\
%=&
% \Big(
% \mathring\nabla_{(A}  \mathring\nabla^CV^{-(3)}_{B)C}
% +\frac{1}{2}\mathring\nabla_{(A}\mathring\nabla^C U^{(3)}_{B)C}
% +\mathring\nabla_{(A} W^{+(3)}_{B)}
%\Big)_{\mathrm{tf}}
\\
=&
 \Big(
 \frac{1}{16}\mcD_{(A}\Delta_s \Delta_s v^{(5)}_{ B)}
-\frac{1}{8} \mcD_{A}\mcD_{ B}\Delta_s\mcD^{ C} v^{(5)}_{ C}
 -\frac{3}{4}\mcD_{A}\mcD_B \mcD^{ C}v^{(5)}_{ C}
 -\frac{5}{16}\mcD_{(A}v^{(5)}_{ B)}
+ \frac{1}{4}\mcD_{(A}\Delta_s v^{(5)}_{ B)}
\Big)_{\mathrm{tf}}
\,.
\label{no-logs-1st-order}
\end{align}
We compute the divergence,
\begin{equation*}
\Big(\Delta_s\Delta_s  \Delta_s 
+ 5\Delta_s  \Delta_s 
 -\Delta_s -5\Big)v^{(5)}_{ A}
-2\Big(\mcD_{A}\Delta_s \Delta_s
+8 \mcD_{A}\Delta_s
 +12\mcD_{A}\Big)\mcD^{ B}v^{(5)}_{ B}
=0
\,.
\end{equation*}
Taking divergence and curl  yields equations for divergence and curl of $v^{(5)}_A$.
\begin{align}
(\Delta_s+6)(\Delta_s+2)\Delta_s\mcD^A v^{(5)}_{ A} =& 0
\,,
\label{no-logs-1st-order1}
\\
(\Delta_s+6)(\Delta_s+2)\Delta_s(\epsilon^{AB}\mcD_A v^{(5)}_{ B}) =& 0
\,.
\label{no-logs-1st-order2}
\end{align}
Comparison with \eq{no_log_rad1}-\eq{no_log_rad2} shows that the source terms vanish  for $n=1$).
For the no-logs condition to be fulfilled  $v^{(5)}_A$ needs to admit a Hodge decomposition $v^{(5)}_A=\mcD_A\ul v^{(5)} + \epsilon_A{}^B\mcD_B\ol v^{(5)}$,
where $\ul v^{(5)}$ and $\ol v^{(5)}$ are linear combinations of $\ell=0,1,2$-spherical harmonics.
However, since $v_A$ arises as a divergence of a trace-free, symmetric tensor,  $\ell=0,1$-spherical harmonics  cannot arise.
%, while $\ell=0$-spherical harmonics are irrelevant.
We observe that 
\eq{no-logs-1st-order1}-\eq{no-logs-1st-order2} is  equivalent to \eq{no-logs-1st-order}. Altogether, the no-logs condition
 \eq{no-logs-1st-order} holds if and only if 
\begin{equation}
\Xi^{(5)}_{ AB} = (\mcD_A\mcD_B\ul \Xi^{(5)})_{\mathrm{tf}} + \epsilon_{(A}{}^C\mcD_{B)}\mcD_C\ol \Xi^{(5)}
\quad \text{for some   $\ell=2$ spherical harmonics $\ul v^{(5)}$ and $\ol v^{(5)}$.}
\label{restrictions_on_Xi5}
\end{equation}

\subsection{Second-order transverse derivatives on $\scri^-$}
\label{sec_confG_secondoder}

In anticipation of later computations, let us also compute the second-order transverse derivatives of frame and connection coefficients, which we obtain by differentiating  \eq{evolution1}-\eq{evolution7} by $\tau$,
\begin{align}
\partial^2_{\tau}e^{\mu}{}_1|_{\scri^-} =&\mathfrak{O}(r^{\infty})
\,,
\label{evolution_2nd_trans_1}
\\
\partial^2_{\tau}e^{\tau}{}_A|_{\scri^-}  =& -\frac{1}{r}v_A + \mathfrak{O}(r^{\infty})
\,,
\\
\partial^2_{\tau}e^{r}{}_A|_{\scri^-}  =& 
-\frac{1}{2}v_A + \mathfrak{O}(r^{\infty})
\,,
\\
\partial^2_{\tau}e^{\mathring A}{}_A|_{\scri^-}  =&
\frac{1}{2}\Big(\partial_r-\frac{1}{r}\Big)\Xi_A{}^B e^{\mathring A}{}_B
+\frac{1}{4r^2}|\Xi|^2 e^{\mathring A}{}_A+ \mathfrak{O}(r^{\infty})
\,,
\label{evolution_2nd_trans_4}
%\\
%\partial^2_{\tau}\widehat\Gamma_{1}{}^1{}_0 |_{\scri^-} 
% &=&
%-4r W_{0101}
%\,,
%\\
%\partial^2_{\tau}\widehat\Gamma_{1}{}^A{}_0 |_{\scri^-} 
% &=&
%4r W^A{}_{001}+ 2r W^A{}_{101}
%\,,
%\\
%\partial^2_{\tau}\widehat\Gamma_{1}{}^0{}_0|_{\scri^-} 
% &=&
%-2r W_{0101}
%\,,
\\
\partial^2_{\tau}\widehat\Gamma_{1}{}^0{}_1|_{\scri^-} 
 =&
  -4r W_{0101}+ \mathfrak{O}(r^{\infty})
\,,
\\
\partial^2_{\tau}\widehat\Gamma_{1}{}^1{}_1 |_{\scri^-} 
 =&
-2r W_{0101}  + \mathfrak{O}(r^{\infty})
\,,
\\
\partial^2_{\tau}(\widehat\Gamma_{1}{}^0{}_A-\widehat\Gamma_{1}{}^1{}_A)|_{\scri^-} 
 =&
  -2 (W^+_A + W^-_A) r+ \mathfrak{O}(r^{\infty})
\,,
\\
\partial^2_{\tau}(\widehat\Gamma_{1}{}^0{}_A+\widehat\Gamma_{1}{}^1{}_A) |_{\scri^-} 
 =&
-4W^-_Ar+ \mathfrak{O}(r^{\infty})
\,,
\\
\partial^2_{\tau}(\widehat\Gamma_{1}{}^A{}_B)_{\mathrm{tf}} |_{\scri^-} 
 =&
2r W^A{}_{B01}+ \mathfrak{O}(r^{\infty})
\,,
%\\
%\partial^2_{\tau}\widehat\Gamma_{A}{}^0{}_0 |_{\scri^-} 
% &=&
%- \frac{1}{2r}v_A
% - \frac{1}{2r^2}\Xi_{A}{}^{B}v_B
%-2r W_{010A}
%\,,
%\\
%\partial^2_{\tau}\widehat\Gamma_{A}{}^1{}_0 |_{\scri^-} 
% &=&
%-   \frac{1}{2r^2}\Xi_A{}^Bv_B  -4r W_{010A}
%\,,
%\\
%\partial^2_{\tau}\widehat\Gamma_{A}{}^B{}_0 |_{\scri^-} 
% &=&
%\frac{1}{8r^2}|\Xi|^2\Big(  \frac{1}{r}\Xi_A{}^B-3\delta^B{}_A\Big)
%+  \frac{1}{4r}\Xi_C{}^B\partial_r\Xi_A{}^C
%+  \frac{1}{2r}\Xi_A{}^C\partial_r\Xi_C{}^B 
%\nonumber
%\\
%&& +4r W^B{}_{00A}-2r W^B{}_{10A}
%\,,
%\\
%&=&
% \frac{1}{8}r^3|\Xi^{(2)}|^2\Xi^{(2)}_A{}^B+ \frac{3}{8}r^2|\Xi^{(2)}|^2\delta^B{}_A
%+4r W^B{}_{00A}
\\
\partial^2_{\tau}\widehat\Gamma_{A}{}^0{}_1 |_{\scri^-} 
 =&
-  \frac{1}{2r^2}\Xi_A{}^Bv_B
  -2r (W^+_A+W^-_A)+ \mathfrak{O}(r^{\infty})
\,,
\\
\partial^2_{\tau}\widehat\Gamma_{A}{}^1{}_1|_{\scri^-} 
 =&
-  \frac{1}{2r}v_A
 - \frac{1}{2r^2}\Xi_A{}^Bv_B
-r (W^+_A+W^-_A) + \mathfrak{O}(r^{\infty})
\,,
\\
\partial^2_{\tau}(\widehat\Gamma_{A}{}^0{}_B-\widehat\Gamma_{A}{}^1{}_B|)|_{\scri^-} 
 =&
  \frac{1}{4r^3} |\Xi|^2\Xi_{AB} - \frac{1}{4r^2} |\Xi|^2\eta_{AB}
+ \frac{1}{r}\Xi_{(A}{}^C\partial_{|r|}\Xi_{B)C}
+\frac{1}{2}\Big(\partial_r-\frac{1}{r}\Big)\Xi_{AB}
\nonumber
\\
&
  - 2V^+_{AB}r-2V^-_{AB}r
  +2r W_{0101}\eta_{AB}  + \mathfrak{O}(r^{\infty})
\,,
\\
\partial^2_{\tau}(\widehat\Gamma_{A}{}^0{}_B+\widehat\Gamma_{A}{}^1{}_B)|_{\scri^-} 
 =&
- \frac{1}{2r^2} |\Xi|^2\eta_{AB}
+ \frac{1}{2r}\Xi_A{}^C\partial_r\Xi_{BC}
-\frac{1}{2}\Big(\partial_r-\frac{1}{r}\Big)\Xi_{AB}
\nonumber
\\
&
  - 4V^+_{AB}r
  +2r U_{AB} + \mathfrak{O}(r^{\infty})
\,,
%\\
%\partial^2_{\tau}\widehat\Gamma_{A}{}^B{}_1|_{\scri^-} 
% &=&
%\frac{1}{8r^2}|\Xi|^2\Big(  \frac{1}{r}\Xi_{A}{}^{B}
%+\delta ^B{}_A\Big)
%+ \frac{1}{4r}\Xi_C{}^B\partial_r\Xi_A{}^C
%+\frac{1}{2}\Big(\partial_r-\frac{1}{r}\Big)\Xi_A{}^B
%\nonumber
%\\
%&&
%     +2r W^B{}_{10A}
%\,,
%\\
% &=&
%  \frac{1}{8}r^3|\Xi^{(2)}|^2\Xi^{(2)}_{A}{}^{B}
%+\frac{3}{8}r^2|\Xi^{(2)}|^2\delta ^B{}_A
%+\frac{1}{2}r\Xi^{(2)}_A{}^B
%     +2r W^B{}_{10A}
\\
\partial^2_{\tau}\widehat\Gamma_{A}{}^B{}_C|_{\scri^-} 
 =&
-\frac{1}{2r} \delta^B{}_ Cv_A
+\frac{1}{4r^2}  \mathring\Gamma_A{}^B{}_C|\Xi|^2
+\frac{1}{2} \mathring\Gamma_D{}^B{}_C\Big(\partial_r-\frac{1}{r}\Big)\Xi_A{}^D
 - \frac{1}{2r^2}\delta^B{}_{C}\Xi_A{}^Dv_D
\nonumber
\\
&
   -r (W^+_A+W^-_A) \delta^B{}_{C}  -2r (W^+_{[C}-W^-_{[C})\eta_{B]A}+ \mathfrak{O}(r^{\infty})
\,.
\label{evolution_2nd_trans_19}
\end{align}

\subsection{A sufficient condition for the non-appearance  of logs: Approaching $I^-$ from~$\scri^-$}

%\subsubsection{Approaching $I^-$ from~$\scri^-$}
\label{sec_exp_dec_rad_field}

We have seen in Section~\ref{sec_no-logs_V-scri_gen} that whether the no-logs conditions \eq{no_log_rad1}-\eq{no_log_rad2},
 regarded as equations on the expansion coefficients of the radiation field at $I^-$,
can be fulfilled or not, depends on the harmonic
decomposition of the source terms in \eq{no_log_rad1}-\eq{no_log_rad2}. The source terms in turn are determined from lower-order expansion coefficients
(and  certain integration functions such as  the mass aspect etc.).
Here, we want to construct data for which these lower order terms can be controlled, and in order to be able to do that, besides imposing the
asymptotically Minkowski-like conformal Gauss gauge at each order, we assume
that the radiation field vanishes at each order at $I^-$, equivalently,
\begin{equation}
\Xi_{AB} = \Xi^{(2)}_{AB} r^2 +  \mathfrak{O}(r^{\infty})
\,.
\label{exp_dec_data1}
\end{equation}
The idea is that for these class of data no terms of order $n+2$ arise on the right-hand side of \eq{derivs_V-} and produce logarithmic  terms, i.e.\ the source terms in  \eq{no_log_rad1}-\eq{no_log_rad2} vanish (as we have already shown above for $n=1$).

The asymptotic expansions of frame and connection coefficients as well as the Schouten tensor at $I^-$ (including their 1st-oder transverse derivatives) follow straightforwardly 
from \eq{gauge_frame1}-\eq{schouten_scri6} and
\eq{frame_1storder1}-\eq{frame_1storder14}
by inserting \eq{exp_dec_data1}. For convenience let us give the expansion of the rescaled Weyl tensor  (cf.\ \eq{constr_spec_gauge1}-\eq{constr_spec_gauge6}),
\begin{align}
W_{0101} |_{\scri^-} =& 2M + \mathfrak{O}(r^{\infty})
\,,
\\
W_{01AB} |_{\scri^-} =&2N\epsilon_{AB}+ \mathfrak{O}(r^{\infty})
\,,
\\
W^-_A |_{\scri^-} =&    \mathfrak{O}(r^{\infty})
\,,
\label{W-_exp_scri}
\\
W^+_A |_{\scri^-} =&-2\mathcal{M}_A+ 2L_A r  +\mathfrak{O}(r^{\infty})
\,,
\label{W+_exp_scri}
\\
V^+_{AB} |_{\scri^-} =&  \mathfrak{O}(r^{\infty})
\,,
\\
V^-_{AB} |_{\scri^-} =&(\mcD_{(A}\mathcal{M}_{B)})_{\mathrm{tf}}
-\underbrace{\Big(2(\mcD_{(A}L_{B)})_{\mathrm{tf}}+3(M\Xi^{(2)}_{AB}+N\epsilon_{(A}{}^C\Xi^{(2)}_{B)C})
\Big)}_{=:H_{AB}}r
+ 2c^{(2,0)}_{AB} r^2
+ \mathfrak{O}(r^{\infty})
\,,
\label{exp_values15}
\end{align}
For its transverse derivatives  we obtain from  \eq{trans_Weyl_spec1}-\eq{trans_Weyl_spec5},
\begin{align}
\partial_{\tau}W_{0101}  |_{\scri^-}   
 =& 
\Delta_s M  
  -\mcD^AL_A r
+ \mathfrak{O}(r^{\infty})
\,,
\\
\partial_{\tau}W_{01AB}  |_{\scri^-}       =&
\Delta_s N\epsilon_{AB}  + 2\mcD_{[A} L_{B]} r 
+ \mathfrak{O}(r^{\infty})
\,,
\\
\partial_{\tau}W^-_A |_{\scri^-}   
 =&\ol{\mathcal{M}}_A
+ \mathfrak{O}(r^{\infty})
\,,
\\
\partial_{\tau}W^+_A |_{\scri^-}   
 =& -\frac{1}{2}(\Delta_s-1)\mathcal{M}_{A}
+(\mcD^BH_{AB} - 2L_A )r
- 2 \mcD^Bc^{(2)}_{AB} r^2
+ \mathfrak{O}(r^{\infty})
\,,
\\
\partial_{\tau}  V^+_{AB} |_{\scri^-}   
 =&
 \mathfrak{O}(r^{\infty})
\,.
\end{align}
Finally, integration of \eq{V-_eqn_scri_exp} yields
\begin{align}
\partial_{\tau}V^-_{AB}|_{\scri^-}
=&
V_{AB}^{-(0,1)}
+V_{AB}^{-(1,1)}r
+V_{AB}^{-(2,1)}r^2
+  c^{(3,1)}_{AB} r^3
+ \mathfrak{O}(r^{\infty})
\,,
\end{align}
where the precise form of the coefficients is irrelevant for our purposes.

\newpage

\begin{lemma}
\label{lemma_estimates1}
Assume that \eq{exp_dec_data1}  holds. Then, in an asymptotically Minkowski-like conformal Gauss gauge at each order the following relations hold for all  $k\geq 1$
(cf.\ Lemma~\ref{lemma_estimates2}),
\begin{align}
\partial^{k}_{\tau}\widehat\Gamma_{1}{}^1{}_1|_{\scri^-}
 =&
\mathcal{P}^{k-1}+ \mathfrak{O}(r^{\infty})
\,,
\\
\partial^{k}_{\tau}\widehat\Gamma_{1}{}^0{}_1 |_{\scri^-}
 =&
 \mathcal{P}^{k-1}+  \mathfrak{O}(r^{\infty})
\,,
\\
\partial^{k}_{\tau}\widehat\Gamma_{A}{}^1{}_1 |_{\scri^-}
 =&
 \mathcal{P}^{k}+  \mathfrak{O}(r^{\infty})
\,,
\\
\partial^{k}_{\tau}\widehat\Gamma_{A}{}^0{}_1 |_{\scri^-}
 =&
 \mathcal{P}^{k}+  \mathfrak{O}(r^{\infty})
\,,
\\
\partial^{k}_{\tau}(\widehat\Gamma_{1}{}^0{}_A +\widehat\Gamma_{1}{}^1{}_A )|_{\scri^-}
 =&
 \mathcal{P}^{k-2}+ \mathfrak{O}(r^{\infty})
\,,
\\
\partial^{k}_{\tau}(\widehat\Gamma_{1}{}^0{}_A -\widehat\Gamma_{1}{}^1{}_A )|_{\scri^-}
 =&
 \mathcal{P}^{k}+ \mathfrak{O}(r^{\infty})
\,,
\\
\partial^{k}_{\tau}(\widehat\Gamma_{A}{}^0{}_B +\widehat\Gamma_{A}{}^1{}_B )|_{\scri^-}
 =&
 \mathcal{P}^{k-1}+  \mathfrak{O}(r^{\infty})
\,,
\\
\partial^{k}_{\tau}(\widehat\Gamma_{A}{}^0{}_{B} -\widehat\Gamma_{A}{}^1{}_{B} )|_{\scri^-}
 =&
 \mathcal{P}^{k+1}+  \mathfrak{O}(r^{\infty})
\,,
\\
\partial^{k}_{\tau}\widehat\Gamma_{1}{}^A{}_B |_{\scri^-}
 =&
 \mathcal{P}^{k-1}+  \mathfrak{O}(r^{\infty})
\,,
\\
\partial^{k}_{\tau}\widehat\Gamma_{A}{}^B{}_C |_{\scri^-}
 =&
\mathcal{P}^{k}+  \mathfrak{O}(r^{\infty})
\,,
\\
\partial^{k}_{\tau}\widehat L_{1i}|_{\scri^-} 
=&
\mathcal{P}^{k}+  \mathfrak{O}(r^{\infty})
\,,
\\
\partial^{k}_{\tau}(\widehat L_{A0} +\widehat L_{A1})|_{\scri^-}
=&
\mathcal{P}^{k+1}+ \mathfrak{O}(r^{\infty})
\,,
\\
\partial^{k}_{\tau}(\widehat L_{A0}-\widehat L_{A1}) |_{\scri^-}
=&
\mathcal{P}^{k}+  \mathfrak{O}(r^{\infty})
\,,
\\
\partial^{k}_{\tau}\widehat L_{AB} |_{\scri^-}
=&
 \mathcal{P}^{k+1}+  \mathfrak{O}(r^{\infty})
\,,
\\
\partial^k_{\tau}e^{\tau}{}_1|_{\scri^-}=&-\delta^k{}_1 + 
\mathcal{P}^{ k-2} + O(r^{\infty})
\\
\partial^{k}_{\tau}e^{\alpha}{}_1|_{\scri^-}=&   \mathcal{P}^{k-1}+ \mathfrak{O}(r^{\infty})
\\
\partial^{k}_{\tau}e^{\tau}{}_A |_{\scri^-}=& \mathcal{P}^{k-1}+  \mathfrak{O}(r^{\infty})
\\
\partial^{k}_{\tau}e^{\alpha}{}_A|_{\scri^-}=&   \mathcal{P}^{k}+  \mathfrak{O}(r^{\infty})
\,,
\\
\partial^k_{\tau}U_{AB} |_{\scri^-}=& \mathcal{P}^{ k} + \mathfrak{O}(r^{\infty})
\,,
\label{Weyl_poly1}
\\
\partial^k_{\tau}W^{\pm}_{A} |_{\scri^-}=& \mathcal{P}^{ k\pm1} +  \mathfrak{O}(r^{\infty})
\,,
%\\
%\partial^k_{\tau}W^-_{A}|_{\scri^-} =& \mathcal{P}^{ k-1} +  \mathfrak{O}(r^{\infty})
%\,,
\\
\partial^k_{\tau}V^{\pm}_{AB} |_{\scri^-}=& \mathcal{P}^{ k\mp 2} + \mathfrak{O}(r^{\infty})
\,,
\label{Weyl_poly3}
%\\
%\partial^k_{\tau}V^-_{AB} |_{\scri^-}=& \mathcal{P}^{ k+2} 
%%+ c^{(k+2)}_{AB} r^{k+2}
%+  \mathfrak{O}(r^{\infty})
%\,,
\end{align}
where $ \mathcal{P}^{ k}$ denotes a polynomial in $r$ of degree $\leq k$ (the zero-polynomial if $k$ is negative).
\end{lemma}

\begin{proof}
This is proved by induction. The case $k=1$ follows from the above considerations.
So let us assume that the assertion holds for $1\leq k\leq n-1$ 
We then apply $\partial^{n-1}_{\tau}$ to the evolution equations
\eq{evolution1}-\eq{evolution7}, \eq{evolutionW1b}-\eq{evolutionW5b}.
Taking the induction hypothesis as well as the 0th-order expansions into account,
 we  straightforwardly deduce that the assertion of the lemma holds for $k=n$ for all components, excluding for the time being
$\partial^{n}_{\tau}(\widehat\Gamma_{A}{}^0{}_B +\widehat\Gamma_{A}{}^1{}_B )$,  $\partial^{n}_{\tau}(\widehat\Gamma_{1}{}^0{}_A -\widehat\Gamma_{1}{}^1{}_A )$ and  $\partial^{n}_{\tau}V^-_{AB}$.

Let us consider the behavior of these connection components somewhat more detailed.
From \eq{evolution1}-\eq{evolution7}
we have
\begin{equation*} 
\partial^n_{\tau}(\widehat\Gamma_{A}{}^0{}_B +\widehat\Gamma_{A}{}^1{}_B ) |_{\scri^-}
 =
\partial^{n-1}_{\tau}\widehat\Gamma_{A}{}^0{}_{B}+ \partial^{n-1}_{\tau}\widehat L_{AB} +\mathcal{P}^{ n-1} + \mathfrak{O}(r^{\infty})
\,.
\end{equation*}
We need to show that the terms of order $n$ in $\partial^{n-1}_{\tau}\widehat\Gamma_{A}{}^0{}_{B}$ and $\partial^{n-1}_{\tau}\widehat L_{AB}$ cancel each other.
We have
\begin{align*} 
\partial^n_{\tau}\widehat L_{AB}  |_{\scri^-}
=&  (n-1) r \partial^{n-2}_{\tau}V^-_{AB}
- \partial^{n-1}_{\tau}( \widehat\Gamma_A{}^C{}_0\widehat L_{CB}  )
+\mathcal{P}^{ n} +  \mathfrak{O}(r^{\infty})
\,,
\\
\partial^n_{\tau}\widehat\Gamma_{A}{}^0{}_B  |_{\scri^-}
 =&
-  \partial^{n-1}_{\tau}( \widehat\Gamma_C{}^0{}_B\widehat\Gamma_{A}{}^C{}_{0} ) -(n-1) r  \partial^{n-2}_{\tau}V^-_{AB} +\mathcal{P}^{ n} +  \mathfrak{O}(r^{\infty})
\,,
\end{align*}
whence
\begin{equation*} 
\partial^n_{\tau}(\widehat L_{AB} +\widehat\Gamma_{A}{}^0{}_B ) |_{\scri^-}
=
- \partial^{n-1}_{\tau}[( \widehat\Gamma_A{}^C{}_0(\widehat L_{CB}+ \widehat\Gamma_C{}^0{}_B)  ]
  +\mathcal{P}^{ n} + \mathfrak{O}(r^{\infty})
\,.
\end{equation*}
By induction we conclude (one easily checks by using \eq{frame_1storder1}-\eq{frame_1storder14} that this is satisfied for $n=1$),
\begin{equation*}
\partial^n_{\tau}(\widehat L_{AB} +\widehat\Gamma_{A}{}^0{}_B )=\mathcal{P}^{ n} + \mathfrak{O}(r^{\infty})
\,,
\end{equation*}
whence it follows readily that
\begin{equation*} 
\partial^n_{\tau}(\widehat\Gamma_{A}{}^0{}_B +\widehat\Gamma_{A}{}^1{}_B )
 =
\mathcal{P}^{ n-1} +  \mathfrak{O}(r^{\infty})
\,.
\end{equation*}

Similarly, we have
\begin{align*}
\partial^n_{\tau}(\widehat\Gamma_{1}{}^0{}_A +\widehat\Gamma_{1}{}^1{}_A ) |_{\scri^-}
 =&
\partial^{n-1}_{\tau}(\widehat\Gamma_{1}{}^0{}_{A} + \widehat L_{1A}  )
+\mathcal{P}^{ n-2} + \mathfrak{O}(r^{\infty})
\,,
\\
\partial^n_{\tau}\widehat\Gamma_{1}{}^0{}_A  |_{\scri^-}
 =&
-\partial^{n-1}_{\tau}(  \widehat\Gamma_B{}^0{}_A\widehat\Gamma_{1}{}^B{}_{0}  ) - \frac{r}{2}\partial^{n-1}_{\tau}[(1-\tau^2)   W^+_{A}]
+\mathcal{P}^{ n-1} + \mathfrak{O}(r^{\infty})
\,,
\\ 
\partial^n_{\tau}\widehat L_{1A} |_{\scri^-}
=&
-  \partial^{n-1}_{\tau}(\widehat\Gamma_1{}^B{}_0\widehat L_{BA} )
+ r\partial^{n-1}_{\tau}[(1+ \tau) W^+_{A} ]
+\mathcal{P}^{ n-1} +  \mathfrak{O}(r^{\infty})
\,,
\end{align*}
whence
\begin{align*}
\partial^n_{\tau}(\widehat L_{1A} +\widehat\Gamma_{1}{}^0{}_A ) |_{\scri^-}
 =&
-\partial^{n-1}_{\tau}[ (\widehat L_{BA} +\widehat\Gamma_B{}^0{}_A)\widehat\Gamma_{1}{}^B{}_{0}  ] + \frac{r}{2}\partial^{n-1}_{\tau}[(1+\tau)^2   W^+_{A}]
+\mathcal{P}^{ n-1} + \mathfrak{O}(r^{\infty})
\\
=&\mathcal{P}^{ n-1} + \mathfrak{O}(r^{\infty})
\,,
\end{align*}
and
\begin{equation*}
\partial^n_{\tau}(\widehat\Gamma_{1}{}^0{}_A +\widehat\Gamma_{1}{}^1{}_A ) |_{\scri^-}
=
\mathcal{P}^{ n-2} +  \mathfrak{O}(r^{\infty})
\,.
\end{equation*}

Finally, we consider the equation \eq{evolutionW6b} for $V^-_{AB}$, to which we apply $\partial^{n}_{\tau}$,
\begin{equation*}
r^{n+3}\partial_r(r^{-n-2}\partial^{n}_{\tau} V^-_{AB} )|_{\scri^-}
=
  \mathcal{P}^{ n+1} + \mathfrak{O}(r^{\infty})
\,,
\end{equation*}
and we observe that no logarithmic terms  arise,
\begin{equation*}
\partial^{n}_{\tau}V^-_{AB} |_{\scri^-}=  \mathcal{P}^{ n+1} + c^{(n+2,n)}_{AB}r^{n+2}+\mathfrak{O}(r^{\infty})
\,,
\end{equation*}
and the lemma is proved.
\qed
\end{proof}

Let us return to   Proposition~\ref{prop_n-1_no-logs}.
A \emph{sufficient} condition which makes sure that  \eq{derivs_V-} holds  is  given by the following result. It is a corollary of the previous lemma, which shows
that the source term in \eq{derivs_V-} is of the form $ \mathcal{P}^{ n+1} + \mathfrak{O}(r^{\infty})$:

\begin{corollary}
\label{cor_vanishing_radiation}
Assume that the radiation field vanishes at any order at $I^-$,  \eq{exp_dec_data1}.
Then the restrictions to $\scri^-$ of all the fields appearing  in the GCFE including their
transverse derivatives  of all orders admit smooth extensions through $I^-$  in an asymptotically Minkowski-like conformal gauss gauge at each order (cf.\ Corollary~\ref{cor_vanishing_radiation2}).
\end{corollary}

\begin{remark}
{\rm
If $\partial_r^k\Xi_{AB}|_{I^-}=0$ only for $3\leq k\leq n$, then  the restrictions to $\scri^-$ of all the fields appearing  in the GCFE including their
transverse derivatives up to and including the order $n-4$ admit smooth extensions through $I^-$  in an asymptotically Minkowski-like conformal gauss gauge at each order.
}
\end{remark}

Let us assume that $\Xi^{(k)}_{AB}$ vanishes  for $3\leq k\leq n-1$, $n\geq 5$ (we already know that $\Xi^{(3)}_{AB}$ and $\Xi^{(4)}_{AB}$ necessarily need to vanish).
Then, by the previous considerations, the no-logs conditions for $\partial_{\tau}^{n-4}V^-_{AB}|_{\scri^-}$, \eq{no_log_rad1}-\eq{no_log_rad2}, take the form
\begin{align}
\prod_{\ell =0}^{n-3}\Big(\Delta_s+\ell (\ell+1)\Big)\mcD^Av^{(n)}_A
=&0
\,,
\\
 \prod_{\ell =0}^{n-3}\Big(\Delta_s+\ell (\ell+1)\Big)(\epsilon^{AB}\mcD_{A}v^{(n)}_{B})
=& 0
\,.
\end{align}
Unfortunately, since the operator appearing on the left-hand side has a non-trivial kernel, we cannot  conclude at this stage  that the converse to the corollary is also true, i.e.\
that \eq{exp_dec_data1} is  \emph{necessary} for the non-appearance of logarithmic terms. 
The two functions $\ul v^{(n)}$ and $\ol v^{(n)}$ appearing in the Hodge decomposition of $v^{(n)}_A=\mcD_A\ul v^{(n)} +\epsilon_A{}^B\mcD_B\ol v^{(n)}$
are  allowed to have spherical harmonics with $2\leq \ell\leq n-3$.
That yields a weakened  converse to the above corollary.
\begin{corollary}
Consider smooth  initial data of the form $\Xi_{AB}=\mathfrak{O}(r^2)$.
Then the expansion coefficients admit a Hodge decomposition of the form  $\Xi^{(n)}_{AB}=(\mcD_A\mcD_B\ul \Xi^{(n)})_{\mathrm{tf}}+ \epsilon_{(A}{}^C\mcD_{B)}\mcD_C \ol \Xi^{(n)}$.
Assume that $\ul \Xi^{(n)}$ and $\ol \Xi^{(n)}$ do not contain spherical harmonics with $2\leq \ell\leq n-3$ for $n\geq 3$.
Then, in an asymptotically Minkowski-like conformal Gauss gauge,  the restrictions to $\scri^-$ of all the fields appearing  in the GCFE including their
transverse derivatives  of all orders admit smooth extensions through $I^-$ if and only if  \eq{exp_dec_data1} holds.
\end{corollary}

\begin{remark}
{\rm
In Section~\ref{sec_spin2} we will see that for the massless spin-2 equation a vanishing radiation field is \emph{not} necessary for the non-appearance 
of logarithmic terms.
The results in Section~\ref{sec_const_mass} show that non-linear effects, or rather a non-vanishing mass,
do impose additional restrictions on the radiation field:
Elements in the kernel of  the operator in \eq{no_log_rad1}-\eq{no_log_rad2}   cause a violation of the no-logs condition in higher orders.
% It is therefore  at least conceivable that its vanishing is necessary.
}
\end{remark}

\subsection{Transport equations on $I$}

Here we want to analyze the
appearance of logarithmic terms when approaching $I^-$ from $I$ for all radial derivatives of the relevant fields.
The 0th order has been computed in Section~\ref{sec_appro_I-I}, \eq{frameI_1}-\eq{Schouten_I} and \eq{Weyl_I_1}-\eq{Weyl_I_6}
(we assume that the no-logs conditions \eq{no-logs_spec_gauge}  are satisfied).
The 1st-order radial derivatives have been computed in Section~\ref{sec_radial_equations}. Let us sum up the results in our current 
asymptotically Minkowski-like conformal Gauss gauge at each order

\begin{align}
\partial_r\widehat L_{10} |_I
&=-4M (1+\tau)+ \mathfrak{O}(1+\tau)^2
\,,
\label{radial_I1}
\\
\partial_r\widehat L_{11}  |_I
&= -4 M(1+\tau)+ \mathfrak{O}(1+\tau)^2
\,,
\\
\partial_r\widehat L_{1A}  |_I
&= \mathfrak{O}(1+\tau)^2
\,,
\\
\partial_r(\widehat L_{A0}-\widehat L_{A1}) |_I
&= \mathfrak{O}(1+\tau)^2
\,,
\\
\partial_r(\widehat L_{A0}+\widehat L_{A1}) |_I
&=v_{ A}^{(2)} +4\mathfrak{M}_A(1+\tau) + \mathfrak{O}(1+\tau)^2
\,,
\\
\partial_r\widehat L_{AB}  |_I
&= -\frac{1}{2}\Xi^{(2)}_{ A B}
+ 2\Big(M\eta_{AB} + N\epsilon_{AB}\Big)(1+\tau) +\mathfrak{O}(1+\tau)^2
\,,
\end{align}
and
\begin{align}
\partial_r\widehat\Gamma_{1}{}^i{}_j  |_I
 &=
 \mathfrak{O}(1+\tau)^2
\,,
\\
\partial_r(\widehat\Gamma_{1}{}^0{}_A +\widehat\Gamma_{1}{}^1{}_A )|_I
 &=
\mathfrak{O}(1+\tau)^3
\,,
\\
\partial_r\widehat\Gamma_{A}{}^0{}_1  |_I
 &=
 \frac{1}{2}v_{ A}^{(2)} (1+\tau)  +\mathfrak{O}(1+\tau)^2
\,,
\\
\partial_r\widehat\Gamma_{A}{}^1{}_1  |_I
 &=
\frac{1}{2} v_{ A}^{(2)}(1+\tau)    +\mathfrak{O}(1+\tau)^2
\,,
\\
\partial_r(\widehat\Gamma_{A}{}^0{}_B+\widehat\Gamma_{A}{}^1{}_B  )  |_I
 &=
\mathfrak{O}(1+\tau)^2
\,,
\\
\partial_r(\widehat\Gamma_{A}{}^0{}_B-\widehat\Gamma_{A}{}^1{}_B)  |_I
 &=
\Xi^{(2)}_{ AB}- \Xi^{(2)}_{ A B}(1+\tau)   
   +\mathfrak{O}(1+\tau)^2
\,,
\\
\partial_r\widehat\Gamma_{A}{}^B{}_C  |_I
 &=
\frac{1}{2} \Big( v_{ A}^{(2)} \delta^B{}_{C}
-\mathring\Gamma_D{}^B{}_C \Xi^{(2)}_{ A}{}^{ D}\Big)(1+\tau)
+\mathfrak{O}(1+\tau)^2
\,,
\\
\partial_re^{\tau}{}_1  |_I & = \mathfrak{O}(1+\tau)^3
\,,
\\
\partial_re^{r}{}_1  |_I & =
1
\,,
\\
\partial_re^{\mathring A}{}_1  |_I & = 
\mathfrak{O}(1+\tau)^2
\,,
\\
\partial_re^{\tau}{}_A  |_I & = \mathfrak{O}(1+\tau)^2
\,,
\\
\partial_re^{r}{}_A  |_I & = 
0
\,,
\\
\partial_re^{\mathring A}{}_A  |_I & =
 -\frac{1}{2}\Xi^{(2)}_A{}^B\mathring e^{\mathring A}{}_B(1+\tau) +\mathfrak{O}(1+\tau)^2
\,,
\\
\partial_rW^-_A|_I &= \mathfrak{O}(1+\tau)^2
\,,
\\
\partial_rW^+_A|_I &= 2L_A +\mathfrak{O}(1+\tau)
\,,
\\
\partial_rV^+_{AB} |_I 
 &= \mathfrak{O}(1+\tau)^3
\,,
\\
\partial_rV^-_{AB} |_I
&=
-2(\mcD_{(A}L_{ B)})_{\mathrm{tf}}
- 3 \Xi^{(2)}_{(A}{}^C (M\eta_{B)C} + N\epsilon_{B)C})
+ \mathfrak{O}(1+\tau )
\,,
\\
\partial_{r} U_{AB}|_I
 &=
\mathfrak{O}(1+\tau )
\,.
\label{radial_I23}
\end{align}
Note that, as compared to \eq{V+AB_I},  $\partial_rV^+_{AB}$ has a better decay in an asymptotically Minkowski-like conformal Gauss gauge
at each order.
This  follows straightforwardly from  \eq{eq_trans_derv_V+}.

We consider radial derivatives of higher order.
In our setting the initial data for the transport equations are determined by the data on $\scri^-$, which we assume to be smooth
at any order at $I^-$, cf.\ \eq{gauge_frame1}-\eq{schouten_scri6} and \eq{constr_spec_gauge1}-\eq{constr_spec_gauge6}.
The  missing data for the singular wave equations \eq{sing_wave1}-\eq{sing_wave2}
are provided by $\partial^{p+1}_{\tau}\partial^p_rW^{-}_A|_{I^-} $ and
$\partial^{p-1}_{\tau}\partial^p_rW^{+}_A|_{I^-} $, which are also determined  from $\scri^-$ but  will be  irrelevant for our purposes.

\subsection{A sufficient condition for the non-appearance  of logs: Approaching $I^-$ from $I$}
%\subsubsection{Approaching $I^-$ from $I$}
\label{sec_suff_cond_I}

As in Section~\ref{sec_exp_dec_rad_field} let us assume that the radiation field vanishes at any order at $I^-$, i.e.\
\begin{equation}
\Xi_{AB} = \Xi^{(2)}_{AB} r^2 +  \mathfrak{O}(r^{\infty})
\,.
\label{exp_dec_data1b}
\end{equation}
In that case the following data, relevant  to solve the GCFE on $I$, are induced  on $I^-$ for $p \geq 2$,
\begin{align}
\partial^p_re^{\mu}{}_i|_{I^-} =& 0
\,,
\label{initial_data_transport1}
\quad
\partial^p_r\widehat\Gamma_i{}^j{}_k{}|_{I^-} 
= 0
\,,
\quad
\partial^p_r\widehat L_{ ij}|_{I^-} 
=
0
\,,
\\
\partial^p_r W^{\pm}_A |_{I^-} =& 0
\,,
\quad 
\partial^p_rU_{AB}|_{I^-} 
=
0
\,,
\quad
\partial^p_r V^{+}_{AB}|_{I^-}  =0
\,,
\quad
\partial^{p-2}_{\tau}\partial^p_rV^-_{AB} |_{I^-} =\frac{1}{p!}c^{(p,p-2)}_{AB}
\,.
\label{initial_data_transport7}
\end{align}

In Section~\ref{sec_exp_dec_rad_field} we have shown that  data on $\scri^-$ which satisfy \eq{exp_dec_data1b} do not produce
logarithmic terms when approaching $I^-$ from $\scri^-$. Here we aim  to show that the same class of data also satisfies all no-logs conditions
at $I^-$ when coming from the cylinder $I$.
More specifically, let us assume that no logs arise up to and including the order $n-1$ at $I^-$.
We  want to show that the singular wave equations \eq{sing_wave1}-\eq{sing_wave2}
and the $\partial^n_rV^-_{AB}$-equation \eq{Wely_general_I2} do not produce logarithmic terms.

%The restrictions of Schouten and Weyl tensor, connection and frame coefficients to $I$ have been computed in \eq{frameI_1}-\eq{Schouten_I} and \eq{Weyl_I_1}-\eq{Weyl_I_6}. The first-order radial derivatives  are given by \eq{radial_I1}-\eq{radial_I23}.

\begin{lemma}
\label{lemma_estimates2} 
Assume that   \eq{exp_dec_data1b} holds. Then for $k\geq 1$ the radial derivatives have the following fall-off behavior at $I^-$ in an asymptotically Minkowski-like conformal Gauss gauge at each order (cf.\ Lemma~\ref{lemma_estimates1}):
\begin{align*}
\partial^k_r\widehat L_{1i}  |_I
=&
  \mathfrak{O}(1+\tau)^{k}
\,,
\\
\partial^k_r(\widehat L_{A0}+\widehat L_{A1})  |_I
=&
  \mathfrak{O}(1+\tau)^{k-1}
\,,
\\
\partial^k_r(\widehat L_{A0} -\widehat L_{A1})  |_I
=&  \mathfrak{O}(1+\tau)^{k}
\,,
\\
\partial^k_r\widehat L_{AB}  |_I
=&
 \mathfrak{O}(1+\tau)^{k-1}
\,,
\\
\partial^k_r\widehat\Gamma_{1}{}^0{}_0  |_I
 =&
  \mathfrak{O}(1+\tau)^{k+1}
\,,
\\
\partial^k_r\widehat\Gamma_{A}{}^0{}_0  |_I
 =&
  \mathfrak{O}(1+\tau)^{k}
\,,
\\
\partial^k_r\widehat\Gamma_{1}{}^0{}_1 |_I
 =&
  \mathfrak{O}(1+\tau)^{k+1}
\,,
\\
\partial^k_r(\widehat\Gamma_{1}{}^0{}_A+\widehat\Gamma_{1}{}^1{}_A)  |_I
 =&
  \mathfrak{O}(1+\tau)^{k+2}
\,,
\\
\partial^k_r(\widehat\Gamma_{1}{}^0{}_A -\widehat\Gamma_{1}{}^1{}_A) |_I
 =&
  \mathfrak{O}(1+\tau)^{k}
\,,
\\
\partial^k_r\widehat\Gamma_{A}{}^0{}_1  |_I
 =&
  \mathfrak{O}(1+\tau)^{k}
\,,
\\
\partial^k_r(\widehat\Gamma_{A}{}^0{}_B+ \widehat\Gamma_{A}{}^1{}_B) |_I
 =&
   \mathfrak{O}(1+\tau)^{k+1}
\,,
\\
\partial^k_r(\widehat\Gamma_{A}{}^0{}_{B} -\widehat\Gamma_{A}{}^1{}_{B}) |_I
 =&
  \mathfrak{O}(1+\tau)^{k-1}
\,,
\\
\partial^k_r\widehat\Gamma_{1}{}^A{}_B  |_I
 =&
  \mathfrak{O}(1+\tau)^{k+1}
\,,
\\
\partial^k_r\widehat\Gamma_{A}{}^B{}_C  |_I
 =&
  \mathfrak{O}(1+\tau)^{k}
\,,
\\
\partial^k_re^{\tau}{}_1 |_I=&    \mathfrak{O}(1+\tau)^{k+2}
\,,
\\
\partial^k_re^{\tau}{}_A |_I=&    \mathfrak{O}(1+\tau)^{k+1}
\,,
\\
\partial^k_re^{\alpha}{}_A |_I=&    \mathfrak{O}(1+\tau)^{k}
\,,
\\
\partial^k_re^{\mathring A}{}_1 |_I=&    \mathfrak{O}(1+\tau)^{k+1}
\,,
\\
\partial^k_re^{r}{}_1 |_I=& \delta^k{}_ 1  +  \mathfrak{O}(1+\tau)^{k+1}
\,,
\\
\partial^k_r  V^{\pm}_{AB} |_I=&  \mathfrak{O}(1+\tau)^{k\pm 2}
\,,
\\
\partial^k_r  W^{\pm}_{A} |_I=&  \mathfrak{O}(1+\tau)^{k\mp 1}
\,,
\\
\partial^k_r  U_{AB}|_I =&  \mathfrak{O}(1+\tau)^{k}
%\,,
%\\
%\partial^k_r  V^-_{AB} =& 
%\mathfrak{O}(1+\tau)^{k-2}
%\,,
%\\
%\partial^k_r  W^+_{A} =&  \mathfrak{O}(1+\tau)^{k-1}
\,.
\end{align*}
\end{lemma}

\begin{remark}
{\rm
While in Lemma~\ref{lemma_estimates1} we had polynomials of a sufficiently small degree which ensured that terms
of the critical, logarithms producing order in the $\partial^p_{\tau}V^-_{AB}$-equation do not appear, here the components show a sufficiently fast decay.
}
\end{remark}

\begin{proof}
As in in the proof of Lemma~\ref{lemma_estimates1}
we use an induction argument. The above considerations show that the lemma is true for $k=1$.
So let us assume that it is true for $1\leq k\leq n-1$. We want to show that it is also true for $k=n\geq 2$.

From \eq{evolution1}-\eq{evolution7}, \eq{frameI_1}-\eq{Schouten_I}
and \eq{Weyl_I_1}-\eq{Weyl_I_6} we deduce that
%\tim{tracefree-part...}
%
\begin{align*}
\partial_{\tau}\partial^n_r\widehat L_{10}  |_I
=&
   \mathfrak{O}(1+\tau)^{n-1}
\,,
\\
\partial_{\tau}\partial^n_r\widehat L_{11}  |_I
=&
  \mathfrak{O}(1+\tau)^{n-1}
\,,
\\
\partial_{\tau}\partial^n_r\widehat L_{1A}  |_I
=&
  \mathfrak{O}(1+\tau)^{n-1}
\,,
\\
\partial_{\tau}\partial^n_r(\widehat L_{A0}+\widehat L_{A1})  |_I
=&
   \mathfrak{O}(1+\tau)^{n-2}
\,,
\\
\partial_{\tau}\partial^n_r(\widehat L_{A0} -\widehat L_{A1})  |_I
=&   \mathfrak{O}(1+\tau)^{n-1}
\,,
\\
\partial_{\tau}\partial^n_r\widehat L_{AB}  |_I
=&
n (1+ \tau )\partial^{n-1}_rV^-_{AB}
- \partial^n_r( \widehat\Gamma_A{}^C{}_0\widehat L_{CB} )
+  \mathfrak{O}(1+\tau)^{n-1}= \mathfrak{O}(1+\tau)^{n-2}
\,,
\end{align*}
(the intermediate step in the last line will be relevant below).
For initial data \eq{initial_data_transport1} we obtain the asserted decay for the Schouten tensor.

Moreover,
%\tim{add remark... estimate improved in a second step}
%
\begin{align*}
\partial_{\tau}\partial^n_r\widehat\Gamma_{1}{}^0{}_0  |_I
 =&
-  \partial^n_r\widehat\Gamma_{1}{}^0{}_{1} 
+   \mathfrak{O}(1+\tau)^{n}\,,
\\
\partial_{\tau}\partial^n_r\widehat\Gamma_{A}{}^0{}_0  |_I
 =&
-  \partial^n_r\widehat\Gamma_{A}{}^0{}_{1} 
+     \mathfrak{O}(1+\tau)^{n-1}\,,
\\
\partial_{\tau}\partial^n_r\widehat\Gamma_{1}{}^0{}_1 |_I
 =&
   \mathfrak{O}(1+\tau)^{n}\,,
\\
\partial_{\tau}\partial^n_r\widehat\Gamma_{1}{}^0{}_A|_I
 =&
  \mathfrak{O}(1+\tau)^{n-1}\,,
\\
\partial_{\tau}\partial^n_r(\widehat\Gamma_{1}{}^0{}_A+\widehat\Gamma_{1}{}^1{}_A)  |_I
 =&
\partial^n_r(\widehat\Gamma_{1}{}^0{}_{A} +\widehat L_{1A} )   +    \mathfrak{O}(1+\tau)^{n+1}\,,
\\
\partial_{\tau}\partial^n_r(\widehat\Gamma_{1}{}^0{}_A -\widehat\Gamma_{1}{}^1{}_A) |_I
 =&
-\partial^n_r\widehat\Gamma_{1}{}^0{}_{A} 
+    \mathfrak{O}(1+\tau)^{n-1}\,,
\\
\partial_{\tau}\partial^n_r\widehat\Gamma_{A}{}^0{}_1  |_I
 =&
   \mathfrak{O}(1+\tau)^{n-1}\,,
\\
\partial_{\tau}\partial^n_r\widehat\Gamma_{A}{}^0{}_B |_I
 =&
-  \partial^n_r(\widehat\Gamma_C{}^0{}_B\widehat\Gamma_{A}{}^C{}_{0} ) -n(1+\tau) \partial^{n-1}_r  V^-_{AB}
+  \mathfrak{O}(1+\tau)^{n-1}\,,
\\
\partial_{\tau}\partial^n_r(\widehat\Gamma_{A}{}^0{}_B+ \widehat\Gamma_{A}{}^1{}_B) |_I
 =&
\partial^n_r(\widehat L_{AB} +\widehat\Gamma_{A}{}^0{}_{B} ) +    \mathfrak{O}(1+\tau)^{n}\,,
\\
\partial_{\tau}\partial^n_r(\widehat\Gamma_{A}{}^0{}_{B} -\widehat\Gamma_{A}{}^1{}_{B}) |_I
 =&
 \partial^n_r(\widehat L_{AB}  -\widehat\Gamma_{A}{}^0{}_{B} )+    \mathfrak{O}(1+\tau)^{n-2}\,,
\\
\partial_{\tau}\partial^n_r\widehat\Gamma_{1}{}^A{}_B  |_I
 =&
- \delta^A{}_B\partial^n_r\widehat\Gamma_{1}{}^1{}_{0} 
-  \widehat\Gamma_C{}^A{}_B\partial^n_r\widehat\Gamma_{1}{}^C{}_{0} 
 +   \mathfrak{O}(1+\tau)^{n}\,,
\\
\partial_{\tau}\partial^n_r\widehat\Gamma_{A}{}^B{}_C  |_I
 =&
-  \delta^B{}_C\partial^n_r\widehat\Gamma_{A}{}^1{}_{0} 
-  \widehat\Gamma_D{}^B{}_C\partial^n_r\widehat\Gamma_{A}{}^D{}_{0} 
+    \mathfrak{O}(1+\tau)^{n-1}\,.
\end{align*}
Similar to the proof of Lemma~\ref{lemma_estimates1}
one needs some  auxiliary equations (which also follow from  \eq{evolution1}-\eq{evolution7})
\begin{equation*}
\partial_{\tau}\partial^n_r(\widehat L_{AB} +\widehat\Gamma_A{}^0{}_B)|_I
=
-  \partial^n_r[\widehat\Gamma_{A}{}^C{}_{0}(\widehat L_{CB} + \widehat\Gamma_C{}^0{}_B ) ]
+ \mathfrak{O}(1+\tau)^{n-1}
\,,
\end{equation*}
which we use to show, by induction,
\begin{equation*}
\partial^n_r(\widehat L_{AB} +\widehat\Gamma_A{}^0{}_B)|_I
=
 \mathfrak{O}(1+\tau)^{n}
\,.
\end{equation*}
Note that it follows from \eq{radial_I1}-\eq{radial_I23} that this is true for $n=1$.

Furthermore,
\begin{align*}
\partial_{\tau}\partial^n_r(\widehat L_{1A}+ \widehat\Gamma_{1}{}^0{}_A )
=&
- \partial^n_r[ (\widehat\Gamma_B{}^0{}_A+\widehat L_{BA})\widehat\Gamma_{1}{}^B{}_{0} ]
   + \frac{n}{2}(1+\tau)^2\partial^{n-1}_r  W^+_{A}
 + \mathfrak{O}(1+\tau)^{n}
\\
=&
 \mathfrak{O}(1+\tau)^{n}
\,,
\end{align*}
which we use to show, again by induction (it follows from \eq{radial_I1}-\eq{radial_I23} that this holds for $n=1$)
\begin{equation*}
\partial^n_r(\widehat L_{1A}+ \widehat\Gamma_{1}{}^0{}_A)|_I
=
 \mathfrak{O}(1+\tau)^{n+1}
\,.
\end{equation*}
%
%Again one straightforwardly checks that  it holds for $n=1$.
For initial data \eq{initial_data_transport1} we then end up with  desired  decay for the connection coefficients by integrating all the above ODEs.

For the frame coefficients we find
\begin{align*}
\partial_{\tau}\partial^n_re^{\tau}{}_1 |_I=&     \mathfrak{O}(1+\tau)^{n+1}
\,,
\\
\partial_{\tau}\partial^n_re^{\mathring A}{}_1 |_I=&     \mathfrak{O}(1+\tau)^{n}
\,,
\\
\partial_{\tau}\partial^n_re^{r}{}_1 |_I=&     \mathfrak{O}(1+\tau)^{n}
\,,
\\
\partial_{\tau}\partial^n_re^{\tau}{}_A |_I=&     \mathfrak{O}(1+\tau)^{n}
\,,
\\
\partial_{\tau}\partial^n_re^{\mathring A}{}_A |_I=&     \mathfrak{O}(1+\tau)^{n-1}
\,,
\\
\partial_{\tau}\partial^n_re^{r}{}_A |_I=&    \mathfrak{O}(1+\tau)^{n-1}
\,.
\end{align*}
For  initial data \eq{initial_data_transport1} at $I^-$ that yields the desired result.

It remains to consider the radial derivatives of the rescaled Weyl tensor. Evaluation of the $n$th-order radial derivative of  \eq{evolution8}-\eq{evolution13}
yields 
\begin{align*}
\partial_{\tau}[(1-\tau)^{2-n} \partial^n_rV^+_{AB} ]|_I   =& 
\frac{1}{(1-\tau)^{n-1}}[ \mcD_{(A}\partial^n_rW^-_{B)}]_{\mathrm{tf}}
+   \mathfrak{O}(1+\tau)^{n+1}
\,,
\\
\partial_{\tau}[(1+\tau)^{2-n} \partial^n_rV^-_{AB} ]|_I  =& 
-\frac{1}{(1+\tau)^{n-1}}[\mcD_{(A}\partial^n_r W^+_{B)}]_{\mathrm{tf}}
+   \mathfrak{O}(1)
\,,
\\
(1+\tau)\partial_{\tau}\partial^n_rW^-_{A}|_I     =& 
2\mcD^B\partial^n_r V^+_{AB}
+(n+1)\partial^n_rW^-_{A} 
+   \mathfrak{O}(1+\tau)^{n+2}
\,,
\\
(1-\tau)\partial_{\tau}\partial^n_rW^+_{A}    |_I =& 
-2\mcD^B\partial^n_r  V^-_{AB}
-(n+1)\partial^n_r W^+_{A} 
+   \mathfrak{O}(1+\tau)^{n-2}
\,,
\\
n\partial^n_{r} W_{0101}|_I
 =& 
-\frac{1}{2}(1+\tau)\mcD^A \partial^n_{r}W^{+}_A
-\frac{1}{2}(1-\tau)\mcD^A\partial^n_{r}W^{-}_A
+   \mathfrak{O}(1+\tau)^{n}
\,,
\\
n\partial^n_{r}W_{01AB}|_I =&
(1+\tau)\mcD_{[A}   \partial^n_{r}W^{+}_{B]}
-(1-\tau)\mcD_{[A} \partial^n_{r} W^{-}_{B]}
+   \mathfrak{O}(1+\tau)^{n}
\,.
\end{align*}
As in Section~\ref{sec_structure_eqn_I}
we derive decoupled  equations for $\partial^n_rW^{\pm}_A$ from this system,
\begin{align*}
(1-\tau^2)\partial^2_{\tau}\partial^n_rW^-_{A}   |_I
  =& 
 \Big(\Delta_s +(n-1)n-1\Big)\partial^n_rW^-_{A}
-2[1-(n-1)\tau)]\partial_{\tau}\partial^n_rW^-_{A}   
+   \mathfrak{O}(1+\tau)^{n+1}
\,,
\\
(1-\tau^2)\partial^2_{\tau}\partial^n_rW^+_{A}   |_I
  =& 
\Big(\Delta_s +(n-1)n-1\Big)\partial^n_r W^+_{A}
-2[1+(n-1)\tau)]\partial_{\tau}\partial^n_rW^+_{A}   
+   \mathfrak{O}(1+\tau)^{n-1}
\,.
\end{align*}
It follows from \eq{initial_data_transport7} and Lemma~\ref{lem_sing_wave} (i) that the solutions are smooth,
\begin{align*}
\partial^n_rW^-_{A}   |_I
  =&
   \mathfrak{O}(1+\tau)^{n+1}
\,,
\quad
\partial^n_rW^+_{A}   |_I
  =
  \mathfrak{O}(1+\tau)^{n-1}
\,.
\end{align*}
Using \eq{initial_data_transport7} one obtains the desired decay for the remaining components of the Weyl tensor by solving the above transport equations
(in particular the $\partial^n_rV^-_{AB}$-equation does not produce log-terms), which completes the proof.
\qed

\end{proof}

\begin{corollary}
\label{cor_vanishing_radiation2}
Assume that  \eq{exp_dec_data1} holds, i.e.\ that  the radiation field vanishes at any order at $I^-$, in a spacetime which admits a smooth $I$, and where the
data for the transport equations on $I$ at $I^-$ are induced by the limits of the corresponding data on $\scri^-$.
Then the restrictions to $I$ of all the fields appearing  in the GCFE including their
radial  derivatives  of all orders admit smooth extensions through $I^-$  in an asymptotically Minkowski-like conformal gauss gauge
at each order (cf.\ Corollary~\ref{cor_vanishing_radiation}).
\end{corollary}

\subsection{Conformal Gauss coordinates at $\scri^-$}

We determine the expansion near $I^-$ of a line element of a vacuum spacetime which satisfies
\begin{equation}
\Xi_{AB} = \mathfrak{O}(r^{\infty})
\label{Xi_assumption}
\end{equation}
 in conformal Gaussian coordinates at $\scri^-$  in an asymptotically
 Minkowski-like conformal Gauss gauge at each order. In particular this will be useful to determine Kerr data on $\scri^-$ for this gauge.

The restriction of the metric to $\scri^-$ follows immediately from the gauge data \eq{main_gauge0B}.
\begin{equation}
g|_{\scri^-} = -\mathrm{d}\tau^2 + \frac{2}{r}\mathrm{d}\tau\mathrm{d}r + s_{\mathring A\mathring B}\mathrm{d}x^{\mathring A}x^{\mathring B}+ \mathfrak{O}(r^{\infty})
\,.
\end{equation}
Higher-order derivatives are obtained  from  expansions of the frame coefficients. 
The first-order terms have been computed in \eq{frame_1storder13}-\eq{frame_1storder14}
For the second- and third-order derivatives a computation which uses \eq{Xi_assumption} and \eq{evolution1}-\eq{evolution7}  reveals 
\begin{align}
\partial^2_{\tau}e^{\mu}{}_a|_{\scri^-} =&  \mathfrak{O}(r^{\infty})
\,,
\\
\partial^3_{\tau}e^{\tau}{}_1|_{\scri^-} =& 
12r M+ \mathfrak{O}(r^{\infty})
\,,
\\
\partial^3_{\tau}e^{r}{}_1|_{\scri^-} =& 
8r^2 M+ \mathfrak{O}(r^{\infty})
\,,
\\
\partial^3_{\tau}e^{\mathring A}{}_1|_{\scri^-} =& 
6r^2L^{\mathring A}+ \mathfrak{O}(r^{\infty})
\,,
\\
\partial^3_{\tau}e^{\tau}{}_A|_{\scri^-} =& 
6r^2L_A+ \mathfrak{O}(r^{\infty})
\,,
\\
\partial^3_{\tau}e^{r}{}_A|_{\scri^-} =& 
4r^3L_A+ \mathfrak{O}(r^{\infty})
\,,
\\
\partial^3_{\tau}e^{\mathring A}{}_A|_{\scri^-} =& 
-2r^2\mathring e^{\mathring A}{}_B \eta^{BC} (\mcD_{(A}L_{C)})_{\mathrm{tf}} +2r^3\mathring e^{\mathring A}{}_B L_{A}{}^{B} - 4r M\mathring e^{\mathring A}{}_A+ \mathfrak{O}(r^{\infty})
\,.
\end{align}
We then obtain
\begin{align}
g=& -\mathrm{d}\tau^2 + \frac{2}{r}\mathrm{d}\tau\mathrm{d}r + s_{\mathring A\mathring B}\mathrm{d}x^{\mathring A}x^{\mathring B}
 + (1+\tau)\Big(\frac{2}{r^2}\mathrm{d}r^2-\frac{2}{r}\mathrm{d}\tau\mathrm{d} r\Big)
\nonumber
\\
&
-\frac{1}{r^2} (1+\tau)^2 \mathrm{d} r^2
+\frac{2}{3}(1+\tau)^3\Big[
2M\mathrm{d}\tau\mathrm{d}r
+r^2 L_{\mathring A}\mathrm{d}\tau\mathrm{d}x^{\mathring A}
-6r^{-1} M\mathrm{d}r^2
\nonumber
\\
&
-6 rL_{\mathring A} \mathrm{d}r  \mathrm{d}x^{\mathring A}
+\Big(2r Ms_{\mathring A\mathring B} + r^2(\mcD_{(\mathring A}L_{\mathring B)})_{\mathrm{tf}} -r^3 c^{(2,0)}_{\mathring A\mathring B} \Big) \mathrm{d}x^{\mathring A}\mathrm{d} x^{\mathring B}
\Big]
+ \mathfrak{O}(1+\tau)^4
+ \mathfrak{O}(r^{\infty})
\,.
\label{conformal_Gauss_coord}
\end{align}

\subsubsection{Example: Kerr spacetime}
\label{sec_example_Kerr}

It is quite illuminating to calculate  which data on $\scri^-$ are needed in our gauge to generate a spacetime which belongs to the Kerr family.
For this purpose let us compute the Kerr metric in conformal Gauss coordinates  in an asymptotically Minkowski-like conformal Gauss gauge at each order, or rather  its asymptotic expansion at $\scri^-$.

In \emph{Kerr-Schild Cartesian coordinates} the Kerr line elements reads (cf.\ e.g.\ \cite{visser}),
\begin{equation}
\widetilde g
% &=& -\mathrm{d}t^2+ \mathrm{d}x^2 + \mathrm{d}y^2 + \mathrm{d}z^2 + \frac{2mr^3}{r^4 + a^2 z^2}
%\Big[ \mathrm{d}t + \frac{r(x\mathrm{d}x + y\mathrm{d}y)}{a^2 + r^2} + \frac{a(y\mathrm{d}x- x\mathrm{d}y)}{a^2 + r^2} + \frac{z}{r}\mathrm{d}z\Big]^2
%\\
= -(\mathrm{d}y^0)^2+(\mathrm{d}y^1)^2 + (\mathrm{d}y^2)^2 +(\mathrm{d}y^3)^2 + \frac{2mR^3}{R^4 + a^2 (y^3)^2}
\ell\otimes\ell
\end{equation}
where
\begin{equation}
\ell= \mathrm{d}y^0 + \frac{Ry^1+ay^2}{a^2+R^2}\mathrm{d}y^1+  \frac{Ry^2-ay^1}{a^2+R^2}\mathrm{d}y^2 + \frac{y^3}{R}\mathrm{d}y^3
\,.
\end{equation}
The function $R$ is given by
\begin{equation}
(y^1)^2+(y^2)^2+(y^3)^2 = R^2 + a^2\Big(1-\frac{(y^3)^2}{R^2}\Big)
\,.
\end{equation}
Observe that 
\begin{equation}
\eta^{\sharp}(\ell,\ell) =0
\,.
\end{equation}
We apply the same coordinate transformation \eq{Mink_trafo1} as for Minkowski and choose the same conformal factor
\eq{Mink_conf_fac}.
That yields
\begin{equation}
 g 
=  -\mathrm{d}\tau^2 -2\frac{\tau}{r}\mathrm{d}\tau\mathrm{d} r + \frac{1-\tau^2}{r^2}\mathrm{d}r^2 + \mathrm{d}\theta^2 + \sin^2\theta\mathrm{d}\phi^2+ \frac{2mR^3\Theta^4}{R^4\Theta^2 + a^2 \cos^2\theta}
\ell\otimes\ell
\,,
\label{Kerr_rescaled}
\end{equation}
where
\begin{align*}
R =&-\frac{1}{\sqrt{2}}\sqrt{(y^1)^2+(y^2)^2+(y^3)^2-a^2 +\sqrt{((y^1)^2+(y^2)^2+(y^3)^2-a^2)^2 + 4 a^2(y^3)^2}}
\\
=&-\frac{1}{\sqrt{2}\,\Theta}\sqrt{1 -a^2 \Theta^2+\sqrt{(1 -a^2\Theta^2)^2 +4 a^2\cos^2\theta \Theta^2}}
\,,
\end{align*}
whence
\begin{equation}
R \Theta = 
-1 +  2a^2 r^2\sin^2\theta(1+\tau)^2  +O(1+\tau)^4
\,.
\end{equation}
We have
\begin{align}
\ell_{\tau} =&r \frac{(a^2+R^2)(2\tau - R \Theta (1+\tau^2))-2 a^2 \tau \sin^2\theta}{R (a^2+R^2)\Theta^3}
=
-\Big(\frac{1}{4r} + r a^2\sin^2\theta \Big)+ O(1+\tau)
\,,
\\
\ell_r =& \frac{-(a^2+R^2)(1- \tau R \Theta)+ a^2\sin\theta}{R(a^2+R^2)r\Theta^2}
=  \frac{1}{ 2 r^2} + O(1+\tau)
\,,
\\
\ell_{\theta} =& -\frac{a^2\cos\theta\sin\theta}{R (a^2+R^2)\Theta^2}=O(1+\tau)
\,,
\\
\ell_{\phi} =& -\frac{a\sin^2\theta}{(a^2+R^2)\Theta^2}
= - a\sin^2\theta + O(1+\tau)^2
\,.
\end{align}
For the prefactor of $\ell\otimes\ell$ in \eq{Kerr_rescaled} we find the expansion
\begin{equation}
\frac{2mR^3\Theta^4}{R^4\Theta^2 + a^2 \cos^2\theta}
=
-16 m r^3(1+\tau)^3 + O(1+\tau)^4
%+6m r^3(1+\tau)^4 + O(1+\tau)^5
\,.
\end{equation}
Altogether the Kerr metric adopts the form
\begin{align}
 g 
=&  -\mathrm{d}\tau^2 +\frac{2}{r}\mathrm{d}\tau\mathrm{d} r + \mathrm{d}\theta^2 + \sin^2\theta\mathrm{d}\phi^2
+(1+\tau)\Big(  \frac{2}{r^2}\mathrm{d}r^2 
 -\frac{2}{r}\mathrm{d}\tau\mathrm{d} r \Big)
\nonumber
\\
&
- \frac{1}{r^2} (1+\tau)^2\mathrm{d}r^2 
+ 4mr(1+\tau)^3\Big[
-\frac{ r^2}{4}  \Big(\frac{1}{r} +4 r a^2\sin^2\theta \Big)^2 \mathrm{d}\tau^2
\nonumber
\\
&
-r^{-2}\mathrm{d}r^2
- 4 a^2 r^2 \sin^4\theta \mathrm{d}\phi^2
+ \Big(\frac{1}{r} +4 r a^2\sin^2\theta \Big)\mathrm{d}\tau\mathrm{d}r
\nonumber
\\
&
-2  a r^2\sin^2\theta \Big(\frac{1}{r} +4 r a^2\sin^2\theta \Big)\mathrm{d}\tau\mathrm{d}\phi
+4 a \sin^2\theta  \mathrm{d}r\mathrm{d}\phi
\Big]
 + O(1+\tau)^4
\,.
\end{align}
We need to make sure that this is the right gauge.
Comparison with \eq{conformal_Gauss_coord}
 shows that this is the case only up to and including terms of order $(1+\tau)^2$.
Straightforward transformation of $\tau$, $r$ and $\phi$ 
\begin{align}
\tau \enspace&\mapsto \enspace \tau +2mr \Big(\frac{1}{3}+ 2 a^2 r^2\sin^2\theta  \Big)(1+\tau)^4
\,,
\\
r \enspace&\mapsto \enspace  r+mr^2 \Big(\frac{3}{8} +2 a^2 r^2\sin^2\theta  -2a^4 r^4 \sin^4\theta  \Big)(1+\tau)^4
\,,
\\
\phi \enspace&\mapsto \enspace \phi - 2ma r^2\Big(\frac{1}{3} +2 a^2 r^2\sin^2\theta \Big)  (1+\tau)^4
\,,
\end{align}
 accompanied by a conformal transformation
$\Theta\mapsto\Omega\Theta$ with 
\begin{equation}
\Omega=1 + 2 mr\Big(\frac{1}{3} + 2 a^2 r^2\sin^2\theta  \Big)(1+\tau)^3
\end{equation}
brings the line element into the desired form
\begin{align}
 g 
=&  -\mathrm{d}\tau^2 +\frac{2}{r}\mathrm{d}\tau\mathrm{d} r + \mathrm{d}\theta^2 + \sin^2\theta\mathrm{d}\phi^2
+(1+\tau)\Big(  \frac{2}{r^2}\mathrm{d}r^2 
 -\frac{2}{r}\mathrm{d}\tau\mathrm{d} r \Big)
- \frac{1}{r^2} (1+\tau)^2\mathrm{d}r^2 
\nonumber
\\
&
+4 m r(1+\tau)^3\Big[\frac{1}{3r}  \mathrm{d}\tau\mathrm{d}r -  r^{-2} \mathrm{d}r^2
+ \frac{1}{3}   \Big(  \mathrm{d}\theta^2 + \sin^2\theta\mathrm{d}\phi^2\Big)
-\frac{2}{3} a r\sin^2\theta  \mathrm{d}\tau\mathrm{d}\phi
+4 a \sin^2\theta  \mathrm{d}r\mathrm{d}\phi
\nonumber
\\
&
 + 2a^2 r^2\sin^2\theta \Big(  \mathrm{d}\theta^2 - \sin^2\theta\mathrm{d}\phi^2\Big)
\Big]
 + O(1+\tau)^4
\,.
\end{align}
Comparison with \eq{conformal_Gauss_coord}
shows that
\begin{align}
M=& m
\,,
\quad
L_{\theta} = 0
\,,
\quad
L_{\phi} = -4 m a\sin^2\theta
\,,
\label{kerr_parameters1}
\\
c_{\theta\theta}=& -12ma ^2 \sin^2\theta
\,,
\quad
c_{\phi\phi}= 12ma^2\sin^4\theta
\,,
\quad
c_{\theta\phi} = 0
\,.
\label{kerr_parameters2}
\end{align}
Note that $L_A$ is a conformal Killing 1-form,
and that (we have not attempted to compute the higher-order integration functions $c^{(p+2,p)}_{\mathring A \mathring B}$)
\begin{equation}
c^{(2,0)}_{\mathring A\mathring B} = \frac{3}{2 m}(L_{\mathring A} \otimes L_{\mathring B})_{\mathrm{tf}}
\,.
\end{equation}

\section{Toy model: Massless Spin-2 equation}
\label{sec_spin2}
\label{section7}

%\subsection{Linear wave equation on Minkowski background}
%
%\begin{equation*}
%\Box \phi=0 \quad \Longleftrightarrow \quad \eta^{ij} e^{\mu}{}_ie^{\nu}{}_j \partial_{\mu}\partial_{\nu}\phi -\eta^{ij}\widehat \Gamma_i{}^k{}_je^{\mu}{}_k\partial_{\mu}\phi  -2f^ie^{\mu}{}_i\partial_{\mu}\phi =0
%\end{equation*}
%
%\begin{equation*}
%(1-\tau^2) \partial_{\tau}^2\phi  +2\tau r\partial_{r}\partial_{\tau}\phi - r^2\partial_{r}^2\phi - \Delta_s\phi  -\tau\partial_{\tau}\phi
%  + r\partial_{r}\phi=0
%\end{equation*}
%
%$\partial^n_{\tau}$
%\begin{equation*}
%(2n+1) \partial_{\tau}^{n+1}\phi  -2a_n\partial_{\tau}^n\phi  -2 r\partial_{r}\partial^{n+1}_{\tau}\phi 
%+ 2n r\partial_{r}\partial^n_{\tau}\phi - r^2\partial_{r}^2\partial^n_{\tau}\phi - \Delta_s\partial^n_{\tau}\phi 
%  -n\partial^n_{\tau}\phi+ r\partial_{r}\partial^n_{\tau}\phi=0
%\end{equation*}

We have seen so far that when computing all the fields and their  transverse and radial derivatives at $I^-$
one generically should expect logarithmic terms even if the seed data, i.e.\ the radiation field, is smooth at $I^-$. These
logarithmic terms can arise at arbitrary high order. 
Although we have established a sufficient condition which ensures that no logarithmic terms arise in the formal expansions 
it seems much harder to establish necessary-and-sufficient conditions.
The main issue is that the appearance of log terms does \emph{not} only depend on the leading-order terms, which are the only 
ones which are controllable without too much effort.

The purpose of this section is to consider a similar problem which is much simpler to deal with, namely we consider
the \emph{massless spin-2 equation} (cf.\ \cite{F_spin, kroon})
\begin{equation}
 \widehat\nabla_i W^i{}_{jkl}  = \frac{1}{4}\widehat\Gamma_i{}^p{}_pW^i{}_{jkl} 
\end{equation}
on a flat background, as which we take the Minkowski metric in the form \eq{Mink_cyl}
\begin{equation}
 \eta = 
-\mathrm{d}\tau^2 - 2\frac{\tau}{  r}\mathrm{d}\tau\mathrm{d}r + \frac{1-\tau^2}{r^2} \mathrm{d}r^2
 +\mathrm{d}\theta^2  + \sin^2\theta^2\mathrm{d}\phi^2
\,.
\label{Mink_cylB}
\end{equation}
In the gauge \eq{Mink_gauge} we have
\begin{align}
e_0=&\partial_{\tau}
\,,\quad
e_1=-\tau\partial_{\tau} + r\partial_r
 \,,\quad
e_A = \mathring e_A
\,,
\quad
f_1=1\,, \quad  f_A =0
\,.
\end{align}
We give a list of the non-vanishing Christoffel symbols
%\begin{eqnarray*}
%\Gamma^0_{00} &=&%g^{01}\partial_0g_{01}=
%\tau
%\\
%\Gamma^0_{01} &=&%\frac{1}{2}g^{01}\partial_0g_{11}=
%\frac{\tau^2}{r}
%\\
%\Gamma^0_{0A} &=&0
%\\
%\Gamma^0_{11} &=&%\frac{1}{2}g^{00}(2\partial_1g_{10}-\partial_{0}g_{11})+ \frac{1}{2}g^{01}\partial_1g_{11}=
%-\frac{\tau}{r^2}(1-\tau^2)
%\\
%\Gamma^0_{1A} &=&0
%\\
%\Gamma^0_{AB} &=&0
%\\
%\Gamma^1_{00} &=&%g^{11}\partial_0g_{01}=
%-r
%\\
%\Gamma^1_{01} &=&%\frac{1}{2}g^{11}\partial_0g_{11}=
%-\tau
%\\
%\Gamma^1_{0A} &=&0
%\\
%\Gamma^1_{11} &=&%\frac{1}{2}g^{01}(2\partial_1g_{10}-\partial_{0}g_{11})+ \frac{1}{2}g^{11}\partial_1g_{11}=
%-\frac{1+\tau^2}{r}
%\\
%\Gamma^1_{1A} &=&0
%\\
%\Gamma^1_{AB} &=&0
%\\
%\Gamma^A_{00} &=&0
%\\
%\Gamma^A_{01} &=&0
%\\
%\Gamma^B_{0A} &=&0
%\\
%\Gamma^A_{11} &=&0
%\\
%\Gamma^B_{1A} &=&0
%\\
%\Gamma^C_{AB} &=&\mathring \Gamma^C_{AB}
%\end{eqnarray*}
\begin{align*}
\Gamma^{\tau}_{\tau\tau} =&\tau
\,,
\quad
\Gamma^{\tau}_{\tau r} =\frac{\tau^2}{r}
\,,
\quad
\Gamma^{\tau}_{rr} =-\frac{\tau}{r^2}(1-\tau^2)
\,,
\quad
\Gamma^r_{\tau\tau} =
-r
\,,
\\
\Gamma^r_{\tau r} =&
-\tau
\,,
\quad
\Gamma^r_{rr} =-\frac{1+\tau^2}{r}
\,,
\quad
\Gamma^{\mathring C}_{\mathring A\mathring B} =\mathring \Gamma^{\mathring C}_{\mathring A\mathring B}
\,.
\end{align*}
For the Weyl connection coefficients we then find
%\begin{equation}
%\widehat\Gamma_i{}^k{}_j e^{\mu}{}_k
%=e^{\nu}{}_i (\partial_{\nu} e^{\mu}{}_j  + \Gamma^{\mu}_{\nu\sigma} e^{\sigma}{}_j )
% + 2e^{\mu}{}_{(i}f_{j)} - \eta_{ij} f^{\mu}
%\,.
%\end{equation}
%
\begin{align}
\widehat\Gamma_1{}^0{}_0 
=&
 1
\,,
\quad
\widehat\Gamma_1{}^0{}_1=0
\,,
\quad
\widehat\Gamma_1{}^0{}_A=0
\,,
\quad
\widehat\Gamma_1{}^1{}_A
=
0
\,,
\quad
\widehat\Gamma_1{}^B{}_A
=
\delta^B{}_A
\,,
\\
\widehat\Gamma_A{}^0{}_1 
=&
0
\,,
\quad
\widehat\Gamma_A{}^1{}_1 
=
0
\,,
\quad
\widehat\Gamma_A{}^B{}_0
=
0
\,,
\quad
\widehat\Gamma_A{}^B{}_1 
=
\delta^B{}_A
\,,
\quad
\widehat\Gamma_A{}^C{}_B 
=
\mathring \Gamma_A{}^C{}_B 
\,.
\end{align}
The spin-2 equation then takes the form (cf.\ \eq{evolutionW1b}-\eq{evolution_WA-_2})
\begin{align}
\mathfrak{W}:=& \partial_{\tau}W_{0101} 
+\frac{1}{2}\mcD^A W^{+}_A
-\frac{1}{2}\mcD^AW^{-}_A=0
\,,
\label{spin2_evolution1}
\\
\mathfrak{W}_{AB}:=&\partial_{\tau}W_{01AB}    
-\mcD_{[A}  W^{+}_{B]}
-\mcD_{[A} W^{-}_{B]}=0
\,,
\label{spin2_evolution2}
\\
\mathfrak{W}^-_A :=&\partial_{\tau}W^-_A
- \mcD^BV^+_{AB}
 -\frac{1}{2}\mcD^B U_{BA}
-W^-_{A} =0
\,,
\label{spin2_evolution3}
\\
\mathfrak{W}^+_A:=&\partial_{\tau}W^+_A
+ \mcD^BV^-_{AB}
 +\frac{1}{2}\mcD^B U_{AB}
 +W^+_A =0
\,,
\label{spin2_evolution4}
\\
\mathfrak{W}^+_{AB}:=&( 1-\tau )\partial_{\tau}  V^+_{AB}
+(r\partial_{r}-2) V^+_{AB}
- (  \mcD_{(A} W^-_{B)})_{\mathrm{tf}}=0
\,,
\label{spin2_evolution5}
\\
\mathfrak{W}^-_{AB}:=&(1+\tau)\partial_{\tau}
 V^-_{AB} 
-(r\partial_{r} -2)V^-_{AB} 
+(\mcD_{(A}W^+_{B)})_{\mathrm{tf}}=0
\,,
\label{spin2_evolution6}
\end{align}
and,
\begin{align}
(1+\tau)\partial_{\tau} W_{0101}
 =& 
r\partial_{r} W_{0101}+
\mcD^AW^{-}_A
\,,
\label{spin2_evconstr1}
\\
(1+\tau)\partial_{\tau}  W_{01AB}   =&
 r\partial_{r} W_{01AB}  +
2\mcD_{[A} W^{-}_{B]}
\,,
\label{spin2_evconstr2}
\\
(1+\tau)\partial_{\tau}W^-_{A} 
 =& 
(r\partial_{r}+1)W^-_{A}  + 2\mcD^BV^+_{AB}
\,,
\label{spin2_evconstr3}
\\
(1+\tau)\partial_{\tau}  W^+_{A}   =&
(r\partial_{r}-1)W^+_{A} -\mcD^B U_{AB}
\,.
\label{spin2_evconstr4}
\end{align}
Using the evolution equations we rewrite the latter ones as the following set of constraint equations,
\begin{align}
\mathfrak{C}:=& r\partial_{r} W_{0101}
+\frac{1}{2}(1+\tau)\mcD^A W^{+}_A
+\frac{1}{2}(1-\tau)\mcD^AW^{-}_A
=0
\,,
\label{spin2_constr1}
\\
\mathfrak{C}_{AB}:=& r\partial_{r} W_{01AB}  
-(1+\tau)\mcD_{[A}  W^{+}_{B]}
+(1-\tau)\mcD_{[A} W^{-}_{B]}
=0
\,,
\label{spin2_constr2}
\\
\mathfrak{C}^-_{A}:=&(r\partial_{r}-\tau )W^-_{A}   
+(1-\tau)\mcD^BV^+_{AB}
 -\frac{1}{2}(1+\tau)\mcD^B U_{BA}
=0
\,,
\label{spin2_constr3}
\\
\mathfrak{C}^+_{A}:=&(r\partial_{r}+\tau )W^+_{A}  
+ (1+\tau)\mcD^BV^-_{AB}
 -\frac{1}{2}(1-\tau)\mcD^B U_{AB}
=0
\,,
\label{spin2_constr4}
\end{align}
which one easily checks to be preserved under evolution \eq{spin2_evolution1}-\eq{spin2_evolution6}.
Here, though, we will use a different set of evolution equations.

%\subsection{Preservation of the constraints}
%%
%\begin{align}
%\partial_{\tau}C
%=&
%\frac{1}{2}\mcD^AD_A
%-\frac{1}{2}\mcD^AE_A
%\,,
%\\
%\partial_{\tau}C_{AB}
%=&
%\mcD_{[A}D_{B]}+\mcD_{[A}E_{B]}
%\,,
%\\
%\partial_{\tau}D_A
%=&
%D_A
% +\frac{1}{2}\mcD_AC
% -\frac{1}{2}\mcD^BC_{AB}
%\,,
%\\
%\partial_{\tau}E_A
%=&
%-E_A
% -\frac{1}{2}\mcD_A C
% -\frac{1}{2}\mcD^BC_{AB}
%\,.
%\end{align}

\subsection{Rewriting the equations}

We want to decouple the evolution equations.
For this we take the divergence of \eq{spin2_evolution5} and \eq{spin2_evolution6}
and eliminate $V^{\pm}_{AB}$ via \eq{spin2_evconstr3} and a  linear combination of \eq{spin2_evolution4} and \eq{spin2_evconstr4},
\begin{equation}
(1-\tau)\partial_{\tau}  W^+_{A}   =
-(r\partial_{r}+1)W^+_{A}-2 \mcD^BV^-_{AB}
\,,
\label{spin2_alg_V-}
\end{equation}
respectively.
That yields decoupled equations for $W^{\pm}_A$ which we supplement by the remaing evolution equations for $V^{\pm}_{AB}$ and $U_{AB}$,
\begin{align}
( 1-\tau^2 )\partial^2_{\tau} W^-_{A} 
=&2 \Big(1-\tau (r\partial_{r}-1) \Big)\partial_{\tau}W^-_{A}
+\Big(\Delta_s+(r\partial_{r}-2) (r\partial_{r}+1)+1\Big)W^-_A
\,,
\label{spin2_wave1}
\\
(1-\tau^2)\partial_{\tau}^2  W^+_{A} 
=&
2\Big( -1-\tau (r\partial_{r} -1)\Big)\partial_{\tau}  W^+_{A} 
+\Big(\Delta_s+(r\partial_{r} -2)(r\partial_{r}+1)+1\Big)W^+_A
\,,
\label{spin2_wave2}
\\
( 1-\tau )\partial_{\tau}  V^+_{AB}
 =&
-(r\partial_{r}-2) V^+_{AB}
+ (  \mcD_{(A} W^-_{B)})_{\mathrm{tf}}
\,,
\\
(1+\tau)\partial_{\tau}
 V^-_{AB} 
=&
(r\partial_{r} -2)V^-_{AB} 
-(\mcD_{(A}W^+_{B)})_{\mathrm{tf}}
%\\
% V^+_{A}
% =& 
%\frac{1}{2}(1+\tau)\partial_{\tau}W^-_{A} -\frac{1}{2}(r\partial_{r}+1)W^-_{A}  
%\,,
%\\
%V^-_{A}
% =& 
% -\frac{1}{2}(1-\tau)\partial_{\tau}  W^+_{A} -\frac{1}{2}(r\partial_{r}+1)W^+_{A} 
%\,,
\\
\partial_{\tau}W_{0101} 
 =& 
-\frac{1}{2}\mcD^A W^{+}_A
+\frac{1}{2}\mcD^AW^{-}_A
\,,
\label{spin2_evolution1B}
\\
\partial_{\tau}W_{01AB}     =&
\mcD_{[A}  W^{+}_{B]}
+\mcD_{[A} W^{-}_{B]}
\,.
\label{spin2_evolution2B}
\end{align}
%
%xxxxxxxxxxxxxxxxxx
%\begin{eqnarray*}
%( 1-\tau^2 )\partial^2_{\tau} \phi
%-2 \Big(\pm 1-\tau (r\partial_{r}-1) \Big)\partial_{\tau}\phi
%-\Big(\Delta_s+r^2\partial^2_r\Big)\phi=0
%\end{eqnarray*}
%\begin{eqnarray*}
%\eta^{ij}\nabla_i\Theta\nabla_i\phi=(1-\tau^2)\tau r\partial_{\tau}\phi
%+ r^2(1+\tau^2)\partial_{r}\phi
%\end{eqnarray*}
%\begin{eqnarray*}
%\Box \phi &=& \eta^{ij}e^{\mu}{}_ie^{\nu}{}_j\partial_{\mu}\partial_{\nu}\phi - \eta^{ij}\widehat \Gamma_i{}^k{}_je^{\mu}{}_k \partial_{\mu}\phi
%-2f^ke^{\mu}{}_k \partial_{\mu}\phi
%\\
% &=&-(1-\tau^2)\partial^2_{\tau}\phi
%-2  \tau r\partial_{r}\partial_{\tau}\phi
%+(\Delta_s+r^2\partial_{r}^2)\phi
%+\tau \partial_{\tau}\phi
%-r \partial_{r}\phi
%\end{eqnarray*}
%xxxxxxxxxxxxxx
%
By derivation these equation follow from the spin-2 equation.
To obtain conditions which ensure that they are also sufficient we find that they imply the following set of equations
(note that \eq{spin2_evolution1}-\eq{spin2_evolution2} and \eq{spin2_evolution5}-\eq{spin2_evolution6}
are trivially satisfied),
%\begin{eqnarray*}
%2\mcD^BR^+_{AB} +\Big((1-\tau)\partial_{\tau}+r\partial_r-2\Big) \Big((1+\tau)Q^-_A-D_A\Big)&=&0
%\\
%2\mcD^BR^-_{AB} -\Big((1+\tau)\partial_{\tau}-r\partial_r+2\Big) \Big((1-\tau)Q^+_A+E_A\Big)&=&0
%\\
%R^+_{AB}&=&0
%\\
%R^-_{AB}&=& 0
%\\
%P&=& 0
%\\
%P_{AB} &=& 0
%\end{eqnarray*}
\begin{align}
\Big((1-\tau)\partial_{\tau}+r\partial_r-2\Big) \Big((1+\tau)\mathfrak{W}^-_A-\mathfrak{C}^-_A\Big)-2\mcD^B\mathfrak{W}^+_{AB}=&0
\,,
\label{constr_pres1}
\\
\Big((1+\tau)\partial_{\tau}-r\partial_r+2\Big) \Big((1-\tau)\mathfrak{W}^+_A+\mathfrak{C}^+_A\Big)-2\mcD^B\mathfrak{W}^-_{AB}=&0
\,,
%\\
%R^+_{AB}&=&0
%\\
%R^-_{AB}&=& 0
%\\
%P&=& 0
%\\
%P_{AB} &=& 0
\\
\partial_{\tau} \mathfrak{C} - \frac{1}{2}(1+\tau)\mcD^A\mathfrak{W}^+_A -\frac{1}{2}(1-\tau)\mcD^A\mathfrak{W}^-_A- \frac{1}{2}\mcD^A\mathfrak{C}^-_A
+\frac{1}{2}\mcD^A\mathfrak{C}^+_A =&0
\,,
\\
\partial_{\tau}\mathfrak{C}_{AB}
+(1+\tau)\mcD_{[A}  \mathfrak{W}^{+}_{B]}
-(1-\tau)\mcD_{[A} \mathfrak{W}^{-}_{B]}
- \mcD_{[A}\mathfrak{C}^-_{B]}-\mcD_{[A}\mathfrak{C}^+_{B]}
=&0
\,,
\\
(\partial_{\tau}-1)\mathfrak{C}^-_A
-(r\partial_r-\tau)\mathfrak{W}^-_A
 -\frac{1}{2}\mcD_A\mathfrak{C}
 +\frac{1}{2}\mcD^B\mathfrak{C}_{AB}
=&0
\,,
\\
(\partial_{\tau}+1)\mathfrak{C}^+_A
-(r\partial_r+\tau)\mathfrak{W}^+_A
 +\frac{1}{2}\mcD_A \mathfrak{C}
 +\frac{1}{2}\mcD^B\mathfrak{C}_{AB}
=&0
\,.
\label{constr_pres6}
\end{align}
One needs to characterize data which ensure that the trivial solution to \eq{constr_pres1}-\eq{constr_pres6} is the only one.
Here, however, we are interested  in the appearance of logarithmic terms:
Once we know that, for a given solution of the spin-2 equation, the  $W^{\pm}_A$-components  are smooth at $I^-$, it follows immediately from 
\eq{spin2_constr3}, \eq{spin2_alg_V-}, \eq{spin2_evolution1B}-\eq{spin2_evolution2B} that the other components need to be smooth there as well.
 So our focus will be on an analysis of \eq{spin2_wave1}-\eq{spin2_wave2}  near $I^-$.

\subsection{Appearance of logarithmic terms}
We consider \eq{spin2_wave1}-\eq{spin2_wave2}.
Expanding $W^{\pm}_A$ in terms of $r$ one obtains \eq{sing_wave1}-\eq{sing_wave2}.
The crucial difference is that in this linearized case  there is no source term: The no-logs condition is a condition on the $n$th-order expansion
coefficient of the radiation field, and independent of all expansion coefficients of different orders.
The no-logs condition at a given order is thus  completely independent of lower order terms.
As a corollary of Lemma~\ref{lem_sing_wave} we obtain the following

\begin{proposition}
\label{prop_spin2}
Let $W_{ijkl}$ be a smooth solution of the massless spin-2 equation on the Minkowski background \eq{Mink_cylB} which is smooth  at $\scri^-$.
The data at $\scri^-$ are a tracefree, symmetric, tensor $V^+_{AB} |_{\scri^-}$ which admits a Hodge decomposition of the form $V^+_{AB} |_{\scri^-}= (\mcD_A\mcD_B\ul{V})_{\mathrm{tf}}+ \epsilon_{(A}{}^C\mcD_{B)}\mcD_C \ol{V}$.
%The solution  $W_{ijkl}$, including all of its transverse and radial derivatives, is smooth
The solution  $W_{ijkl}$ satisfies all  no-logs conditions 
 when approaching $I^-$ from $\scri^-$  if and only if the data are of the following form
\begin{align}
\ul{V} \sim \sum_{n=0}^{\infty}\ul{V} ^{(n)} r^n\,,\quad\ul{V} ^{(n)}=&\sum_{\ell=2}^{n-1}\sum_{m=-\ell}^{+\ell}\ul{V} ^{(n)}_{\ell m}Y_{\ell m}(\theta,\phi)
\,,
\label{spin2_nolog_data1}
\\
\ol{V} \sim \sum_{n=0}^{\infty}\ol{V} ^{(n)} r^n\,,\quad\ol{V} ^{(n)}=&\sum_{\ell=2}^{n-1}\sum_{m=-\ell}^{+\ell}\ol{V} ^{(n)}_{\ell m}Y_{\ell m}(\theta,\phi)
\,.
\label{spin2_nolog_data2}
\end{align}
In that case also all  no-logs conditions are fulfilled   when approaching $I^-$ from $I$, supposing that $I$ is smooth and that the
data for the transport equations on $I$ are induced by the limits of the corresponding fields on $I^-$.
\end{proposition}

\begin{remark}
\label{rem_cotton}
{\rm
A corresponding analysis which analyzes the appearance of logarithmic terms starting from an ordinary Cauchy problem for the spin-2 equation
has been carried out in \cite{kroon}, cf.\ \cite{F_i0, F_spin}.
It is shown there that no logs arise if and only  if  all symmetrized trace-free derivatives of the linearized Cotton tensor of the induced metric on the initial surface vanish at spatial infinity, to which \eq{spin2_nolog_data1}-\eq{spin2_nolog_data2} is the analog at $I^-$.
}
\end{remark}

\section{Constant (ADM) mass aspect and vanishing dual mass aspect}
\label{sec_const_mass}
\label{section8}

Let us  compare Proposition~\ref{prop_spin2}  with the full non-linear case:  In an asymptotically Minkowski-like conformal Gauss gauge at each order which admits a smooth $I^-$ 
data of the form \eq{spin2_nolog_data1}-\eq{spin2_nolog_data2} provide the maximal  part of the radiation field which one can freely prescribe.
Spherical harmonics with $\ell\geq n$ which may appear in the harmonic decomposition of $\ul V^{(n)}$ and $\ol V^{(n)}$
are determined by the no-logs conditions \eq{no_log_rad1}-\eq{no_log_rad2}.

While we have shown that a radiation field which has a trivial expansion at $I^-$ does not produce log-terms at any order,
 it is not clear at all how necessary-and-sufficient conditions look like, even on the level of formal expansions we are interested in.
This is due to the fact that the no-logs conditions are not decoupled equations for the expansion coefficients
$\ul V^{(n)}$ and $\ol V^{(n)}$ as they are in the spin-2 case considered above.

If the $n$th-order expansion coefficient of the radiation field is the first non-trivial one, no log-terms are produced if and only if 
its Hodge-decomposition scalars have only $0\leq k\leq n-1$ spherical harmonics in their decomposition.
However, when passing to the $(n+1)$st-order, the $n$th-order expansion coefficient appears in the source, and the source is not allowed to
have $0\leq k\leq n$-spherical harmonics.
One should therefore expect additional restrictions arising from this as compared to the spin-2 case, where this cannot happen,
unless there are some magic cancellations or  Laplacian-like operators which project out  all the problematic terms.
In fact, it is shown in \cite{kroon2},
where smoothness is analyzed from an ordinary Cauchy problem, that the ``linear spin-2 condition'', i.e.\ the condition on the Cotton tensor mentioned in Remark~\ref{rem_cotton}, is not sufficient to exclude logarithmic terms. 
%So, supposing  that this is not a gauge artifact,
One therefore \emph{must}  expect  non-trivial additional restrictions
as compared to the spin-2 case.

To make computations feasible we restrict attention henceforth to  a more restricted class of initial data where
\begin{equation}
M=\mathrm{const.}\ne 0 \quad \text{and}  \quad  N=0
\,.
\label{constant_mass}
\end{equation}
(By  definition of $N$ as the Laplacian of a certain function it is not allowed to have $\ell=0$-spherical harmonics, so $N=\mathrm{const.}$ implies  $N=0$ in our setting.)
As before we assume an asymptotically Minkowski-like conformal Gauss gauge at each order.
All the results of Section~\ref{section_aMlcGg} apply, now with \eq{constant_mass}.
A crucial  advantage of \eq{constant_mass}  is that some of the transport equations on $I$ can  be  solved explicitly.
In particular  the Bianchi equation on $I$ can be solved.
From Section~\ref{sec_solve_Bianchi_I} we deduce that in our current setting
\begin{equation}
W^{\pm}_A|_I=0\,, \quad V^{\pm}_{AB}|_I =0\,, \quad W_{0101}|_I=2 M\,, \quad W_{01AB}|_I =0
\,.
\end{equation}
We will analyze this problem approaching $I^-$ from $\scri^-$, which is more natural in our setting. Of course one could do a similar analysis from $I$.
However, it turns out the the equations which arise on $\scri^-$ are somewhat more manageable.
For instance on $I$ a decomposition in  spherical harmonics comes in at a very  early stage when solving the Bianchi equation, while on $\scri^-$ it
needs to be taken into account only when the actual no-logs condition is derived.

\subsection{Second-order transverse derivatives on $\scri^-$}
\label{sec_constat_mass_ex}

In order to get some insights concerning the expected  additional  ``non-linear'' restrictions%
\footnote{
In fact similar restrictions should be expected if one linearizes around e.g.\ the Schwarzschild metric, i.e.\ the additional restrictions
seem  mainly be  due to a (dual)  mass rather than non-linearities.
}
 let us consider the \emph{second-order} transverse derivatives on $\scri^-$ first, before we analyze the general case.
We want to search for additional  restrictions on $\Xi^{(5)}_{AB}$, in addition to
\eq{restrictions_on_Xi5}, 
%
%\begin{equation*}
%\Xi^{(5)}_{ AB} = (\mcD_A\mcD_B\ul \Xi^{(5)})_{\mathrm{tf}} + \epsilon_{(A}{}^C\mcD_{B)}\mcD_C\ol \Xi^{(5)}
%\,,
%\quad \text{for some   $\ell=2$ spherical harmonics $\ul v^{(5)}$ and $\ol v^{(5)}$,}
%\end{equation*}
%
which ensure that no log-terms arise in the expansions of the second-order transverse derivatives.
%To get some insight whether a non-trivial $\Xi^{(5)}_{AB}$ can produce log terms in higher order, let us compute the second-order transverse derivatives.
Since we already know how the contribution from $\Xi^{(6)}_{AB}$ looks like, and we are only interested in the source, i.e.\ the right-hand side of
\eq{some_no_logs_cond},
we also assume $\Xi^{(6)}_{AB}=0$ for the computation.

The second-order transverse derivative of \eq{evolutionW6b} reads (using \eq{connection_scri1}-\eq{connection_scri9} and \eq{frame_1storder1}-\eq{frame_1storder14}),
\begin{align*}
r^5\partial_{r} (r^{-4}\partial^2_{\tau}V^-_{AB} )|_{\scri^-}
=&
- \partial^2_{\tau}e^{\mu}{}_1\partial_{\mu}V^-_{AB} 
+(\mcD_{(A} \partial^2_{\tau}W^+_{B)})_{\mathrm{tf}}
\\
&
+(\partial^2_{\tau} e^{\mu}{}_{(A}\partial_{(\mu} W^+_{B)}-\partial^2_{\tau}\widehat \Gamma_{(A}{}^C{}_{B)}W^+_C)_{\mathrm{tf}}
\\
&
+\partial^2_{\tau} (\widehat\Gamma_1{}^0{}_0+ 
2 \widehat\Gamma_1{}^1{}_{0} 
  + \widehat\Gamma_C{}^C{}_0 
- \widehat\Gamma_C{}^C{}_1 ) V^-_{AB} 
\\
&
+ \Big[ -\partial^2_{\tau}( \widehat\Gamma_C{}^0{}_{(A} + \widehat\Gamma_C{}^1{}_{(A}- 2\widehat\Gamma_{1C(A})V^-_{B)}{}^C
+ \frac{3}{2}\partial^2_{\tau}( \widehat\Gamma_C{}^0{}_{(A}  -  \widehat\Gamma_C{}^1{}_{(A} )U_{B)}{}^C
\\
&
+\partial^2_{\tau}\Big(
- 2\widehat\Gamma_1{}^0{}_{(A} 
- 2 \widehat\Gamma_{(A}{}^0{}_0
- \widehat\Gamma_{(A}{}^0{}_{1}
+2 \widehat\Gamma_1{}^1{}_{(A} \Big)W^+_{B)}
\Big]_{\mathrm{tf}}
+\mathfrak{O}(r^{5})
\,.
\end{align*}
The   first-order transverse derivative of \eq{evolutionW4b} gives with \eq{trans_Weyl_spec1}-\eq{trans_Weyl_spec5}
(alternatively, one could  compute $(\partial^2_{\tau}W^+_A)^{(4)}$ from \eq{evolution11})
\begin{align*}
(\partial^2_{\tau}W^+_A)^{(4)}
 =&
( -\mcD^B \partial_{\tau}V^-_{AB}
-\frac{1}{2}\mcD_A\partial_{\tau}W_{0101}
-\frac{1}{2}\mcD^B\partial_{\tau}W_{01AB}
-\partial_{\tau}W^+_A
)^{(4)}
+\frac{3}{2}M v^{(5)}_A 
\\
=&
\frac{3}{4}M (\Delta_s+15) v^{(5)}_{A}
\,,
\end{align*}
where we used that by \eq{V-_eqn_scri_exp}
\begin{equation*}
(\partial_{\tau}V^-_{AB})^{(4)}
=
-\frac{3}{2}M \Big(
 (\mcD_{(A}  v^{(5)}_{B)})_{\mathrm{tf}}
+5\Xi^{(5)}_{AB}
\Big)
\,.
\end{equation*}
%
%Using  \eq{Weyl_frame_scri1}-\eq{Weyl_frame_scri6} and \eq{evolution1}-\eq{evolution7}
Using \eq{evolution_2nd_trans_1}-\eq{evolution_2nd_trans_19} to determine the second-order
transverse derivatives of connection and frame coefficients, a calculation reveals that
%
%\begin{eqnarray}
%\partial^2_{\tau}e^{\mu}{}_1|_{\scri^-} &=&0
%\,,
%\label{evolution_2nd_trans_1}
%\\
%\partial^2_{\tau}e^{\tau}{}_A|_{\scri^-}  &=& -\frac{1}{r}v_A 
%\,,
%\\
%\partial^2_{\tau}e^{r}{}_A|_{\scri^-}  &=& 
%-\frac{1}{2}v_A 
%\,,
%\\
%\partial^2_{\tau}e^{\mathring A}{}_A|_{\scri^-}  &=&
%\frac{1}{2}\Big(\partial_r-\frac{1}{r}\Big)\Xi_A{}^B e^{\mathring A}{}_B
%\,,
%\label{evolution_2nd_trans_4}
%\\
%\partial^2_{\tau}\widehat\Gamma_{1}{}^0{}_1|_{\scri^-} 
% &=&
%  -4r W_{0101}
%\,,
%\\
%\partial^2_{\tau}\widehat\Gamma_{1}{}^0{}_A|_{\scri^-} 
% &=&
%  -4r W_{010A} + 2r W_{011A}
%\,,
%\\
%\partial^2_{\tau}\widehat\Gamma_{1}{}^b{}_c |_{\scri^-} 
% &=&
%2r W^b{}_{c01}-2r\delta^b{}_{c}W_{0101}  
%\,,
%\\
%\partial^2_{\tau}\widehat\Gamma_{A}{}^0{}_1 |_{\scri^-} 
% &=&
%  -4r W_{010A}
%\,,
%\\
%\partial^2_{\tau}\widehat\Gamma_{A}{}^0{}_B|_{\scri^-} 
% &=&
%  -4r W_{0A0B} - 2r W_{0AB1}
%\,,
%\\
%\partial^2_{\tau}\widehat\Gamma_{A}{}^1{}_1|_{\scri^-} 
% &=&
%-  \frac{1}{2r}v_A
%-2rW_{010A}   
%\,,
%\\
%\partial^2_{\tau}\widehat\Gamma_{A}{}^1{}_B|_{\scri^-} 
% &=&
%-\frac{1}{2}\Big(\partial_r-\frac{1}{r}\Big)\Xi_{AB}
%   -2r W_{0AB1}
%\,,
%\\
%\partial^2_{\tau}\widehat\Gamma_{A}{}^B{}_C|_{\scri^-} 
% &=&
%-\frac{1}{2r} \delta^B{}_ Cv_A
%+\frac{1}{2} \mathring\Gamma_D{}^B{}_C\Big(\partial_r-\frac{1}{r}\Big)\Xi_A{}^D  -r\delta^B{}_{C}(W^-_A+W^+_A)  -4r \eta_{A[B}W_{C]110}
%\,.
%\label{evolution_2nd_trans_19}
%\end{eqnarray}
%
%
\begin{align*}
(\partial_{r}(r^{-4}\partial^2_{\tau}V^-_{AB} ))^{(-1)}
=&
(\mcD_{(A} (\partial^2_{\tau}W^+_{B)})^{(4)})_{\mathrm{tf}}
-  \frac{15}{4}M\Big(
(\Delta_s -4)\mcD_{( A} v^{(5)}_{ B)}
-2\mcD_{ A}\mcD_{ B} \mcD^{ C} v^{(5)}_{ C}
\Big)_{\mathrm{tf}}
+ 15\Xi^{(5)}_{AB} M
\\
=&
- 3 M(\Delta_s -8)(\mcD_{( A} v^{(5)}_{ B)})_{\mathrm{tf}}
+ \frac{15}{2}M(\mcD_{ A}\mcD_{ B} \mcD^{ C} v^{(5)}_{ C})_{\mathrm{tf}}
+ 15\Xi^{(5)}_{AB} M
\,.
\end{align*}
We deduce that for this order to be smooth at $I^-$ the following smoothness condition needs to be satisfied
(in addition to the requirement on $\Xi^{(6)}_{AB}$ to arise from a linear combination of $2\leq \ell\leq 3$ spherical harmonics),
\begin{equation}
(\Delta_s -8)(\mcD_{( A} v^{(5)}_{ B)})_{\mathrm{tf}}
- \frac{5}{2}(\mcD_{ A}\mcD_{ B} \mcD^{ C} v^{(5)}_{ C})_{\mathrm{tf}}
-5\Xi^{(5)}_{AB} =0
\,.
\end{equation}
We compute the divergence
\begin{equation}
(\Delta_s -5)(\Delta_s+1) v^{(5)}_{ A}
- \frac{5}{2}\mcD_{ A}(\Delta_s+2)\mcD^{ B} v^{(5)}_{ B}
-10v^{(5)}_{A} =0
\,.
\end{equation}
Divergence and curl of this equation read
\begin{align}
(3\Delta_s\Delta_s + 14 \Delta_s + 36) \mcD^Av^{(5)}_{ A}
=&0
\quad\Longrightarrow \quad \mcD^Av^{(5)}_{ A} =0
\,,
\\
(\Delta_s\Delta_s-2\Delta_s -18)\epsilon^{AB}\mcD_{[A} v^{(5)}_{ B]}
 =&0
\quad\Longrightarrow \quad \epsilon^{AB}\mcD_{[A} v^{(5)}_{ B]}=0
\,,
\end{align}
and we deduce the smoothness condition
\begin{equation}
\Xi^{(5)}_{AB} =0
\,.
\end{equation}
While a non-trivial $\Xi^{(5)}_{AB}=(\mcD_A\mcD_B\ul{\Xi}^{(5)})_{\mathrm{tf}}+ \epsilon_{(A}{}^C\mcD_{B)}\mcD_C \ol{\Xi}^{(5)}$
does not produce logarithmic terms in the expansion of $\partial_{\tau}V^-_{AB}|_{\scri^-}$  as long as
$\ul \Xi^{(5)}$ and $\ol \Xi^{(5)}$ are linear combinations of $\ell=2$-spherical  harmonics, it does produce log terms in the next order, namely for
$\partial^2_{\tau}V^-_{AB}|_{\scri^-}$, at least supposing that the mass aspect $M$  is constant and non-zero,  and that the dual mass aspect $N$ vanishes.

\subsection{A necessary condition for the non-appearance of log terms}

We want to generalize the above computation to any order.
For this we will  extend the computations of Section~\ref{sec_no-logs_V-scri_gen}
to determine  the $\Xi^{(m_0+2)}_{AB}$-contribution to the no-logs condition \eq{no_log_rad1}-\eq{no_log_rad2},
\begin{equation}
\prod_{\ell =0}^{m_0}\Big(\Delta_s+\ell (\ell+1)\Big)v^{(m_0+3)}
=   O_{\Xi}(m_0+2)
\,, \quad v^{(m_0+3)}\in \{\mcD^Av^{(m_0+3)}_A,\epsilon^{AB}\mcD_{A}v^{(m_0+3)}_{B}\}
\,.
\end{equation}
For this we consider a scenario where the $\Xi^{(k)}_{AB}$'s vanishes for $3\leq k\leq m_0+1$. 
In an asymptotically Minkowski-like conformal Gauss gauge  at  each order we then  compute the $\Xi^{(m_0+2)}_{AB}$-contribution, which, somewhat surprisingly, can be done explicitly.

More precisely, let us  assume that \eq{constant_mass} holds and consider initial data of the form 
\begin{equation}
\label{assumption_data}
\Xi_{AB} = \Xi_{AB}^{m_0+2} r^{m_0+2} +  \Xi_{AB}^{m_0+3} r^{m_0+3} + \mathfrak{O}(r^{m_0+4})\,, \quad m_0\geq 3
\,. 
\end{equation}
%
%The previous example analyzes the appearance of logarithmic terms by approaching $I^-$ from $\scri^-$.
%It turn out that the general case is somewhat easier to handle when approaching $I^-$ form $I$. One reason for this is that
%$\Xi_{AB}^{(m)}=0$ for $m< m_0+2$, whence we can use the decay computed in Lemma~\ref{lemma_estimates2} for the lower-order radial derivatives on $I$.
It follows from the results in Section~\ref{sec_exp_dec_rad_field}
that in order for $I^-$ to be smooth
the two functions appearing in the Hodge decomposition of $ \Xi_{AB}^{(m_0+2)}= (\mcD_{A}\mcD_{B}\ul {\Xi}^{(m_0+2)})_{\mathrm{tf}}
+ \epsilon_{(A}{}^C \mcD_{B)}\mcD_{C}\ol{ \Xi}^{(m_0+2)}$ need to be linear combinations of $1\leq \ell\leq m_0-1$-spherical harmonics,
\begin{equation}
\ul{\Xi}^{(m_0+2)}=\sum_{\ell=1}^{m_0-1}\sum_{m=-\ell}^{+\ell}\ul{\Xi}^{(m_0+2)}_{\ell m}Y_{\ell m}(\theta,\phi)
\,, \quad
\ol{\Xi}^{(m_0+2)}=\sum_{\ell=1}^{m_0-1}\sum_{m=-\ell}^{+\ell}\ol{\Xi}^{(m_0+2)}_{\ell m}Y_{\ell m}(\theta,\phi)
\,.
\label{assumption_data2B}
\end{equation}
If this is the case, no log terms arise up to and including the $(m_0-2)$-nd order transverse derivatives of the rescaled Weyl tensor.
Our goal is to  compute the $(m_0-1)$-st order transverse derivatives.
We will see that a non-trivial $\Xi^{(m_0+2)}_{AB}$-contribution   does produce non-trivial   $2\leq \ell\leq m_0-1$-spherical harmonics  in the source term of the
no-logs condition of order $m_0+1$.
% supposing that $\langle \ol \Xi^{(m_0+3)}, Y_{\ell m}\rangle\ne0$ or  $\langle \ul \Xi^{(m_0+3)}, Y_{\ell m}\rangle\ne0$ for some $-\ell \leq m\leq \ell$.
We will thus be able to deduce that in a smooth setting the radiation field \emph{necessarily} needs to vanish at all orders at $I^-$.
%Nevertheless, it provides some insights into the structure  of the equations.

\subsubsection{First-order radial derivatives on $I$}
\label{sec_crucial_behavior}

Recall the expressions \eq{frameI_1}-\eq{Schouten_I} of the $0th$-order radial derivatives on $I$, 
\begin{align}
e^{\tau}{}_{1}|_{I} =&- \tau 
\,, 
\quad
e^{r}{}_{1}|_{I} =0
\,,
\quad
e^{\mathring A}{}_{1}|_{I} =0
\,,
\quad
e^{\tau}{}_{A}|_{I}  = 0
\,,
\quad
e^{r}{}_{A}|_{I}  =   0
\,,
\quad
e^{\mathring A}{}_{A}|_{I}  =   \mathring e^{\mathring A}{}_A
\,,
\label{frameI_1B}
\\
\widehat\Gamma_1{}^i{}_j |_{I}  =& \delta^i{}_j
\,,
\quad
\widehat\Gamma_a{}^b{}_0 |_{I}  =0
\,,
\quad
 \widehat\Gamma_A{}^1{}_1   |_{I}=0
\,,
\quad
\widehat\Gamma_A{}^B{}_1   |_{I}= \delta^B{}_A
\,,
\quad
\widehat\Gamma_A{}^C{}_B  |_{I} =    \mathring\Gamma_A{}^C{}_B 
\,,
\\
\widehat L_{ ij}|_{I} 
=& 0
\,.
\label{Schouten_IB}
\end{align}
and that in our current setting we have
\begin{equation}
W^{\pm}_A|_I=0\,, \quad V_{AB}^{\pm}|_I =0\,, \quad W_{0101}|_I= 2M\,, \quad W_{01AB}|_I =0
\,.
\end{equation}
From \eq{evolution1}-\eq{evolution7} we compute the first-order radial derivatives on $I$ for Schouten tensor, connection and frame coefficients
(recall \eq{frameI_1}-\eq{Schouten_I}, $\Theta=r(1-\tau^2)$, $b_0=-2r\tau$, $b_1=2r$, $b_A=0$).
%We have
%%
%\begin{align}
%\partial_{\tau}\partial_r\widehat L_{10} |_I
%&= -2 M
%\,,
%\\
%\partial_{\tau}\partial_r\widehat L_{A0} |_I
%&= 0
%\,,
%\\
%\partial_{\tau}\partial_r\widehat L_{11} |_I
%&= 2\tau  M
%\,,
%\\
%\partial_{\tau}\partial_r\widehat L_{1A} |_I
%&=0
%\,,
%\\
%\partial_{\tau}\partial_r\widehat L_{A1} |_I
%&= 0
%\,,
%\\
%\partial_{\tau}\partial_r\widehat L_{AB} |_I
%&= -\tau  M \eta_{AB}
%\,,
%\end{align}
%%
%which yields with trivial data as computed from \eq{schouten_scri3}-\eq{schouten_scri6}
For trivial data as computed from \eq{schouten_scri3}-\eq{schouten_scri6} we find for the Schouten tensor
\begin{align}
\partial_r\widehat L_{10} |_I
&= -4 M(1+\tau)
\,,
\\
\partial_r\widehat L_{A0} |_I
&= 0
\,,
\\
\partial_r\widehat L_{11} |_I
&= -2M(1-\tau^2)
\,,
\\
\partial_r\widehat L_{1A} |_I
&=0
\,,
\\
\partial_r\widehat L_{A1} |_I
&= 0
\,,
\\
\partial_r\widehat L_{AB} |_I
&= M(1-\tau^2) \eta_{AB} 
\,.
\end{align}
%
%Moreover,
%%
%\begin{align}
%\partial_{\tau}\partial_r\widehat\Gamma_{1}{}^0{}_1 |_I
% &=
%- 2M (1-\tau^2) 
%\,,
%\\
%\partial_{\tau}\partial_r\widehat\Gamma_{1}{}^0{}_A |_I
% &=
%0
%\,,
%\\
%\partial_{\tau}\partial_r\widehat\Gamma_{A}{}^0{}_1 |_I
% &=
%0
%\,,
%\\
%\partial_{\tau}\partial_r\widehat\Gamma_{A}{}^0{}_B |_I
% &=
%M (1-\tau^2) \eta_{AB}
%\,,
%\\
%\partial_{\tau}\partial_r\widehat\Gamma_{1}{}^1{}_1|_I
% &=
%- \partial_r\widehat\Gamma_{1}{}^1{}_{0} -2 M(1+\tau)
%\,,
%\\
%\partial_{\tau}\partial_r\widehat\Gamma_{1}{}^1{}_A |_I
% &=
%\partial_r\widehat\Gamma_{1}{}^0{}_{A}   
%\,,
%\\
%\partial_{\tau}\partial_r\widehat\Gamma_{A}{}^1{}_1 |_I
% &=
%-  \partial_r\widehat\Gamma_{A}{}^1{}_{0}   
%\,,
%\\
%\partial_{\tau}\partial_r\widehat\Gamma_{A}{}^1{}_B |_I
% &=
%\partial_r\widehat\Gamma_{A}{}^0{}_{B} 
%\,,
%\\
%\partial_{\tau}\partial_r\widehat\Gamma_{1}{}^A{}_B|_I 
% &=
%-  \delta^A{}_B\partial_r\widehat\Gamma_{1}{}^1{}_{0}
%-  \mathring \Gamma_C{}^A{}_B\partial_r\widehat\Gamma_{1}{}^C{}_{0}
%  -2 M(1+\tau)\delta^A{}_{B}
%\,,
%\\
%\partial_{\tau}\partial_r\widehat\Gamma_{A}{}^B{}_C |_I
% &=
%-  \delta^B{}_C\partial_r\widehat\Gamma_{A}{}^0{}_{1} 
%-  \mathring\Gamma_D{}^B{}_C\partial_r\widehat\Gamma_{A}{}^D{}_{0} 
%\,.
%\end{align}
%
With trivial data as induced by \eq{connection_scri1}-\eq{connection_scri9} 
we  end up with the following expressions for the connection coefficients,
\begin{align}
\partial_r\widehat\Gamma_{1}{}^0{}_1 |_I
 &=
-4M\Big( (1+\tau)^2- \frac{1}{3} (1+\tau)^3\Big)
\,,
\\
\partial_r\widehat\Gamma_{1}{}^0{}_A |_I
 &=
0
\,,
\\
\partial_r\widehat\Gamma_{A}{}^0{}_1 |_I
 &=
0
\,,
\\
\partial_r\widehat\Gamma_{A}{}^0{}_B |_I
 &=
2M\Big( (1+\tau)^2-\frac{1}{3}(1+\tau)^3\Big)\eta_{AB}
\,,
\\
\partial_r\widehat\Gamma_{1}{}^1{}_1|_I
 &=
 -2 M\Big((1+\tau)^2 - \frac{2}{3} (1+\tau)^3+  \frac{1}{6} (1+\tau)^4\Big)
\,,
\\
\partial_r\widehat\Gamma_{1}{}^1{}_A |_I
 &=
0
\,,
\\
\partial_r\widehat\Gamma_{A}{}^1{}_1 |_I
 &=
0
\,,
\\
\partial_r\widehat\Gamma_{A}{}^1{}_B |_I
 &=
\frac{2}{3}M\Big((1+\tau)^3-\frac{1}{4}(1+\tau)^4\Big)\eta_{AB}
\,,
\\
\partial_r(\widehat\Gamma_{1}{}^A{}_B)_{\mathrm{tf}}|_I 
 &=
0
\,,
\\
\partial_r\widehat\Gamma_{A}{}^B{}_C |_I
 &=
- \frac{2}{3}M\Big( (1+\tau)^3-\frac{1}{4}(1+\tau)^4\Big) \mathring\Gamma_A{}^B{}_C
\,.
\end{align}
Finally, we have
\begin{align}
\partial_re^{\tau}{}_1|_I&= 2M\Big((1+\tau)^3 - \frac{5}{6} (1+\tau)^4+  \frac{1}{6} (1+\tau)^5\Big)
\,,
\\
\partial_re^{r}{}_1|_I&= 1
\,,
\\
\partial_re^{\mathring A}{}_1|_I&=
0
\,,
\\
\partial_re^{\tau}{}_A|_I&=
0
\,,
\\
\partial_re^{r}{}_A|_I&=
0
\,,
\\
\partial_re^{\mathring A}{}_A|_I&=
 -\frac{2}{3}M\Big( (1+\tau)^3-\frac{1}{4}(1+\tau)^4\Big)e^{\mathring A}{}_A
\,.
\end{align}
We will also need some second-order radial derivatives,
\begin{align}
\partial^2_re^{r}{}_1|_I&=
\frac{8}{3}M\Big( (1+\tau)^3- \frac{1}{4} (1+\tau)^4\Big)
\,,
\\
\partial^2_re^{r}{}_A|_I&=0
\,.
\end{align}
These results will be crucial for the computations on $\scri^-$ because it provides information concerning the decay of connection and  frame coefficients,
in particular we e.g.\ find that 
$$\partial^n_{\tau}\Gamma_i{}^j{}_k|_{\scri^-}=O(r^2) \quad \text{for all $n\geq 5$.}$$
Because of this only a bounded  number of terms will contribute to the critical logarithmic terms producing order in the Bianchi equation for transverse derivatives
of any order.

\subsubsection{Some expansion coefficients at $I^-$}
\label{constM_higher_orders}

In analogy to the proof of Lemma~\ref{lemma_estimates1}, replacing $r^{\infty}$ there by $m_0+m_1$, where $m_1$ depends on the first-order contribution,
 cf.\ \eq{frame_1storder1}-\eq{V-_eqn_scri_exp},  one shows that for $1\leq k\leq m_0+1$ ($1\leq k\leq m_0$ for $\partial^k_{\tau}V^{-}_{AB} |_{\scri^-}$ since the order $m_0+1$ may have log terms)
\begin{align}
\partial^{k}_{\tau}\widehat\Gamma_{1}{}^1{}_1|_{\scri^-}
 =&
\mathcal{P}^{k-1}+ \mathfrak{O}(r^{m_0+1})
\,,
\label{behavior_nec1}
\\
\partial^{k}_{\tau}\widehat\Gamma_{1}{}^0{}_1 |_{\scri^-}
 =&
 \mathcal{P}^{k-1}+  \mathfrak{O}(r^{m_0+1})
\,,
\\
\partial^{k}_{\tau}\widehat\Gamma_{A}{}^1{}_1 |_{\scri^-}
 =&
 \mathcal{P}^{k}+  \mathfrak{O}(r^{m_0+1})
\,,
\\
\partial^{k}_{\tau}\widehat\Gamma_{A}{}^0{}_1 |_{\scri^-}
 =&
 \mathcal{P}^{k}+  \mathfrak{O}(r^{m_0+1})
\,,
\\
\partial^{k}_{\tau}(\widehat\Gamma_{1}{}^0{}_A +\widehat\Gamma_{1}{}^1{}_A )|_{\scri^-}
 =&
 \mathcal{P}^{k-2}+ \mathfrak{O}(r^{m_0+1})
\,,
\\
\partial^{k}_{\tau}(\widehat\Gamma_{1}{}^0{}_A -\widehat\Gamma_{1}{}^1{}_A )|_{\scri^-}
 =&
 \mathcal{P}^{k}+ \mathfrak{O}(r^{m_0+1})
\,,
\\
\partial^{k}_{\tau}(\widehat\Gamma_{A}{}^0{}_B +\widehat\Gamma_{A}{}^1{}_B )|_{\scri^-}
 =&
 \mathcal{P}^{k-1}+  \mathfrak{O}(r^{m_0+1})
\,,
\\
\partial^{k}_{\tau}(\widehat\Gamma_{A}{}^0{}_{B} -\widehat\Gamma_{A}{}^1{}_{B} )|_{\scri^-}
 =&
 \mathcal{P}^{k+1}+  \mathfrak{O}(r^{m_0+1})
\,,
\\
\partial^{k}_{\tau}\widehat\Gamma_{1}{}^A{}_B |_{\scri^-}
 =&
 \mathcal{P}^{k-1}+  \mathfrak{O}(r^{m_0+1})
\,,
\\
\partial^{k}_{\tau}\widehat\Gamma_{A}{}^B{}_C |_{\scri^-}
 =&
\mathcal{P}^{k}+  \mathfrak{O}(r^{m_0+1})
\,,
\\
\partial^{k}_{\tau}\widehat L_{1i}|_{\scri^-} 
=&
\mathcal{P}^{k}+  \mathfrak{O}(r^{m_0+1})
\,,
\\
\partial^{k}_{\tau}(\widehat L_{A0} +\widehat L_{A1})|_{\scri^-}
=&
\mathcal{P}^{k+1}+ \mathfrak{O}(r^{m_0+1})
\,,
\\
\partial^{k}_{\tau}(\widehat L_{A0}-\widehat L_{A1}) |_{\scri^-}
=&
\mathcal{P}^{k}+  \mathfrak{O}(r^{m_0+1})
\,,
\\
\partial^{k}_{\tau}\widehat L_{AB} |_{\scri^-}
=&
 \mathcal{P}^{k+1}+  \mathfrak{O}(r^{m_0+1})
\,,
\\
\partial^k_{\tau}e^{\tau}{}_1|_{\scri^-}=&-\delta^k{}_1 + 
\mathcal{P}^{ k-2} + O(r^{m_0+1})
\,,
\\
\partial^{k}_{\tau}e^{\mathring A}{}_1|_{\scri^-}=&   \mathcal{P}^{k-1}+ \mathfrak{O}(r^{m_0+1})
\,,
\\
\partial^{k}_{\tau}e^{\tau}{}_A |_{\scri^-}=& \mathcal{P}^{k-1}+  \mathfrak{O}(r^{m_0+1})
\,,
\\
\partial^{k}_{\tau}e^{\mathring A}{}_A|_{\scri^-}=&   \mathcal{P}^{k}+  \mathfrak{O}(r^{m_0+1})
\,,
\\
\partial^{k}_{\tau}e^{r}{}_1|_{\scri^-}=&   \mathcal{P}^{k-1}+ \mathfrak{O}(r^{m_0+2})
\,,
\\
\partial^{k}_{\tau}e^{r}{}_A|_{\scri^-}=&   \mathcal{P}^{k}+  \mathfrak{O}(r^{m_0+2})
\,,
\\
\partial^k_{\tau}U_{AB} |_{\scri^-}=& \mathcal{P}^{ k} + \mathfrak{O}(r^{m_0})
\,,
\label{behavior_nec21}
\\
\partial^k_{\tau}W^{\pm}_{A} |_{\scri^-}=& \mathcal{P}^{ k\pm1} +  \mathfrak{O}(r^{m_0})
\,,
\\
\partial^k_{\tau}V^{\pm}_{AB} |_{\scri^-}=& \mathcal{P}^{ k\mp 2} + \mathfrak{O}(r^{m_0})
\,.
\label{behavior_nec23}
\end{align}
\begin{remark}
{\rm
To obtain the error term for  $\partial^{k}_{\tau}e^{r}{}_a|_{\scri^-}$ one uses that $\partial_{\tau}^n\Gamma_a{}^b{}_0|_{\scri^-}=O(r)$
and $\partial_{\tau}^ne^r{}_a|_{\scri^-}=O(r)$ for all $n$ by \eq{frameI_1B}-\eq{Schouten_IB}.
}
\end{remark}

%\subsection{Better equations}

Recall the notation introduced in \eq{dfn_deriv}. 
We further introduce the notation
$$
V^{\pm}_A:=\mcD^BV^{\pm}_{AB}
\,.
$$
As a consequence of  Lemma~\ref{lemma_estimates1}
% \eq{behavior_nec21}-\eq{behavior_nec23}
 and 
\eq{double_exp_W0101}-\eq{recursion4}
%\eq{double_exp_W0101}-\eq{double_exp_V-} 
we  have
%\tim{we know the $\Xi^{(m_0+2)}$-contribution}
%
\begin{lemma}
\label{lemma_recursive_data}
Let  $n\geq1$, then
\begin{align*}
 V^+_{A}{}^{(m_0,n)} =& \frac{1}{2} (n-m_0-1) W^-_{A}{}^{(m_0,n)} 
\,,
\quad  n \leq m_0\,,
\\
V^-_{A} {}^{(m_0,n)}
=&
-\frac{1}{2(n-m_0+2)  }(\Delta_s+1)W^+_{A}{}^{(m_0,n)}
\,,
\quad  n+3 \leq m_0\,,
\end{align*}
and
\begin{align*}
nW_{0101} {}^{(m_0,n)}
 =& 
-\frac{1}{2}\mcD^A W^{+}_A{}^{(m_0,n-1)}
+\frac{1}{2}\mcD^AW^{-}_A{}^{(m_0,n-1)}
\,,
\quad  n+1 \leq m_0
\,,
\\
nW_{01AB}{}^{(m_0,n)}     =&
\mcD_{[A}  W^{+}_{B]}{}^{(m_0,n-1)}
+\mcD_{[A} W^{-}_{B]}{}^{(m_0,n-1)}
\,,
\quad  n+1 \leq m_0
\,,
\\
(m_0-n+1)W^-_{A}{}^{(m_0,n)} 
 =&
-\frac{1}{2n} \Big(\Delta_s+(m_0-n) (m_0-n+1)-1\Big)  W^-_{A}{}^{(m_0,n-1)}
\,,
\quad n \leq m_0
\,,
\\
(m_0-n+2)  V^+_{A}{}^{(m_0,n)} 
 =&
-\frac{1}{2n} \Big(\Delta_s+(m_0-n) (m_0-n+1)-1\Big)  V^+_{A}{}^{(m_0,n-1)}  
\,,
\quad n-1 \leq m_0\,,
\\
(m_0-n-1) W^+_{A} {}^{(m_0,n)}
 =&  -\frac{1}{2n}\Big(\Delta_s 
+(m_0-n) (m_0-n+1)-1\Big)W^+_{A}{}^{(m_0,n-1)}
\,,
\quad  n+2 \leq m_0\,,
\\
(m_0-n-2) V^-_{A} {}^{(m_0,n)}
=&
-\frac{1}{2n}\Big(\Delta_s 
+(m_0-n) (m_0-n+1)-1\Big)V^-_{A} {}^{(m_0,n-1)}
\,,
\quad  n+3 \leq m_0
\,.
\end{align*}
\end{lemma}
\begin{remark}
{\rm
Given $m_0$, for transverse derivatives of order $n$, with $n$ in a certain range, we have the same recursion relations at $I^-$ as for the spin-2 equation.
}
\end{remark}
\begin{proof}
For initial data $\Xi_{AB}$ which vanish asymptotically up to an including the order $m_0+1$  the
 error terms in \eq{double_exp_W0101}-\eq{recursion4} are the same as for a radiation field which vanishes asymptotically
at any order (in fact if $\Xi^{(m_0+2)}$ is the first non-trivial term this influences the expansion in $r$ of transverse derivatives of any order
only for $m\geq m_0+m_1$ as in \eq{behavior_nec1}-\eq{behavior_nec23}).
 It then  follows from Lemma~\ref{lemma_estimates1} that these
 error terms  are polynomials as in \eq{Weyl_poly1}-\eq{Weyl_poly3}. In particular if $m_0$ is sufficiently large as compared
to the number of transverse derivatives (as in the formulation of the lemma), the polynomials are in the kernel and  their contribution vanishes.
\qed
\end{proof}

Applying this formula recursively, the corresponding expansion coefficients can be expressed in terms of the initial data given at $\scri^-$:
\begin{corollary}
\label{cor_recursive_data}
The following relations hold at $I^-$,
\begin{align*}
W^-_{A}{}^{(m_0,m_0-k)} 
% =&
%-\frac{1}{2(m_0-k)(k+1)} \Big(\Delta_s+k (k+1)-1\Big)  W^-_{A}{}^{(m_0,m_0-k-1)}
%\\
 =&
\frac{(-1)^{m_0-k}k!}{2^{m_0-k}(m_0-k)!m_0!} \prod_{\ell=k}^{m_0-1}\Big(\Delta_s+\ell (\ell+1)-1\Big)  W^-_{A}{}^{(m_0,0)}
\,, \quad k\geq 0
\,,
\\
 V^+_{A}{}^{(m_0,m_0-k)} 
 %=&
%-\frac{1}{2(m_0-k)(k+2) } \Big(\Delta_s+k (k+1)-1\Big)  V^+_{A}{}^{(m_0,m_0-k-1)}  
%\\
 =&
\frac{(-1)^{m_0-k}(k+1)!}{2^{m_0-k}(m_0-k)!(m_0+1)! } \prod_{\ell=k}^{m_0-1} \Big(\Delta_s+\ell (\ell+1)-1\Big)  V^+_{A}{}^{(m_0,0)}  
\,, \quad k\geq -1
\,,
\\
 W^+_{A} {}^{(m_0,m_0-k)}
% =&  -\frac{1}{2(m_0-k)(k-1)}\Big(\Delta_s 
%+k (k+1)-1\Big)W^+_{A}{}^{(m_0,m_0-k-1)}
%\\
 =&  \frac{(-1)^{m_0-k}(k-2)!}{2^{m_0-k}(m_0-k)!(m_0-2)!} \prod_{\ell=k}^{m_0-1}\Big(\Delta_s 
+\ell (\ell+1)-1\Big)W^+_{A}{}^{(m_0,0)}
\,, \quad k\geq 2
\,,
\\
 V^-_{A} {}^{(m_0,m_0-k)}
%=&
%-\frac{1}{2(m_0-k)(k-2)}\Big(\Delta_s  +k(k+1)-1\Big)V^-_{A} {}^{(m_0,m_0-k-1)}
%\\
=&
-\frac{(-1)^{m_0-k}(k-3)!}{2^{m_0-k}(m_0-k)!(m_0-3)!}\prod_{\ell=k}^{m_0-1}\Big(\Delta_s 
+\ell (\ell+1)-1\Big)V^-_{A} {}^{(m_0,0)}
\,, \quad k\geq 3
\,.
\end{align*}
Recall that by \eq{constr_spec_gauge1}-\eq{constr_spec_gauge6} for $m_0\geq 3$
\begin{align*}
V^{+(m_0,0)}_{A}  =& -\frac{m_0(m_0+1)}{8}v^{(m_0+2)}_{ A }
\,,
\\
W^{-(m_0,0)}_A =&  \frac{m_0}{4}v^{(m_0+2)}_A
\,,
\\
W^{+(m_0,0)}_A=&\frac{1}{4(m_0-1)}\Big(( \Delta_s-1) v^{(m_0+2)}_{ A}
 -2 \mcD_A\mcD^{ B}v^{(m_0+2)}_{ B}\Big)
\,,
\\
V^{-(m_0,0)}_{A}=&  \frac{1}{8(m_0-1)(m_0-2)} (\Delta_s+1)\Big(( \Delta_s-1) v^{(m_0+2)}_{ A}
 -2 \mcD_A\mcD^{ B}v^{(m_0+2)}_{ B}\Big)
\,,
\end{align*}
so that the above expansion coefficients of the rescaled Weyl tensor can directly be expressed in terms of the asymptotic  initial data $\Xi_{AB}$.
\end{corollary}

\subsubsection{Higher-order transverse  derivatives on $\scri^-$}

We apply $\partial_{\tau}^n$ to  \eq{evolutionW6b}
and employ the formulas derived in Section~\ref{sec_crucial_behavior} \& \ref{constM_higher_orders}
which ensure that only very few terms contribute to the critical order $m_0+1$. For $n \leq m_0$ we find
\begin{align*}
&
\hspace{-3em} (r\partial_{r}-n-2) \partial_{\tau}^nV^-_{AB} |_{\scri^-}
\\
=&
- \begin{pmatrix} n \\ 3\end{pmatrix} \partial_{\tau}^3 e^{\tau}{}_1\partial_{\tau}^{n-2}  V^-_{AB} 
- \begin{pmatrix} n \\ 4\end{pmatrix} \partial_{\tau}^4 e^{\tau}{}_1\partial_{\tau}^{n-3}  V^-_{AB} 
- \begin{pmatrix} n \\ 5\end{pmatrix} \partial_{\tau}^5 e^{\tau}{}_1\partial_{\tau}^{n-4}  V^-_{AB} 
\\
&
- \begin{pmatrix} n \\ 3\end{pmatrix}\partial_{\tau}^{3}  e^{r}{}_1\partial_{r} \partial_{\tau}^{n-3}  V^-_{AB} 
- \begin{pmatrix} n \\ 4\end{pmatrix}\partial_{\tau}^{4}  e^{r}{}_1\partial_{r} \partial_{\tau}^{n-4}  V^-_{AB} 
\\
&+
 \begin{pmatrix} n \\ 2\end{pmatrix} \partial_{\tau}^{2}\Big(3\widehat\Gamma_1{}^0{}_0+ 2 \widehat\Gamma_1{}^1{}_{0}    + \frac{1}{2}\widehat\Gamma_C{}^C{}_0  - \frac{1}{2}\widehat\Gamma_C{}^C{}_1 \Big) \partial_{\tau}^{n-2}V^-_{AB} 
\\
&
+ \begin{pmatrix} n \\ 3\end{pmatrix} \partial_{\tau}^{3}\Big(3\widehat\Gamma_1{}^0{}_0+ 2 \widehat\Gamma_1{}^1{}_{0}    + \frac{1}{2}\widehat\Gamma_C{}^C{}_0  - \frac{1}{2}\widehat\Gamma_C{}^C{}_1 \Big) \partial_{\tau}^{n-3}V^-_{AB} 
\\
&
+ \begin{pmatrix} n \\ 4\end{pmatrix} \partial_{\tau}^{4}\Big(3\widehat\Gamma_1{}^0{}_0+ 2 \widehat\Gamma_1{}^1{}_{0}    + \frac{1}{2}\widehat\Gamma_C{}^C{}_0  - \frac{1}{2}\widehat\Gamma_C{}^C{}_1 \Big) \partial_{\tau}^{n-4}V^-_{AB} 
\\
&
+ \begin{pmatrix} n \\ 3\end{pmatrix}\Big( \partial_{\tau}^3e^{\mathring A}{}_{(A}\partial_{\mathring A}\partial_{\tau}^{n-3}W^+_{B)}-\partial_{\tau}3\Gamma_{(A}{}^C{}_{B)}\partial_{\tau}^{n-3}W^+_C \Big)_{\mathrm{tf}}
\\
&
+ \begin{pmatrix} n \\ 4\end{pmatrix}\Big( \partial_{\tau}^4e^{\mathring A}{}_{(A}\partial_{\mathring A}\partial_{\tau}^{n-4}W^+_{B)}-\partial_{\tau}^4\Gamma_{(A}{}^C{}_{B)}\partial_{\tau}^{n-4}W^+_C \Big)_{\mathrm{tf}}
\\
&
+( \mcD_{(A}\partial^n_{\tau}W^+_{B)})_{\mathrm{tf}}
+  3M\partial^n_{\tau}( \widehat\Gamma_{(A}{}^0{}_{B)} -  \widehat\Gamma_{(A}{}^1{}_{B)} )_{\mathrm{tf}}
+\mathcal{P}^{ n+1} + \mathfrak{O}(r^{m_0+2})
\\
=&
-12Mr\Big[\begin{pmatrix} n \\ 3\end{pmatrix}  +2 \begin{pmatrix} n \\ 2\end{pmatrix}\Big] \partial_{\tau}^{n-2}V^-_{AB} 
+8Mr\Big[5 \begin{pmatrix} n \\ 4\end{pmatrix} 
- (m_0-5)\begin{pmatrix} n \\ 3\end{pmatrix}  \Big] \partial_{\tau}^{n-3}V^-_{AB} 
\\
&
-4Mr\Big[
10\begin{pmatrix} n \\ 5\end{pmatrix} 
-(2m_0-7) \begin{pmatrix} n \\ 4\end{pmatrix} \Big] \partial_{\tau}^{n-4}V^-_{AB} 
\\
&
-4Mr \begin{pmatrix} n \\ 3\end{pmatrix}(\mcD_{(A}\partial_{\tau}^{n-3}W^+_{B)} )_{\mathrm{tf}}
+4Mr  \begin{pmatrix} n \\ 4\end{pmatrix}( \mcD_{(A}\partial_{\tau}^{n-4}W^+_{B)})_{\mathrm{tf}}
\\
&
+( \mcD_{(A}\partial^n_{\tau}W^+_{B)})_{\mathrm{tf}}
+  3M\partial^n_{\tau}( \widehat\Gamma_{(A}{}^0{}_{B)} -  \widehat\Gamma_{(A}{}^1{}_{B)} )_{\mathrm{tf}}
+\mathcal{P}^{ n+1} + \mathfrak{O}(r^{m_0+2})
\,.
\end{align*}
A similar computation when  $\partial_{\tau}^n$ is applied to \eq{evolution11}
yields for 
$ n \leq m_0$,
%\tim{different equation because...}
%$n\leq m_0-2$,
%
\begin{align*}
2\partial_{\tau}^{n+1}W^+_{A} |_{\scri^-}
 =&
- \begin{pmatrix} n \\ 3\end{pmatrix} \partial_{\tau}^3e^{\tau}{}_1\partial_{\tau}^{n-2} W^+_{A} 
- \begin{pmatrix} n \\ 4\end{pmatrix} \partial_{\tau}^4e^{\tau}{}_1\partial_{\tau}^{n-3} W^+_{A} 
- \begin{pmatrix} n \\ 5\end{pmatrix} \partial_{\tau}^5e^{\tau}{}_1\partial_{\tau}^{n-4} W^+_{A} 
\\
&
- \begin{pmatrix} n \\ 3\end{pmatrix}   \partial_{\tau}^3e^{r}{}_1\partial_{r} \partial_{\tau}^{n-3} W^+_{A} 
- \begin{pmatrix} n \\ 4\end{pmatrix}  \partial_{\tau}^4e^{r}{}_1\partial_{r} \partial_{\tau}^{n-4} W^+_{A}
\\
&
-2 \begin{pmatrix} n \\ 3\end{pmatrix} \Big(\partial_{\tau}^3e^{\mathring A}{}_B\partial_{\mathring A}\partial_{\tau}^{n-3}V^-_{A}{}^{B}-\partial_{\tau}^3\Gamma_B{}^C{}_A\partial_{\tau}^{n-3}V^-_{C}{}^{B}+ \partial_{\tau}^3\Gamma_B{}^B{}_C\partial_{\tau}^{n-3}V^-_{A}{}^{C} \Big)
\\
&
-2  \begin{pmatrix} n \\ 4\end{pmatrix}\Big(\partial_{\tau}^4e^{\mathring A}{}_B\partial_{\mathring A}\partial_{\tau}^{n-4}V^-_{A}{}^{B}-\partial_{\tau}^4\Gamma_B{}^C{}_A\partial_{\tau}^{n-4}V^-_{C}{}^{B}+ \partial_{\tau}^4\Gamma_B{}^B{}_C\partial_{\tau}^{n-4}V^-_{A}{}^{C} \Big)
\\
&
+\begin{pmatrix} n \\ 2\end{pmatrix}\partial_{\tau}^2\Big(  3\widehat\Gamma_1{}^0{}_{0} - 2\widehat\Gamma_B{}^B{}_1  -2\widehat\Gamma_B{}^B{}_0  + \widehat\Gamma_1{}^1{}_0\Big) \partial_{\tau}^{n-2}W^+_{A}
\\
&
+\begin{pmatrix} n \\ 3\end{pmatrix}\partial_{\tau}^3\Big(  3\widehat\Gamma_1{}^0{}_{0} - 2\widehat\Gamma_B{}^B{}_1  -2\widehat\Gamma_B{}^B{}_0  + \widehat\Gamma_1{}^1{}_0\Big) \partial_{\tau}^{n-3}W^+_{A}
\\
&
+\begin{pmatrix} n \\ 4\end{pmatrix}\partial_{\tau}^4\Big(  3\widehat\Gamma_1{}^0{}_{0} - 2\widehat\Gamma_B{}^B{}_1  -2\widehat\Gamma_B{}^B{}_0  + \widehat\Gamma_1{}^1{}_0\Big) \partial_{\tau}^{n-4}W^+_{A}
\\
&
-2\mcD_B\partial_{\tau}^nV^-_{A}{}^{B}
-(r\partial_r-n+1)\partial_{\tau}^{n}W^+_{A} 
-3M\partial_{\tau}^n(  \widehat\Gamma_1{}^0{}_{A} - \widehat\Gamma_1{}^1{}_{A}  )
+\mathcal{P}^{ n+1} + \mathfrak{O}(r^{m_0+2})
\\
=&
-12Mr \Big[\begin{pmatrix} n \\ 3\end{pmatrix} +3\begin{pmatrix} n \\ 2\end{pmatrix}\Big] \partial_{\tau}^{n-2}W^+_{A}
+8Mr\Big[5\begin{pmatrix} n \\ 4\end{pmatrix} 
-( m_0-8)\begin{pmatrix} n \\ 3\end{pmatrix}   \Big]\partial_{\tau}^{n-3}W^+_{A}
\\
&
-8Mr\Big[5\begin{pmatrix} n \\ 5\end{pmatrix} 
-(m_0-5)\begin{pmatrix} n \\ 4\end{pmatrix}  \Big]\partial_{\tau}^{n-4}W^+_{A}
\\
&
+8Mr \begin{pmatrix} n \\ 3\end{pmatrix} \mcD_B\partial_{\tau}^{n-3}V^-_{A}{}^{B}
-8Mr\begin{pmatrix} n \\ 4\end{pmatrix}\mcD_B\partial_{\tau}^{n-4}V^-_{A}{}^{B}
\\
&
-2\mcD_B\partial_{\tau}^nV^-_{A}{}^{B}
-(r\partial_r-n+1)\partial_{\tau}^{n}W^+_{A} 
-3M\partial_{\tau}^n(  \widehat\Gamma_1{}^0{}_{A} - \widehat\Gamma_1{}^1{}_{A}  )
+\mathcal{P}^{ n+1} + \mathfrak{O}(r^{m_0+2})
\,.
\end{align*}
For  $n\leq m_0-1$ 
we combine both equations to obtain the expansion coefficient in $r$ of order~$m_0+1$.
\begin{align*}
&
\hspace{-2em}
(n+1) (m_0-n-1)W^+_{A} {}^{(m_0+1,n+1)}
+ \frac{1}{2}\Big(\Delta_s+  (m_0-n) (m_0-n+1) - 1\Big)W^+_{A} {}^{(m_0+1,n)}
\\
=&
2(n+4)MV^-_{A}{}^{(m_0,n-2)}
-\frac{1}{3}( 7n- 6m_0+7)MV^-_{A} {}^{(m_0,n-3)}
+\frac{2}{3}M(n-m_0) V^-_{A}{}^{(m_0,n-4)}
\\
&
-M(m_0-n-1)  (n +7) W^+_{A}{}^{(m_0,n-2)}
+\frac{1}{6}M(m_0-n-1) (5n -4m_0+17 )W^+_{A}{}^{(m_0,n-3)}
\\
&
+\frac{1}{6}M (m_0-n-1)^2W^+_{A}{}^{(m_0,n-4)}
+\frac{1}{3}M(\Delta_s + 1)W^+_{A}{}^{(m_0,n-3)} 
- \frac{1}{12}M(\Delta_s + 1)W^+_{A} {}^{(m_0,n-4)}
\\
&
-   3M \mcD^B( \widehat\Gamma_{(A}{}^0{}_{B)} -  \widehat\Gamma_{(A}{}^1{}_{B)} )_{\mathrm{tf}}^{(m_0+1,n)}
-\frac{3}{2}M (m_0-n-1) (  \widehat\Gamma_1{}^0{}_{A} - \widehat\Gamma_1{}^1{}_{A}  ){}^{(m_0+1,n)}
\\
=&
-M(m_0-n-1)  (n +7) W^+_{A}{}^{(m_0,n-2)}
-M\frac{n+4}{n-m_0}(\Delta_s+1)W^+_{A}{}^{(m_0,n-2)}
\\
&
+\frac{1}{6}M(m_0-n-1) (5n -4m_0+17 )W^+_{A}{}^{(m_0,n-3)}
+ \frac{1}{6}M\frac{9n- 8m_0+5}{n-m_0-1}(\Delta_s+1)W^+_{A}{}^{(m_0,n-3)}
\\
&
+\frac{1}{6}M (m_0-n-1)^2W^+_{A}{}^{(m_0,n-4)}
-\frac{1}{12}M \frac{5n-5m_0-2}{n-m_0-2}(\Delta_s+1)W^+_{A}{}^{(m_0,n-4)}
\\
&
-   3M \mcD^B( \widehat\Gamma_{(A}{}^0{}_{B)} -  \widehat\Gamma_{(A}{}^1{}_{B)} )_{\mathrm{tf}}^{(m_0+1,n)}
-\frac{3}{2}M (m_0-n-1) (  \widehat\Gamma_1{}^0{}_{A} - \widehat\Gamma_1{}^1{}_{A}  ){}^{(m_0+1,n)}
\,,
\end{align*}
where we have employed   Lemma~\ref{lemma_recursive_data}.
This formula holds for $0\leq n\leq m_0-1$ if one defines derivatives of  negative order to be zero.

We still need to find expressions for $( \widehat\Gamma_{(A}{}^0{}_{B)} -  \widehat\Gamma_{(A}{}^1{}_{B)} )_{\mathrm{tf}}$ and
$\widehat\Gamma_1{}^0{}_{A} - \widehat\Gamma_1{}^1{}_{A} $ in terms of $W^{\pm}_A$.
From the $(n-1)$-st order transverse derivatives of \eq{evolution1}-\eq{evolution7} we find for $n\geq 1$ (recall  \eq{global_functions_confG}),
\begin{align*}
\partial_{\tau}^n\widehat L_{1A} |_{\scri^-}
=&
-\partial_{\tau}^{n-1}(  \widehat\Gamma_1{}^1{}_0\widehat L_{1A} ) - \partial_{\tau}^{n-1}( \widehat\Gamma_1{}^B{}_0\widehat L_{BA} )
 - 2r\partial_{\tau}^{n-1} W^-_{A} 
\\
&
 +(n-1)r\partial_{\tau}^{n-2} (W^+_{A} + W^-_{A} )
+\mathfrak{O}(r^{\infty})
\,,
\\
\partial_{\tau}^n\widehat\Gamma_{1}{}^0{}_A |_{\scri^-}
 =&
- \partial_{\tau}^{n-1}( \widehat\Gamma_1{}^0{}_A\widehat\Gamma_{1}{}^1{}_{0})
-  \partial_{\tau}^{n-1}(\widehat\Gamma_B{}^0{}_A\widehat\Gamma_{1}{}^B{}_{0})
 + \partial_{\tau}^{n-1}\widehat L_{1A}
   -(n-1)r \partial_{\tau}^{n-2} (W^+_{A}+W^-_{A})
\\
&
   +\frac{r}{2}(n-1)(n-2)  \partial_{\tau}^{n-3}(W^+_{A} +W^-_{A})
+\mathfrak{O}(r^{\infty})
\,,
\\
\partial_{\tau}^n\widehat\Gamma_{1}{}^1{}_A |_{\scri^-}
 =&
-  \partial_{\tau}^{n-1}(\widehat\Gamma_1{}^1{}_A\widehat\Gamma_{1}{}^1{}_{0})
  - \partial_{\tau}^{n-1}( \widehat\Gamma_B{}^1{}_A\widehat\Gamma_{1}{}^B{}_{0} )
   + (n-1)r  \partial_{\tau}^{n-2}(W^+_{A}-  W^-_{A})
\\
&
  -\frac{r}{2}(n-1)(n-2)   \partial_{\tau}^{n-3}   (W^+_{A}-W^-_{A})
+\mathfrak{O}(r^{\infty})
\,,
\end{align*}
as well as
\begin{align*}
(\partial_{\tau}^n\widehat L_{(AB)})_{\mathrm{tf}} |_{\scri^-}
=& 
-\partial_{\tau}^{n-1}  (\widehat\Gamma_{(A}{}^1{}_0\widehat L_{1B)})_{\mathrm{tf}} 
 -  \partial_{\tau}^{n-1}(\widehat\Gamma_{(A}{}^C{}_0\widehat L_{CB)})_{\mathrm{tf}}  
-2r\partial_{\tau}^{n-1} V^+_{AB}
\\
&
+(n-1)r \partial_{\tau}^{n-2} (V^+_{AB}+V^-_{AB})
+\mathfrak{O}(r^{\infty})
\,,
\\
(\partial_{\tau}^n\widehat\Gamma_{(A}{}^0{}_{B)} )_{\mathrm{tf}}  |_{\scri^-}
 =&
- \partial_{\tau}^{n-1} (\widehat\Gamma_1{}^0{}_{(A}\widehat\Gamma_{B)}{}^1{}_{0} )_{\mathrm{tf}} 
-\partial_{\tau}^{n-1} ( \widehat\Gamma_C{}^0{}_{(A}\widehat\Gamma_{B)}{}^C{}_{0} )_{\mathrm{tf}} 
 -(n-1)r\partial_{\tau}^{n-2}  (V^+_{AB} +V^-_{AB})
\\
&
+\partial_{\tau}^{n-1} (\widehat L_{AB})_{\mathrm{tf}}   
 +\frac{r}{2}(n-1)(n-2)  \partial_{\tau}^{n-3}  (V^+_{AB} +  V^-_{AB})
+\mathfrak{O}(r^{\infty})
\,,
\\
(\partial_{\tau}^n\widehat\Gamma_{(A}{}^1{}_{B)} )_{\mathrm{tf}}  |_{\scri^-}
 =&
- \partial_{\tau}^{n-1} (\widehat\Gamma_1{}^1{}_{(A}\widehat\Gamma_{B)}{}^1{}_{0})_{\mathrm{tf}} 
-\partial_{\tau}^{n-1} ( \widehat\Gamma_C{}^1{}_{(A}\widehat\Gamma_{B)}{}^C{}_{0})_{\mathrm{tf}} 
 -(n-1)r\partial_{\tau}^{n-2}(V^+_{AB} -V^-_{AB})
\\
&
+\frac{r}{2}(n-1)(n-2)   \partial_{\tau}^{n-3}  (V^+_{AB} - V^-_{AB})
+\mathfrak{O}(r^{\infty})
\,.
\end{align*}
From this we deduce, for $ 2\leq n \leq m_0+2$,
\begin{align*}
\partial_{\tau}^n(\widehat L_{1A}-\widehat\Gamma_{1}{}^0{}_A) |_{\scri^-}
=&
 - 2r\partial_{\tau}^{n-1} W^-_{A} 
 +2(n-1) r\partial_{\tau}^{n-2}  (W^+_{A}+W^-_{A} )
\\
&
   -\frac{r}{2}(n-1)(n-2)  \partial_{\tau}^{n-3}(W^+_{A} +W^-_{A})
 - \partial_{\tau}^{n-1}\widehat L_{1A}
+\mathcal{P}^{ n} + \mathfrak{O}(r^{m_0+2})
\\
=&
 - 2r\partial_{\tau}^{n-1} W^-_{A} 
 +2nr\partial_{\tau}^{n-2} W^-_{A} 
 +2(n-1) r\partial_{\tau}^{n-2}  W^+_{A}
\\
&
   -\frac{r}{2}(n+1)(n-2) \partial_{\tau}^{n-3}(W^+_{A} +W^-_{A})
+\mathcal{P}^{ n} + \mathfrak{O}(r^{m_0+2})
\,,
\end{align*}
and
\begin{align*}
(\partial_{\tau}^n(\widehat L_{(AB)}-\widehat\Gamma_{(A}{}^0{}_{B)}))_{\mathrm{tf}} |_{\scri^-}
=& 
-2r\partial_{\tau}^{n-1} V^+_{AB}
+2(n-1)r \partial_{\tau}^{n-2}( V^+_{AB}+V^-_{AB})
\\
&
 -\frac{r}{2}(n-1)(n-2)  \partial_{\tau}^{n-3}  (V^+_{AB} +V^-_{AB})
-\partial_{\tau}^{n-1} (\widehat L_{AB})_{\mathrm{tf}}   
+\mathcal{P}^{ n+1} + \mathfrak{O}(r^{m_0+2})
\\
=& 
-2r\partial_{\tau}^{n-1} V^+_{AB}
+2nr \partial_{\tau}^{n-2} V^+_{AB}
+2(n-1)r \partial_{\tau}^{n-2}V^-_{AB}
\\
&
  -\frac{r}{2}(n+1)(n-2) \partial_{\tau}^{n-3} ( V^+_{AB} +V^-_{AB})
+\mathcal{P}^{ n+1} + \mathfrak{O}(r^{m_0+2})
\,.
\end{align*}
We then determine, for $ 3\leq n \leq m_0+2$ (one checks that this actually  holds  for $n\geq 1$),
\begin{align*}
&\hspace{-3em}\partial_{\tau}^n(\widehat\Gamma_{1}{}^0{}_A-\widehat\Gamma_{1}{}^1{}_A) |_{\scri^-}
\\
 =&
- \partial_{\tau}^{n-1}[( \widehat\Gamma_1{}^0{}_A-\widehat\Gamma_1{}^1{}_A)\widehat\Gamma_{1}{}^1{}_{0}]
-  \partial_{\tau}^{n-1}[(\widehat\Gamma_B{}^0{}_A- \widehat\Gamma_B{}^1{}_A)\widehat\Gamma_{1}{}^B{}_{0}]
 + \partial_{\tau}^{n-1}\widehat L_{1A}
\\
&
   -2(n-1)r \partial_{\tau}^{n-2} W^+_{A}
   +r(n-1)(n-2) \partial_{\tau}^{n-3}W^+_{A} 
+\mathfrak{O}(r^{\infty})
\\
 =&
 \partial_{\tau}^{n-1}(\widehat L_{1A}-\widehat\Gamma_{1}{}^0{}_{A})
   -2(n-1)r \partial_{\tau}^{n-2} W^+_{A}
   +r(n-1)(n-2) \partial_{\tau}^{n-3}W^+_{A} 
+\mathcal{P}^{ n} + \mathfrak{O}(r^{m_0+2})
\\
 =&
   -2(n-1)r \partial_{\tau}^{n-2} W^+_{A}
   +r(n+1)(n-2)\partial_{\tau}^{n-3}W^+_{A} 
    -\frac{r}{2}n(n-3) \partial_{\tau}^{n-4}W^+_{A} 
\\
&
  - 2r\partial_{\tau}^{n-2} W^-_{A} 
 +2(n-1)r\partial_{\tau}^{n-3} W^-_{A} 
  -\frac{r}{2}n(n-3) \partial_{\tau}^{n-4} W^-_{A}
+\mathcal{P}^{ n} + \mathfrak{O}(r^{m_0+2})
\,,
\end{align*}
and, for $ 3\leq n \leq m_0+2$ (this formula is wrong for $n\in\{1,2\}$),
\begin{align*}
&\hspace{-2em} (\partial_{\tau}^n(\widehat\Gamma_{(A}{}^0{}_{B)}-\widehat\Gamma_{(A}{}^1{}_{B)}) )_{\mathrm{tf}}  |_{\scri^-}
\\
 =&
- \partial_{\tau}^{n-1} [(\widehat\Gamma_1{}^0{}_{(A}-\widehat\Gamma_1{}^1{}_{(A})\widehat\Gamma_{B)}{}^1{}_{0} ]_{\mathrm{tf}} 
-\partial_{\tau}^{n-1}[ ( \widehat\Gamma_C{}^0{}_{(A}-\widehat\Gamma_C{}^1{}_{(A})\widehat\Gamma_{B)}{}^C{}_{0} ]_{\mathrm{tf}} 
+\partial_{\tau}^{n-1} (\widehat L_{AB})_{\mathrm{tf}}   
\\
&
- 2(n-1)r\partial_{\tau}^{n-2} V^-_{AB}
 +r(n-1)(n-2)  \partial_{\tau}^{n-3}  V^-_{AB}
+\mathfrak{O}(r^{\infty})
\\
 =&
[\partial_{\tau}^{n-1}(\widehat L_{AB}-\widehat\Gamma_{A}{}^0{}_{B)} )]_{\mathrm{tf}} 
- 2(n-1)r\partial_{\tau}^{n-2} V^-_{AB}
 +r(n-1)(n-2) \partial_{\tau}^{n-3}  V^-_{AB}
+\mathcal{P}^{ n+1} + \mathfrak{O}(r^{m_0+2})
\\
 =&
- 2(n-1)r\partial_{\tau}^{n-2} V^-_{AB}
 +r(n+1)(n-2)  \partial_{\tau}^{n-3}  V^-_{AB}
  -\frac{r}{2}n(n-3)   \partial_{\tau}^{n-4}  V^-_{AB}
\\
&
-2r\partial_{\tau}^{n-2} V^+_{AB}
+2(n-1)r \partial_{\tau}^{n-3} V^+_{AB}
  -\frac{r}{2}n(n-3)   \partial_{\tau}^{n-4}  V^+_{AB} 
+\mathcal{P}^{ n+1} + \mathfrak{O}(r^{m_0+2})
\,.
\end{align*}
That yields
%\tim{check ranges of $n$ here and elsewhere}
%
\begin{align*}
&\hspace{-2em}(\widehat\Gamma_{1}{}^0{}_A-\widehat\Gamma_{1}{}^1{}_A) {}^{(m_0+1,n)}
\\
 =&
   -\frac{2}{n} W^+_{A}{}^{(m_0,n-2)}
   +\frac{n+1}{n(n-1)}W^+_{A} {}^{(m_0,n-3)}
    -\frac{1}{2(n-1)(n-2)}W^+_{A} {}^{(m_0,n-4)}
\\
&
  - \frac{2}{n(n-1)} W^-_{A} {}^{(m_0,n-2)}
 +\frac{2}{n(n-2)}W^-_{A} {}^{(m_0,n-3)}
  -\frac{1}{2(n-1)(n-2)} W^-_{A}{}^{(m_0,n-4)}
\,, \quad 3\leq n \leq m_0+2
\,.
\end{align*}
This formula  holds for $n\in\{1,2\}$  if one defines terms with $W^{\pm}_{A} {}^{(m_0,n)}$, $n<0$, to vanish.
Using Lemma~\ref{lemma_recursive_data} we also obtain
\begin{align*}
&\hspace{-3em}\mcD^B(\widehat\Gamma_{(A}{}^0{}_{B)}-\widehat\Gamma_{(A}{}^1{}_{B)} )_{\mathrm{tf}} {}^{(m_0+1,n)}
\\
 =&
- \frac{2}{n}V^-_{A} {}^{(m_0,n-2)}
 +\frac{n+1}{n(n-1)}V^-_{A} {}^{(m_0,n-3)}
  -\frac{1}{2(n-1)(n-2)}   V^-_{A} {}^{(m_0,n-4)}
\\
&
- \frac{2}{n(n-1)} V^+_{A} {}^{(m_0,n-2)}
+\frac{2}{n(n-2)} V^+_{A} {}^{(m_0,n-3)}
  -\frac{1}{2(n-1)(n-2)}     V^+_{A}  {}^{(m_0,n-4)}
\\
 =&
\frac{1}{n(n-m_0)  }(\Delta_s+1)W^+_{A}{}^{(m_0,n-2)}
 -\frac{n+1}{2n(n-1)(n-m_0-1)  }(\Delta_s+1)W^+_{A}{}^{(m_0,n-3)}
\\
&
  +  \frac{1}{4(n-1)(n-2)(n-m_0-2)  }(\Delta_s+1)W^+_{A}{}^{(m_0,n-4)}
- \frac{n-m_0-3}{n(n-1)}  W^-_{A}{}^{(m_0,n-2)} 
\\
&
+\frac{n-m_0-4}{n(n-2)}  W^-_{A}{}^{(m_0,n-3)} 
  -\frac{n-m_0-5}{4(n-1)(n-2)}     W^-_{A}{}^{(m_0,n-4)} 
\,, \quad 3\leq n \leq m_0-1
\,.
\end{align*}
Moreover (cf.\ Section~\ref{sec_confG_firstoder} \& \ref{sec_confG_secondoder}),
\begin{align*}
\mcD^B(\widehat\Gamma_{(A}{}^0{}_{B)}-\widehat\Gamma_{(A}{}^1{}_{B)})_{\mathrm{tf}} {}^{(m_0+1,2)}
=&
\frac{(m_0+1)(m_0+2)}{8}v^{(m_0+2)}_{ A }
\\
&
\hspace{-4em}
- \frac{1}{8(m_0-1)(m_0-2)} (\Delta_s+1)\Big(( \Delta_s-1) v^{(m_0+2)}_{ A}
 -2 \mcD_A\mcD^{ B}v^{(m_0+2)}_{ B}\Big)
\,,
\\
\mcD^B(\widehat\Gamma_{(A}{}^0{}_{B)}-\widehat\Gamma_{(A}{}^1{}_{B)} )_{\mathrm{tf}} {}^{(m_0+1,1)}
=&-\frac{m_0+2}{2}v^{(m_0+2)}_A
\,.
\end{align*}
Altogether we end up with the following  recursion relation for  $3\leq n\leq m_0-1$,
\begin{align*}
&
\hspace{-2em}
(n+1) (m_0-n-1)W^+_{A} {}^{(m_0+1,n+1)}
+ \frac{1}{2}\Big(\Delta_s+  (m_0-n) (m_0-n+1) - 1\Big)W^+_{A} {}^{(m_0+1,n)}
\\
=&
-\frac{M}{n}\Big[(m_0-n-1)  (n^2 +7n-3) 
+\frac{n^2+4n+3}{n-m_0}(\Delta_s+1)\Big]W^+_{A}{}^{(m_0,n-2)}
\\
&
+\frac{M}{6} (m_0-n-1) \Big(5n -4m_0+17- \frac{9(n+1)}{n(n-1)} \Big)W^+_{A}{}^{(m_0,n-3)}
\\
&
+\frac{M}{6}\frac{1}{n-m_0-1}\Big( 9n- 8m_0+5+    \frac{9(n+1)}{n(n-1)  }\Big)(\Delta_s+1)W^+_{A}{}^{(m_0,n-3)}
\\
&
+\frac{M}{6}  (m_0-n-1) \Big(m_0-n-1 +    \frac{9}{2(n-1)(n-2)}\Big)W^+_{A}{}^{(m_0,n-4)}
\\
&
-\frac{M}{12(n-m_0-2)}\Big( 5n-5m_0-2+     \frac{9}{(n-1)(n-2)  }\Big)(\Delta_s+1)W^+_{A}{}^{(m_0,n-4)}
\\
&
-3M  \Big[
  \frac{4}{n(n-1)} W^-_{A} {}^{(m_0,n-2)}
 -\frac{5}{n(n-2)}W^-_{A} {}^{(m_0,n-3)}
  +\frac{3}{2(n-1)(n-2)} W^-_{A}{}^{(m_0,n-4)}\Big]
\,.
\end{align*}
Using Lemma~\ref{lemma_recursive_data} this can be written for $4\leq n\leq m_0-1$
\begin{align*}
&
\hspace{-2em}
(n+1) (m_0-n-1)W^+_{A} {}^{(m_0+1,n+1)}
+ \frac{1}{2}\Big(\Delta_s+  (m_0-n) (m_0-n+1) - 1\Big)W^+_{A} {}^{(m_0+1,n)}
\\
=&
-\frac{M}{6}\Big\{ \frac{1}{2(n-3)(m_0-n+2) }\Big[\frac{3}{n(n-2)(m_0-n+1) }\Big((m_0-n-1)  (n^2 +7n-3) 
+\frac{n^2+4n+3}{n-m_0}(\Delta_s+1)\Big)
\\
&
\qquad \times \Big(\Delta_s  +(m_0-n+2) (m_0-n+3)-1\Big)
\\
&
+ (m_0-n-1) \Big(5n -4m_0+17- \frac{9(n+1)}{n(n-1)} \Big)
-\frac{1}{m_0-n+1}\Big( 9n- 8m_0+5+    \frac{9(n+1)}{n(n-1)  }\Big)(\Delta_s+1)\Big]
\\
&
\qquad \times 
\Big(\Delta_s 
+(m_0-n+3) (m_0-n+4)-1\Big)
\\
&
- (m_0-n-1) \Big(m_0-n-1 +    \frac{9}{2(n-1)(n-2)}\Big)
-\frac{1}{2(m_0-n+2)}\Big( 5n-5m_0-2+     \frac{9}{(n-1)(n-2)  }\Big)(\Delta_s+1)\Big\}
\\
&
\qquad \times W^+_{A}{}^{(m_0,n-4)}
\\
&
- \frac{3M}{n-2}\Big[\frac{1}{n(n-3)(m_0-n+4)} \Big(\frac{1}{(n-1)(m_0-n+3)} \Big(\Delta_s+(m_0-n+2) (m_0-n+3)-1\Big)+  \frac{5}{2} \Big)
\\
&
\quad \times
 \Big(\Delta_s+(m_0-n+3) (m_0-n+4)-1\Big) 
+\frac{3}{2(n-1)}\Big] W^-_{A}{}^{(m_0,n-4)}
\,.
\end{align*}
For $n\in\{0,1,2,3\}$ one has (we assume  $m_0\geq 4$ which is fine as the case $m_0=3$ is covered by Section~\ref{sec_constat_mass_ex})
\begin{align}
&
\hspace{-2em}
4 (m_0-4)W^+_{A} {}^{(m_0+1,4)}
+ \frac{1}{2}\Big(\Delta_s+  (m_0-3) (m_0-2) - 1\Big)W^+_{A} {}^{(m_0+1,3)}
\nonumber
\\
=&
\frac{M}{6}(23m_0^2-39m_0-212)W^+_A{}^{(m_0,0)}  + \frac{M(11m_0^2-227m_0+486)}{6(m_0-2)(m_0-3) }(\Delta_s+1)W^+_{A}{}^{(m_0,0)}
\nonumber
\\
&
-\frac{4M}{(m_0-2)(m_0-3)}(\Delta_s+1)(\Delta_s +1) W^+_A{}^{(m_0,0)}
\nonumber
\\
&
+M \Big( (m_0+4)W^-_A{}^{(m_0,0)}+ \frac{1}{m_0}(\Delta_s-1)W^-_A{}^{(m_0,0)}\Big)
\,,
\label{ab3}
\\
&
\hspace{-2em}
3 (m_0-3)W^+_{A} {}^{(m_0+1,3)}
+ \frac{1}{2}\Big(\Delta_s+  (m_0-2) (m_0-1) - 1\Big)W^+_{A} {}^{(m_0+1,2)}
\nonumber
\\
=&
-\frac{15}{2}M(m_0-3)   W^+_{A}{}^{(m_0,0)}
+M\frac{15}{2(m_0-2)}(\Delta_s+1)W^+_{A}{}^{(m_0,0)}
-M\frac{3(3m_0+1)}{m_0}W^-_A{}^{(m_0,0)} 
\,,
\label{ab2}
\\
&
\hspace{-2em}
2 (m_0-2)W^+_{A} {}^{(m_0+1,2)}
+ \frac{1}{2}\Big(\Delta_s+  (m_0-1) m_0 - 1\Big)W^+_{A} {}^{(m_0+1,1)}
=
 \frac{6}{m_0} M (m_0+2)W^{-(m_0,0)}_A
\,,
\label{ab1}
\\
&
\hspace{-2em}
 (m_0-1)W^+_{A} {}^{(m_0+1,1)}
+ \frac{1}{2}\Big(\Delta_s+  m_0 (m_0+1) - 1\Big)W^+_{A} {}^{(m_0+1,0)}
=
-   \frac{12}{m_0}M W^{-(m_0,0)}_A
\,,
\label{ab0}
\end{align}
where we have used \eq{evolution_2nd_trans_1}-\eq{evolution_2nd_trans_19},
\eq{constr_spec_gauge1}-\eq{constr_spec_gauge6}, and \eq{trans_Weyl_spec1}-\eq{trans_Weyl_spec5},
% and that, by \eq{constr_spec_gauge1}-\eq{constr_spec_gauge6} and \eq{trans_Weyl_spec1}-\eq{trans_Weyl_spec5} for $m_0\geq 3$,
%
\begin{align*}
W^-_A{}^{(m_0,0)} =&\frac{m_0}{4}v_A{}^{(m_0+2)}
\,,
\\
V^+_A{}^{(m_0,0)} =& -\frac{m_0+1}{2}W^-_A{}^{(m_0,0)}
\,,
%\\
%W_{0101}{}^{(m_0,0)}=& -\frac{1}{4}\mcD^Bv^{(m_0+2)}_B
%\,,
%\\
%W_{01AB}{}^{(m_0,0)}=& -\frac{1}{2}\mcD_{[A}v^{(m_0+2)}_{B]}
%\,,
\\
W^+_A{}^{(m_0,0)} =& \frac{1}{m_0-1} \mcD^B U_{AB}{}^{(m_0,0)}
\,,
\\
V^-_A{}^{(m_0,0)} =& \frac{1}{2(m_0-2)} (\Delta_s +1) W^+_A{}^{(m_0,0)}
\,,
\\
W^-_A{}^{(m_0,1)} =&  
-\frac{m_0-1}{2}W^-_A{}^{(m_0,0)}- \frac{1}{2m_0}(\Delta_s-1)W^-_A{}^{(m_0,0)}
\,,
\\ 
W^+_A{}^{(m_0,1)} =& -\frac{m_0+1}{2}W^+_A{}^{(m_0,0)}  -\frac{1}{2(m_0-2)} (\Delta_s +1) W^+_A{}^{(m_0,0)}
\,.
\end{align*}

\subsubsection{Recursion relation}

We observe that the recursion relation has the following structure for $4\leq n\leq m_0-1$,
\begin{align*}
&
\hspace{-2em}
(n+1) (m_0-n-1)W_A^+ {}^{(m_0+1,n+1)}
+ \frac{1}{2}\Big(\Delta_s+  (m_0-n) (m_0-n+1) -1\Big)W_A^+ {}^{(m_0+1,n)}
\\
=&
 a_{m_0,n}W_A^+{}^{(m_0,n-4)} +  b_{m_0,n}W_A^-{}^{(m_0,n-4)}
\,,
\label{gen_rec_formula}
\end{align*}
where
$a_{m_0,n}$ and $b_{m_0,n}$ are operators, more precisely they are polynomials in the  Laplacian $\Delta_s$.
In fact, now it is convenient to set 
\begin{equation}
W_A^-{}^{(m_0,n)} := W_A^-{}^{(m_0,0)} \quad \text{for $n\leq -1$}
\label{dfn_sepc_cases}
\end{equation}
as then the formula remains true  for $n\in\{0,1,2,3\}$ with appropriately chosen  $a_{m_0,n}$ and $b_{m_0,n}$ which can be read off from \eq{ab3}-\eq{ab0}.

The no-logs-condition for $\partial_{\tau}^{m_0-1}V^-_{AB}|_{\scri^-}$ can be written as (cf.\ \eq{no_logs_W+}) 
\begin{equation}
(\Delta_s+1)W^+_{A} {}^{(m_0+1,m_0-1)}
-2 a_{m_0,m_0-1}W_A^+{}^{(m_0,m_0-5)} -2  b_{m_0,m_0-1}W_A^-{}^{(m_0,m_0-5)}=0
\,.
\label{no-log_constM}
\end{equation}
By recursion  one shows that
%\begin{align*}
%W^+_{A} {}^{(m_0+1,n+1)}
%=&
% - \frac{1}{2(n+1) (m_0-n-1)}\Big(\Delta_s+  (m_0-n) (m_0-n+1) - 1\Big)W^+_{A} {}^{(m_0+1,n)}
%\\
%&
%+\frac{1}{(n+1) (m_0-n-1)}\Big( a_{m_0,n}W_A^+{}^{(m_0,n-4)} +  b_{m_0,n}W_A^-{}^{(m_0,n-4)}\Big)
%\,,
%\end{align*}
%
\begin{align*}
&
\hspace{-2em}
W^+_{A} {}^{(m_0+1,m_0-1)}
\\
=&
-\frac{1}{2(m_0-1)}(\Delta_s+5)W^+_{A} {}^{(m_0+1,m_0-2)}
\\
&
+ \frac{1}{m_0-1} \Big( a_{m_0,m_0-2}W_A^+{}^{(m_0,m_0-6)} +  b_{m_0,m_0-2}W_A^-{}^{(m_0,m_0-6)}\Big)
\\
=&
\frac{1}{8(m_0-1)(m_0-2)}(\Delta_s+5)(\Delta_s+11)W^+_{A} {}^{(m_0+1,m_0-3)}
\\
&
-\frac{1}{4(m_0-1)(m_0-2)}(\Delta_s+5)\Big[ a_{m_0,m_0-3}W_A^+{}^{(m_0,m_0-7)} +  b_{m_0,m_0-3}W_A^-{}^{(m_0,m_0-7)}\Big]
\\
&
+ \frac{1}{m_0-1} \Big( a_{m_0,m_0-2}W_A^+{}^{(m_0,m_0-6)} +  b_{m_0,m_0-2}W_A^-{}^{(m_0,m_0-6)}\Big)
%\\
%=&
% - \frac{1}{2(m_0-3) 3}\frac{1}{8(m_0-1)(m_0-2)}(\Delta_s+5)(\Delta_s+11)\Big(\Delta_s+19\Big)W^+_{A} {}^{(m_0+1,m_0-4)}
%\\
%&
%+\frac{1}{(m_0-3)3}\frac{1}{8(m_0-1)(m_0-2)}(\Delta_s+5)(\Delta_s+11)\Big[ a_{m_0,m_0-4}W_A^+{}^{(m_0,m_0-8)} +  b_{m_0,m_0-4}W_A^-{}^{(m_0,m_0-8)}
%\Big]
%\\
%&
%-\frac{1}{4(m_0-1)(m_0-2)}(\Delta_s+5)\Big[ a_{m_0,m_0-3}W_A^+{}^{(m_0,m_0-7)} +  b_{m_0,m_0-3}W_A^-{}^{(m_0,m_0-7)}\Big]
%\\
%&
%+ \frac{1}{(m_0-1)} \Big( a_{m_0,m_0-2}W_A^+{}^{(m_0,m_0-6)} +  b_{m_0,m_0-2}W_A^-{}^{(m_0,m_0-6)}\Big)
\\
=&\dots
\\
=&
\frac{(-1)^{m_0-1}}{2^{m_0-1}(m_0-1)!^2}\prod_{\ell=2}^{m_0}\Big(\Delta_s + \ell(\ell+1)-1\Big) W^+_{A} {}^{(m_0+1,0)}
\\
& 
+\sum_{k=0}^{m_0-2} \frac{(-1)^k(m_0-k-2)!}{2^{k}(m_0-1)!(k+1)!} a_{m_0,m_0-2-k}\prod_{\ell_1=2}^{k+1}\Big(\Delta_s + \ell_1(\ell_1+1)-1\Big)  W^+_{A} {}^{(m_0,m_0-k-6)}
\\
& 
+\sum_{k=0}^{m_0-2}  \frac{(-1)^k(m_0-k-2)!}{2^{k}(m_0-1)!(k+1)!}  b_{m_0,m_0-2-k} \prod_{\ell_1=2}^{k+1}\Big(\Delta_s + \ell_1(\ell_1+1)-1\Big)  W^-_{A} {}^{(m_0,m_0-k-6)}
\,.
\end{align*}
We use  Corollary~\ref{cor_recursive_data}   to conclude that the no-logs condition \eq{no-log_constM} becomes
(recall \eq{dfn_sepc_cases}),
\begin{align}
0
=&
\frac{(-1)^{m_0-1}}{2^{m_0-1}(m_0-1)!^2}\prod_{\ell=1}^{m_0}\Big(\Delta_s + \ell(\ell+1)-1\Big) W^+_{A} {}^{(m_0+1,0)}
\nonumber
\\
& 
+\frac{(-1)^{m_0}}{2^{m_0-6}}\sum_{k=-1}^{m_0-6} \Big[ \frac{(m_0-k-2)!(k+4)!}{(m_0-1)!(m_0-2)!(k+1)!(m_0-k-6)!}  a_{m_0,m_0-2-k}
\nonumber
\\
&
\qquad \times \prod_{\ell_1=1}^{k+1}\Big(\Delta_s + \ell_1(\ell_1+1)-1\Big)  \prod_{\ell_2=k+6}^{m_0-1}\Big(\Delta_s 
+\ell_2 (\ell_2+1)-1\Big)W^+_{A}{}^{(m_0,0)}
\Big]
\nonumber
\\
& 
+ \frac{(-1)^{m_0}}{2^{m_0-6}} \sum_{k=-1}^{m_0-6}\Big[ \frac{(m_0-k-2)!(k+6)!}{m_0!(m_0-1)!(k+1)!(m_0-k-6)!}  b_{m_0,m_0-2-k}
\nonumber
\\
&
\qquad \times  \prod_{\ell_1=1}^{k+1}\Big(\Delta_s + \ell_1(\ell_1+1)-1\Big) \prod_{\ell_2=k+6}^{m_0-1}\Big(\Delta_s+\ell_2 (\ell_2+1)-1\Big)  W^-_{A}{}^{(m_0,0)}
\Big]
\nonumber
\\
&
+\sum_{k=m_0-5}^{m_0-2} \frac{(-1)^k(m_0-k-2)!}{2^{k}(m_0-1)!(k+1)!} a_{m_0,m_0-2-k}\prod_{\ell_1=1}^{k+1}\Big(\Delta_s + \ell_1(\ell_1+1)-1\Big)  W^+_{A} {}^{(m_0,0)}
\nonumber
\\
& 
+\sum_{k=m_0-5}^{m_0-2}  \frac{(-1)^k(m_0-k-2)!}{2^{k}(m_0-1)!(k+1)!}b_{m_0,m_0-2-k} \prod_{\ell_1=1}^{k+1}\Big(\Delta_s + \ell_1(\ell_1+1)-1\Big)  W^-_{A} {}^{(m_0,0)}
\,.
\label{some-no-log-form}
\end{align}
%
%%
%\begin{align*}
%W^-_{A}{}^{(m_0,m_0-k)} 
% =&
%\frac{(-1)^{m_0-k}k!}{2^{m_0-k}(m_0-k)!m_0!} \prod_{\ell=k}^{m_0-1}\Big(\Delta_s+\ell (\ell+1)-1\Big)  W^-_{A}{}^{(m_0,0)}
%\,, \quad k\geq 0
%\,,
%\\
% W^+_{A} {}^{(m_0,m_0-k)}
% =&  \frac{(-1)^{m_0-k}(k-2)!}{2^{m_0-k}(m_0-k)!(m_0-2)!} \prod_{\ell=k}^{m_0-1}\Big(\Delta_s 
%+\ell (\ell+1)-1\Big)W^+_{A}{}^{(m_0,0)}
%\,, \quad k\geq 2
%\,.
%\end{align*}
%
%
Because of \eq{assumption_data2B}  the Hodge decomposition scalars of $W_A^{\pm(m_0,0)}$ only  involve
spherical harmonics up to and including $\ell=m_0-1$.
With \eq{constr_spec_gauge1}-\eq{constr_spec_gauge6} we deduce that a necessary condition for the non-appearance of logarithmic terms at this order is
\begin{equation}
\ul{\Xi}^{(m_0+3)}=\sum_{\ell=1}^{m_0}\sum_{m=-\ell}^{+\ell}\ul{\Xi}^{(m_0+3)}_{\ell m}Y_{\ell m}(\theta,\phi)
\,, \quad
\ol{\Xi}^{(m_0+3)}=\sum_{\ell=1}^{m_0}\sum_{m=-\ell}^{+\ell}\ol{\Xi}^{(m_0+3)}_{\ell m}Y_{\ell m}(\theta,\phi)
\,,
\end{equation}
and then  the  term in the first line  vanishes.

%Also because of \eq{assumption_data2B} the last two lines vanish,  so that  the coefficients $a_{m_0,k}$ and $b_{m_0,k}$ with $k\in\{0,1,2,3\}$ are not needed.

We further  observe that for each  $1\leq \ell\leq m_0-1$ in the harmonic decomposition of $W^{\pm(m_0,0)}_A$
there are terms in the above sum, for which the Laplacian  does not project out their contribution.
To deduce that
the no-log condition is actually violated by a non-trivial $\Xi_{AB}^{(m_0+2)}$, we also need to make sure that the coefficients are non-zero.
%We want to analyze whether the  no-logs condition is satisfied.
Since the initial data \eq{assumption_data} enter this condition linearly,
and $a_{m,n}$ and $b_{m,n}$ only involve the Laplacian and $x^{\mathring A}$-independent coefficients, 
 we may assume w.l.o.g.\ initial data of the following form
\begin{equation}
\Xi_{AB} = \Xi_{AB}^{(m_0+2)} r^{m_0+2} + \mathfrak{O}(r^{m_0+3})\,, \quad m_0\geq 3
\,,
\label{specific_data}
\end{equation}
where $ \Xi_{AB}^{(m_0+2)}= (\mcD_{A}\mcD_{B}\ul {\Xi}^{(m_0+2)}_{\widehat\ell})_{\mathrm{tf}}
+ \epsilon_{(A}{}^C \mcD_{B)}\mcD_{C}\ol{ \Xi}^{(m_0+2)}_{\widehat\ell}$ with
\begin{equation*}
\ul{\Xi}^{(m_0+2)}_{\widehat\ell}=\sum_{m=-\widehat \ell}^{+\widehat\ell}\ul{\Xi}^{(m_0+2)}_{\widehat\ell m}Y_{\widehat\ell m}(\theta,\phi)
\,, \quad
\ol{\Xi}^{(m_0+2)}_{\widehat\ell}=\sum_{m=-\widehat\ell}^{+\widehat\ell}\ol{\Xi}^{(m_0+2)}_{\widehat\ell m}Y_{\widehat\ell m}(\theta,\phi)
\,,
\quad 2\leq \widehat \ell \leq m_0-1
\,.
\end{equation*}
%
%(for each $\ell=\widehat\ell$ at most 4 terms from the ``main'' terms which come along with $W^+_{A} {}^{(m_0,0)}$ and $W^-_{A} {}^{(m_0,0)}$, respectively, possibly supplemented by at most 3 terms, respectively, from the last 2 lines).
Let us first analyze the case where  $1\leq \widehat \ell \leq m_0-4$. Then the last two lines in \eq{some-no-log-form} vanish.
We further observe that there are at most 4 terms  in \eq{some-no-log-form} which come along with $W^+_{A} {}^{(m_0,0)}$ and $W^-_{A} {}^{(m_0,0)}$, respectively, which are not projected out by the Laplacians. The no-logs condition thus becomes
\begin{align*}
0
=&
\sum_{k=\widehat \ell -5}^{\widehat \ell -2} \Big[ \frac{(m_0-k-2)!(k+4)!}{(m_0-1)!(m_0-2)!(k+1)!(m_0-k-6)!}  a_{m_0,m_0-2-k}
\\
&
\qquad \times \prod_{\ell_1=1}^{k+1}\Big(\Delta_s + \ell_1(\ell_1+1)-1\Big)  \prod_{\ell_2=k+6}^{m_0-1}\Big(\Delta_s 
+\ell_2 (\ell_2+1)-1\Big)W^+_{A}{}^{(m_0,0)}
\Big]
\\
& 
+ \sum_{k=\widehat \ell -5}^{\widehat \ell -2}\Big[ \frac{(m_0-k-2)!(k+6)!}{m_0!(m_0-1)!(k+1)!(m_0-k-6)!}  b_{m_0,m_0-2-k}
\\
&
\qquad \times  \prod_{\ell_1=1}^{k+1}\Big(\Delta_s + \ell_1(\ell_1+1)-1\Big) \prod_{\ell_2=k+6}^{m_0-1}\Big(\Delta_s+\ell_2 (\ell_2+1)-1\Big)  W^-_{A}{}^{(m_0,0)}
\Big]
\,.
\end{align*}
To ensure that  $k$ is in $[-1, m_0-6]$ as required by \eq{some-no-log-form} it is convenient to set
\begin{equation}
a_{m_0,n}=b_{m_0,n}=0 \quad \text{for} \quad n\leq 3 \quad \text{and} \quad  n\geq m_0
\,.
\label{rest_ab}
\end{equation}
It seems important to emphasize that this is only for the evaluation of the above term, the coefficients do not vanish when the  contributions from 
the   last two lines in \eq{some-no-log-form}  are determined  below.

We take divergence and curl of this equation.
Then we replace the  Laplacian $\Delta_s$ in the resulting formula by the corresponding eigenvalue  (be aware  that it also appears in the coefficients $a_{m_0,k}$ and $b_{m_0,k}$).
 Assuming that $\ul {\Xi}^{(m_0+2)}_{\widehat\ell}\ne 0$ and 
 $\ol {\Xi}^{(m_0+2)}_{\widehat\ell}\ne 0$, respectively,  the no-logs condition reads
\begin{align}
0
=&
\sum_{k=\widehat \ell -5}^{\widehat \ell -2}  \frac{(k+4)!(m_0-k-2)!}{(k+1)!(m_0-k-6)!} 
 \prod_{\ell_1=1}^{k+1}\Big( \ell_1(\ell_1+1)-\widehat\ell(\widehat\ell+1)\Big)  \prod_{\ell_2=k+6}^{m_0-1}\Big(\ell_2 (\ell_2+1)-\widehat\ell(\widehat\ell+1)\Big)
\nonumber
\\
&
\qquad \times \Big(\widehat\ell(\widehat\ell+1)a_{m_0,m_0-2-k}^{\widehat\ell} 
\pm  (k+5)(k+6) b_{m_0,m_0-2-k}^{\widehat\ell}\Big)
\,,
\label{simplified_no_logs}
\end{align}
where  we used that (cf.\  Corollary~\ref{cor_recursive_data})
%
%\begin{align*}
%v_{A}^{(m_0+2)}=\frac{1}{2}(\Delta_s+1)\mcD_{A}\ul {\Xi}^{(m_0+2)}_{\widehat\ell}
%+\frac{1}{2} \epsilon_{A}{}^B (\Delta_s+1)\mcD_{B}\ol{ \Xi}^{(m_0+2)}_{\widehat\ell}
%\end{align*}
%
\begin{align*}
\mcD^AW^{-(m_0,0)}_A =&  \frac{m_0}{4}\mcD^Av^{(m_0+2)}_A
= \frac{m_0}{8}\Big(\widehat\ell(\widehat\ell+1)-2\Big)\widehat\ell(\widehat\ell+1)\ul {\Xi}^{(m_0+2)}_{\widehat\ell}
\,,
\\
\mcD^AW^{+(m_0,0)}_A=& -\frac{1}{4(m_0-1)}\Delta_s\mcD^A v^{(m_0+2)}_{ A}
= \frac{1}{8(m_0-1)}\Big(\widehat\ell(\widehat\ell+1)-2\Big)\widehat\ell^2(\widehat\ell+1)^2\ul {\Xi}^{(m_0+2)}_{\widehat\ell}
\,,
\\
\epsilon^{AB}\mcD_{[A}W^{-(m_0,0)}_{B]} =&  \frac{m_0}{4}\epsilon^{AB}\mcD_{[A}v^{(m_0+2)}_{B]}
=  \frac{m_0}{8}\Big(\widehat\ell(\widehat\ell+1)-2\Big)\widehat\ell(\widehat\ell+1)\ol{ \Xi}^{(m_0+2)}_{\widehat\ell}
\,,
\\
\epsilon^{AB}\mcD_{[A}W^{+(m_0,0)}_{B]}=&\frac{1}{4(m_0-1)}\Delta_s\epsilon^{AB}\mcD_{[A} v^{(m_0+2)}_{ B]}
=-\frac{1}{8(m_0-1)}\Big(\widehat\ell(\widehat\ell+1)-2\Big)\widehat\ell^2(\widehat\ell+1)^2\ol{ \Xi}^{(m_0+2)}_{\widehat\ell}
\,.
\end{align*}
The ``$+$''-sign appears for  the divergence, the ``$-$''-sign for  the curl.
For $4\leq n\leq m_0-1$ the coefficients in \eq{simplified_no_logs} are given by
\begin{align*}
a_{m_0,n}^{\widehat\ell}
=&
-M\frac{\Big((m_0-n+3) (m_0-n+4)-\widehat\ell(\widehat\ell+1)\Big)\Big((m_0-n+2) (m_0-n+3)-\widehat\ell(\widehat\ell+1)  \Big)}{4n(n-2)(n-3)(m_0-n+1)(m_0-n+2) }
\\
&\qquad \times\Big((m_0-n-1)  (n^2 +7n-3) 
+\frac{n^2+4n+3}{m_0-n}(\widehat\ell(\widehat\ell+1)-2)\Big)
\\
&
-M (m_0-n-1)\frac{\Big((m_0-n+3) (m_0-n+4)-\widehat\ell(\widehat\ell+1)\Big)}{12(n-3)(m_0-n+2) } \Big(5n -4m_0+17- \frac{9(n+1)}{n(n-1)} \Big)
\\
&
+M\frac{\Big((m_0-n+3) (m_0-n+4)-\widehat\ell(\widehat\ell+1)\Big)}{12(n-3)(m_0-n+1)(m_0-n+2) }
\Big( 9n- 8m_0+5+    \frac{9(n+1)}{n(n-1)  }\Big)(-\widehat\ell(\widehat\ell+1)+2)
\\
&
+\frac{M}{6} (m_0-n-1) \Big(m_0-n-1 +    \frac{9}{2(n-1)(n-2)}\Big)
\\
&
+M\frac{\widehat\ell(\widehat\ell+1)-2}{12(m_0-n+2)}\Big( 5m_0-5n+2-     \frac{9}{(n-1)(n-2)  }\Big)
\,,
\\
b_{m_0,n}^{\widehat\ell}
=&
-3M \frac{\Big((m_0-n+3) (m_0-n+4)-\widehat\ell(\widehat\ell+1)\Big) \Big((m_0-n+2) (m_0-n+3)-\widehat\ell(\widehat\ell+1)\Big)}{n(n-1)(n-2)(n-3)(m_0-n+3)(m_0-n+4)} 
\\
&
-15M \frac{(m_0-n+3) (m_0-n+4)-\widehat\ell(\widehat\ell+1)}{2n(n-2)(n-3)(m_0-n+4)}  
- \frac{9M}{2(n-1)(n-2)}
\,.
\end{align*}
%
%The remaining relevant coefficients are
%\tim{check again}
%%
%\begin{align*}
%a_{m_0,3}^{\widehat\ell}
%=&
%\frac{M}{12}(23m_0^2-39m_0-212)  -M\Big(   \frac{35m_0-146}{12(m_0-2)}-\frac{2(m_0+1)}{m_0-3}\Big) (\widehat\ell(\widehat\ell+1) -2) 
%\\
%&
%-M\frac{2(\widehat\ell(\widehat\ell+1)-2)(\widehat\ell(\widehat\ell+1)-2)}{(m_0-2)(m_0-3)} 
%\,,
%\\
%a_{m_0,2}^{\widehat\ell}
%=&
%-3M\Big(m_0-3
%+\frac{3(\widehat\ell(\widehat\ell+1) - 2)}{2(m_0-2)}\Big)
%\,,
%\\
%a_{m_0,1}^{\widehat\ell}
%=& 0
%\,,
%\\
%b_{m_0,3}^{\widehat\ell}
%=&
%\frac{M}{2}  \Big( m_0+4- \frac{\widehat\ell(\widehat\ell+1)}{m_0}\Big)
%\,,
%\\
%b_{m_0,2}^{\widehat\ell}
%=& -3M\Big(2  +   \frac{m_0+1}{m_0}\Big) 
%\,,
%\\
%b_{m_0,1}^{\widehat\ell}
%=& M\frac{3(m_0+2)}{m_0}  
%\,`.
%\end{align*}
%
One checks that \eq{simplified_no_logs} is equivalent to the following equation,
\begin{align*}
0
=&
 (\widehat\ell-2)(\widehat\ell-3)(\widehat\ell+1) (2\widehat\ell+3)(m_0-\widehat\ell+3)(m_0-\widehat\ell+2)(m_0-\widehat\ell+1)
\Big(a_{m_0,m_0-\widehat\ell+3}^{\widehat\ell} 
\pm  b_{m_0,m_0-\widehat\ell+3}^{\widehat\ell}\Big)
\\
&
 -3(\widehat\ell-1)(\widehat\ell-2)(2\widehat\ell+3)(m_0-\widehat\ell+2)(m_0-\widehat\ell+1)(m_0-\widehat\ell-1)
  \Big(\widehat\ell a_{m_0,m_0-\widehat\ell+2}^{\widehat\ell} 
\pm (\widehat\ell+2) b_{m_0,m_0-\widehat\ell+2}^{\widehat\ell}\Big)
\\
&
+3(\widehat\ell-1)(2\widehat\ell -1)(m_0-\widehat\ell+1)(m_0-\widehat\ell-1)(m_0-\widehat\ell-2)
\Big(\widehat\ell(\widehat\ell+1)a_{m_0,m_0-\widehat\ell+1}^{\widehat\ell} 
\pm  (\widehat\ell+2)(\widehat\ell+3) b_{m_0,m_0-\widehat\ell+1}^{\widehat\ell}\Big)
\\
&
 -\widehat\ell(2\widehat\ell -1) (m_0-\widehat\ell-1)(m_0-\widehat\ell-2)(m_0-\widehat\ell-3)
\Big(\widehat\ell(\widehat\ell+1)a_{m_0,m_0-\widehat\ell}^{\widehat\ell} 
\pm  (\widehat\ell+3)(\widehat\ell+4) b_{m_0,m_0-\widehat\ell}^{\widehat\ell}\Big)
\,.
\end{align*}
Taking into account that the relevant range of $\widehat \ell$ is $2\leq \widehat \ell\leq m_0-4$
we observe that \eq{rest_ab} is actually not needed as the corresponding coefficients vanish anyway.
This equation can be determined explicitly. A \texttt{mathematica} computation shows that there is no contribution by the $b_{m_0,n}^{\widehat\ell}$-terms,%
\footnote{
Because of this the no-logs conditions for divergence and curl take an identical  form. In Section~\ref{sec_constat_mass_ex} this was not the case.
The reason for that is that for the derivation in this section we have used \eq{assumption_data2B}, while we have not used in Section~\ref{sec_constat_mass_ex} that $\ol\Xi^{(5)}$ and   $\ul\Xi^{(5)}$ are $\ell=2$-spherical harmonics.
}
and that  
the right-hand side is given by the surprisingly simple expression
%
%\begin{align*}
%\mathcal{P}_{m_0}^{\widehat \ell} =&-\frac{M}{4}\Big[
%5\Big(6 - \widehat \ell- \widehat \ell^2\Big)m_0^3 + \Big(132 -24 \widehat \ell - 55\widehat \ell^2 + 3\widehat \ell^3 + 4\widehat \ell^4\Big)m_0^2
%\\
%&+\Big(186 - 61 \widehat \ell - 193 \widehat \ell^2 + 77 \widehat \ell^3 + 159 \widehat \ell^4+ 52 \widehat \ell^5\Big)m_0
%+ 84 - 54 \widehat \ell - 193 \widehat \ell^2 + 87 \widehat \ell^3 + 257 \widehat \ell^4 + 147 \widehat \ell^5 + 32 \widehat \ell^6
%\Big]
%\,.
%\end{align*}
%
\begin{align}
\mathcal{P}_{m_0}^{\widehat \ell} =&\frac{M}{2}\Big[
5\Big(\widehat \ell^2 +  \widehat \ell- 6\Big)m_0^3 - \Big(132 -24 \widehat \ell - 55\widehat \ell^2 + 3\widehat \ell^3 + 4\widehat \ell^4\Big)m_0^2
\nonumber
\\
&-\Big(186 - 61 \widehat \ell - 193 \widehat \ell^2 + 77 \widehat \ell^3 + 159 \widehat \ell^4+ 52 \widehat \ell^5\Big)m_0
- 84 + 54 \widehat \ell + 193 \widehat \ell^2 - 87 \widehat \ell^3 - 257 \widehat \ell^4 - 147 \widehat \ell^5 - 32 \widehat \ell^6
\Big]
\,,
\label{main_polynomial}
\end{align}
valid for $ 1\leq \widehat \ell \leq m_0-4$.

For  $m_0-3\leq \widehat \ell \leq m_0-1$
%we only need to compute the contributions from the   last two lines in \eq{some-no-log-form} as the contribution from the remaining terms is precisely the above polynomial.
we find with \eq{specific_data}  that divergence and curl of the no-logs condition \eq{some-no-log-form} become
\begin{align}
0
=&
  360\frac{(m_0-4)!}{(m_0-7)!} 
\Big((m_0-1)m_0-\widehat\ell(\widehat\ell+1)\Big)
\Big((m_0-2)(m_0-1)-\widehat\ell(\widehat\ell+1)\Big)
\nonumber
\\
&
\quad\times
\Big( \widehat\ell(\widehat\ell+1) a^{\widehat\ell}_{m_0,6}\pm  (m_0-3)(m_0-2) b^{\widehat\ell}_{m_0,6}\Big)
\nonumber
\\
&
+  120\frac{(m_0-3)!}{(m_0-6)!} 
  \prod_{\ell_1=\mathrm{max}(1,m_0-6)}^{m_0-6}\Big( \ell_1(\ell_1+1)-\widehat\ell(\widehat\ell+1)\Big)\Big((m_0-1)m_0-\widehat\ell(\widehat\ell+1)\Big)
\nonumber
\\
&
\quad\times
\Big( \widehat\ell(\widehat\ell+1) a^{\widehat\ell}_{m_0,5}\pm  (m_0-2)(m_0-1) b^{\widehat\ell}_{m_0,5}\Big)
\nonumber
\\
&
+ 24 \frac{(m_0-2)!}{(m_0-5)!} 
  \prod_{\ell_1=\mathrm{max}(1,m_0-6)}^{m_0-5}\Big( \ell_1(\ell_1+1)-\widehat\ell(\widehat\ell+1)\Big)
\Big( \widehat\ell(\widehat\ell+1) a^{\widehat\ell}_{m_0,4}\pm  (m_0-1)m_0 b^{\widehat\ell}_{m_0,4}\Big)
\nonumber
\\
&
- 3(m_0-2)(m_0-3)\prod_{\ell_1=\mathrm{max}(1,m_0-6)}^{m_0-4}\Big( \ell_1(\ell_1+1)-\widehat\ell(\widehat\ell+1)\Big)  \Big(  \widehat\ell(\widehat\ell+1) a^{\widehat\ell}_{m_0,3}\pm  m_0(m_0-1) b^{\widehat\ell}_{m_0,3}\Big)
\nonumber
\\
&
+ \frac{1}{2}(m_0-2)\prod_{\ell_1=\mathrm{max}(1,m_0-6)}^{m_0-3}\Big( \ell_1(\ell_1+1)-\widehat\ell(\widehat\ell+1)\Big)  \Big(  \widehat\ell(\widehat\ell+1) a^{\widehat\ell}_{m_0,2}\pm  m_0(m_0-1) b^{\widehat\ell}_{m_0,2}\Big)
\nonumber
\\
&
- \frac{1}{8}\prod_{\ell_1=\mathrm{max}(1,m_0-6)}^{m_0-2}\Big( \ell_1(\ell_1+1)-\widehat\ell(\widehat\ell+1)\Big)  \Big(  \widehat\ell(\widehat\ell+1) a^{\widehat\ell}_{m_0,1}\pm m_0(m_0-1) b^{\widehat\ell}_{m_0,1}\Big)
\,.
\label{no-logs_sepecial_cases}
\end{align}
To ensure that  $k$ is in is in the right range as required by \eq{some-no-log-form} we  set (cf.\ \eq{rest_ab}),
\begin{equation}
a_{m_0,n}=b_{m_0,n}=0 \quad \text{for} \quad  n\geq m_0
\,.
\label{rest_ab2}
\end{equation}
We have already considered the case $m_0=3$ in Section~\ref{sec_constat_mass_ex}.
To avoid a tedious case distinction it is convenient to check first  that the  above condition is violated
for  $m_0=4,5,6,7$ and $m_0-3\leq \widehat \ell \leq m_0-1$, which is just a matter of computation. 
We may then assume $m_0\geq 8$, for which  we obtain from \eq{no-logs_sepecial_cases}
\begin{align}
0
=&
  -20(m_0-5) (m_0-6) (m_0-2)(2m_0-3)
\Big( a^{m_0-3}_{m_0,6}\pm  b^{m_0-3}_{m_0,6}\Big)
\nonumber
\\
&
+  20(m_0-4) (m_0-5) (2m_0-3)
\Big( (m_0-3)a^{m_0-3}_{m_0,5}\pm  (m_0-1) b^{m_0-3}_{m_0,5}\Big)
\nonumber
\\
&
-4(m_0-4) (2m_0-7)
\Big((m_0-2)(m_0-3)a^{m_0-3}_{m_0,4}\pm  (m_0-1)m_0 b^{m_0-3}_{m_0,4}\Big)
\nonumber
\\
&
-(m_0-3)(2m_0-7) \Big( (m_0-2)(m_0-3) a^{m_0-3}_{m_0,3}\pm  m_0(m_0-1) b^{m_0-3}_{m_0,3}\Big)
\,,
\label{pre-poly1}
\\
0
=&
   -20(m_0-1)(m_0-4)(m_0-5)
\Big(  a^{ m_0-2}_{m_0,5}\pm b^{ m_0-2}_{m_0,5}\Big)
\nonumber
\\
&
+12 (m_0-3)(m_0-4)
\Big( (m_0-2) a^{ m_0-2}_{m_0,4}\pm  m_0 b^{m_0-2}_{m_0,4}\Big)
\nonumber
\\
&
+3 (m_0-3)(2m_0-5)\Big(   (m_0-2) a^{\widehat\ell}_{m_0,3}\pm  m_0b^{\widehat\ell}_{m_0,3}\Big)
\nonumber
\\
&
+(m_0-2)(2m_0-5) \Big(  (m_0-2)a^{\widehat\ell}_{m_0,2}\pm  m_0b^{\widehat\ell}_{m_0,2}\Big)
\,,
\\
0
=&
8 (m_0-3)(m_0-4)\Big(  a^{ m_0-1}_{m_0,4}\pm  b^{ m_0-1}_{m_0,4}\Big)
+ 6 (m_0-2)(m_0-3)\Big(  a^{ m_0-1}_{m_0,3}\pm  b^{ m_0-1}_{m_0,3}\Big)
\nonumber
\\
&
+2(2m_0-3) (m_0-2) \Big( a^{ m_0-1}_{m_0,2}\pm  b^{ m_0-1}_{m_0,2}\Big)
\nonumber
\\
&
+ (2m_0-3) (m_0-1)  \Big( a^{ m_0-1}_{m_0,1}\pm b^{ m_0-1}_{m_0,1}\Big)
\,,
\label{pre-poly3}
\end{align}
where, in addition to the above expressions for $a^{\widehat \ell}_{m_0,n}$ and $ b^{\widehat \ell}_{m_0,n}$ with $n\geq 4$, we obtain from \eq{ab3}-\eq{ab0},
\begin{align*}
a^{\widehat \ell}_{m_0,3}= &
%\frac{M}{3}(7m_0^2-6m_0-88)+ \frac{M(5m_0^2-86m_0+181)}{3(m_0-2)(m_0-3) }(-\widehat\ell(\widehat\ell+1)+2)
%\\
%&
%-\frac{8M}{3(m_0-2)(m_0-3)}(-\widehat\ell(\widehat\ell+1)+2)^2
%\\
%=&
\frac{M}{6}(23m_0^2-39m_0-212)
  + \frac{M(11m_0^2-227m_0+486)}{6(m_0-2)(m_0-3) }(-\widehat\ell(\widehat\ell+1)+2)
\\
&
-\frac{4M}{(m_0-2)(m_0-3)}(-\widehat\ell(\widehat\ell+1)+2)^2
\,,
\\
b^{\widehat \ell}_{m_0,3}=&  M  (m_0+4)- \frac{M}{m_0}\widehat\ell(\widehat\ell+1)
\,,
\\
a^{\widehat \ell}_{m_0,2}=& -\frac{15}{2}M(m_0-3) 
+M\frac{15}{2(m_0-2)}(-\widehat\ell(\widehat\ell+1)+2)
\,,
\\
b^{\widehat \ell}_{m_0,2}=& -M\frac{3(3m_0+1)}{m_0}
\,,
\\
a^{\widehat \ell}_{m_0,1}=&0\,, 
\\ 
 b^{\widehat \ell}_{m_0,1}=&M\frac{6 (m_0+2)}{m_0} 
\,.
\end{align*}
The right-hand sides of \eq{pre-poly1}-\eq{pre-poly3} yield the following polynomials, where, again, the $b$-terms drop out so that we get the same polynomials for divergence and curl,
\begin{align}
\mathcal{P}_{m_0}^{m_0-1} =& \frac{M}{3}  \Big(22 m_0^6-275 m_0^5+1345 m_0^4-3358 m_0^3+4777 m_0^2-3657 m_0+1146\Big)
\,,
\label{polynomial1}
\\
\mathcal{P}_{m_0}^{m_0-2} =&\frac{M}{3}  \Big(22 m_0^5-147 m_0^4+334 m_0^3-354 m_0^2+343 m_0-90\Big)
\,,
\label{polynomial2}
\\
\mathcal{P}_{m_0}^{m_0-3} =&-\frac{2}{3}M  \Big(11 m_0^4-26 m_0^3+4 m_0^2-m_0+48\Big)
\,.
\label{polynomial3}
\end{align}

\subsubsection{Analysis of the no-logs condition}

We want to show that the no-logs condition \eq{some-no-log-form} is violated  for data of the form \eq{specific_data} with a non-trivial $\Xi_{AB}^{(m_0+2)}$-term, i.e.\ that such data inevitably produce 
logarithmic terms.
This will be the case if and only if the polynomials $\mathcal{P}_{m_0}^{\widehat \ell}$ given by \eq{main_polynomial}, \eq{polynomial1}-\eq{polynomial3} do not have integer roots $\widehat\ell$
in the interval $[2, m_0-1]$ for any integer $m_0\geq 4$ (the case $m_0=3$ was treated in Section~\ref{sec_constat_mass_ex}).

We start with   \eq{polynomial1}-\eq{polynomial3}.
For $m_0\geq 17$  we have
\begin{align}
22 m_0^6 >& 275 m_0^5+1345 m_0^4+ 3358 m_0^3+4777 m_0^2+ 3657 m_0+ 1146
\,,
\\
22 m_0^5>& 147 m_0^4+334 m_0^3+354 m_0^2+343 m_0+90
\,,
\\
11 m_0^4>& 26 m_0^3+4 m_0^2+m_0+48
\,,
\end{align}
so that the polynomials cannot have any roots.
A straightforward computation shows that they do not have integer roots in the interval $[4,16]$.

For the analysis of \eq{main_polynomial}  it is convenient to treat $\widehat \ell$ as a parameter and read $\mathcal{P}_{m_0}^{\widehat \ell}$
as a polynomial in $m_0$.
First of all for $\widehat\ell=2$ we have
\begin{align*}
\mathcal{P}_{m_0}^{2} =&6M\Big(4m_0^2-343m_0-897\Big)
\,,
\end{align*}
which has no integer roots.
It remains to consider the cases where $\widehat \ell \geq 3$.
One easily checks that  $\mathcal{P}_{m_0}^{\widehat \ell}$ is negative at $m_0=0$, goes to plus infinity as $m_0\rightarrow \infty$
and has one stationary point for $m_0 >0$. It follows that $\mathcal{P}_{m_0}^{\widehat \ell} $ has exactly one root  for $m_0>0$. We needs to ensure
that this cannot be an integer.
%
%We compute the derivative
%%
%\begin{align*}
%\frac{\partial}{\partial m_0}\mathcal{P}_{m_0}^{\widehat \ell} =&-\frac{M}{4}\Big[
%15\Big(6 - \widehat \ell- \widehat \ell^2\Big)m_0^2 + 2\Big(132 -24 \widehat \ell - 55\widehat \ell^2 + 3\widehat \ell^3 + 4\widehat \ell^4\Big)m_0
%\\
%&+186 - 61 \widehat \ell - 193 \widehat \ell^2 + 77 \widehat \ell^3 + 159 \widehat \ell^4+ 52 \widehat \ell^5
%\Big]
%\,,
%\end{align*}
%%
%and  the derivative in $m_0$ vanishes at
%%
%\begin{equation*}
%m_0^{\pm} = \frac{132 - 24 \widehat \ell - 55\widehat \ell^2 + 3\widehat \ell^3 + 4\widehat \ell^4 \pm \sqrt{684+1944 \widehat \ell+ 5301 \widehat \ell^2 - 7308 %\widehat \ell^3 - 12113 \widehat \ell^4 - 1662 \widehat \ell^5 + 2734 \widehat \ell^6 + 804 \widehat \ell^7 + 16 \widehat \ell^8}}{15(\widehat \ell^2 + \widehat %\ell-6)}
%\end{equation*}
%%
%We have
%%
%\begin{equation*}
%m_0^-<0
%\,,
%\quad
%m_0^+ >\frac{\widehat\ell^2}{6}
%\,.
%\end{equation*}
As a polynomial of degree 3 the root can be computed explicitly,
\begin{align}
\overset{(0)}{m_0}(\widehat\ell)
%=&
%\frac{4 \widehat\ell^4+3 \widehat\ell^3-55 \widehat\ell^2-24 \widehat\ell+132}{15 (\widehat\ell^2+\widehat\ell-6)}
%\\
%&
%+\frac{1}{15 (\widehat\ell^2+\widehat\ell-6)} \Big\{\Big(a_{\widehat\ell} 
%-i\sqrt{b_{\widehat\ell} ^3 - a_{\widehat\ell} ^2}\Big)^{1/3}
%+ b_{\widehat\ell} 
%\Big(a_{\widehat\ell} 
%-i\sqrt{b_{\widehat\ell} ^3 - a_{\widehat\ell} ^2}\Big)^{-1/3}
%\Big\}
%\\
=&
\frac{4 \widehat\ell^4+3 \widehat\ell^3-55 \widehat\ell^2-24 \widehat\ell+132}{15 (\widehat\ell^2+\widehat\ell-6)}
\nonumber
\\
&
+\frac{1}{15 (\widehat\ell^2+\widehat\ell-6)} \Big\{\Big(a_{\widehat\ell} 
+i\sqrt{b_{\widehat\ell} ^3 - a_{\widehat\ell} ^2}\Big)^{1/3}
+ 
\Big(a_{\widehat\ell} 
-i\sqrt{b_{\widehat\ell} ^3 - a_{\widehat\ell} ^2}\Big)^{1/3}
\Big\}
\\
=&
\frac{4 \widehat\ell^4+3 \widehat\ell^3-55 \widehat\ell^2-24 \widehat\ell+132}{15 (\widehat\ell^2+\widehat\ell-6)}
+\frac{2\sqrt{b_{\widehat\ell}}}{15 (\widehat\ell^2+\widehat\ell-6)} \cos\Big(\frac{1}{3}\sqrt{b_{\widehat\ell} ^3/a_{\widehat\ell}^2 -1}\Big)
\\
%=& \frac{4}{5}\widehat \ell^2 + \frac{64}{5} \widehat \ell-\frac{709}{4}+ O(\widehat \ell^{-1})
=&\frac{16\widehat \ell^2 +  256\widehat \ell -3545 }{20}+ q(\widehat\ell)\widehat \ell^{-1}
%+ O(\widehat \ell^{-1})
\,,
\label{root}
\end{align}
with 
\begin{align*}
a_{\widehat\ell} :=&64 \widehat\ell^{12}+4824 \widehat\ell^{11}+30768 \widehat\ell^{10}+9180 \widehat\ell^9-284037 \widehat\ell^8-406854 \widehat\ell^7
\\
&
+585521 \widehat\ell^6+1228797 \widehat\ell^5+291384 \widehat\ell^4-293463 \widehat\ell^3-191484 \widehat\ell^2-8748 \widehat\ell+ 6048
\,,
\\
b_{\widehat\ell} :=&16 \widehat\ell^8+804 \widehat\ell^7+2734 \widehat\ell^6-1662 \widehat\ell^5-12113 \widehat\ell^4-7308 \widehat\ell^3+5301 \widehat\ell^2+1944 \widehat\ell+684
\,.
\end{align*}
We want to derive an estimate for $q(\widehat\ell)$.
Note that, $a_{\widehat\ell}$, $b_{\widehat\ell} $ and $ b_{\widehat\ell} ^3/a_{\widehat\ell} ^2-1$ are positive in the range of interest, so  $\overset{(0)}{m_0}(\widehat \ell)$   is a real function.
We set $x:=1/\widehat \ell$, then
\begin{align*}
q(\widehat\ell(x)) 
%\equiv&\overset{(0)}{m_0}(\widehat\ell)- \frac{16\widehat \ell^2 +  256\widehat \ell -3545 }{20}
%\\
%=&
%-\frac{32\widehat\ell^4 + 804 \widehat\ell^3 - 9935 \widehat\ell^2- 15147 \widehat\ell + 63282- 8\sqrt{b_{\widehat\ell}}\cos\Big(\frac{1}{3}\sqrt{b_{\widehat\ell} ^3/a_{\widehat\ell}^2 -1}\Big)}{60 (\widehat\ell^2+\widehat\ell-6)}
%\\
=&
-\frac{32 + 804x - 9935x^2- 15147 x^3 + 63282x^4- 8\sqrt{b_{x}}\cos\Big(\frac{1}{3}\sqrt{b_{x} ^3/a_{x}^2 -1}\Big)}{60x^3 (1+x-6x^2)}
\,,
\end{align*}
with
\begin{align*}
a_{x} :=&64+4824 x+30768x^2+9180 x^3-284037x^4-406854 x^5
\\
&
+585521 x^6+1228797x^7+291384x^8-293463x^9-191484x^{10}-8748 x^{11}+ 6048x^{12}
\,,
\\
b_{x} :=&16+804x+2734x^2-1662x^3-12113x^4-7308 x^5+5301x^6+1944 x^7+684x^8
\,.
\end{align*}
A Taylor expansion yields for $x< 10^{-3}$
\begin{align*}
\sqrt{b_{x}} =&  4+ \frac{201}{2}x -\frac{29465}{32}x^2 +q^{(1)}(x)x^3
\,, \quad |q^{(1)}| < 2\times 10^5
\,,
\\
\frac{1}{{60 (1+x-6x^2)}} =& \frac{1}{60}-\frac{1}{60} x +\frac{7}{60} x^2 + q^{(2)}(x)x^3 
\,, \quad |q^{(2)}| < 2
\,.
\end{align*}
We further find for $x<10^{-6}$  (set $c_x:=b_{x}^3/a_{x} ^2 -1$)
\begin{equation*}
%c_x =& \frac{92475 }{64}x^2 + q(x)x^3\,, \quad |q| \leq 10^9
%\,,
%\\
|\partial_x^3 c_x|< 6.8\times 10^{5}
\,,
\quad
\Big|\frac{(\partial_xc_x)^3}{c_x}\Big| < 20
\,,
\quad
\Big|\frac{(\partial_xc_x)^3}{c_x^2}-2\frac{\partial_xc_x\partial^2_xc_x}{c_x}\Big|< 9.1\times 10^{5}
\,,
\end{equation*}
which yields for $x<10^{-6}$
\begin{align*}
\Big|\partial^3_x\cos\Big(\frac{\sqrt{c_x}}{3}\Big)\Big|=&\frac{1}{24}\Big|\Big[\cos\Big(\frac{\sqrt{c_x}}{3}\Big)-\frac{3}{\sqrt{c_x}}\sin\Big(\frac{\sqrt{c_x}}{3}\Big)\Big]\Big(\frac{(\partial_xc_x)^3}{c_x^2}-2\frac{\partial_xc_x\partial^2_xc_x}{c_x}\Big)
\\
&
+ \frac{1}{9\sqrt{c_x}}\sin\Big(\frac{\sqrt{c_x}}{3}\Big)\Big(\frac{(\partial_xc_x)^3}{c_x}-36\partial_x^3c_x\Big)\Big|
\\
<& 10^{6}
\,.
\end{align*}
From this we obtain
\begin{equation*}
\cos\Big(\frac{1}{3}\sqrt{b_{x}^3/a_{x} ^2 -1}\Big) = 1-\frac{10275}{128}x^2 + q^{(3)}(x)x^3\,, \quad |q^{(3)}|< 1.7\times 10^{5}
\,,
\end{equation*}
and finally, for $x< 10^{-6}$
\begin{align*}
|q(\widehat\ell(x)) |
%=&
%x^{-2}\Big(32 + 804x - 9935x^2- 15147 x^3 + 63282x^4- 8\sqrt{b_{x}}\cos\Big(\frac{1}{3}\sqrt{b_{x} ^3/a_{x}^2 -1}\Big)\Big)
%\\
%&
%\Big(\frac{1}{60}-\frac{1}{60}x + \frac{7}{60}x^2 + qx^3\Big)
%\\
< 1.2\times 10^{5}
\,,
\end{align*}
%$|q|<2$
%
It follows that, for  $\widehat\ell > 10^6$
\begin{equation*}
\Big| \overset{(0)}{m_0}(\widehat\ell)-\frac{16\widehat \ell^2 +  256\widehat \ell -3545 }{20}\Big|
<1.2\times 10^{5}\times  \widehat\ell^{-1}
\,.
\end{equation*}
The numerator is always an odd number, so
the fraction differs from an integer by at least $1/20$, i.e.\ the polynomial cannot have integer roots if the right-hand side is smaller than $1/20$, i.e.\  for $\widehat \ell > 2.4 \times 10 ^{6}$.
It remains to be checked whether there are integer roots for  $3\leq \widehat \ell \leq  2.4 \times 10 ^{6}$.
Note that 
given $\widehat \ell $ there is only one root and this  is given by \eq{root}.
A \texttt{mathematica} computation shows that there are no integer roots in this range of $\widehat \ell$.
Here a remark is in order: Due to the appearance of roots there might arise a problem to recognize $\overset{(0)}{m_0}(\widehat\ell)$ as an integer
due to numerical errors. We therefore rounded $\overset{(0)}{m_0}(\widehat\ell)$ to the nearest integers and plugged it in into the polynomial to
check whether it is a root.

By way of summary, at least in the setting of constant (ADM)  mass aspect $M$ and vanishing dual (ADM) mass aspect $N$ we conclude that while $2\leq \ell\leq m_0-1$-spherical harmonics in $\Xi^{m_0+2}$
do not produce logarithmic terms at order $m_0-2$, they do produce logarithmic terms in the next order.
Taking Lemma~\ref{lem_inv} into account we end up with the following result:

\begin{theorem}
Consider a solution  $(e^{\mu}{}_i,  \widehat \Gamma_i{}^j{}_k, \widehat L_{ij}, W_{ijkl})$ of  the GCFE  with constant mass aspect $M$ and vanishing dual mass aspect $N$
which is smooth  at $\scri^-\cup I^-\cup I$
in some  weakly  asymptotically Minkowski-like conformal Gauss gauge.
Then the expansion of the radiation field
vanishes  at $I^-$ at any order.
\end{theorem}

\begin{remark}
{\rm
If the solution is only $C^{3m_0-1}$  the above computation shows that $\Xi^{(m+2)}_{AB}=0$ for $m \leq m_0$.
}
\end{remark}

\begin{remark}
{\rm
We expect this result to remain true for arbitrary $M$ and $N$.
}
\end{remark}

\section*{Acknowledgments}
The author wishes to thank  Helmut Friedrich and Juan A.\ Valiente Kroon for useful discussions and comments
on the manuscript. The author is also thankful to the Max Planck Institute for Gravitational Physics in Golm, Germany, for hospitality, where part
of the work on this paper has been done.
Financial support by the Austrian Science Fund (FWF) P 28495-N27 is gratefully acknowledged, as well.

\appendix

\section{Asymptotic initial value problem}
\label{app_app}

\subsection{Characteristic constraint equations on $\scri^-$}
\label{app_charact_constraints}
\label{app1}

There are different versions of the CFE using different sets of unknowns.
Here let us focus  on  the \emph{metric conformal field equations (MCFE)} \cite{F3},
where, besides  rescaled Weyl tensor $W_{\mu\nu\sigma\rho}$,  Schouten tensor $L_{\mu\nu}$ and conformal factor $\Theta$, the metric $g_{\mu\nu}$ and a certain scalar $s$ are regarded as unknowns,
\begin{align}
 & \nabla_{\rho} W_{\mu\nu\sigma}{}^{\rho} =0\,,
 \label{conf1}
\\
 & \nabla_{\mu} L_{\nu\sigma} - \nabla_{\nu}L_{\mu\sigma} = \nabla_{\rho}\Theta \, W_{\nu\mu\sigma}{}^{\rho}\,,
 \label{conf2}
\\
 & \nabla_{\mu}\nabla_{\nu}\Theta = -\Theta L_{\mu\nu} + s g_{\mu\nu}\,,
 \label{conf3}
\\
 & \nabla_{\mu} s = -L_{\mu\nu}\nabla^{\nu}\Theta\,,
 \label{conf4}
\\
 & 2\Theta s  -\nabla_{\mu}\Theta\nabla^{\mu}\Theta = \lambda/3 \,,
 \label{conf5}
\\
 & R_{\mu\nu\sigma}{}^{\kappa}[ g] = \Theta W_{\mu\nu\sigma}{}^{\kappa} + 2\left(g_{\sigma[\mu} L_{\nu]}{}^{\kappa}  - \delta_{[\mu}{}^{\kappa}L_{\nu]\sigma} \right)
 \label{conf6}
\,.
\end{align}
In the following we assume  that the cosmological constant vanishes,
$$
\lambda=0
\,.
$$
We will recall the constraint equations induced by the MCFE   in a \emph{generalized wave-map gauge} 
in  \emph{adapted null coordinates} \cite{F4, CCM2}  $(\tau,r,x^{\mathring A})$ on $\scri^-$, cf.\ Section~\ref{sec_constr_gauss_gauge}.
 It seems worth
to emphasize that we do \emph{not} assume the existence of a regular point $i^-$ representing past timelike infinity. We
assume that the null geodesics generating $\scri^-$ emanate from $O=\{ \tau=-1, r=r_0\}$ which could be a point  which represents a (possibly regular) $i^-$, but which also could be a topological 2-sphere.
Spatial infinity (at least its ``intersection'' with $\scri^-$), which also could be a point or a 2-sphere, is located at $i^0=\{ \tau=-1, r=r_1\}$.

In comparison with \cite{ttp1} we present here a slightly modified system, which permits gauges where the scalar %$s$
$
s := \frac{1}{4}\Box_g \Theta + \frac{1}{24} R\Theta
$
vanishes on $\scri^-$, as crucial in view of  a cylinder representation of spatial infinity (in fact the function $s$ is not needed in this scheme).
To do that it is convenient to regard  
\begin{equation}
\Sigma:=\nabla^r\Theta|_{\scri^-}\ne 0\,, \quad  \kappa\,, \quad  \text{and the gauge source functions $W^{\mu}$  \cite{F4}}
\end{equation}
as the ``gauge data'' (rather than $\kappa$, $s|_{\scri^-}$, $W^{\mu}$).
These data are supplemented by  ``non-gauge'' data.
On $\scri^-$, we take
(note that this differs slightly from the data used in \cite{ttp1}, and that further data need to be prescribed on an incoming null hypersurface as described in Appendix~\ref{app_ADM})
\begin{equation}
\Xi_{\mathring A\mathring B} := -2( \Gamma{}^r_{\mathring A\mathring B})_{\mathrm{tf}} =\nu^{\tau} (\partial_{\tau} g_{\mathring A\mathring B})_{\mathrm{tf}}
-2 \nu^{\tau}(\rnabla_{(\mathring A}\nu_{\mathring B)})_{\mathrm{tf}}
\;.
\end{equation}
It is related to the radiation field via \eq{expression_d1A1B}.

The constraint equations on $\scri^-$ form a hierarchical system of ODEs and algebraic equations  (cf.\ \cite{ttp1}, but note that in \cite{ttp1} a regular vertex has been assumed whence some equations take a slightly different form here),
\begin{align}
 \sigma_{\mathring A\mathring B} =&0
\,,
\label{constraint1}
%\\
%  s|_{\scri^-} =& (\partial_r+\kappa)\Sigma
%\,,
%\label{constraint2}
\\
\theta^+ =&  \frac{2 }{\Sigma} (\partial_r+\kappa)\Sigma
\,,
\label{constraint3}
\\
 L_{rr} |_{\scri^-}
=& 
%    -\frac{\partial_r s}{\Sigma}
%=
-\frac{1}{2}\Big(\partial_r + \frac{1}{2}\theta^+ -\kappa\Big)\theta^+
\,,
\label{constraint_Lrr}
\\
  \Big(\partial_r +\frac{1}{2}\theta^+ + \kappa\Big)\nu^{\tau}  =&-  \frac{1}{2} W^{\tau}
\,,
\\
\partial_{\tau}\Theta |_{\scri^-} =& \nu_{\tau} \Sigma
\,,
\\
 \xi_{\mathring A} 
=& 2\rnabla_{\mathring A}\log|\Sigma| 
\,,
\label{constraint_xiA}
\\ 
\Big(\partial_r   +\frac{1}{2}\theta^+ + \kappa \Big) g^{r\mathring A}|_{\scri^-}
 =&
 \frac{1}{2}\Big( \xi^{\mathring A} -  W{}^{\mathring A} +  g^{\mathring B\mathring C}\rGamma^{\mathring A}_{\mathring B\mathring C} \Big)
\,,
\\
 L_{r\mathring A} |_{\scri^-}
 =& - \frac{1}{2}\Big(\rnabla_A  +\frac{1}{2}\xi_{\mathring A}\Big)\theta^+ 
\,,
\label{constraint_LrA}
%\\
%R |_{\scri^-}   =&  3\Big(\partial_r + \frac{1}{2}\theta^+ + \kappa\Big)\theta^-  +  3\Big(\rnabla_{\mathring A}- \frac{1}{2}\xi_{\mathring A}\Big)\xi^{\mathring A} 
%\,,
\\
   g^{\mathring A\mathring B} L_{\mathring A\mathring B} |_{\scri^-}=& \frac{1}{4}\theta^+\theta^-  + \frac{1}{2}\not \hspace{-0.2em}R
\,,
\label{constraint_LABtr}
\\
\Big(\partial_r +\frac{1}{2}\theta^++\kappa  \Big) g^{rr}|_{\scri^-}=& \frac{1}{2} \theta^-  -   W{}^r
\,,
\\
4 L_r{}^r |_{\scri^-}
 =& (\partial_r + \kappa)\theta^-  + \Big(\rnabla_{\mathring A}- \frac{1}{2}\xi_{\mathring A}\Big)\xi^{\mathring A}  - g^{rr}\Big(\partial_r+\frac{1}{2}\theta^+ -\kappa \Big)\theta^+
- \not \hspace{-0.2em}R
\,.
\label{constraint12}
\end{align}
For completeness let us also provide the constraint for the function $s$,
\begin{equation}
s|_{\scri^-} =\frac{1}{2}\theta^+\Sigma\,.
\end{equation}
Here $\not \hspace{-0.2em}R$ denotes the curvature scalar associated to the Riemannian family $r\mapsto \not \hspace{-0.2em}g=g_{\mathring A\mathring B}\mathrm{d}x^{\mathring A}\mathrm{d}^{\mathring B}|_{\scri^-}$.
Moreover, $\sigma_{\mathring A\mathring B}$ denotes the shear, while the divergences $\theta^+$ and $\theta^-$ are defined in \eq{dfn_theta+}-\eq{dfn_theta-}. $\kappa$ and $\xi_{\mathring A}$ may be regarded here as auxiliary quantities.
%The vector field $V$ is determined by the gauge source functions (see \cite{ttp1, paetz_thesis} for the details).

In fact, it is more convenient to  regard the  metric coefficients $g^{r \mu}|_{\scri^-}$  as gauge functions which determine $W^{\mu}$ on $\scri^-$ \cite{ChPaetz}.
Off $\scri^-$ we use a conformal Gauss gauge, whence the gauge source
functions are basically irrelevant for our purposes.

%As compared to the ordinary characteristic Cauchy problem \cite{CCM2}, one does not need to deal with a non-linear Raychaudhuri-like equation on $\scri^-$.
 %The gauge functions $\Sigma$ and $\kappa$  cannot be prescribed completely arbitrarily, if one wants to end up with  a representation of timelike or spacelike infinity.
%\tim{add sth}

In a wave-map gauge %(cf.\ \cite{ttp1})
one usually regards the curvature scalar $R$ as a gauge function. Its restriction to $\scri^-$ is related to $\theta^-$ (which we regard here as gauge function) as follows,
$$
 R|_{\scri^-} = 3\Big(\partial_r + \frac{1}{2}\theta^+ + \kappa\Big) \theta^-+ 3\Big(\rnabla_{\mathring A}-\frac{1}{2}\xi_{\mathring A}\Big)\xi^{\mathring A}  
\,.
$$
%($\theta$ and $\xi_{\mathring A}$ can be expressed in terms of $\Sigma$ and $\kappa$).
One therefore needs to make sure that also the integration function which arises when solving this equation can be regarded as a conformal gauge freedom in order to make sure that $\theta^-$ can, indeed, be treated as a gauge function.
However, this is precisely what we have  accomplished  in Section~\ref{sec_constr_gauss_gauge}.

We observe that the constraint imply the following useful relations
\begin{equation}
%\theta^- =&2 g^{\mathring A\mathring B} \Gamma^r_{\mathring A\mathring B} +\theta^+  g^{rr}
%\;,
%\\
\Big(\partial_r+\frac{1}{2}\theta^+-\kappa\Big) 
\Sigma =0
\;,
\label{tau_s_relation}
\quad
\partial_r\xi_{\mathring A} = \rnabla_{\mathring A}(\theta^+-2\kappa)
\;.
\end{equation}
Note further  that when prescribing $ g^{r\mu}|_{\scri^-}$ and $\theta^-$ instead of $ W^{\mu}|_{\scri^-}$ and $R|_{\scri^-}$ the whole system becomes a system of algebraic equations
(in that case there remains the gauge freedom  to extend $ W^{\mu}|_{\scri^-}$  and $R|_{\scri^-}$, as computed algebraically from the constraints, off $\scri^-$).
%Moreover,
%%
%\begin{align}
%\partial_{\tau} g_{rr}|_{\scri^-} =& 2(\partial_r-\kappa)\nu_{\tau}
%\,,
%\\
% \partial_{\tau} g_{r\mathring A}|_{\scri^-}  =& (\rnabla_{\mathring A}+\xi_{\mathring A})\nu_{\tau} + (\partial_r - \theta^+)\nu_{\mathring A}
%\,,
%\\
% g^{\mathring A\mathring B}\partial_{\tau} g_{\mathring A\mathring B}|_{\scri^-}  =& 2 \rnabla_{\mathring A}\nu^{\mathring A}  - \nu_{\tau}(\theta^+  g^{rr} +%\theta^-) 
%\,.
%\end{align}
%
The components considered so far do \emph{not} involve the ``physical'' data $\Xi_{\mathring A\mathring B}$ and are purely determined by the gauge.

Before we continue let us provide a list of  the Christoffel symbols, or rather their restriction to $\scri$, which is straightforwardly obtained by rewriting the expressions given in
\cite[Appendix~A]{CCM2}, and which employ that the shear tensor vanishes on $\scri$ (the $\Gamma^{\mu}_{\tau\tau}$-components are not needed)
\begin{align}
\Gamma{}^{\tau}_{rr}|_{\scri^-} =& 0 \,=\, \Gamma{}^{\tau}_{r\mathring A}|_{\scri^-} \,=\,\Gamma{}^{\mathring C}_{rr}|_{\scri^-} 
\,,
\label{coordChristoffel1}
\\
\Gamma{}^{\mathring  C}_{r\mathring A} |_{\scri^-}=& \frac{1}{2}\theta^+\delta_A{}^C
\,,
\label{coordChristoffel2}
\\
\Gamma{}^{\tau}_{\mathring A\mathring B}|_{\scri^-} =&-\frac{1}{2}\theta^+\nu^{\tau} g_{\mathring A\mathring B}
\,,
\\
\Gamma{}^{\mathring  C}_{\mathring A\mathring B}|_{\scri^-} =& \rGamma^{ \mathring C}_{\mathring A\mathring B} +\frac{1}{2}\theta^+ \nu^{\tau}\nu^{\mathring C} g_{\mathring A\mathring B}
\,,
\label{coordChristoffel4}
\\
\Gamma{}^r_{rr}|_{\scri^-} =& \kappa
\,,
\label{dfn_kappa}
\\
\Gamma{}^{\tau}_{\tau r} |_{\scri^-}=& \nu^{\tau}\partial_r\nu_{\tau} - \kappa
\,,
\label{dfn_kappa2}
\\
\Gamma{}^r_{r\mathring A}|_{\scri^-} =& -\frac{1}{2}\xi_{\mathring A}
\,,
\label{dfn_xi}
\\
\Gamma{}^{\tau}_{\tau \mathring A} |_{\scri^-}=& \frac{1}{2}\xi_{\mathring A} + \nu^{\tau}\rnabla_{\mathring A}\nu_{\tau} - \frac{1}{2}\theta^+\nu^{\tau}\nu_{\mathring A}
\,,
\\
\Gamma{}^{\mathring C}_{\tau r}|_{\scri^-} =& \frac{1}{2}\nu_{\tau}\xi^{\mathring  C}+ \Big(\partial_r +\frac{1}{2}\theta^+ + \kappa -  \nu^{\tau}\partial_r\nu_{\tau} \Big)  \nu^{\mathring C}
\,,
\\
\Gamma{}^r_{\tau r} |_{\scri^-}=&
-\frac{1}{2}\nu_{\mathring A}\xi^{\mathring A} 
-\frac{1}{2}\nu_{\tau}(\partial_r+2\kappa) g^{rr}
\,,
\\
\Gamma{}^r_{\mathring A\mathring B}|_{\scri^-} =& -\frac{1}{2}\Xi_{\mathring A\mathring B} + \frac{1}{4}(\theta^-  - \theta^+  g^{rr}) g_{\mathring A\mathring B}
\,
\label{coordChristoffel11}
\\
\Gamma{}^{\mathring C}_{\tau \mathring A}|_{\scri^-} =& 
\frac{1}{2} \nu_{\tau} \Xi_{\mathring A}{}^{\mathring C}+ 
\Big(\rnabla_{\mathring A}- \frac{1}{2}\xi_{\mathring A}
+ \frac{\theta^+}{2}\nu^{\tau}\nu_{\mathring A}  -  \nu^{\tau}\rnabla_{\mathring A}\nu_{\tau} \Big)\nu^{\mathring C}
 - \frac{1}{4}(\theta^-  + \theta^+  g^{rr})  \nu_{\tau} \delta_{\mathring A}{}^{\mathring C}
\,,
\\
\Gamma{}^r_{\tau \mathring  A} |_{\scri^-}=&
-  \frac{1}{2}\nu_{\tau}(\rnabla_{\mathring A}-\xi_{\mathring A}) g^{rr}
 -\frac{1}{2}\nu^{\mathring B}\Xi_{\mathring A\mathring B} + \frac{1}{4}(\theta^-  - \theta^+  g^{rr}) \nu_{\mathring A}
\label{coordChristoffel13}
%\,,
%\\
%\Gamma{}^0_{00} |_{\scri^-}=& \nu^0\ol{\partial_0g_{01} }-\frac{1}{2}\nu^0\partial_1\ol g_{00}
%\,,
%\\
%\Gamma{}^C_{00}|_{\scri^-} =&  g^{CD}\ol{\partial_0g_{0D}}-\frac{1}{2}\widehat\nabla^C\ol g_{00} 
%-\nu^C \ol\Gamma{}^0_{00} 
%\,,
%\\
%\Gamma{}^1_{00} |_{\scri^-} =& \frac{1}{2}\nu^0\ol{\partial_0g_{00}} - \nu^0\nu_A\ol\Gamma{}^A_{00}
%-\nu^0\ol g_{00}\ol\Gamma{}^0_{00}
\,.
\end{align}
A  somewhat lengthy calculation, which makes heavily use of these formulas for the Christoffel symbols
and the constraint equations
reveals that
(this computation as the ones below have not been carried out in \cite{ttp1} for a general wave-map gauge),%
\footnote{It seems worth to stress that an analog to the to a large extent gauge-independent field $\Xi_{\mathring A \mathring B}$ can be defined for the ordinary characteristic Cauchy problem as well:
For this one sets on a characteristic initial surface $\Sigma$,  $\Xi_{\mathring A\mathring B}:=-2( \Gamma^r_{\mathring A\mathring B})_{\mathrm{tf}}- g^{rr}\sigma_{\mathring A\mathring B}|_{\Sigma}$.
One then checks that it satisfies the equation, 
\begin{equation*}
\Big( \partial_r - \frac{1}{2}\theta^+ +\kappa\Big) \Xi_{\mathring A\mathring B} %-2(\sigma_{(\mathring A}{}^C\Xi_{\mathring B)\mathring C})_{\mathrm{tf}}
- (\rnabla_{(\mathring A}\xi_{\mathring B)})_{\mathrm{tf}}
+ \frac{1}{2}(\xi_{\mathring A}\xi_{\mathring B})_{\mathrm{tf}}  - \frac{1}{2}\theta^-\sigma_{\mathring A\mathring B}  = -2(L_{\mathring A\mathring B})_{\mathrm{tf}}
\,,
\end{equation*}
cf.\ \eq{gen_eqn_Xi}. Note that in the vacuum case the right-hand side is determined  by the Einstein equations.
}
\begin{align}
( L_{\mathring A\mathring B})_{\mathrm{tf}}|_{\scri^-}
 =&\frac{1}{2}({\partial_{\mu}\Gamma^{\mu}_{\mathring A\mathring B}}  - \partial_{\mathring A}\Gamma^{\mu}_{\mathring B\mu} +\Gamma^{\nu}_{\mathring A\mathring B}\Gamma^{\mu}_{\nu\mu}
-  \Gamma^{\mu}_{\mathring A\nu}\Gamma^{\nu}_{\mathring B\mu})_{\mathrm{tf}}
\\
=&
-\frac{1}{2}\Big(\partial_{r}-\frac{1}{2} \theta^+ +\kappa\Big)\Xi_{\mathring A\mathring B}
+\frac{1}{2}(\rnabla_{(\mathring A}\xi_{\mathring B)})_\mathrm{tf}
-\frac{1}{4}  (\xi_{\mathring A}\xi_{\mathring B})_{\mathrm{tf}}
\label{LAB_constraint}
\,,
\\
 L_{\mathring  A}{}^r|_{\scri^-} =&\nu^{\tau}   L_{\tau \mathring A} +  g^{rr}  L_{r\mathring A} +  g^{r\mathring B}  L_{\mathring A\mathring  B}
\\
=&
 -\nu^{\tau} g_{\mathring A\mathring D}  g^{\mathring B\mathring C} R_{\tau \mathring B\mathring C}{}^{\mathring D} 
+ \nu_{\mathring A}\nu^{\tau}( 2  L_{ r}{}^r -  2g^{rr}  L_{rr}- g^{r\mathring B} L_{r\mathring B} )
+  g^{rr}  L_{r\mathring A} 
\\
=&
\frac{1}{2} \Big(\rnabla^{\mathring B} - \frac{1}{2}\xi^{\mathring B} \Big)\Big(\Xi_{\mathring A\mathring B} + \frac{1}{2}\theta^- g_{\mathring A\mathring B} \Big)
- \frac{1}{4}  g^{rr} \Big( \rnabla_{\mathring A} +\frac{1}{2}\xi_{\mathring A}\Big)   \theta^+
\,.
\label{LAr_constraint}
\end{align}
Let us  compute the independent components of the rescaled Weyl tensor.
First of all, we have
\begin{equation}
W_{r\mathring Ar\mathring B}|_{\scri^-} 
 =-\frac{1}{2\Sigma}\Big[
\partial_r\Big((\partial_{r}-\theta^+ +\kappa)\Xi_{\mathring A\mathring B}-(\rnabla_{(\mathring A}\xi_{\mathring B)})_{\mathrm{tf}}+\frac{1}{2}  (\xi_{\mathring A}\xi_{\mathring B})_{\mathrm{tf}}\Big)
-(\rnabla_{\mathring A}\rnabla_{\mathring B}\theta^+  )_{\mathrm{tf}}
 \Big]
\,.
\label{expression_d1A1B}
\end{equation}
In \cite{ttp1} we have used certain components of $\nabla_{\rho}W_{\mu\nu\sigma}{}^{\rho}=0$ to determine
$W_{\tau rr \mathring A}$ and $ W_{\tau r \mathring A\mathring B}$ on $\scri^-$  by integrating ODEs along the null geodesic generators of $\scri^-$.
However, it is  more convenient to employ approrpriate components of $2\nabla_{[\mu}L_{\nu]\sigma} =\nabla_{\rho}\Theta W_{\nu\mu\sigma}{}^{\rho}$,
 which yields algebraic equations from the outset,
\begin{align}
W_{r\mathring Ar}{}^r|_{\scri^-}
=& \frac{1}{\Sigma}\Big[(\partial_{r}+\kappa) L_{\mathring A}{}^{r}
- (\rnabla_{\mathring A}-\frac{1}{2}\xi_{\mathring A}) L_{r}{}^{r}   
+  \frac{1}{2}(\rnabla_{\mathring A}-\xi_{\mathring A}) g^{rr}   L_{rr} 
\nonumber
\\
&
-\frac{1}{2}\xi^{\mathring B}  L_{\mathring A\mathring B}
  -\frac{1}{2}\Big(\frac{1}{2} \theta^- +  (\partial_r -\frac{1}{2} \theta^+ +2\kappa)  g^{rr} \Big)  L_{r\mathring A}
+\frac{1}{2}\Xi_{\mathring A}{}^{\mathring B}  L_{r\mathring B }   
\Big]
\;,
\label{expression_d1A10}
\\
W_{\mathring A\mathring Br}{}^r|_{\scri^-} =& 
\frac{2}{\Sigma}\Big(( \rnabla_{[\mathring A} -\frac{1}{2}\xi_{[\mathring A}) L_{\mathring B]}{}^{r}
 -  \frac{1}{2}[(\rnabla_{[\mathring A}-\xi_{[\mathring A}) g^{rr}] L_{\mathring B]r}
  -\frac{1}{2}\Xi_{[\mathring A}{}^{\mathring C} L_{\mathring B]\mathring C} \Big)
\;.
\label{expression_d01AB}
\end{align}
While all the previous constraints can be read as algebraic equations, the remaining ones will be ODEs along the null geodesic generators of $\scri^-$.
To obtain them, we evaluate certain components of the equation $\nabla_{\rho}W_{\mu}{}^{\nu}{}_{\sigma}{}^{\rho}=0$ on $\scri^-$ which yields  ODEs for $W_r{}^r{}_r{}^r|_{\scri^-}$
and $W_{\mathring A}{}^r{}_r{}^r|_{\scri^-}$,
%
%\begin{align}
%0 =\nabla_{\mu}W_r{}^r{}_r{}^{\mu}|_{\scri^-}
%=&
% \Big(\partial_r+\frac{3}{2} \theta^+\Big)W_r{}^r{}_r{}^r + \Big(\rnabla^{\mathring A}-\frac{1}{2}\xi^{\mathring A}\Big)W_{r\mathring Ar}{}^r
%-\frac{1}{2}\Xi^{\mathring A\mathring B} W_{r\mathring Ar\mathring B}
%\,,
%\label{adm_ode}
%\\
%0 =\nabla_{\mu}W_{\mathring A}{}^r{}_{r}{}^{\mu}|_{\scri^-}
%=&  \Big(\partial_{r}   +\frac{1}{2}\theta^+ +\kappa\Big) W_{\mathring A}{}^r{}_{r}{}^{r}
%-  \frac{1}{2} g^{rr}\Big( \rnabla^{\mathring B}+\frac{1}{2}\xi^{\mathring B}\Big)W_{r\mathring Ar\mathring B}
%-\frac{1}{2}\Big(  \rnabla_{\mathring A} -\frac{3}{2}\xi_{\mathring A}\Big) W_{r}{}^r{}_{r}{}^r
%\nonumber
%\\
%&
%+\frac{1}{2}\Big(  \rnabla^{\mathring B} -\frac{3}{2}\xi^{\mathring B}\Big)W_{\mathring A\mathring Br}{}^{r}
%  +\Xi_{\mathring A}{}^{\mathring B}W_{r\mathring Br}{}^r  
%+\frac{1}{2}(\partial_r- \theta^+   +2\kappa) g^{rr}\, W_{r\mathring Ar}{}^{r}  
%\,.
%\label{adm2_ode}
%\end{align}
%
\begin{align}
 \Big(\partial_r+\frac{3}{2} \theta^+\Big)W_r{}^r{}_r{}^r|_{\scri^-} 
=&
- \Big(\rnabla^{\mathring A}-\frac{1}{2}\xi^{\mathring A}\Big)W_{r\mathring Ar}{}^r
+\frac{1}{2}\Xi^{\mathring A\mathring B} W_{r\mathring Ar\mathring B}
\,,
\label{adm_ode}
\\
 \Big(\partial_{r}   +\frac{1}{2}\theta^+ +\kappa\Big) W_{\mathring A}{}^r{}_{r}{}^{r}|_{\scri^-}
=& 
  \frac{1}{2} g^{rr}\Big( \rnabla^{\mathring B}+\frac{1}{2}\xi^{\mathring B}\Big)W_{r\mathring Ar\mathring B}
+\frac{1}{2}\Big(  \rnabla_{\mathring A} -\frac{3}{2}\xi_{\mathring A}\Big) W_{r}{}^r{}_{r}{}^r
  -\Xi_{\mathring A}{}^{\mathring B}W_{r\mathring Br}{}^r  
\nonumber
\\
&
-\frac{1}{2}(\partial_r- \theta^+   +2\kappa) g^{rr}\, W_{r\mathring Ar}{}^{r}  
-\frac{1}{2}\Big(  \rnabla^{\mathring B} -\frac{3}{2}\xi^{\mathring B}\Big)W_{\mathring A\mathring Br}{}^{r}
\,.
\label{adm2_ode}
\end{align}
From the CFE $\nabla_r L^{rr}  -   \nabla_{\tau} L^{r\tau }
 -  g^{rr} \nabla_{r} L_r{}^r -  g^{r\mathring A} \nabla_{\mathring A} L_r{}^r=\Sigma W_r{}^r{}_r{}^r$
and  using the Bianchi identity as well as, one more time, the above formulas for the Christoffels symbols,  we obtain
\begin{align}
2\Big(\partial_r +\frac{1}{2}\theta^++2\kappa\Big)  L^{rr}|_{\scri^-}   =& \Big(  g^{rr}\partial_r  +  2(\partial_r +\frac{1}{4}\theta^+ + 2\kappa ) g^{rr}  +  \frac{1}{2}\theta^- \Big)  L_r{}^{r} 
 -\frac{1}{2} g^{rr} [ (\partial_r+2\kappa) g^{rr}] L_{rr}
\nonumber
\\
&
 +  \frac{1}{2} [(\rnabla^{\mathring A}- 2\xi^{\mathring A}) g^{rr}] L_{r\mathring A}
 - \Big(  \rnabla^{\mathring A}-\frac{5}{2}\xi^{\mathring A}  \Big)  L_{\mathring A}{}^{r}  
 +\frac{1}{2}\Xi^{\mathring A\mathring B}  L_{\mathring A\mathring B}
\nonumber
\\
&
- \frac{1}{8}(\theta^-  - \theta^+  g^{rr})\Big(\not\hspace{-.2em} R + \frac{1}{2}\theta^+\theta^-\Big)
    +   \frac{1}{6}{\nabla^r R}
+\Sigma W_r{}^r{}_r{}^r
\;.
\end{align}
Finally,
% (the constraint for $L_{\tau\tau}$ is not needed here),
 let us derive an equation for $(W_{\mathring A}{}^r{}_{\mathring B}{}^r)_{\mathrm{tf}}$,  somehwat more explicitly as compared to \cite{ttp1}.
From the algebraic symmetries of the Weyl tensor it follows that
\begin{align*}
W^r{}_{(\mathring A\mathring B)r}|_{\scri^-} =&  g^{r\mu} W_{\mu(\mathring A\mathring B)r}
= -  \frac{1}{2}  g^{rr}W_{r\mathring Ar\mathring B} -\frac{1}{4} g_{\mathring A\mathring B} g^{\mathring C\mathring D} g^{\mathring E\mathring F} W_{\mathring C\mathring E\mathring F\mathring D}
\,,
\\
W^r{}_{\mathring A\mathring B\mathring C}|_{\scri^-} =&  g^{r\mu}W_{\mu \mathring A\mathring B\mathring C}
= ( W^r{}_{r}{}^r{}_{\mathring B}
 -  g^{rr}  W^r{}_{rr\mathring B} 
) g_{\mathring A\mathring C}
+ f_{\mathring C}  g_{\mathring A\mathring B}
\,,
\end{align*}
where the specific form of $f_{\mathring C}$ is irrelevant.
The Bianchi identity and the algebraic symmetries  of the rescaled Weyl tensor imply
\begin{align*}
({\nabla_{r} W_{\mathring A}{}^r{}_{\mathring B}{}^{r}} )_{\mathrm{tf}} |_{\scri^-}  =& (\nu^{\tau}{\nabla_{\tau} W^r{}_{(\mathring A\mathring B)r}}   + {\nabla_{\mathring C} W^1{}_{(\mathring A\mathring B)}{}^{\mathring C}} )_{\mathrm{tf}} 
\\
  =& \Big(\nu^{\tau}{\nabla_{\tau} W^r{}_{(\mathring A\mathring B)r}}   -\frac{1}{2} g^{rr} g^{r\mathring C} { \nabla_{\mathring C} W_{r\mathring Ar\mathring B}} +{\nabla_{(\mathring A}   W_{\mathring B)}{}^r{}_{r}{}^r}
 -   g^{rr} { \nabla_{(\mathring A}  W_{\mathring B)rr}{}^r}
\Big)_{\mathrm{tf}}  
\,,
\end{align*}
as well as
\begin{align*}
({\nabla_{\tau } W^r{}_{(\mathring A\mathring B)r}})_{\mathrm{tf}}|_{\scri^-}  =&  -({ \nabla_{\tau} W^{\tau}{}_{(\mathring A\mathring B)\tau}}  + {\nabla_{\tau} W^{\mathring C}{}_{(\mathring A\mathring B)\mathring C}})_{\mathrm{tf}}
\\
=&  - ({\nabla_{\tau} W^{\tau}{}_{(\mathring A\mathring B)\tau}} - \nu_{(\mathring A}{\nabla_{|\tau|} W^{\tau}{}_{\mathring B)r}{}^{r}})_{\mathrm{tf}}
\\
=& ({\nabla_{\alpha} W^{\alpha}{}_{(\mathring A\mathring B)\tau}} 
- \nu_ {(\mathring A}{\nabla_{|\alpha|} W^{\alpha}{}_{\mathring B)r}{}^{r}}
 )_{\mathrm{tf}}
\\
=& \Big(\frac{1}{2}\nu_{\tau}( g^{rr})^2{\nabla_{r} W_{r\mathring Ar\mathring B} }
-\nu_{\tau}{\nabla_{r} W_{\mathring A}{}^r{}_{\mathring B}{}^{r} }
+2\nu_{(\mathring A}{\nabla_{|r|}  W_{\mathring B)}{}^r{}_{r}{}^r}
\nonumber
\\
&
 -  g^{rr}\nu_{(\mathring A}{\nabla_{|r|}   W_{\mathring B)rr}{}^r}
+ {\nabla_{\mathring C} W^{\mathring C}{}_{(\mathring A\mathring B)\tau}}
 - \nu_{(\mathring A}{\nabla_{|\mathring C|} W^{\mathring C}{}_{\mathring B)r}{}^{r}}
 \Big)_{\mathrm{tf}}
\\
=&\nu_{\tau} \Big(\frac{1}{2} g^{rr}(g^{rr}\nabla_{r}+ g^{r\mathring C}\nabla_{\mathring C}) W_{r\mathring Ar\mathring B}
-{\nabla_{r} W_{\mathring A}{}^r{}_{\mathring B}{}^{r} }
+\nabla_{(\mathring A}W_{\mathring B)}{}^{r}{}_{r}{}^r
 \Big)_{\mathrm{tf}}
\,.
\end{align*}
Altogether that yields
\begin{equation}
({\nabla_{r} W_{\mathring A}{}^r{}_{\mathring B}{}^{r}} )_{\mathrm{tf}} 
  \,=\, \frac{1}{4}( g^{rr})^2\nabla_{r} W_{r\mathring Ar\mathring B} 
 -  \frac{1}{2} g^{rr}  ({\nabla_{(\mathring A}  W_{\mathring B)rr}{}^r})_{\mathrm{tf}}  
+({\nabla_{(\mathring A}W_{\mathring B)}{}^{r}{}_{r}{}^r}
)_{\mathrm{tf}}  
\;,
\end{equation}
equivalently,
%\tim{seems to be correct as compared to \cite{ttp1}!!!!!!!!}
%
\begin{align}
 &\hspace{-4em} \Big(\partial_{r} -\frac{1}{2}\theta^+ +2\kappa \Big)\Big ( ( W_{\mathring A}{}^r{}_{\mathring B}{}^{r} )_{\mathrm{tf}}- \frac{1}{4} ( g^{rr})^2 W_{r\mathring Ar\mathring B} \Big)
\nonumber
\\
  \hspace{2em}=& 
 \Big[\Big(\rnabla_{(\mathring A} -\frac{5}{2}\xi_{(\mathring A} \Big)(W_{\mathring B)}{}^{r}{}_{r}{}^r -\frac{1}{2} g^{rr}  W_{\mathring B)rr}{}^r)
\Big]_{\mathrm{tf}}  
+\frac{3}{4}\Xi_{\mathring A\mathring B}W_{r}{}^{r}{}_{r}{}^r
-\frac{3}{4}\Xi_{(\mathring A}{}^{\mathring C}W_{\mathring B)\mathring C r}{}^r
\;.
\label{ODE_Wfinal}
\end{align}

A remark concerning the integration functions is in order which arise when integrating those constraints which are ODEs rather than algebraic equations.
We are interested in an analysis of the constraints near spatial infinity, and an asymptotic expansion of the solutions.
The integration functions bring in a global aspect which encodes information  of the data on the whole of null infinity and not just its asymptotic part near 
spatial infinity (such as e.g.\ the (ADM)  mass aspect).
In this context let us mention two possibilities to set up an asymptotic characteristic initial value problem:
\begin{enumerate}
\item[(i)]
The first one \cite{kannar} is to start with two characteristic surfaces intersecting a spherical cross section $S$, and with one of these surfaces representing $\scri^-$.
In that case the initial data for the constraint ODEs are determined at $S$. Some of the data will be determined by continuity requirements at $S$ whereas
other can be prescribed freely (cf.\ Appendix~\ref{app_ADM}).
\item[(ii)]
Alternatively \cite{ChPaetzInfCone, F_T} one may prescribe data on $\scri^-$, regarded as a future light-cone emanating from
some point $i^-$ which represents past timelike infinity.
Assuming this point to be regular in the spacetime to be constructed, the initial data are determined by regularity conditions there \cite{C1}.
\end{enumerate}

\subsection{Asymptotic initial value problem with prescribed  (ADM) mass and dual mass aspect}
\label{app_ADM}
\label{app2}

The conformal field equations (CFE) \cite{F1,F2} permit the formulation of an \emph{asymptotic Cauchy problem} where
some of the data are prescribed on (a piece of) null infinity $\scri$.
The simplest situation arises when considering two null hypersurfaces which intersect transversally in a smooth spherical cross section $S$, one of them 
representing an incoming null hypersurface and the other one null infinity.
In \cite{kannar} K\'ann\'ar has proved local well-posedness for the CFE in some future-neighbourhood of $S$.
For this he employs  that the CFE contain a symmetric hyperbolic system of evolution equations and a system of constraint equations, preserved under evolution.
Solutions to the constraints are constructed from suitably chosen freely prescribable ``seed data'', and well-posedness for the evolution equations follows
from Rendall's result \cite{Rendall}, which guarantees  existence  in some neighborhood to the future of the intersection sphere $S$.
This result has been improved recently in \cite{CCTW}, where it is shown that a solution exists in fact in some neighborhood to the future of the whole initial surface
(or rather of that part where the constraints admit a solution as there might be obstructions due to the non-linear Raychaudhuri equation).

The purpose of this appendix  is to split the required data on the initial surface into ``gauge data'', whose description is just a matter of choice,  and the remaining ``physical data''.
For this it is convenient to somewhat reformulate the asymptotic Cauchy problem  where the freedom to prescribe data on the incoming null hypersurface is, to some extent, shifted
to the freedom to prescribe certain global quantities such as the  mass aspect on the critical set $I^-$.
Our aim is to set up a scheme where as many data as possible can be freely  prescribed on $\scri^-$ and its future boundary $I^-$. 
Such a scheme turns out to be  very  convenient for the analysis of   the appearance of logarithmic terms at the critical sets of spatial infinity.
% and to set up asymptotic initial data sets where these are  smooth.
In doing so we will choose a wave-map  gauge which admits a representation of spatial infinity as a cylinder \`a la Friedrich.
%In fact,  the results in \cite{CCTW} do not  guarantee existence of a solution up and including  to the cylinder. However, we merely want to analyze the critical set $I^-$ where the cylinder ``touches''  null infinity.

\subsubsection{Gauge freedom}
\label{gauge_adapted_null}

Consider two   null hypersurfaces $\mathcal{N}$ and $\Sigma$  with transverse intersection along a smooth submanifold $S\cong {S}^2$.
We introduce adapted null coordinates $(\tau, r, x^{\mathring A})$ (cf.\ \cite{CCM2})
 so that $\Sigma$ coincides with the set $\{\tau=-1\}$, while $\mathcal{N}$ is given by $\{r=r_{\mathcal{N}}>0\}$ and the intersection
sphere $S$ corresponds to the set $\{\tau=-1, r=r_{\mathcal{N}}\}$.
The conformal factor is to be chosen in such a way that $\Sigma$ can be identified with $\scri^-$  and its future boundary $I^-$  in the emerging vacuum spacetime.

The conformal gauge freedom hidden in the MCFE \eq{conf1}-\eq{conf6}  arises from the freedom to choose the conformal factor~$\Theta$. It can be exploited in such a way
that e.g.
\begin{equation}
R=R^* \,, \quad\theta^+_{\mathcal{N}}=0\,, \quad \theta^+_{\Sigma}=0
\,, \quad g_{\mathring A\mathring B}|_S=s_{\mathring A\mathring B}
\,.
\label{conformal_freedom}
\end{equation}
Such a gauge can be realized as follows:
Assume we have been given a spacetime $(\mcM, g)$ and a conformal factor $\Theta$.
We apply a conformal rescaling $\Theta\mapsto \phi\Theta$ and $g\mapsto\phi^2g$ with $\phi>0$.
To realize the condition $R=R^*$ the function $\phi$ needs to satisfy a wave equation.
This leaves the freedom to prescribe $\phi$ on $\mathcal{N}\cup\Sigma$.
On $\mathcal{N}$ we have $\theta^+_{\mathcal{N}}\mapsto  \theta^+_{\mathcal{N}} + 2\partial_{\tau}\log\phi|_{\mathcal{N}}$, which becomes zero if 
the restriction of $\phi$ to $\mathcal{N}$ satisfies an appropriate ODE along each of the null geodesic generators on $\mathcal{N}$. 
Since any Riemannian metric on the 2-sphere is conformal to the standard metric, the initial data for $\phi|_{\mathcal{N}}$ at $S$ can be employed to
arrange that $g_{\mathring A\mathring B}|_S=s_{\mathring A\mathring B}$.
Finally, $\theta^+_{\Sigma}\mapsto  \theta^+_{\Sigma} + 2\partial_{r}\log\phi|_{\Sigma}$, whence $\theta^+_{\Sigma}=0$
is realized by  a function $\phi|_{\Sigma}$ which satisfies an ODE. The initial data follow from $\phi|_{\mathcal{N}}$ and continuity at $S$.
(Note that the solutions of both ODEs $\phi|_{\mathcal{N}}$ and $\phi|_{\Sigma}$ will be positive  since $\phi|_S>0$, which in turn implies that $\phi$
will be positive at least sufficiently close to $\mathcal{N}\cup\Sigma$.)

Next, we want exploit the coordinate gauge freedom in such a way that
\begin{equation}
 \kappa_{\mathcal{N}}=0\,, \quad \kappa_{\Sigma}=-\frac{2}{r}
\,,\quad \partial_{\tau}\Theta|_S=2 \,, \quad g_{\tau r}|_S =g_{\tau r}^*
\,.
\label{coordinate_freedom}
\end{equation}
For this, we consider coordinate transformations of the form $\tau\mapsto \tilde \tau=\tilde \tau(\tau, x^{\mathring A})$ and  $r\mapsto \tilde r=\tilde r(r,x^{\mathring A})$.
First of all we observe that \eq{conformal_freedom} remains invariant.
The gauge conditions  $\kappa_{\mathcal{N}}=0$ and $ \kappa_{\Sigma}=-\frac{2}{r}$ are arranged by solving second-order ODEs
for $\tilde \tau$ and $\tilde r$ along the null geodesic generators of $\mathcal{N}$ and $\Sigma$, respectively.
This still leaves the freedom to apply  transformations of the form
$\tau \mapsto p^{(1)} \tau + p^{(2)}$ on $\mathcal{N}$ and $r\mapsto \frac{q^{(1)} r}{1 + q^{(2)} r}$ on $\Sigma$ with $p^{(a)}$ and $q^{(a)}$ some functions on ${S}^2$.
We have imposed the conditions that $\scri^-=\{\tau=-1\}$, $S=\{ \tau=-1, r=r_{\mathcal{N}}\}$ and $I^-= \{\tau=-1, r=0\}$.
This requires $p^{(2)}=p^{(1)}-1$ and $q^{(1)}=1+q^{(2)}r_{\mathcal{N}}$. 
Applying both transformations  we find that
\label{coord_trafo_wavemap}
\begin{equation}
\partial_{\tau}\Theta|_S \enspace\mapsto\enspace \frac{1}{p^{(1)}} \partial_{\tau}\Theta
\,,
\quad
 g_{\tau r}|_S \enspace\mapsto\enspace  \frac{1}{p^{(1)}(1+q^{(2)}r_{\mathcal{N}})}g_{\tau r}
\,,
\label{gauge_freedom_S}
\end{equation}
 which clearly can be employed to realize \eq{coordinate_freedom}. The remaining gauge freedom will be fixed below.
Let us choose
\begin{equation}
r_{\mathcal{N}}=1\,.
\end{equation}

In addition to \eq{conformal_freedom} and \eq{coordinate_freedom} there remains the freedom to prescribe the gauge source functions $W^{\mu}$
(cf.\ \cite{CCM2, F4, F5})
 which capture
the freedom to choose coordinates off the initial surface. 
The gauge source functions (or rather their restrictions to $\mathcal{N}\cup\Sigma$) can be chosen in such a way that \cite{ttp1, paetz_thesis}
\begin{equation}
g^{\tau r}|_{\mathcal{N}}=r_{\mathcal{N}}=1\,, \quad g^{\tau r}|_{\Sigma}=r\,, \quad  g^{\tau\mathring A}|_{\mathcal{N}}=g^{\tau\tau}|_{\mathcal{N}}=g^{r\mathring A}|_{\Sigma}=0
\,,
\quad 
g^{rr}|_{\Sigma} =\chi(r)\,.
\end{equation}
The function $\chi(r)$ is a smooth, non-increasing cut-off function which is one on $[0,1/3]$ and zero on $[2/3 ,1]$.
The reason for the cut-off is that $\partial_{\tau}$ should be a null vector on $\Sigma$ close to $S$, so that it provides a parameterization of the null geodesic generators of $\mathcal{N}$, while we want it to be timelike close to $I^-$ to get there conformal Gauss coordinates  based on a congruence of  timelike conformal geodesics.

Finally we choose
%\tim{the gauge choice for $R$, $s$ etc. are compatible with the conformal Gaussian coordinates constructed in xx}
%
\begin{equation}
R^*|_{\mathcal{N}} =0\,, \quad R^*|_{\Sigma}=0
\,, \quad \partial_{\tau} R^*|_{\Sigma}=0
\,.
\label{gauge_cond_R}
\end{equation}
The gauge source functions and the curvature scalar are then extended to smooth spacetime functions, e.g.\ in such a way that one obtains conformal Gauss coordinates near spatial infinity. In which way they are chosen off the initial surface will be irrelevant for the following considerations.
Once this has been done, the gauge is  fixed, apart from the freedom to choose coordinates $(x^{\mathring A})$
on ${S}^2$ (which will be irrelevant for us).

\subsubsection{Constraint equations on an incoming null hypersurface}
%\subsection{Constraint equations on $\scri$}

In the gauge described in Section~\ref{gauge_adapted_null}  the constraint equations on $\scri^-\cong\Sigma$   have been derived in  Appendix~\ref{app_charact_constraints}.
%
% imply the following relations on $\Sigma$
%(note that $\zeta_{\Sigma}|_S=-\theta_{\mathcal{N}}|_S/2=0$),
%%
%\begin{eqnarray}
%  g_{\mathring A\mathring B}|_{\Sigma} &=& s_{\mathring A\mathring B}
%\,,
%\\
%  s|_{\Sigma} &=& 0
%\,,
%\\
%  L_{r\alpha}|_{\Sigma}
%&=&    0
%\,,
%%\\
%%\nu^0  &=& r
%%\,,
%\\
%\partial_{\tau}\Theta|_{\Sigma} &=&2r
%\,,
%\\
% \xi^{\Sigma}_{\mathring A} |_{\Sigma}
%&=& 0
%\,,
%%\\ 
%%\ol g^{1A}
%% &=&
%%0
%%\,,
%\\
%\zeta_{\Sigma}|_{\Sigma}    &=& 0
%\,,
%\\
%   g^{\mathring A\mathring B} L_{\mathring A\mathring B}|_{\Sigma} &=& 1
%\,,
%%\\
%%\overline g^{11} &=&1
%%\,,
%\\
% L_r{}^r |_{\Sigma}
% &=& - \frac{1}{2}
%\,.
%\end{eqnarray}
%
On $\Sigma$ it is convenient to regard $\Xi^{\Sigma}_{\mathring A\mathring B}$, equivalently the radiation field $W_{r\mathring Ar\mathring B}|_{\Sigma}$  supplemented by $\Xi^{\Sigma}_{\mathring A\mathring B}|_S$ and $\partial_r\Xi^{\Sigma}_{\mathring A\mathring B}|_S$, as the free ``physical'' initial data. 
In the ``standard'' approach  the initial data  for the ODEs for $W_r{}^r{}_r{}^r$, $W_{\mathring A}{}^r{}_r{}^r$,
$W_{\mathring A}{}^r{}_{\mathring B}{}^r$ and $L^{rr}$
 cannot be specified freely, but follow from the data given at $\mathcal{N}$
and the continuity requirement at $S$.
%\tim{all of them?}
Here, we want to present an approach where this procedure is reserved: We  prescribe  initial data for the ODEs  at $I^-$, i.e.\ at $r=0$. Then we solve the ODEs
and determine the data they induce at $S$, i.e.\ at $r=r_{\mathcal{N}}=1$, and choose the data on $\mathcal{N}$ in such a way that all the field are continuous at $S$
so that the results in \cite{kannar, Rendall, CCTW}
apply.
For this it becomes necessary to discuss the constraint equations on $\mathcal{N}$ as well.

In the gauge constructed above we have
\begin{equation}
R |_{\mathcal{N}}=0\,, \quad \theta^+_{\mathcal{N}}=\kappa_{\mathcal{N}}=0\,, \quad g^{\tau r}|_{\mathcal{N}}=1\,, \quad g^{\tau \mathring A}|_{\mathcal{N}}=g^{\tau\tau}|_{\mathcal{N}}=0
\,, \quad g_{\mathring A\mathring B}|_{S}=s_{\mathring A\mathring B}
\,.
\label{gauge_conds_N}
\end{equation}
As \emph{``physical'' initial data} we regard as e.g.\ in \cite{Rendall, CCM2} 
the
\begin{equation}
\text{conformal class of $g_{\mathring A\mathring B}|_{\mathcal{N}}\mathrm{d}x^{\mathring A}\mathrm{d}x^{\mathring B}$,}
\end{equation}
which is a smooth 1-parameter family of Riemannian metrics,
defined at least in some neighborhood of $S$, i.e.\ on $\mathcal{N} \cong [0,\varepsilon)\times {S}^2$.
Denote by $\gamma_{\mathring A\mathring B}\mathrm{d}x^{\mathring A}\mathrm{d}^{\mathring B}$ a representative of that conformal class.
The conformal factor $\Omega>0$ relating $\gamma_{\mathring A\mathring B}$ and $g_{\mathring A\mathring B}$, $g_{\mathring A\mathring B}=\Omega^2\gamma_{\mathring A\mathring B}$  needs to be chosen in such a way that 
$\theta^{\mathcal{N}}=0$, or, equivalently,
\begin{equation}
\partial_{\tau} \log\Omega  =-\frac{1}{4}\gamma^{\mathring A\mathring B} \partial_{\tau}\gamma_{\mathring A\mathring B}
\,.
\end{equation}
For a given initial datum $\Omega|_S>0$, which is computed from the gauge condition $\Omega^2\gamma_{\mathring A\mathring B}|_S=s_{\mathring A\mathring B}$, this equation determines a positive  function $\Omega$ and thus a Riemannian family $g_{\mathring A\mathring B}|_{\mathcal{N}}$.

%In the following we will simply assume that $g_{\mathring A\mathring B}|_{\mathcal{N}}$ with $\theta^{\mathcal{N}}=0$ has been given.
For smooth seed data $\gamma_{\mathring A\mathring B}$, $g_{\mathring A\mathring B}|_{\mathcal{N}}$ admits an expansion 
at the intersection sphere $S\cong{S}^2$ 
of the form
\begin{equation}
g_{\mathring A\mathring B} |_{\mathcal{N}}\sim  s_{\mathring A\mathring B} + \sum_{n=1}^{\infty} h^{(n)}_{\mathring A\mathring B}(1+\tau)^n 
\,.
\label{expansion_g}
\end{equation}
In fact this expansion will be the only relevant part of the data with regard to the problem we are interested in.
As ``non-gauge''-part of the asymptotic expansion \eq{expansion_g} one may regard  the trace-free part of the $h^{(n)}$'s:
Indeed, instead of $\gamma_{\mathring A\mathring B}$  we may prescribe  a set of $s$-tracefree tensors $( h^{(n)}_{\mathring A\mathring B})_{\mathrm{tf}}$, $n\in\mathbb{N}$, on $\mathbb{S}^2$.
The gauge condition $\theta^+_{\mathcal{N}}=0$ then determines all the traces $s^{\mathring A\mathring B}h^{(n)}_{\mathring A\mathring B}$ by solving a hierarchical system of algebraic equations.
This determines the expansion \eq{expansion_g} which then can be extended in any way  to a $\theta^+_{\mathcal{N}}=0$-family of Riemannian metrics on $\mathcal{N}$.

Continuity of $\partial_{\tau}g_{AB}$ at $S$ requires
\begin{equation}
\Xi^{\Sigma}_{\mathring A\mathring B}|_S = (\partial_{\tau} g_{\mathring A\mathring B})_{\mathrm{tf}}|_S = ( h^{(1)}_{\mathring A\mathring B} )_{\mathrm{tf}}
\,.
\end{equation}
The shear $\sigma^{\mathcal{N}}_{\mathring A}{}^{\mathring B}$ of $\mathcal{N}$ depends only  on the conformal class of $\gamma_{\mathring A\mathring B}$, cf.\ \cite{CCM2}. Its expansion at $S$ reads
\begin{equation}
\sigma^{\mathcal{N}}_{\mathring A}{}^{\mathring B} \equiv \frac{1}{2}\big( g^{\mathring B\mathring C}\partial_{\tau} g_{\mathring A\mathring C}\big)_{\mathrm{tf}}\big|_{\mathcal{N}} = \frac{1}{2}(  h^{(1)}_{\mathring A}{}^{\mathring B} )_{\mathrm{tf}} +
( h^{(2)}_{\mathring A}{}^{\mathring B} )_{\mathrm{tf}}(1+\tau)+ O(1+\tau)^2
\,.
\end{equation}
Angular indices which refer to fields defined on $\mathcal{N}$ are raised and lowered with $g_{\mathring A\mathring B}|_{\mathcal{N}}$ while those of their expansion coefficients
at $S$ are raised and lowered with $s_{\mathring A\mathring B}$.

Let us  determine all the remaining fields on $\mathcal{N}$ which are needed as initial data for the symmetric hyperbolic system of evolution equations implied by the MCFE.
From the definition of the Schouten tensor in terms of $g$ one finds that
\begin{equation}
L_{\tau\tau} |_{\mathcal{N}}=\frac{1}{2} R_{\tau\tau}[g]= -\frac{1}{2}|\sigma^{\mathcal{N}}|^2
\,,
\label{L00_eqn_N}
\end{equation}
where $|\sigma|^2:=\sigma_{\mathring A}{}^{\mathring B}\sigma_{\mathring B}{}^{\mathring A}$.
Here (and in what follows) we make extensively use of the expressions for the Christoffel symbols in adapted null coordinates computed in \cite[Appendix~A]{CCM2}.

Next, we evaluate the $(\tau\tau)$-component of \eq{conf3}, 
\begin{equation}
\partial^2_{\tau\tau}\Theta |_{\mathcal{N}}= -\Theta L_{\tau\tau}
\,, \quad \text{with} \quad  \Theta|_S=0\,, \quad \partial_{\tau}\Theta|_S=2
\,.
\end{equation}
The first initial datum makes sure that $S$ correspond to a cross-section of $\scri^-$ while the second datum is our gauge condition \eq{coordinate_freedom}.
In particular this yields the expansion
\begin{equation}
\Theta = 2(1+ \tau) + \frac{1}{24}|h^{(1)}_{\mathrm{tf}}|^2 (1+\tau)^3 + O(1+\tau)^4
\,.
\end{equation}
The $(\tau A)$-component of  \eq{conf3} together  with the definition of the Schouten tensor yields an expression for $ \xi^{\mathcal{N}}_A\equiv-\frac{1}{2}\Gamma^{\tau}_{\tau A}|_{\mathcal{N}}$ and $L_{\tau A}|_{\mathcal{N}}$,
\begin{align}
\Big(\partial_{\tau} - 2\frac{\partial_{\tau}\Theta}{\Theta}\Big) \xi^{\mathcal{N}}_{\mathring A}  =&
 2\rnabla_{\mathring B}\sigma^{\mathcal{N}}_{\mathring A}{}^{\mathring B}
+ 4\frac{\partial_{\mathring A}\partial_{\tau}\Theta}{\Theta} -4\sigma^{\mathcal{N}}_{\mathring A}{}^{\mathring B}\frac{\partial_{\mathring B}\Theta}{\Theta} 
\,,
\\
L_{\tau\mathring A} |_{\mathcal{N}}=&
\frac{1}{2} \rnabla_B\sigma^{\mathcal{N}}_{\mathring A}{}^{\mathring B} -\frac{1}{4} \partial_{\tau} \xi^{\mathcal{N}}_{\mathring A} 
\,.
\end{align}
The Levi-Civita connection associated to $g_{\mathring A\mathring B}|_{\mathcal{N}}$ is denoted by $\rnabla$.
The ODE for $\xi^{\mathcal{N}}_{\mathring A}$ takes the asymptotic form
\begin{equation}
\Big(\partial_{\tau}-\frac{2}{1+ \tau} + O(1+\tau)\Big) \xi^{\mathcal{N}}_{\mathring A} 
=\mcD_{\mathring B}( h^{(1)}_{\mathring A}{}^{\mathring B})_{\mathrm{tf}}
+ 2(1+\tau)\mcD_{\mathring B}(  h^{(2)}_{\mathring A}{}^{\mathring B})_{\mathrm{tf}} + O(1+\tau)^2
\,,
\end{equation}
where $\mcD$ denotes the Levi-Civita connection of $s_{\mathring A\mathring B}$.
This is a Fuchsian ODE and there remains a gauge freedom to prescribe 
\begin{equation}
\varsigma_{\mathring A} := \partial^2_{\tau\tau}\xi^{\mathcal{N}}_{\mathring A}|_S
\,.
\end{equation}
This corresponds to the freedom to prescribe the torsion 1-form on the intersection surface  of two null hypersurfaces intersecting transversally  in the 
physical spacetime $(\widetilde \mcM, \widetilde g)$ (cf.\ e.g.\ \cite{ChPaetz}).

In general, the asymptotic expansion of the solution of the $\xi^{\mathcal{N}}_A$-equation will involve logarithmic terms. The solution will be smooth at $S$ if and only
if a \emph{no-logs-condition holds},
\begin{equation}
\mathring\nabla_{\mathring B}(  h^{(2)}_{\mathring A}{}^{\mathring B})_{\mathrm{tf}}=0 \quad \Longleftrightarrow \quad (  h^{(2)}_{\mathring A\mathring B})_{\mathrm{tf}} =0
\,.
\label{no-logs_conditionA}
\end{equation}
This recovers the no-logs-condition derived in \cite{ChPaetz2}
expressed in the conformally rescaled spacetime and in our current gauge.
We assume that this condition is satisfied. Then
\begin{equation}
 \xi^{\mathcal{N}}_{\mathring A} =-\mcD_{\mathring B}( h^{(1)}_{\mathring A}{}^{\mathring B})_{\mathrm{tf}} (1+\tau) + \frac{1}{2}\varsigma_{\mathring A}(1+\tau)^2 + O(1+\tau)^3
\,.
\end{equation} 
The results in \cite{CCTW, ChPaetz2, ttp3}
then tell us that this already implies that there exists a smooth extension through $\scri^-$.
% no further log terms appear, supposing that the gauge is chosen appropriately, and we shall see that our gauge \eq{gauge_conds_N} is appropriate in this respect.

Taking the trace of the $(\mathring A\mathring B)$-component of  \eq{conf3} and combining it with  \eq{conf5} and the $\tau$-component of \eq{conf4},
the definition of the Schouten tensor and the gauge condition $ R|_{\mathcal{N}}=0$, we obtain the following system
($\not \hspace{-0.2em}R^{\mathcal{N}}$ denotes the curvature scalar associated to $g_{\mathring A\mathring B}|_{\mathcal{N}}$)
\begin{align}
\Big(\partial_{\tau}- \frac{\partial_{\tau}\Theta}{\Theta}  -  \frac{\partial^2_{\tau\tau}\Theta}{\partial_{\tau}\Theta}\Big)\partial_{\tau}\theta ^{-}_{\mathcal{N}} 
=&- \Big(\partial_{\tau} -  \frac{\partial^2_{\tau\tau}\Theta}{\partial_{\tau}\Theta}\Big)\Big( \not \hspace{-0.2em}R^{\mathcal{N}} +\rnabla^{\mathring A}\xi^{\mathcal{N}}_{\mathring A}-\frac{1}{2}|\xi^{\mathcal{N}}|^2 
\Big)   
- \frac{2\partial_{\tau} \not\hspace{-0.2em}\Delta\Theta }{\Theta} 
\nonumber
\\
&
+ \frac{2\partial^2_{\tau\tau}\Theta }{\Theta^2 \partial_{\tau}\Theta}\Big( \Theta \not\hspace{-0.2em} \Delta\Theta -   \rnabla_{\mathring A}\Theta \rnabla^{\mathring A}\Theta\Big)
 -\frac{4 L_{\tau\mathring A}\rnabla^{\mathring A}\Theta }{\Theta}
\,,
\\
s|_{\mathcal{N}} =& \frac{\Theta}{4} \Big( \not \hspace{-0.2em}R^{\mathcal{N}} + \rnabla^{\mathring A}\xi^{\mathcal{N}}_{\mathring A}-\frac{1}{2}|\xi^{\mathcal{N}}|^2 
+\partial_{\tau}\theta^{-}_{\mathcal{N}}\Big) + \frac{ \not\hspace{-0.2em} \Delta\Theta }{2}
 - \frac{\theta^{-}_{\mathcal{N}}}{4}\partial_{\tau}\Theta 
\,,
\\
   \partial_{r}\Theta |_{\mathcal{N}}  =&\frac{1}{\partial_{\tau}\Theta}\Big[\Theta s- \frac{1}{2}   \rnabla_{\mathring A}\Theta \rnabla^{\mathring A}\Theta\Big]
\,,
\\
g^{\mathring A\mathring B}L_{\mathring A\mathring B} |_{\mathcal{N}}= &\frac{1}{2}\not \hspace{-0.2em}R^{\mathcal{N}} +\frac{1}{2} \rnabla^{\mathring A}\xi^{\mathcal{N}}_{\mathring A}-\frac{1}{4}|\xi^{\mathcal{N}}|^2 
+ \frac{1}{2}\partial_{\tau}\theta^{-}_{\mathcal{N}}
\label{boundary_partial_zeta}
\,,
\\
    L_{\tau r} |_{\mathcal{N}} =&-\frac{1}{2} g^{\mathring A\mathring B}L_{\mathring A\mathring B}
\,.
\end{align}
Near $S$ the ODE for $\theta^{-\mathcal{N}}$ takes the form   (note that $\not\hspace{-0.2em}  R^{\mathcal{N}} +  \rnabla^{\mathring A}\xi_{\mathring A}^{\mathcal{N}} = 2+ O(1+\tau)^2$)
\begin{equation}
\Big(\partial_{\tau} -\frac{1}{1+\tau} +O(1+\tau)\Big)\partial_r\theta^{-}_{\mathcal{N}}=O(1+\tau)
\,,
\end{equation}
 and the boundary conditions
are  $\theta^{-}_{\mathcal{N}}|_S=-2\theta^{+}_{\Sigma}=0$ and $\varsigma:=\partial^2_{\tau\tau}\theta^{-}_{\mathcal{N}}|_S$,
whence
\begin{equation}
\theta^-_{\mathcal{N}}= \frac{1}{2}\varsigma (1+\tau)^2+O(1+\tau)^3
\,.
\end{equation}
Moreover,
\begin{equation}
g^{\mathring A\mathring B}L_{\mathring A\mathring B} |_{\mathcal{N}}=1
+ \frac{1}{2}\varsigma (1+\tau)+ O(1+\tau)^2
\,.
\end{equation}

The tracefree-part of the $(\mathring A\mathring B)$-component of \eq{conf3} combined with the definition of the Schouten tensor in terms of $g$
provides the following equations
%\tim{dfn $\Xi_{AB}$ and $\zeta$... sign}
%
\begin{align}
\Big( \partial_{\tau}-\frac{\partial_{\tau}\Theta}{\Theta}\Big) \Xi^{\mathcal{N}}_{\mathring A\mathring B} 
%-2(\sigma^{\mathcal{N}}_{(\mathring A}{} ^{\mathring C}\Xi^{\mathcal{N}}_{\mathring B)\mathring C})_{\mathrm{tf}}
=&
\Big( \rnabla_{(\mathring A}\xi^{\mathcal{N}}_{\mathring B)}-\frac{1}{2}\xi^{\mathcal{N}}_{\mathring A}\xi^{\mathcal{N}}_{\mathring B} +2\Theta^{-1} \rnabla_{\mathring A} \rnabla_{\mathring B}\Theta
\Big)_{\mathrm{tf}}
+\sigma^{\mathcal{N}}_{\mathring A\mathring B}\Big( \frac{\theta^-_{\mathcal{N}}}{2}+2\frac{\partial_{r}\Theta}{\Theta}\Big)
\,,
\label{gen_eqn_Xi}
\\
(L_{\mathring A\mathring B})_{\mathrm{tf}}= &\Big(\frac{1}{2}  \rnabla_{(\mathring A}\xi^{\mathcal{N}}_{\mathring B)}-\frac{1}{4}\xi^{\mathcal{N}}_{\mathring A}\xi^{\mathcal{N}}_{\mathring B}
+\frac{1}{4}\theta^-_{\mathcal{N}}\sigma^{\mathcal{N}}_{\mathring A\mathring B}
 -\frac{1}{2}\partial_{\tau}\Xi^{\mathcal{N}}_{\mathring A\mathring B}
%+\sigma^{\mathcal{N}}_{(\mathring A}{} ^{\mathring C}\Xi^{\mathcal{N}}_{\mathring B)\mathring C}
\Big)_{\mathrm{tf}}
\,,
\end{align}
where
\begin{equation}
\Xi^{\mathcal{N}}_{\mathring A\mathring B}:=-2( \Gamma^{\tau}_{\mathring A\mathring B})_{\mathrm{tf}}- g^{\tau\tau}\sigma^{\mathcal{N}}_{\mathring A\mathring B}|_{\Sigma}\,.
\end{equation}
Near $S$, the ODE for $\Xi^{\mathcal{N}}_{\mathring A\mathring B}$ is of the form 
\begin{equation}
\Big(\partial_{\tau} -\frac{1}{1+\tau}+ O(1)\Big)\Xi^{\mathcal{N}}_{\mathring A\mathring B}=O(1+\tau)\,,
\end{equation}
 and the data $\partial_{\tau}\Xi^{\mathcal{N}}_{\mathring A\mathring B}|_S$
 are determined by the data $\Xi^{\Sigma}_{\mathring A\mathring B}$ given on $\Sigma$, 
$$
\partial_{\tau}\Xi^{\mathcal{N}}_{\mathring A\mathring B}|_S=
\partial_r\Xi^{\Sigma}_{\mathring A\mathring B}|_S=: \Sigma^{(1)}_{\mathring A\mathring B}
\,.
$$

From the definition of the Weyl tensor we find
\begin{equation}
\Theta W_{\tau \mathring A\tau \mathring B}|_{\mathcal{N}}= R_{\tau \mathring A\tau \mathring B}- g_{\mathring A\mathring B}L_{\tau\tau}
=
 -g_{\mathring B\mathring C}\partial_{\tau} \sigma^{\mathcal{N}}_{\mathring A}{}^{\mathring C} 
=
-( h^{(2)}_{\mathring A\mathring B} )_{\mathrm{tf}} + O(1+\tau)\,=\, O(1+\tau)
\,,
\label{eqn_N_W0A0B}
\end{equation}
as follows from the no-logs condition \eq{no-logs_conditionA}.

%In our adapted null coordinates one may regard
%%
%\begin{equation}
%W_{\tau r\tau r }\,, \quad W_{\tau r\tau \mathring A}\,, \quad W_{\tau r r \mathring A}\,, \quad W_{\tau r \mathring A \mathring B}\,, \quad ( W_{\tau\mathring A \tau\mathring B})_{\mathrm{tf}}\,, \quad ( W_{r\mathring A  r \mathring B})_{\mathrm{tf}}
%\end{equation}
%%
%as the 10  independent components of the rescaled Weyl tensor on $\mathcal{N}$.
%The algebraic symmetries of the Weyl tensor imply the following relations 
%%
%\begin{eqnarray}
%g^{\mathring A\mathring B}W_{\tau \mathring  A \tau \mathring B}|_{\mathcal{N}} &=& 0
%\,,
%\\
%g^{\mathring A\mathring B}W_{r\mathring Ar\mathring B}|_{\mathcal{N}} &=& 0
%\,,
%\\
%W_{\tau \mathring A\mathring B\mathring C}|_{\mathcal{N}} &=& 2W_{r\tau\tau[C}g_{B]A}
%\,,
%\\
%W_{r[\mathring A\mathring B]\tau}|_{\mathcal{N}} &=& -\frac{1}{2} W_{r\tau\mathring A\mathring B}
%\,,
%\\
%W_{r(\mathring A\mathring B)\tau} |_{\mathcal{N}}&=& -\frac{1}{2}g_{\mathring A\mathring B}W_{\tau r \tau r}
%\,,
%\\
%W_{r\mathring A\mathring B\mathring C}|_{\mathcal{N}} &=& -2W_{r\tau r[\mathring C}g_{\mathring B]\mathring A}
%\,,
%\\
%W_{\mathring A\mathring B\mathring C\mathring D}|_{\mathcal{N}} &=&-2W_{\tau r\tau r} g_{\mathring C[\mathring A} g_{\mathring B]\mathring D}
%\,.
%\end{eqnarray}
%
Using the  algebraic symmetries of the Weyl tensor we extract  from \eq{conf2} the following set of equations,
\begin{align}
\partial_{\tau} \Theta \, W_{r\tau\tau \mathring A}|_{\mathcal{N}} =&   \partial_{\tau} L_{\tau\mathring A} -\Big(  \rnabla_{\mathring A} +\frac{1}{2}\xi^{\mathcal{N}} _{\mathring A}\Big)L_{\tau\tau}  + \sigma^{\mathcal{N}} _{\mathring A}{}^{\mathring B}L_{\tau \mathring B}+ \rnabla^{\mathring B}\Theta \, W_{\tau \mathring A\tau \mathring  B}
\,,
\\
\partial_{\tau}\Theta \, W_{r\tau\mathring A\mathring B}|_{\mathcal{N}}  =& 2  \Big(\rnabla_{[\mathring A}+\frac{1}{2}\xi^{\mathcal{N}} _{[\mathring A}\Big) L_{\mathring B]\tau} -2\sigma^{\mathcal{N}} _{[\mathring A}{}^{\mathring C}(L_{\mathring B]\mathring C})_{\mathrm{tf}}+ 2\rnabla_{[\mathring A}\Theta \, W_{\mathring B]\tau\tau r}
\,,
\\
 \partial_{\tau}\Theta \, W_{\tau r \tau r}|_{\mathcal{N}} 
=&
 \sigma^{\mathcal{N}} _{\mathring A\mathring B}( L^{\mathring A\mathring B})_{\mathrm{tf}}-2\partial_{\tau} L_{\tau r}
 - \Big( \rnabla^{\mathring A} -\frac{1}{2}\xi^{\mathring A}\Big)L_{\tau \mathring  A} +  \frac{1}{2}\theta^{-\mathcal{N}}L_{\tau\tau}
+  \rnabla^{\mathring A}\Theta  W_{r\tau\tau\mathring  A},
\\
  \partial_{\tau} L_{r\mathring A}|_{\mathcal{N}} 
=&-\partial_r\Theta \, W_{r\tau\tau\mathring  A}-\frac{1}{2} \rnabla_{\mathring A}\Theta \, W_{\tau r\tau r}+\frac{1}{2} \rnabla^{\mathring B}\Theta \, W_{r\tau\mathring  A\mathring  B}+ \rnabla_{\mathring A}L_{\tau r}
\nonumber
\\
&
+ \frac{1}{2}\xi^{\mathcal{N}} _{\mathring B} (L_{\mathring A}{}^{\mathring B})_{\mathrm{tf}}
-\frac{1}{2}\Xi^{\mathcal{N}} _{\mathring  A}{}^{\mathring B}L_{\tau \mathring  B}+ \frac{1}{4}\theta^{-\mathcal{N}}  L_{\tau \mathring  A}
\,,
\label{ODE_L0A_0}
\\
\partial_{\tau}\Theta \, W_{r\tau r\mathring  A} |_{\mathcal{N}} 
=& 
 - \rnabla_{\mathring A} L_{\tau r}-  \rnabla^{\mathring B} ( L_{\mathring A\mathring B})_{\mathrm{tf}}
+  \sigma^{\mathcal{N}} _{\mathring A}{}^{\mathring B}L_{r\mathring  B}
+\frac{1}{2}\Xi^{\mathcal{N}} _{\mathring A}{}^{\mathring B} L_{\tau \mathring  B}
+\frac{1}{4}\theta^{-\mathcal{N}}  L_{\tau \mathring A}
\nonumber
\\
&
+ \partial_r\Theta \, W_{r\tau\tau\mathring  A} + \rnabla_{\mathring A}\Theta \,W_{\tau r\tau r}
\,.
\end{align}
The intial data for \eq{ODE_L0A_0} follow from the data on $\Sigma$ by continuity,  $L_{r\mathring A}|_S=0$.

The $(\tau r r)$-component of \eq{conf2} together with the contracted second Bianchi identity %$\nabla_{\nu}L_{\mu}{}^{\nu}=\frac{1}{6}\nabla_{\mu}R$
provides an ODE for $L_{rr}|_{\mathcal{N}}$,
\begin{equation}
 2\partial_{\tau} L_{rr}|_{\mathcal{N}} =\partial_r\Theta \, W_{ \tau r\tau r}+  \rnabla^{\mathring A}\Theta \, W_{r\tau r \mathring  A}
+\theta^{-\mathcal{N}} L_{\tau r} 
 -\Big(  \rnabla_{\mathring A}-\frac{3}{2}\xi^{\mathcal{N}}_{\mathring A}\Big) L_{ r}{}^{\mathring A}
 +\frac{1}{2}\Xi^{\mathcal{N}}_{\mathring A\mathring B} (L^{\mathring A\mathring B})_{\mathrm{tf}} 
+ \frac{1}{6}\partial_rR
\,.
\end{equation}
Again, the  data at the intersection sphere follow from those on $\Sigma$,  $L_{rr}|_S=0$.

The remaining components for the rescaled Weyl tensor follow from the  Bianchi equation and algebraic symmetries of the Weyl tensor,
%\tim{check again}
%
\begin{equation}
%\Big(
\partial_{\tau} ( W_{r\mathring A r\mathring B})_{\mathrm{tf}} 
%- 3 (\sigma^{\mathcal{N}}_{(\mathring A}{}^{\mathring C} W_{\mathring B) r\mathring  C r})_{\mathrm{tf}}\Big)
\Big|_{\mathcal{N}}=
\frac{3}{4}\Xi^{\mathcal{N}}_{(\mathring A}{}^{\mathring C}W_{r\tau  \mathring B)\mathring  C}+\frac{3}{4}\Xi^{\mathcal{N}}_{\mathring  A\mathring  B}W_{\tau r\tau r}
 + \Big( (\rnabla_{(\mathring A} -\frac{7}{2}\xi^{\mathcal{N}}_{(\mathring A}) W_{\mathring B)r\tau r}
\Big)_{\mathrm{tf}}
\end{equation}
where the initial data are determined by the radiation field at $S$.

From all these equations one may determine smooth expansions of all the relevant fields on $\mathcal{N}$ near $S$ (assuming, as a matter of course, that the no-logs condition \eq{no-logs_conditionA} holds).

We obtain the following
\begin{proposition}
\begin{enumerate}
\item
Consider two smooth null hypersurfaces  $\mathcal{N}$ and $\Sigma$ in a $3+1$-dimensional manifold with transverse intersection along a
smooth submanifold $S\cong {S}^2$ in adapted null coordinates (so that $\mathcal{N}=\{r=1\}$, $\Sigma=\{\tau=-1\}$ and $S=\{\tau=-1, r=1\}$).
Given an  initial data set which consists of 
\begin{enumerate}
\item[(i)] a smooth family $\tau\mapsto \gamma_{\mathring A\mathring B}(\tau)$ of Riemannian metrics on $\mathcal{N}$,
\item[(ii)] a smooth family  $r\mapsto W_{\mathring A\mathring B}(r) $ of  symmetric, $s$-tracefree tensor fields on $\Sigma$ representing the radiation field,
\item[(iii)] a  function $\varsigma$, a 1-form $\varsigma_{\mathring A}$ and two symmetric trace-free tensors $ \Sigma^{(n)}_{\mathring A\mathring B}$, $n=0,1$ on $S$,
\end{enumerate}
and assume that the no-logs condition \eq{no-logs_conditionA} is satisfied at $S$. Then
there exists a unique smooth (continuous at $S$) solution $(\Theta, s, g_{\mu\nu}, L_{\mu\nu}, W_{\mu\nu\sigma\rho})$ to the characteristic constraint equations induced by the MCFE on $\mathcal{N}\cup\Sigma$ in the gauge described in Section~\ref{gauge_adapted_null}
such  that 
\begin{enumerate}
\item[(a)] $\gamma_{\mathring A\mathring B}=[g_{\mathring A\mathring B}]|_{\mathcal{N}}$,
\item[(b)] $W_{\mathring A\mathring B} = (W_{r \mathring  A r \mathring B})_{\mathrm{tf}}|_{\Sigma}$,
\item[(c)] $\varsigma=4 W_{\tau r \tau r}|_S$,  $\varsigma_{\mathring A}=8W_{\tau r \tau\mathring A}|_S$,
$ \Sigma^{(0)}_{\mathring A\mathring B}=\Xi^{\Sigma}_{\mathring A\mathring B}|_S$
and $ \Sigma^{(1)}_{\mathring A\mathring B}=\partial_r\Xi^{\Sigma}_{\mathring A\mathring B}|_S$.
\item[(d)] $\Sigma=\scri^-$.
\end{enumerate}
One may regard the set $(\gamma_{\mathring A\mathring B}, W_{\mathring A\mathring B}, \varsigma, \varsigma_{\mathring A}, \Sigma^{(0)}_{\mathring A\mathring B}, \Sigma^{(1)}_{\mathring A\mathring B} )$ as  ``physical'' seed data
 for the evolution equations.%
\footnote{
As  trace-free symmetric tensors on $S^2$, $ \Sigma^{(0)}_{\mathring A\mathring B}$ and $ \Sigma^{(1)}_{\mathring A\mathring B}$  are both determined (via Hodge decomposition) in terms of  2 functions.
Here we use a gauge where $r_{\mathcal{N}}=1$.
In fact this gauge freedom can alternatively be employed to prescribe one of these functions.
This freedom will  be relevant when it is shifted to $I^-$ in the next section, and  it is   shown in Section~\ref{sec_yet_another}
that this is possible.
%The field $\Sigma_{\mathring A\mathring B}$ involves a gauge freedom which arises from the freedom to choose $r_{\mathcal{N}}$.
}
\item Given seed data $(\gamma_{\mathring A\mathring B}, W_{\mathring A\mathring B}, \varsigma, \varsigma_{\mathring A}, \Sigma^{(0)}_{\mathring A\mathring B}, \Sigma^{(1)}_{\mathring A\mathring B} )$  it follows from the results in \cite{CCTW, kannar,  Rendall}
that an (up to conformal diffeomorphisms) unique solution to the MCFE exists in some neighborhood to the future of $\scri^-=\{\tau=-1, 0 < r\leq 1\}$ (and $\mathcal{N}$).
\end{enumerate}
\end{proposition}

\begin{remark}
{\rm
Note that the metric becomes singular at the critical set $I^-=\{\tau=-1,r=0\}$ so nothing can be said about spatial infinity in this scheme.
The solution also exists in some neighborhood to the future of $\mathcal{N}$, supposing that the Raychaudhuri equation does not produce conjugate points.
Since our main interest lies in the behavior near $\scri^-$ and the critical set $I^-$, this is irrelevant for our purposes.
}
\end{remark}

\subsubsection{An alternative initial data set}
\label{app_alternative_data}

In view of an analysis of the behavior of the CFE at spatial infinity  it is convenient to prescribe 
as many data as possible at the critical set $I^-\cong{S}^2$.
Instead of $(\tau\mapsto\gamma_{\mathring A\mathring B}(\tau), r\mapsto W_{\mathring A\mathring B}(r), \varsigma, \varsigma_{\mathring A}, \Sigma^{(0)}_{\mathring A\mathring B}, \Sigma^{(1)}_{\mathring A\mathring B}  )$ let us therefore consider the following initial data set
\begin{enumerate}
\item[(i)] 
 a smooth family  $r\mapsto W_{\mathring A\mathring B}(r) $ of  symmetric, $s$-tracefree tensor fields on $\scri^-$ representing the radiation field,
\item[(ii)] two  functions $M=\frac{1}{2}W_r{}^r{}_r{}^r|_{I^-}$ (which can be identified as the (ADM) mass aspect, cf.\ Section~\ref{sec_solution_constr_gen}) and  $N:=-\frac{1}{8}\epsilon^{\mathring A\mathring B}\mcD_{[\mathring A}\mcD^{\mathring C}\partial^2_r \Xi^{\scri^-}_{\mathring B]\mathring C}|_{I^-}$ (which can be identified as the dual mass aspect, i.e.\ a NUT-like parameter, cf.\ Section~\ref{sec_dual_mass}).
\item[(iii)]
a 1-form
$L_{\mathring A}=\frac{1}{2}\partial^2_rW_{\mathring A}{}^{r}{}_{ r}{}^r|_{I^-}$
 on $I^-$ (which is related to the angular momentum), and
\item[(iv)] a set $\{c^{(n+2,n)}_{\mathring A\mathring B}\}$
of   symmetric, $s$-tracefree tensors on $I^-$, where $c^{(n+2,n)}_{\mathring A\mathring B}$ corresponds to the 
 $(n+2)$th-order expansion coefficient of
$\partial^n_{\tau}(W_{\tau \mathring A\tau \mathring B})_{\mathrm{tf}}|_{I^-}$, $n\geq 0$.
\end{enumerate}
$M$ and $L_{\mathring A}$ provide the initial data for the ODEs  \eq{adm_ode}-\eq{adm2_ode}
for $W_r{}^r{}_r{}^r|_{\scri^-}$ and $W_{\mathring A}{}^{r}{}_{ r}{}^r|_{\scri^-}$ (note  that the latter one is of Fuchsian type at $I^-$). They substitute  the data $\varsigma$ and $\varsigma_{\mathring A}$ at $S$, to which they are, via the constraints, in one-to-one correspondence.
Similarly, the freedom to prescribe $\Sigma^{(0)}_{\mathring A\mathring B}\equiv \Xi^{\scri^-}_{\mathring A\mathring B}|_S$ and $\Sigma^{(1)}_{\mathring A\mathring B}\equiv \partial_r\Xi^{\scri^-}_{\mathring A\mathring B}|_S$ can be shifted by the second-order ODE \eq{expression_d1A1B} to the freedom to prescribe $\partial_r \Xi^{\scri^-}_{\mathring A\mathring B}|_{I^-}$ and 
$\partial^2_r \Xi^{\scri^-}_{\mathring A\mathring B}|_{I^-}$.
In the previous section we have chosen a gauge where $r_{\mathcal{N}}=1$. This gauge freedom, which arises from a freedom to rescale $r$
can be used to prescribe  $\mcD^{\mathring A} \mcD^{\mathring B}\partial^2_r \Xi_{\mathring A\mathring B}|_{I^-}$  instead (cf.\ Section~\ref{sec_yet_another}), so
that, via Hodge decomposition, the function $N$ is left as ``physical'' part of the data.
The second datum  $\partial_r \Xi^{\scri^-}_{\mathring A\mathring B}|_{I^-}$ needs to vanish if one requires the rescaled Weyl tensor to be bounded at $I^-$
(cf.\ \eq{boundedness_Weyl}), whence we do not  consider it here (it is tacitly assumed to be trivial).

Let us derive equations for $\partial_{\tau}^n (W_{\mathring A}{}^r{}_{\mathring B}{}^r)_{\mathrm{tf}}|_{\scri^-}$.
For this  set 
$
\nabla^{(n)}_{\tau} := \underbrace{\nabla_{\tau}\dots\nabla_{\tau}}_{\text{$n$ times}}
$.
Suppose  that the fields $(\partial^{k+1}_{\tau}\Theta , \partial^{k}_{\tau} s, \partial^{k+1}_{\tau} g_{\mu\nu} ,\partial^{k}_{\tau} L_{\mu\nu}, \partial^{k}_{\tau} W_{\mu\nu\sigma\rho})|_{\scri^-}$, $k\leq n-1$,
have been computed from appropriate smooth seed data.
We  employ the MCFE to compute  $(\partial^{n+1}_{\tau}\Theta , \partial^{n}_{\tau} s, \partial^{n+1}_{\tau} g_{\mu\nu} ,\partial^{n}_{\tau} L_{\mu\nu}, \partial^{n}_{\tau} W_{\mu\nu\sigma\rho})$ on $\scri^-$. It turns out that  the equations are algebraic, except the ones  for $\partial^n_{\tau}L_{\tau\tau}$  and for  certain
components of the metric and the  rescaled Weyl tensor. For $n\geq 1$ 
we have
\begin{align}
\nabla^{(n+1)}_{\tau}\Theta|_{\scri^-}  =& -\nabla^{(n-1)}_{\tau} (\Theta L_{\tau\tau}) + g_{\tau\tau}\nabla^{(n-1)}_{\tau}  s \,,
\label{trans_scri1}
\\
\nabla^{(n)}_{\tau} s|_{\scri^-}  =& - \nabla^{(n-1)}_{\tau} (L_{\tau\nu}\nabla^{\nu}\Theta)
\,,
\label{trans_scri2}
\\
\nabla^{(n)}_{\tau}  L_{\alpha\sigma}|_{\scri^-} =& \nabla^{(n-1)}_{\tau}\nabla_{\alpha}L_{\tau\sigma}  -\nabla^{(n-1)}_{\tau}( \nabla_{\rho}\Theta \, W_{ \tau \alpha \sigma}{}^{\rho})
\,,
\label{trans_scri3}
\\
g^{\tau r} \nabla^{(n)}_{\tau}W_{\mu\nu\sigma r}|_{\scri^-} =&- g^{\tau r}\nabla^{(n-1)}_{\tau} \nabla_{r} W_{\mu\nu\sigma \tau } -  g^{rr}\nabla^{(n-1)}_{\tau}\nabla_{r} W_{\mu\nu\sigma r} 
-  g^{\mathring A\mathring B}\nabla^{(n-1)}_{\tau}\nabla_{\mathring A} W_{\mu\nu\sigma \mathring B} 
\,,
\label{trans_scri4}
\\
2g^{\tau r}\nabla^{(n)}_{\tau}\nabla_{r} L_{\tau\tau} |_{\scri^-} =&\frac{1}{6}\nabla^{(n+1)}_{\tau}R-g^{rr}\nabla^{(n)}_{\tau}\nabla_{r}L_{\tau r} 
-g^{\mathring A\mathring B}\nabla^{(n)}_{\tau}\nabla_{\mathring A}L_{\tau \mathring B} 
 + \nabla^{(n)}_{\tau}( \nabla_{\rho}\Theta \, W_{\tau r\tau}{}^{\rho})
\,,
\label{trans_scri5}
\\
\nabla^{(n)}_{\tau}R^{(H)}_{\mu\nu}[g]|_{\scri^-}  =& 2\nabla^{(n)}_{\tau} L_{\mu\nu} + \frac{1}{6}\nabla^{(n)}_{\tau} Rg_{\mu\nu}
\,,
\label{trans_scri6}
%\\
%\nabla^{(n)}_0R_{\mu\nu\sigma}{}^{\kappa}[ g] |_{\scri^-} &=& \nabla^{(n)}_0(\Theta W_{\mu\nu\sigma}{}^{\kappa}) + 2\left(g_{\sigma[\mu}\nabla^{(n)}_0 L_{\nu]}{}^{\kappa}  - \delta_{[\mu}{}^{\kappa}\nabla^{(n)}_0 L_{\nu]\sigma} \right)
%\,,
\end{align}
where $R^{(H)}_{\mu\nu}$ denotes the wave-map gauge reduced Ricci tensor \cite{CCM2}.
Indeed, we observe that \eq{trans_scri1}-\eq{trans_scri4} provide algebraic equations for   $(\partial^{n+1}_{\tau}\Theta , \partial^{n}_{\tau} s, \partial^{n}_{\tau} L_{\alpha\sigma }, \partial^{n}_{\tau} W_{\mu\nu\sigma r})|_{\scri^-} $ in terms of  the known fields $(\partial^{k+1}_{\tau}\Theta , \partial^{k}_{\tau} s, \partial^{k+1}_{\tau} g_{\mu\nu} ,\partial^{k}_{\tau} L_{\mu\nu}, \partial^{k}_{\tau} W_{\mu\nu\sigma\rho})|_{\scri^-} $, $k\leq n-1$, while \eq{trans_scri5} provides an ODE for $\partial_{\tau}^nL_{\tau\tau}|_{\scri^-} $ 
%(supposing that the gauge source function $R$ has been  given) 
with initial data determined by $L_{\tau\tau}$ on $\mathcal{N}$, \eq{L00_eqn_N}.
Then \eq{trans_scri6} provides  ODEs for $\partial^{n+1}_{\tau}g_{\mu\nu}|_{\scri^-} $.
%(supposing that gauge source functions $V^{\mu}$ have been given).
 The initial data
at $S$ are determined by $g_{\mu\nu}|_{\mathcal{N}}$.

Note that, due to the divergence of $g_{\tau r}|_{\scri^-} $,  these ODEs are of Fuchsian type at $I^-$, and note further that  the solutions to e.g.\  \eq{trans_scri4} might be unbounded at $I^-$. For  our current analysis, though, this does not cause any problems.

Finally, the second Bianchi identity and the algebraic symmetries of the rescaled Weyl tensor yield (cf.\ \cite{ttp1})
\begin{align}
0 =&\nabla^{(n)}_{\tau}  (\nabla_{\rho}W_{0(\mathring A\mathring B)}{}^{\rho})_{\mathrm{tf}}|_{\scri^-}
\\
=&
\frac{1}{2}g^{r r}\nabla^{(n)}_{\tau}  \nabla_{\tau} (W_{r\mathring Ar\mathring B})_{\mathrm{tf}}- g^{\tau r}\nabla^{(n)}_{\tau} \nabla_r(W_{\tau \mathring A\tau \mathring B})_{\mathrm{tf}}
 + \frac{1}{2}g_{\tau r}(g^{rr})^2\nabla^{(n)}_{\tau} \nabla_{r} (W_{r\mathring Ar\mathring B})_{\mathrm{tf}}
\nonumber
\\
&
+g^{\tau r}(\nabla^{(n)}_{\tau} \nabla_{(\mathring A} W_{\mathring B)\tau r\tau})_{\mathrm{tf}}
+g^{rr}\nabla^{(n)}_{\tau}( \nabla_{(\mathring A} W_{\mathring B)rr\tau})_{\mathrm{tf}}
\\
=&
- g^{\tau r}\nabla^{(n)}_{\tau} \nabla_r(W_{\tau \mathring A\tau \mathring B})_{\mathrm{tf}}
+ \frac{1}{4}g_{\tau r} (g^{rr})^2 \nabla^{(n)}_{\tau}\nabla_r (W_{rA\mathring r\mathring B})_{\mathrm{tf}} 
\nonumber
\\
&
+\frac{1}{2} g^{rr} \nabla^{(n)}_{\tau}(\nabla_{(\mathring A} W_{\mathring B)rr\tau })_{\mathrm{tf}}
+g^{\tau r}(\nabla^{(n)}_{\tau} \nabla_{(\mathring A} W_{\mathring B)\tau r\tau})_{\mathrm{tf}}
\,,
\label{trans_scri7}
\end{align}
where we used that
\begin{align*}
0 =&  g_{\tau r}\nabla^{(n)}_{\tau} (\nabla_{\rho}W_{r(\mathring A\mathring B)}{}^{\rho})_{\mathrm{tf}}|_{\scri^-}
\\
=&-  \nabla^{(n)}_{\tau}\nabla_{\tau }(W_{r \mathring A r\mathring B})_{\mathrm{tf}}
- \frac{1}{2} g_{\tau r}g^{rr}\nabla^{(n)}_{\tau}\nabla_r(W_{r \mathring A r\mathring B})_{\mathrm{tf}}
\nonumber
- \nabla^{(n)}_{\tau}(\nabla_{(\mathring A} W_{\mathring B)rr\tau })_{\mathrm{tf}}
\,.
\end{align*}
Equation \eq{trans_scri7} is of the form (recall that $\theta^{+\scri^-}=\sigma^{\scri^-}_{\mathring A\mathring B}=0$, $g^{\tau r}|_{\scri^-}=r$, $g^{r r}|_{\scri^-}=r^2$  and $\kappa=-2/r$, for $r< 1/3$),
\begin{equation}
 \Big( \partial_{r}-\frac{n+2}{r}\Big) \partial^{n}_{\tau}(W_{\tau\mathring  A\tau \mathring B})_{\mathrm{tf}} |_{\scri^-} 
=
\text{known smooth function}
\,.
\label{ODE_for_data}
\end{equation}
This equation is also valid for $n=0$.
In the usual approach  the initial data for these ODEs follow from \eq{eqn_N_W0A0B}
%
%\begin{equation}
%W_{\tau\mathring A\tau\mathring B}|_{\mathcal{N}}=-\frac{1}{\Theta}g_{\mathring  B\mathring  C}\partial_{\tau}\sigma^{\mathcal{N}}_{\mathring A}{}^{\mathring C}
%\label{data_N_W0A0B}
%\end{equation}
%
and the continuity requirement at $S$ (the right-hand side of \eq{eqn_N_W0A0B} divided by $\Theta$
is regular at $S$).
% (cf.\ \cite{ttp3})
%$$
%\sigma^{\mathcal{N}}_{\mathring A}{}^{\mathring B} = \frac{1}{2}(  h^{(1)}_{\mathring A}{}^{\mathring B})_{\mathrm{tf}} + O(1+\tau)^2
%\,.
%$$
What actually matters from the data given on $\mathcal{N}$  is  thus  the expansion of $W_{\tau\mathring A\tau\mathring B}|_{\mathcal{N}} $ at $S$, and this is determined by the functions $ (h^{(k)}_{\mathring A\mathring B})_{\mathrm{tf}}$, $k\geq 3$.

Here we want to  prescribe   data at $I^-=\{r=0\}$.
The data which can be specified for  $\partial^{n}_{\tau}(W_{\tau\mathring A\tau\mathring B})_{\mathrm{tf}} |_{\scri^-}$, $n\geq 0$, 
correspond to its  $(n+2)$nd-order expansion coefficient $c^{(n+2,n)}_{\mathring A\mathring B}$ at $I^-$.
We then compute all the $( h^{(k)}_{\mathring A\mathring B})_{\mathrm{tf}}$'s, $k\geq 3$,
at $S$ by solving the hierarchical system above, and, using Borel summation (cf.\ e.g.\ \cite{cj}),
extend them to data $\gamma_{\mathring  A\mathring B}$ on $\mathcal{N}$. Our analysis at $I^-$ in this work
 does not depend on this  extension.
Note  that $ (h^{(1)}_{\mathring A\mathring B})_{\mathrm{tf}}$ is determined by \eq{expression_d1A1B}  while $(h^{(2)}_{\mathring A\mathring B})_{\mathrm{tf}}$ follows from the no-logs condition.
A solution to \eq{trans_scri7} will generally be polyhomogeneous at $I^-$. If this already happened for some $k< n$ the right-hand side might be polyhomogenous at $I^-$ as well. For our current discussion non-smoothness at $I^-$ is irrelevant.

\begin{proposition}
The data $(W_{\mathring A\mathring B}(r), M,N, L_{\mathring A}, c^{(n+2,n)}_{\mathring A\mathring B})$  for the asymptotic characteristic initial value problem determine a unique (up to gauge) solution of the MCFE supposing that an  extension of the data $\gamma_{\mathring A\mathring B}(r)$ on $\mathcal{N}$ has been given, whose Taylor expansion
at $S$ is determined by $(W_{\mathring A\mathring B}(r), M,N, L_{\mathring A}, c^{(n+2,n)}_{\mathring A\mathring B})$.
All solutions with bounded rescaled Weyl tensor $W_{ijkl}|_{\scri^-}$ at $I^-$ can be generated by such data.
\end{proposition}


\begin{thebibliography}{[10]}
\bibitem{acena}  A.E. Ace\~na, J.A. Valiente Kroon: \emph{Conformal extensions for stationary spacetimes}, Class. Quantum Grav. \textbf{28}
(2011) 225023.
\bibitem{andersson} L. Andersson,  P.T. Chru\'sciel: \textit{On asymptotic behavior of solutions of the constraint equations in general relativity with ``hyperboloidal boundary conditions''},  Dissertationes Math. \textbf{355}  (1996) 1--100.
\bibitem{acf} L. Andersson,  P.T. Chru\'sciel, H. Friedrich: \textit{On the regularity of solutions to the Yamabe equation and the existence of smooth hyperboloidal initial data for Einstein's field equations},  Comm. Math. Phys. \textbf{149} (1992) 587--612.
%\bibitem{ashtekar2} A. Ashtekar:  \textit{Asymptotic structure of the gravitational field at spatial infinity}, in: \textit{General relativity and gravitation -- One hundred years after the birth of Albert Einstein, Volume 2},  A. Held (ed.),  New York: Plenum Press, 1980, 37--70.
%\bibitem{ashtekar} A. Ashtekar, A. Magnon-Ashtekar: \textit{Energy-momentum in general relativity}, Phys. Rev. Lett. \textbf{43} (1979) 181--184.
\bibitem{ah}   A. Ashtekar, R.O. Hansen: \emph{A unified treatment of null and spatial infinity in general relativity. I. Universal structure, asymptotic symmetries, and conserved quantities at spatial infinity},J. Math. Phys. \textbf{19} (1978) 1542--1566.
\bibitem{ashtekar} A. Ashtekar,  A. Sen: \emph{NUT 4-momenta are forever}, Journal of Mathematical Physics \textbf{23} (1982) 2168--2178.
\bibitem{beig} R. Beig: \emph{A remarkable property of spherical harmonics},   J. Math. Phys.  \textbf{26} (1984) 769--770.
\bibitem{CCTW} A. Cabet,  P.T. Chru\'sciel, R. Tagne Wafo: \emph{On the characteristic initial value problem for nonlinear symmetric hyperbolic systems, including Einstein equations}, Dissertationes Mathematicae \textbf{515} (2016) 1--72.
\bibitem{CCM2} Y. Choquet-Bruhat, P.T. Chru\'sciel,  J.M. Mart\'in-Garc\'ia: \emph{The Cauchy problem on a characteristic cone for the Einstein equations in arbitrary dimensions},  Ann. Henri Poincar\'e \textbf{12} (2011) 419--482.
\bibitem{c_lecture}  P.T. Chru\'sciel: \textit{An introduction to the Cauchy problem for the Einstein equations},  lecture notes, Roscoff (2010),
 \url{http://homepage.univie.ac.at/piotr.chrusciel/teaching/Cauchy/Roscoff.pdf}.
\bibitem{C1}  P.T. Chru\'sciel: \textit{The existence theorem for the general relativistic Cauchy problem on the light-cone},
Forum of Mathematics, Sigma \textbf{2} (2014) e10 (50pp).
\bibitem{cd}  P.T. Chru\'sciel, E. Delay: \emph{On mapping properties of the general relativistic constraints operator in weighted function spaces, with applications},  M\' emoires de la Soci\' et\' e Math\' ematique de France \textbf{94} (2003) 1--103.
\bibitem{cj} P.T. Chru\'{s}ciel,  J. Jezierski: \emph{On free general relativistic initial data on the light cone}, J. Geom. Phys. \textbf{62} (2012) 57--593.
\bibitem{ChPaetz} P.T. Chru\'{s}ciel, T.-T. Paetz: \emph{The many ways of the characteristic Cauchy problem}, Class.\ Quantum Grav. \textbf{29} (2012) 145006.
\bibitem{ChPaetzInfCone} P.T. Chru\'{s}ciel, T.-T. Paetz: \emph{Solutions of the vacuum Einstein equations with initial data on past null infinity}, Class. Quantum Grav. \textbf{30} (2013) 235037.
\bibitem{ChPaetz2} P.T. Chru\'{s}ciel, T.-T. Paetz: \emph{Characteristic initial data and smoothness of Scri. I. Framework and results}, 
Ann. Henri Poincar\'e \textbf{16} (2015) 2131--2162.
\bibitem{dain} S. Dain: \textit{Initial data for stationary spacetimes near space-like infinity}, Class. Quantum Grav. \textbf{18} (2001) 4329--4338.
\bibitem{d_s}  T. Damour, B. Schmidt: \emph{Reliability of perturbation theory in general relativity}, J. Math. Phys.  \textbf{31} (1990) 2441--2453.
\bibitem{david} M. David: \emph{A Study of the Inhomogeneous  Hypergeometric Differential Equation}, Research Thesis (2017), \url{https://kb.osu.edu/dspace/bitstream/handle/1811/80568/1/Thesis_Final3.pdf}.
\bibitem{F1} H. Friedrich: \textit{On the regular and the asymptotic characteristic initial value problem for Einstein's vacuum field equations},  Proc. R. Soc. Lond. A \textbf{375} (1981) 169--184.
 \bibitem{F2} H. Friedrich: \textit{The asymptotic characteristic initial value problem for Einstein's vacuum field equations as an initial value problem for a first-order quasilinear symmetric hyperbolic system}, Proc. R. Soc. Lond. A \textbf{378} (1981) 401--421.
\bibitem{F_hyp}  H. Friedrich: \textit{Cauchy problems for the conformal vacuum field equations in general relativity}, Comm. Math. Phys. \textbf{91} (1983) 445--472.
 \bibitem{F4} H. Friedrich: \textit{On the hyperbolicity of Einstein's and other gauge field equations}, Comm. Math. Phys. \textbf{100} (1985) 525--543.
% \bibitem{F7} H. Friedrich: \textit{On purely radiative space-times}, Comm. Math. Phys. \textbf{103} (1986) 35--65.
\bibitem{F_hyp2} H. Friedrich: \textit{On the existence of $n$-geodesically complete or future complete solutions
of Einstein's field equations with smooth asymptotic structure}, Comm. Math. Phys. \textbf{107} (1986) 587--609.
 \bibitem{F5} H. Friedrich: \textit{Hyperbolic reductions for Einstein's equations}, Class. Quantum Grav. \textbf{13} (1996) 1451--1469.
\bibitem{F_AdS} H. Friedrich: \emph{Einstein equations and conformal structure: Existence of Anti-de Sitter-type space-times},
J. Geom. Phys. \textbf{17} (1995) 125--184.
\bibitem{F_i0} H. Friedrich: \textit{Gravitational fields near space-like and null infinity}, 
J. Geom. Phys. \textbf{24} (1998) 83--163.
 \bibitem{F3} H. Friedrich: \textit{Conformal Einstein evolution}, in: \textit{The conformal structure of space-time -- Geometry, analysis, numerics}, J. Frauendiener, H. Friedrich (eds.),  Berlin, Heidelberg: Springer, 2002,  1--50.
\bibitem{F_spin} H. Friedrich: \textit{Spin-2 fields on Minkowski space near spacelike and null infinity}, Class. Quantum Grav. \textbf{20} (2003) 101--117.
\bibitem{F_cg} H. Friedrich: \emph{Conformal geodesics on vacuum space-times},  Comm. Math. Phys. \textbf{235} (2003) 513--543.
\bibitem{F_i0_2} H. Friedrich: \textit{Smoothness at null infinity and the structure of initial data},
 in: \textit{The Einstein equations and the large scale behavior of gravitational fields}, P.T. Chru\'sciel, H. Friedrich (eds.),  Basel: Birkh\"auser, 2004,  121--203.
%\bibitem{F_static} H. Friedrich: \textit{Conformal structure of static vacuum data},  Comm. Math. Phys. \textbf{321} (2013) 419--482.
 \bibitem{F_T} H. Friedrich: \textit{The Taylor expansion at past time-like infinity}, Comm. Math. Phys. \textbf{324} (2013) 263--300.
\bibitem{F_17} H. Friedrich: \emph{Peeling or not peeling -- is that the question?}, Class. Quantum Grav. \textbf{35}  (2018) 083001.
\bibitem{F_Schmidt} H. Friedrich, B. Schmidt: \emph{Conformal geodesics in general relativity}, Proc. Roy. Soc. A \textbf{414} (1987) 171--195.
\bibitem{geroch} R. Geroch: \textit{Asymptotic structure of space-time},  in:  \textit{Asymptotic structure of space-time}, F. P. Esposito, L. Witten (eds.), New York: Plenum Press, 1977,  1--105.
\bibitem{gp} J.B. Griffiths, J. Podolsk\'y: \emph{Exact Space-Times in Einstein's General Relativity}, Cambridge: Cambridge University Press, 2009.
\bibitem{hypergeometric2} L.M. Hall: \emph{Special functions}, lecture notes, \url{http://web.mst.edu/~lmhall/SPFNS/sfch4.pdf}.
 \bibitem{hawking} S.W. Hawking, G.F.R. Ellis: \textit{The large scale structure of space-time}, Cambridge: Cambridge University Press, 1973.
\bibitem{hiva} P. Hintz, A. Vasy: \textit{A global analysis proof of the stability of Minkowski space and the polyhomogeneity of the metric}, (2017),  arXiv:1711.00195 [math.AP].
\bibitem{kannar} J. K\'ann\'ar: \textit{On the existence of $C^{\infty}$ solutions to the asymptotic characteristic initial value
 problem in general relativity}, Proc. R. Soc. Lond. A \textbf{452} (1996) 945--952.
\bibitem{kn} S. Klainerman, F. Nicol\`o: \emph{Peeling properties of asymptotically flat solutions to the Einstein vacuum equations},  Class. Quantum Grav. \textbf{20} (2003) 3215--3257.
\bibitem{kn2} S. Klainerman, F. Nicol\`o: \emph{The evolution problem in general relativity},  Boston: Birk\"auser, 2003. 
\bibitem{hypergeometric} R. Koekoek: \emph{Hypergeometric functions}, lecture notes, \url{http://homepage.tudelft.nl/11r49/documents/wi4006/hyper.pdf}.
\bibitem{luebbe_kroon} C. L\"ubbe. J. A.Valiente Kroon: \emph{On de Sitter-like and Minkowski-like spacetimes},
Class.\ Quantum Grav. \textbf{26} (2009) 145012.
\bibitem{ttp1} T.-T. Paetz: \emph{Conformally covariant systems of wave equations and their equivalence to Einstein's field equations}, 
Ann. Henri Poincar\'e \textbf{16} (2015) 2059--2129.
\bibitem{ttp3} T.-T. Paetz: \emph{Characteristic initial data and smoothness of Scri. II. Asymptotic expansions and construction of conformally smooth data sets},  J. Math. Phys. \textbf{55} (2014) 102503.
%\bibitem{ttpKIDs} T.-T. Paetz: \emph{KIDs prefer special cones }, Class. Quantum Grav. \textbf{31} (2014) 085007.
\bibitem{paetz_thesis} T.-T. Paetz: \emph{On characteristic Cauchy problems in general relativity}, PhD thesis (2014), \url{http://homepage.univie.ac.at/piotr.chrusciel/papers/Tim.pdf}.
%\bibitem{ttp2} T.-T. Paetz: \emph{Killing Initial Data on space-like conformal boundaries},  J. Geom. Phys. \textbf{106} (2016) 51--69.
\bibitem{p1} R. Penrose: \textit{Asymptotic properties of fields and space-time}, Phys. Rev. Lett. \textbf{10} (1963) 66--68.
\bibitem{p2} R. Penrose: \textit{Zero rest-mass fields including gravitation: Asymptotic behavior},  Proc. R. Soc. Lond. A \textbf{284} (1965) 159--203.
%\bibitem{pr} R. Penrose, W. Rindler: \emph{Spinors and space-time: Volume 1 -- Two-spinor calculus and relativistic fields},  Cambridge: Cambridge University Press, 1984.
\bibitem{bochner} P. Petersen: \emph{Riemannian geometry},  Graduate texts in mathematics, vol.\ 171,  New York: Springer, 1998.
%p.178
\bibitem{dual} S. Ramaswamy, A. Sen: \emph{Dual mass in general relativity}, Journal of Mathematical Physics \textbf{22}  (1981) 2612--2619.
 \bibitem{Rendall}  A.D. Rendall: \textit{Reduction of the characteristic initial value problem to the Cauchy problem and its applications to the Einstein equations}, Proc. R. Soc. Lond. A \textbf{427} (1990) 221--239.
\bibitem{kroon} J.A. Valiente Kroon: \textit{Polyhomogeneous expansions close to null and spatial infinity}, in: \textit{The conformal structure of space-time -- Geometry, analysis, numerics}, J. Frauendiener, H. Friedrich (eds.),  Berlin, Heidelberg: Springer, 2002,  135--159.
\bibitem{kroon2} J.A. Valiente Kroon: \textit{A new class of obstructions to the smoothness of null infinity}, Comm. Math. Phys. \textbf{244} (2004) 133--156.
\bibitem{kroon0} J.A. Valiente Kroon: \textit{Regularity Conditions for Einstein's Equations at Spatial Infinity},
Ann. Henri Poincar\'e \textbf{10} (2009) 623--671.
\bibitem{kroon1} J.A. Valiente Kroon: \textit{A Rigidity Property of Asymptotically Simple Spacetimes Arising from Conformally Flat Data},
Comm. Math. Phys. \textbf{298} (2010) 673--706.
\bibitem{kroon_book} J.A. Valiente Kroon: \emph{Conformal Methods in General Relativity}, Cambridge: Cambridge University Press, 2016.
\bibitem{visser} M. Visser: \emph{The Kerr spacetime: A brief introduction}, (2007), arXiv:0706.0622 [gr-qc].
\bibitem{hodge} F.W. Warner: \emph{Foundations of differentiable manifolds and Lie groups}, Graduate texts in mathematics, vol.\ 94,  New York: Springer, 1983.
%arXiv:0909.1967[math.DG] Shonkwiler
\end{thebibliography}
\end{document}